\documentclass[%
aps,
twocolumn,
superscriptaddress,
amsmath,amssymb,
pra,
floatfix,
longbibliography,
10pt
]{revtex4-2}
\usepackage{url}
\usepackage{graphicx}
\usepackage{placeins}
\usepackage{dcolumn}
\usepackage{bm}
\usepackage{amsmath,amsfonts,amssymb}
\usepackage{mathtools}
\usepackage[table]{xcolor}

\usepackage{silence}
\usepackage{physics}

\usepackage{dsfont}
\usepackage{tikz}
\usepackage{pgfplots}
\usepackage{enumerate}
\usepackage{subfigure}
\usepackage{bbold}
\usepackage{makecell}
\usepackage{array}
\usepackage[table]{xcolor}
\usepackage{colortbl}
\usepackage{xcolor}
\usepackage{tabularray}
\usepackage{amssymb}
\usepackage{ytableau}

\usepackage{tikz-cd}
\usepackage{tikz-qtree}
\usetikzlibrary{shadows.blur, shadings, arrows.meta, tikzmark, positioning, matrix}

\usepackage[colorlinks, linkcolor=red, anchorcolor=green, citecolor=blue]{hyperref}

\usepackage{algorithmic}
\usepackage[ruled,linesnumbered]{algorithm2e}

\SetCommentSty{mycommfont}

\makeatletter
\def\pre@bibdata{MTv02Notes}
\newcommand{\thickhline}{%
    \noalign {\ifnum 0=`}\fi \hrule height 1pt
    \futurelet \reserved@a \@xhline
}
\newcolumntype{"}{@{\hskip\tabcolsep\vrule width 1pt\hskip\tabcolsep}}
\makeatother

\makeatletter

\newcommand{\AppendixTocStart}{%
  \let\l@part\AppendixTocSavedPart
  \let\l@section\AppendixTocSavedSection
  \let\l@subsection\AppendixTocSavedSubsection
  \let\l@subsubsection\AppendixTocSavedSubsubsection
}

\newcommand{\appendixtableofcontents}{%
  \begingroup

  \setcounter{tocdepth}{2}%

  \let\AppendixTocSavedPart\l@part
  \let\AppendixTocSavedSection\l@section
  \let\AppendixTocSavedSubsection\l@subsection
  \let\AppendixTocSavedSubsubsection\l@subsubsection

  \let\l@part\@gobbletwo
  \let\l@section\@gobbletwo
  \let\l@subsection\@gobbletwo
  \let\l@subsubsection\@gobbletwo

  \let\appendix\relax

  \begin{center}
    {\small\bfseries CONTENTS}
  \end{center}

  \vspace{0.5cm}

  \@starttoc{toc}%

  \endgroup
}

\makeatother

\usepackage{multirow}
\usepackage{tabu}
\usepackage{textcomp,booktabs}
\usepackage{colortbl}
\usepackage[T1]{fontenc}
\definecolor{mygray}{gray}{0.9}
\definecolor{mypink}{rgb}{0.99,0.91,0.95}
\definecolor{mycyan}{cmyk}{0.3,0,0,0}

\newcommand{\mB}{\mathcal{B}}

\newcommand{\g}{\mathfrak{g}}

\newcommand{\mO}{\mathcal{O}}

\newcommand{\mH}{\mathcal{H}}

\newcommand{\mS}{\mathcal{S}}
\newcommand{\mU}{\mathcal{U}}
\newcommand{\T}{\mathbf{T}}
\newcommand{\0}{\mathbb{0}}
\newcommand{\1}{\mathbb{1}}

\makeatletter
\newcommand*{\rom}[1]{\expandafter\@slowromancap\romannumeral #1@}
\makeatother

\makeatletter
\newcommand*{\Relbarfill@}{\arrowfill@\Relbar\Relbar\Relbar}
\newcommand*{\xeq}[2][]{\ext@arrow 0055\Relbarfill@{#1}{#2}}
\makeatother

\definecolor{mGreen}{RGB}{193, 225, 193}
\definecolor{mRed}{RGB}{255, 204, 203}
\definecolor{mBlue}{RGB}{176, 224, 230}  
\definecolor{mOrange}{RGB}{255, 218, 185} 
\definecolor{mPurple}{HTML}{E5D1FA}
\definecolor{gridColor}{RGB}{210, 210, 210}

\usepackage{amsthm}
\usepackage{tcolorbox}
\tcbuselibrary{theorems}
\newtheorem{thm}{Theorem}
\newtheorem{lem}{Lemma}
\newtheorem{cor}{Corollary}

\newtcbtheorem[number within=section, use counter*=thm]{mythm}{Theorem}%
{colback=red!5,colframe=red!35!red,fonttitle=\bfseries}{thm}

\newtcbtheorem[number within=section, use counter*=thm]{mylem}{Lemma}%
{colback=green!5,colframe=green!35!black,fonttitle=\bfseries}{lem}

\newtcbtheorem[number within=section, use counter*=thm]{mydef}{Definition}%
{colback=yellow!5,colframe=yellow!70!black,fonttitle=\bfseries}{def}

\newtcbtheorem[number within=section, use counter*=thm]{mycor}{Corollary}%
{colback=blue!5,colframe=blue!70!black,fonttitle=\bfseries}{cor}

\newtcbtheorem[number within=section, use counter*=thm]{myprop}{Proposition}%
{colback=pink!5,colframe=pink!70!black,fonttitle=\bfseries}{prop}

\pgfplotsset{compat=1.17}

\begin{document}

\title{Lie Algebraic Uncertainty Relations}


\author{Xiaoli Hu}
\email{xiaolihumath@163.com}
\affiliation{School of Artificial Intelligence, Jianghan University, Wuhan, Hubei 430056, China}

\author{Naihuan Jing}
\email{jing@ncsu.edu}
\affiliation{Department of Mathematics, North Carolina State University, Raleigh, NC 27695, USA}

\author{Jian Zhang}
\email{jzhang@ccnu.edu.cn}
\affiliation{School of Mathematics and Statistics, Central China Normal University, Wuhan, Hubei 430079, China}

\author{Yunlong Xiao}
\email{mathxiao123@gmail.com}
\affiliation{Institute of Advanced Intelligence and Computing (IAIC), Agency for Science, Technology and Research (A*STAR), 1 Fusionopolis Way, \#16-16 Connexis, Singapore 138632, Republic of Singapore}




\begin{abstract}
Uncertainty relations impose fundamental constraints on quantum observables, yet optimal bounds are needed to determine their incompatibility through the minimum attainable uncertainty.
Here we introduce a covariance-matrix formulation from which infinitely many uncertainty relations follow, encompassing Hermitian observables and non-Hermitian operators and recovering the variance sum as a special case.
Within this family, a basis-independent uncertainty for simple and semisimple Lie algebras admits an exact solution: 
its optimal state-independent bound is an explicit pairing of root system with the highest weight of the representation.
The minimum is attained by a highest-weight state, replacing a search over quantum states with a calculation in representation theory.
A complementary result establishes that uncertainty is constant within each weight space, yielding a conservation law.
The same commutation relations can therefore support distinct uncertainty limits, depending on their realization on the state space.
By rendering these limits analytically accessible, the framework turns representation theory into a tool for quantifying quantum incompatibility and may help establish the boundary between attainable and impossible tasks in quantum theory.
\end{abstract}


\maketitle


\noindent \textbf{Introduction}--Quantum mechanics differs from our everyday experience in a fundamental way: 
some physical observables cannot be measured with arbitrary precision at the same time~\cite{Heisenberg1927}. 
Position and momentum embody this uncertainty principle: 
narrowing one distribution necessarily broadens the other~\cite{Kennard1927,weyl1928gruppentheorie}.  
Over nearly a century~\cite{conlon2026100pmdeltatyearsquantum}, the principle has been recast in complementary forms that explore different aspects of uncertainty. 
Variance-based relations capture the spread of measurement outcomes around their mean values~\cite{PhysRev.34.163,PhysRevA.64.052103,PhysRevLett.113.260401,PhysRevA.100.032118}, entropic formulations quantify the uncertainty of measurement outcome distributions~\cite{Bialynicki1975,PhysRevLett.50.631,PhysRevD.35.3070,PhysRevLett.60.1103,Wehner_2010,PhysRevA.93.042125,RevModPhys.89.015002,PhysRevA.98.042305}, while majorization compares such distributions directly without mapping them to a single scalar~\cite{PhysRevA.84.052117,Pucha_a_2013,PhysRevLett.111.230401,Xiao_2020,Yuan2023}. 
Their scope has also expanded to scenarios assisted by quantum memory~\cite{PhysRevLett.103.020402,Berta2010,Xiao_2016}, temporal settings involving interventions~\cite{PhysRevResearch.3.023077,PhysRevLett.130.240201}, and related constraints in quantum thermodynamics~\cite{PhysRevLett.114.158101,PhysRevLett.123.090604,PhysRevLett.125.140602,PhysRevLett.125.260604,PhysRevLett.127.160601}. 
What began with a few canonical observables has grown into a general framework for delineating fundamental limits across quantum theory.

An uncertainty relation identifies a constraint; 
its optimal bound determines the exact limit attainable by quantum states.
This distinction is essential to its physical meaning: 
a non-optimal bound can underestimate the uncertainty that incompatible observables necessarily retain.
For position and momentum, this limit is half the reduced Planck constant, $\hbar/2$; 
for angular momentum, the minimum uncertainty depends on the spin~\cite{PhysRevLett.18.182,Dammeier_2015}.
For general operator sets, however, determining the optimal state-independent bound requires a demanding minimization over all quantum states.
Existing approaches characterize the allowed trade-offs through lower convex hulls of variances~\cite{PhysRevLett.119.170404}, outer approximations to sets of measurement probability distributions~\cite{xiao2019complementaryinformationprinciplequantum}, and semidefinite programming hierarchies based on anticommutativity graphs ~\cite{PhysRevA.107.062211,PhysRevLett.132.200202}.
These methods can constrain the attainable region with increasing accuracy, but need not yield its exact boundary in closed form.
The resulting challenge is to identify which features of the operators determine the uncertainty, and whether that limit can be calculated directly from their structure.

Here we resolve this question for operator sets with Lie algebraic structure~\cite{Humphreys1972}.
Lie algebras encode the commutation relations underlying quantum symmetries and dynamics, making them a physical setting in which to connect noncommutativity to its quantum uncertainty.
Our starting point is a covariance-based formulation that encompasses Hermitian observables and non-Hermitian operators and generates an infinite family of uncertainty relations, including the usual variance sum.
For bases of simple and semisimple Lie algebras, a basis-independent measure admits an optimal state-independent bound determined in closed form by the root system and the highest weight.
Representation theory thus supplies both the irreducible limit and the states that attain it, resolving the minimization over all quantum states.
Within each weight space, uncertainty is constant, turning an inequality into a conservation law.
These results reveal how the commutation structure and its realization on the state space jointly determine quantum uncertainty.
As Lie algebras arise naturally throughout quantum information, computation, and technologies, our results may offer an algebraic lens to reveal further intrinsic constraints on quantum information processing.


\noindent \textbf{Covariance Formulation}--The statistics obtained by measuring a set of observables carry information about both individual uncertainties and pairwise correlations.
Variances quantify the former but do not capture the latter.
Here we introduce a new formulation of quantum uncertainty that brings both contributions together through the covariance matrix,
treating observables and non-Hermitian operators on equal footing.
For a set of operators $\mathbf X:=\{X_i\}_i$ and a quantum state $\psi$, define $\mathrm{Cov}_\psi(X_i,X_j)\coloneqq\langle X_i^\dagger X_j\rangle_\psi-\langle X_i^\dagger\rangle_\psi\langle X_j\rangle_\psi$.
The corresponding covariance matrix, $\mathrm{Cov}_\psi(\mathbf X)\coloneqq(\mathrm{Cov}_\psi(X_i,X_j))_{i,j}$, is positive semidefinite.
It allows us to introduce a universal form of quantum uncertainty $\mU_\psi(\mathbf X)$,
\begin{align}\label{eq:MT_Cov_Uncertainty}
    \mU_\psi(\mathbf X)
    \coloneqq
    \Tr[M\cdot\mathrm{Cov}_\psi(\mathbf X)],
\end{align}
where $M$ is a positive semidefinite matrix that weights the covariance entries and ensures $\mU(\mathbf X)\geqslant0$.
The universality of this construction lies in the freedom to choose $M$, generating an infinite family of uncertainty formulations within a single framework. 
In particular, the identity matrix recovers the sum of generalized variances, reducing to the familiar sum of variances~\cite{PhysRevA.68.032103,PhysRevA.74.010301,PhysRevA.84.022107,Dammeier_2015,PhysRevA.100.032118}, while diagonal matrices with nonnegative entries lead to weighted uncertainty relations~\cite{Xiao2016}.


\noindent \textbf{Lie Algebras}--Commutation relations constrain quantum uncertainty, yet the attainable bounds also depend on how the operators act on the state space.
Lie theory separates these two ingredients: 
the algebra encodes the commutation structure, whereas its representation specifies its realization.
This distinction, central to quantum mechanics~\cite{Hall2013,Hall2015,woit2017quantum,Sakurai_Napolitano_2020}, quantum control~\cite{PhysRevA.63.032308,d2021introduction} and quantum machine learning~\cite{Fontana2024,Ragone2024}, underlies our formulation.
For operator sets that form a basis of a Lie algebra, choosing $M$ in Eq.~\eqref{eq:MT_Cov_Uncertainty} as the inverse Gram matrix yields a basis-independent measure of quantum uncertainty.
We derive its optimal bound from the algebra and its representation, recovering familiar angular momentum uncertainty relations for $\mathfrak{su}(2)$ and extending them to general Lie algebras.
Formally, a Lie algebra $\g$ is a real or complex vector space equipped with a bilinear, antisymmetric bracket $[\cdot,\cdot]$ satisfying the Jacobi identity, $[X,[Y,Z]]+[Y,[Z,X]]+[Z,[X,Y]]=0$.
For matrix Lie algebras, the bracket is the commutator, $[X,Y]\coloneqq XY-YX$, with closure requiring $[X,Y]\in\g$ for all $X,Y\in\g$.
A representation $\pi:\g\to\operatorname{End}(V)$ realizes the algebra on a state space $V$ while preserving the bracket, $\pi([X,Y])=[\pi(X),\pi(Y)]$.

Root systems reveal the internal structure of a complex semisimple Lie algebra.
A Cartan subalgebra $\mathfrak h$ is a maximal commuting subalgebra whose elements are simultaneously diagonalizable in the adjoint representation.
It determines the root space decomposition $\g=\mathfrak h\oplus(\bigoplus_{\alpha\in\Phi}\g_\alpha)$, where $[h,X_\alpha]=\alpha(h)X_\alpha$ for $h\in\mathfrak h$ and $X_\alpha\in\g_\alpha$.
Each root $\alpha\in\mathfrak h^*: \mathfrak h\to \mathbb C$ specifies how the corresponding generators transform under the Cartan subalgebra.
Choosing positive roots $\Phi^+$ selects simple roots $\Pi\coloneqq\{\alpha_i\}_{i=1}^{r}$, with $r\coloneqq\dim\mathfrak h$, such that every positive root is a nonnegative integer combination of the simple roots.
Their geometry follows from the Killing form $\kappa_{\mathrm{Kill}}(X,Y)\coloneqq\Tr[\operatorname{ad}_X\operatorname{ad}_Y]$, where $\operatorname{ad}_X(Y)\coloneqq[X,Y]$.
Its normalized version $\kappa$ induces an inner product $\langle\cdot,\cdot\rangle$ on the real span of the roots, with long roots of squared length two.
This geometry is recorded by the Cartan matrix $C\coloneqq(2\langle\alpha_i,\alpha_j\rangle/\langle\alpha_j,\alpha_j\rangle)_{ij}$ and the vector of squared simple-root lengths $\vec{\alpha}\coloneqq(\langle\alpha_1,\alpha_1\rangle,\ldots,\langle\alpha_r,\alpha_r\rangle)^{\T}$.

\begin{table*}[t]
\centering

\begin{tblr}{
    width = \textwidth,
    colspec = {c || l c || X[l]},
    cells = {m},
    column{2} = {c},
    row{1} = {font=\bfseries, c},
    rowsep = 6pt,
    colsep = 6pt,
    hline{1,Z} = {1}{-}{0.5pt},
    hline{1,Z} = {2}{-}{0.5pt},
    hline{2,6} = {-}{0.4pt},
}

Type
&
Lie Algebra
&
Dynkin Diagram
&
Optimal Bound
\\

$A_n$
&
$\mathfrak{sl}(n+1,\mathbb C)$
&
\hspace*{-1.5em} \raisebox{-0.4ex}{\includegraphics[scale=0.6]{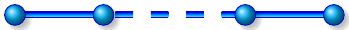}}
&
$\displaystyle\sum_{i=1}^{n} i(n+1-i)\lambda_i$
\\

$B_n$
&
$\mathfrak{so}(2n+1,\mathbb C)$
&
\hspace*{-1.5em}\raisebox{-1ex}{\includegraphics[scale=0.6]{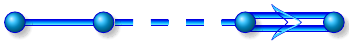}}
&
$\displaystyle\sum_{i=1}^{n-1} i(2n-i)\lambda_i
+\frac{n^2}{2}\lambda_n$
\\

$C_n$
&
$\mathfrak{sp}(2n,\mathbb C)$
&
\hspace*{-1.5em} \raisebox{-1ex}{\includegraphics[scale=0.6]{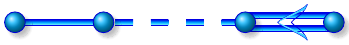}} 
&
$\displaystyle\sum_{i=1}^{n-1}
\frac{i(2n+1-i)}{2}\lambda_i
+\frac{n(n+1)}{2}\lambda_n$
\\

$D_n$
&
$\mathfrak{so}(2n,\mathbb C)$
&
\hspace*{-1.5em} \raisebox{-4.3ex}{\includegraphics[scale=0.6]{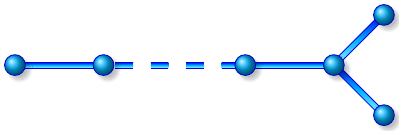}}  
&
$\displaystyle\sum_{i=1}^{n-2} i(2n-i-1)\lambda_i
+\frac{n(n-1)}{2}
\left(\lambda_{n-1}+\lambda_n\right)$
\\

$G_2$
&
$\mathfrak g_2$
&
\hspace*{-1.5em} \raisebox{-1ex}{\includegraphics[scale=0.75]{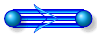}} 
&
$6\lambda_1+\frac{10}{3}\lambda_2$
\\

$F_4$
&
$\mathfrak f_4$
&
\hspace*{-1.5em} \raisebox{-1ex}{\includegraphics[scale=0.6]{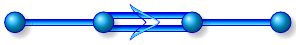}}
&
$16\lambda_1+30\lambda_2+21\lambda_3+11\lambda_4$
\\

$E_6$
&
$\mathfrak e_6$
&
\hspace*{-1.7em} \raisebox{-2.3ex}{\includegraphics[scale=0.6]{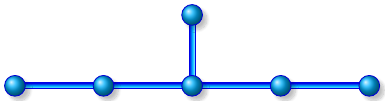}}
&
$16\lambda_1+22\lambda_2+30\lambda_3
+42\lambda_4+30\lambda_5+16\lambda_6$
\\

$E_7$
&
$\mathfrak e_7$
&
\hspace*{-1.9em}\raisebox{-2.3ex}{\includegraphics[scale=0.6]{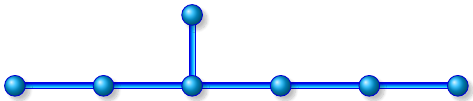}}
&
$34\lambda_1+49\lambda_2+66\lambda_3+96\lambda_4
+75\lambda_5+52\lambda_6+27\lambda_7$
\\

$E_8$
&
$\mathfrak e_8$
&
\hspace*{-2.1em}\raisebox{-2.3ex}{\includegraphics[scale=0.6]{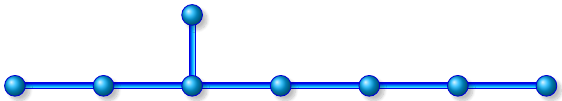}}
&
$92\lambda_1+136\lambda_2+182\lambda_3+270\lambda_4
+220\lambda_5+168\lambda_6+114\lambda_7+58\lambda_8$
\\

\end{tblr}

\caption{
\textbf{Lie Algebraic Uncertainty Relations.}
Classical and exceptional simple Lie algebras, their Dynkin diagrams, and the optimal state-independent bound of Thm.~\ref{thm:Simple_Bound}.
For $B_n$, $C_n$, $F_4$, and $G_2$, arrows point from long roots to short roots.
The rank satisfies $n\geqslant2$ for $B_n$ and $C_n$, and $n\geqslant4$ for $D_n$.
}

\label{tab:Optimal_Bound}

\end{table*}

Within a finite dimensional representation, weights describe the corresponding organization of the state space.
The commuting Cartan generators resolve $V$ into simultaneous eigenspaces, $V=\bigoplus_\mu V_\mu$, where $\pi(h)|\psi_\mu\rangle=\mu(h)|\psi_\mu\rangle$ for every $h\in\mathfrak h$ defines the weight $\mu$.
Root operators connect these spaces through $\pi(X_\alpha)V_\mu\subseteq V_{\mu+\alpha}$, linking the algebraic root decomposition to transitions between weight sectors.
To label the representations, we introduce the fundamental weights $\Omega\coloneqq\{\omega_i\}_{i=1}^{r}$, dual to the simple coroots $\alpha_j^\vee\coloneqq2\alpha_j/\langle\alpha_j,\alpha_j\rangle$ and satisfying $\langle\omega_i,\alpha_j^\vee\rangle=\delta_{ij}$.
Every irreducible representation is uniquely specified, up to equivalence, by a highest weight $\Lambda=\sum_i\lambda_i\omega_i$, with nonnegative integer Dynkin labels $\lambda_i$.
Its highest weight is annihilated by all positive-root operators, whereas repeated application of negative-root operators generates the full representation.
The Dynkin coefficient vector $\vec{\lambda}\coloneqq(\lambda_1,\ldots,\lambda_r)^{\T}$ thus labels the irreducible representation, while $\mu$ identifies a weight space within it.

Connecting this structure to quantum uncertainty requires a common geometric description of the operator bases.
Three related choices enter our analysis: 
a Chevalley basis $\mathbf X_{\mathrm{Ch}}$ of $\g$, a basis $\mathbf X_{\mathrm{aH}}$ of its compact real form $\g_{\mathbb{R}}$ represented by anti-Hermitian operators, and the observable basis $\mathbf X_{\mathrm{ob}}$ obtained by multiplying each anti-Hermitian generator by $\sqrt{-1}$.
For $\mathfrak{sl}(2,\mathbb C)$, the simplest example, its compact real form is $\mathfrak{su}(2)$, and the three corresponding bases are
\begin{align}
\mathbf X_{\mathrm{Ch}}
&=
\left\{
\begin{pmatrix}
    1&0\\0&-1
\end{pmatrix},
\begin{pmatrix}
    0&1\\0&0
\end{pmatrix}, 
\begin{pmatrix}
    0&0\\1&0
\end{pmatrix}
\right\},\notag\\
\mathbf X_{\mathrm{aH}}
&=
\left\{
\begin{pmatrix}
    -\mathrm{i}&0\\0&\mathrm{i}
\end{pmatrix},
\begin{pmatrix}
    0&-\mathrm{i}\\-\mathrm{i}&0
\end{pmatrix}, 
\begin{pmatrix}
    0&-1\\1&0
\end{pmatrix}
\right\},\notag\\
\mathbf X_{\mathrm{ob}}
&=
\left\{
\begin{pmatrix}
    1&0\\0&-1
\end{pmatrix},
\begin{pmatrix}
    0&1\\1&0
\end{pmatrix}, 
\begin{pmatrix}
    0&-\mathrm{i}\\ \mathrm{i}&0
\end{pmatrix}
\right\}.\notag
\end{align}
The observable basis thus consists of the Pauli operators, equivalently twice the spin-$1/2$ angular momentum operators $\{J_z, J_x, J_y\}$ in units of $\hbar=1$.
For general Lie algebras, the construction is presented in the Appendices.

The normalized Killing form $\kappa$ defines the quadratic Casimir element $C_2\coloneqq\sum_i X_iX^i$, where $X_i\in\mathbf X_{\mathrm{aH}}$ and $\{X^i\}_i$ is the dual basis satisfying $\kappa(X_i,X^j)=\delta_{ij}$.
To define the Gram matrix, we equip $\g$ with the positive-definite inner product $(\cdot,\cdot)_{\g}$ induced by $\kappa$, as constructed in the Appendices.
For any operator family $\mathbf X=\{X_i\}_i\subset\g$, its Gram matrix is $\mathrm G(\mathbf X)\coloneqq((X_i,X_j)_{\g})_{ij}$.
In particular, $\mathbf X$ may be chosen as the Chevalley, anti-Hermitian or observable basis introduced above.
Each choice yields an invertible Gram matrix whose inverse supplies the weighting matrix $M=\mathrm G(\mathbf X)^{-1}$ in Eq.~\eqref{eq:MT_Cov_Uncertainty}, giving the same quantum uncertainty.


\noindent \textbf{Uncertainty Relations}--To identify the intrinsic connection between quantum uncertainty and Lie algebraic structure, we seek a formulation independent of the operator basis.
Choosing $M=\mathrm G(\mathbf X)^{-1}$ in Eq.~\eqref{eq:MT_Cov_Uncertainty} for a basis $\mathbf X$ of $\g$ defines the Lie algebraic uncertainty
\begin{align}\label{eq:MT_Lie_Uncertainty}
    \mU_\psi(\mathbf X)
    \coloneqq
    \Tr[\mathrm G(\mathbf X)^{-1}\cdot\mathrm{Cov}_\psi(\mathbf X)].
\end{align}
Although the Gram and covariance matrices each depend on the basis, their contraction is invariant under any invertible basis transformation. 

\begin{lem}[\bf{Basis Invariance}]\label{lem:Basis_Invariance}
Let $\mathbf X$ and $\mathbf Y$ be bases of $\g$ related by $\mathbf Y=N\mathbf X$, with $N$ invertible.
With the state, representation and inner product fixed, the uncertainty defined in Eq.~\eqref{eq:MT_Lie_Uncertainty} satisfies $\mU_\psi(\mathbf X)=\mU_\psi(\mathbf Y)$.
\end{lem}

The resulting uncertainty is therefore a property of the state and representation, independent of the generators used to express it.
We henceforth denote it by $\mU_\psi(\g)$.
Restricting the state to a weight space reveals a further constraint.
Within an irreducible representation of highest weight $\Lambda$, every state $\psi_{\mu}$ in $V_{\Lambda,\mu}$ has the same uncertainty, determined entirely by the quadratic Casimir eigenvalue and the weight $\mu$.
Quantum uncertainty thus obeys a conservation law within each weight space.

\begin{thm}[\bf{Uncertainty Conservation}]\label{thm:Uncertainty_Conservation}
Within the weight space $V_{\Lambda,\mu}$ of an irreducible representation of $\g$ with highest weight $\Lambda$, the Lie algebraic uncertainty is constant: 
every state $\ket{\psi_\mu}\in V_{\Lambda,\mu}$ satisfies
\begin{align}\label{eq:MT_Uncertainty_Conservation}
    \mU_{\psi_\mu}(\g)
    =
    \langle C_2\rangle_{\psi_\mu}
    -
    \langle\mu,\mu\rangle,
\end{align}
where $\langle C_2\rangle_{\psi_\mu}$ is the expectation value of the quadratic Casimir element, and $\langle\mu,\mu\rangle$ is the squared norm of the weight $\mu$.
\end{thm}

Root systems classify complex semisimple Lie algebras and are encoded by Dynkin diagrams (see Tab.~\ref{tab:Optimal_Bound}), which we use as a shorthand in our uncertainty notation.
For $\mathfrak{su}(n+1)$, the diagram is of type $A_n$, a chain of $n$ nodes connected by single bonds, and we write
\begin{align}
    \mU_{\psi_\mu}(\raisebox{-2.8ex}{\includegraphics[scale=0.8]{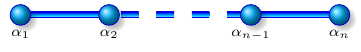}})
    \coloneqq
    \mU_{\psi_\mu}(\mathfrak{su}(n+1)).\notag
\end{align}
Angular momentum provides a direct illustration of Thm.~\ref{thm:Uncertainty_Conservation} in the simplest case, i.e., $A_1$.
In the spin-$j$ representation of $\mathfrak{su}(2)$, the weight states are $\ket{j,m}$, where $j$ is the total angular momentum quantum number and $m\in\{-j,\ldots,j\}$ is the magnetic quantum number, defined by $J_z\ket{j,m}=m\ket{j,m}$ in units of $\hbar=1$.
Under our normalization, the quadratic Casimir expectation is $2j(j+1)$ and the squared weight length is $2m^2$, giving 
\begin{align}
    \mU_{\ket{j,m}}( \raisebox{-0.6ex}{\includegraphics[scale=0.8]{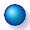}} )
    =
    &2\left(
    \Delta_{\ket{j,m}}^{2}(J_x)+
    \Delta_{\ket{j,m}}^{2}(J_y)+
    \Delta_{\ket{j,m}}^{2}(J_z)
    \right)\notag\\
    =
    &2j(j+1)-2m^2.\notag
\end{align}
Fixing the magnetic quantum number therefore fixes the angular momentum uncertainty exactly.
Among these weight states, the uncertainty is smallest at the extremal weights $m=\pm j$, where it equals $2j$ and saturates the familiar angular momentum uncertainty bound.
This example gives a physical interpretation of the two terms in Thm.~\ref{thm:Uncertainty_Conservation}: 
the representation fixes the Casimir contribution, while the weight determines how much is subtracted from it.

\begin{table}[t]
\centering

\begin{tblr}{
    width = \columnwidth,
    colspec = {c || X[l]},
    cells = {m},
    column{1} = {c},
    row{1} = {font=\bfseries, c},
    rowsep = 5pt,
    colsep = 5pt,
    hline{1,Z} = {1}{-}{0.5pt},
    hline{1,Z} = {2}{-}{0.5pt},
    hline{2} = {-}{0.4pt},
}

Type
&
Defining Representation
\\

$A_n$
&
$\mU_{\psi}^{\mathrm{def}}(\raisebox{-0.4ex}{\includegraphics[scale=0.57]{DK_A_n.pdf}})\geqslant n$
\\

$B_n$
&
$\mU_{\psi}^{\mathrm{def}}( \raisebox{-1ex}{\includegraphics[scale=0.57]{DK_B_n.pdf}} )\geqslant 2n-1$
\\

$C_n$
&
$\mU_{\psi}^{\mathrm{def}}( \raisebox{-1ex}{\includegraphics[scale=0.57]{DK_C_n.pdf}} )\geqslant n$
\\

$D_n$
&
$\mU_{\psi}^{\mathrm{def}}( \raisebox{-4ex}{\includegraphics[scale=0.57]{DK_D_n.pdf}} )\geqslant 2(n-1)$
\\

\end{tblr}

\caption{
\textbf{Defining Representation Bounds.}
Explicit forms of the optimal state-independent bound in Thm.~\ref{thm:Simple_Bound} for the defining representations of the classical Lie algebras $A_n$, $B_n$, $C_n$, and $D_n$.
}

\label{tab:Def_Rep}

\end{table}

A second application of Thm.~\ref{thm:Uncertainty_Conservation} concerns $\mathfrak{su}(3)$, whose defining representation encodes $\mathrm{SU}(3)$-symmetry on the qutrit space $\mathbb C^3$.
The algebra acts through traceless anti-Hermitian $3\times3$ matrices; 
their Hermitian counterparts span all traceless observables of a three-level quantum system and form the Gell-Mann matrices~\cite{Gell-Mann}:
\begin{align}
G_1 &=
\begin{pmatrix}
0&1&0\\
1&0&0\\
0&0&0
\end{pmatrix},
&
G_2 &=
\begin{pmatrix}
0&-\mathrm{i}&0\\
\mathrm{i}&0&0\\
0&0&0
\end{pmatrix},
\notag\\
G_3 &=
\begin{pmatrix}
1&0&0\\
0&-1&0\\
0&0&0
\end{pmatrix},
&
G_4 &=
\begin{pmatrix}
0&0&1\\
0&0&0\\
1&0&0
\end{pmatrix},
\notag\\
G_5 &=
\begin{pmatrix}
0&0&-\mathrm{i}\\
0&0&0\\
\mathrm{i}&0&0
\end{pmatrix},
&
G_6 &=
\begin{pmatrix}
0&0&0\\
0&0&1\\
0&1&0
\end{pmatrix},
\notag\\
G_7 &=
\begin{pmatrix}
0&0&0\\
0&0&-\mathrm{i}\\
0&\mathrm{i}&0
\end{pmatrix},
&
G_8 &=
\frac{1}{\sqrt{3}}
\begin{pmatrix}
1&0&0\\
0&1&0\\
0&0&-2
\end{pmatrix}.\notag
\end{align}
This representation has highest weight $\Lambda=\omega_1$ and quadratic Casimir expectation $\langle C_2\rangle=8/3$.
Its three weights, $\mu_1=\omega_1$, $\mu_2=-\omega_1+\omega_2$ and $\mu_3=-\omega_2$, follow by successively subtracting the simple roots from the highest weight.
Since $\langle\omega_1,\omega_1\rangle=\langle\omega_2,\omega_2\rangle=2/3$ and $\langle\omega_1,\omega_2\rangle=1/3$, all three satisfy $\langle\mu_i,\mu_i\rangle=2/3$.
Thm.~\ref{thm:Uncertainty_Conservation} then yields
\begin{align}
    \mU_{\psi_{\mu_i}}^{\mathrm{def}}(\raisebox{-0.8ex}{\includegraphics[scale=0.8]{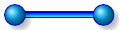}})
    =
    \sum_{j=1}^{8}\Delta_{\psi_{\mu_i}}^{2}(G_j)
    =
    2,
    \quad\forall i\in\{1, 2, 3\}.\notag
\end{align}
Thus, the three weight states have identical uncertainty despite their distinct eigenvalues under the commuting Cartan generators.

Removing the restriction to a fixed weight space, we obtain the optimal state-independent uncertainty bound for every simple Lie algebra and every finite dimensional representation, as established in the following theorem.

\begin{thm}[\bf{Simple Bound}]\label{thm:Simple_Bound}
A complex simple Lie algebra $\g$ of rank $r$, acting in a finite-dimensional irreducible representation with highest weight $\Lambda=\sum_{i=1}^{r}\lambda_i\omega_i$, satisfies
\begin{align}\label{eq:MT_Optimal_Uncertainty}
    \mU_\psi(\g)
    \geqslant
    \vec{\lambda}^{\T}\cdot\vec{\xi},
\end{align}
for every quantum state $\psi$ in the representation space.
Here $\vec{\lambda}$ is the Dynkin coefficient vector and $\vec{\xi}\coloneqq C^{-1}\vec{\alpha}$, where $C$ is the Cartan matrix and $\vec{\alpha}$ collects the squared lengths of the simple roots.
The bound is optimal and is attained by a highest-weight state.
\end{thm}

Bound in Eq.~\eqref{eq:MT_Optimal_Uncertainty} separates the algebraic structure encoded by the commutation relations, through $\vec{\xi}$, from its realization in a given representation, through $\vec{\lambda}$. 
Together, these determine the minimum quantum uncertainty over the full irreducible representation space, extending the state-independent bound of Ref.~\cite{deGuise2018} beyond states restricted to a fixed weight space. 
This establishes a representation-wide limit imposed by noncommutativity, with the optimal bounds for the different Lie algebra types listed in Tab.~\ref{tab:Optimal_Bound}.

\begin{table}[t]
\centering

\begin{tblr}{
    width = \columnwidth,
    colspec = {c || X[l]},
    cells = {m},
    column{1} = {c},
    row{1} = {font=\bfseries, c},
    rowsep = 5pt,
    colsep = 5pt,
    hline{1,Z} = {1}{-}{0.5pt},
    hline{1,Z} = {2}{-}{0.5pt},
    hline{2,6} = {-}{0.4pt},
}

Type
&
Adjoint Representation
\\

$A_n$
&
$\mU_{\psi}^{\mathrm{ad}}(\raisebox{-0.4ex}{\includegraphics[scale=0.57]{DK_A_n.pdf}})\geqslant 2n$
\\

$B_n$
&
$\mU_{\psi}^{\mathrm{ad}}( \raisebox{-1ex}{\includegraphics[scale=0.57]{DK_B_n.pdf}} )\geqslant 4(n-1)$
\\

$C_n$
&
$\mU_{\psi}^{\mathrm{ad}}( \raisebox{-1ex}{\includegraphics[scale=0.57]{DK_C_n.pdf}} )\geqslant 2n$
\\

$D_n$
&
$\mU_{\psi}^{\mathrm{ad}}( \raisebox{-4ex}{\includegraphics[scale=0.57]{DK_D_n.pdf}} )\geqslant 4n-6$
\\

$G_2$
&
$\mU_{\psi}^{\mathrm{ad}}( \raisebox{-1ex}{\includegraphics[scale=0.7]{DK_G_2.pdf}} )\geqslant 6$
\\

$F_4$
&
$\mU_{\psi}^{\mathrm{ad}}( \raisebox{-1ex}{\includegraphics[scale=0.57]{DK_F_4.pdf}} )\geqslant 16$
\\

$E_6$
&
$\mU_{\psi}^{\mathrm{ad}}( \raisebox{-0.4ex}{\includegraphics[scale=0.57]{DK_E_6.pdf}} )\geqslant 22$
\\

$E_7$
&
$\mU_{\psi}^{\mathrm{ad}}( \raisebox{-0.4ex}{\includegraphics[scale=0.57]{DK_E_7.pdf}} )\geqslant 34$
\\

$E_8$
&
$\mU_{\psi}^{\mathrm{ad}}( \raisebox{-0.4ex}{\includegraphics[scale=0.57]{DK_E_8.pdf}} )\geqslant 58$
\\

\end{tblr}

\caption{
\textbf{Adjoint Representation Bounds.}
Explicit optimal state-independent bounds of Thm.~\ref{thm:Simple_Bound} for the adjoint representations of all simple Lie algebras.
}

\label{tab:Adj_Rep}

\end{table}

Defining and adjoint representations connect the optimal bound to complementary settings in quantum physics. 
For the classical families, the defining representation gives the standard matrix action on the underlying state space, linking the abstract algebraic structure to finite dimensional quantum observables. 
The corresponding bounds, collected in Tab.~\ref{tab:Def_Rep}, make the two ingredients of Eq.~\eqref{eq:MT_Optimal_Uncertainty} explicit: 
the root system fixes $\vec{\xi}$, while the representation enters through the fundamental highest weight $\vec{\lambda}=(1,0,\ldots,0)^{\T}$. 
By contrast, the adjoint representation describes the action on the algebra itself, a perspective central to recent investigations of trainability and barren plateaus in parameterized quantum circuits~\cite{Fontana2024,Ragone2024}. 
Observables supported on the dynamical Lie algebra remain closed under Heisenberg evolution, enabling trainability and barren plateaus to be characterized through this action. Table~\ref{tab:Adj_Rep} presents the optimal uncertainty bounds for the adjoint representations of all simple Lie algebras.

For the classical families, the defining representation gives the standard matrix action on the underlying space, linking the abstract algebraic structure to finite-dimensional quantum observables. 
The corresponding bounds, collected in Tab.~\ref{tab:Def_Rep}, make the two ingredients of Eq.~\eqref{eq:MT_Optimal_Uncertainty} explicit: the root system fixes $\vec{\xi}$, while the representation enters through the fundamental highest weight $\vec{\lambda}=(1,0,\ldots,0)^{\T}$. 
The adjoint representation instead describes the algebra's action on its own generators. Its relevance to parameterized quantum circuits arises from the closure of observables in the dynamical Lie algebra under Heisenberg evolution, allowing algebraic approaches to trainability and barren plateaus. Specializing Thm.~\ref{thm:Simple_Bound} to the adjoint highest weight yields the corresponding optimal uncertainty bounds for all simple Lie algebras.

Semisimplicity introduces no additional optimization beyond the simple case. 
Indeed, a semisimple Lie algebra decomposes into mutually orthogonal simple ideals,
while every finite-dimensional irreducible representation factorizes accordingly into
a tensor product of irreducible representations of these ideals. 
The Lie algebraic uncertainty is therefore additive across the simple components, even
when the quantum state is entangled across the corresponding representation factors.
Consequently, the optimal bound for a semisimple Lie algebra follows directly by summing
the optimal bounds of its simple constituents. 
Combining this additivity with Thm.~\ref{thm:Simple_Bound} gives the following corollary.

\begin{cor}[\bf{Semisimple Bound}]
\label{cor:Semisimple_Bound}
A complex semisimple Lie algebra
$\g=\bigoplus_{k=1}^{s}\g_k$
with finite dimensional irreducible representation
$V(\Lambda)\simeq\bigotimes_{k=1}^{s}V_k(\Lambda_k)$
inherits its optimal uncertainty bound from $\g_k$.
Writing $\vec{\lambda}^{(k)}$ for the Dynkin coefficient vector of $\Lambda_k$
and $\vec{\xi}^{(k)}$ for the corresponding algebraic coefficient vector associated with $\g_k$,
every quantum state $\psi$ obeys
\begin{align}
\mU_{\psi}(\g)
\geqslant
\sum_{k=1}^{s}
\bigl(\vec{\lambda}^{(k)}\bigr)^{\T}\vec{\xi}^{(k)}.
\end{align}
Equality is attained by a tensor product of highest-weight states, one for each $\g_k$, establishing the optimality of the bound.
\end{cor}

Corollary~\ref{cor:Semisimple_Bound} reduces semisimple uncertainty to a factor-wise calculation:
each simple component contributes an optimal bound fixed by its representation, and the global optimum is their sum, with no further optimization.
For
$\mathfrak{su}(3)\oplus\mathfrak{so}(8)\oplus\mathfrak e_6$, of types $A_2$, $D_4$, and $E_6$, the adjoint representations contribute $4$, $10$, and $22$ (see Tab.~\ref{tab:Adj_Rep}), respectively, yielding
\begin{widetext}
\begin{align}
\mU_{\psi}^{\mathrm{ad}}
(
\mathfrak{su}(3)\oplus\mathfrak{so}(8)\oplus\mathfrak e_6
)
\geqslant 
\min_{\psi_1}
\mU_{\psi_1}^{\mathrm{ad}}(\raisebox{-0.4ex}{\includegraphics[scale=0.57]{DK_A_2.pdf}})
+
\min_{\psi_2}
\mU_{\psi_2}^{\mathrm{ad}}( \raisebox{-3.6ex}{\includegraphics[scale=0.57]{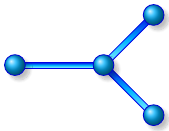}}  )
+
\min_{\psi_3}
\mU_{\psi_3}^{\mathrm{ad}}( \raisebox{-0.4ex}{\includegraphics[scale=0.57]{DK_E_6.pdf}} )
=36.
\end{align}
Replacing the adjoint representations of $A_2$, $D_4$ and $E_6$ by their defining representations gives optimal bounds of $2$, $6$ and $16$ (see Tab.~\ref{tab:Def_Rep}), respectively. 
These imply the optimal uncertainty relation for their semisimple direct sum:
\begin{align}
\mU_{\psi}^{\mathrm{def}}
(
\mathfrak{su}(3)\oplus\mathfrak{so}(8)\oplus\mathfrak e_6
)
\geqslant 
\min_{\psi_1}
\mU_{\psi_1}^{\mathrm{def}}(\raisebox{-0.4ex}{\includegraphics[scale=0.57]{DK_A_2.pdf}})
+
\min_{\psi_2}
\mU_{\psi_2}^{\mathrm{def}}( \raisebox{-3.6ex}{\includegraphics[scale=0.57]{DK_D_4.pdf}}  )
+
\min_{\psi_3}
\mU_{\psi_3}^{\mathrm{def}}( \raisebox{-0.4ex}{\includegraphics[scale=0.57]{DK_E_6.pdf}} )
=24.
\end{align}
\end{widetext}

The contrast between these optimal bounds identifies representation as an essential ingredient of quantum uncertainty: 
identical commutation relations can impose different irreducible bounds on the incompatible observables. 
Minimum uncertainty therefore reflects both the algebraic structure and its realization on the system.


\noindent \textbf{Discussions}--Finding an uncertainty relation is often far easier than determining its sharp quantum limit.
Yet only an optimal bound quantifies the incompatibility associated with a given set of observables.
Since the earliest uncertainty relations, closed-form optimal state-independent bounds have been available only in a comparatively small number of highly structured settings, such as canonical position-momentum pairs and angular momentum.
Beyond these cases, identifying the minimum uncertainty generally requires a nontrivial optimization, whose complexity rapidly increases with the dimension and number of operators.
Here, this optimization is replaced entirely by algebra.
Whenever the operators form a basis of a Lie algebra, the optimal state-independent uncertainty follows directly from the root structure of the algebra and the highest weight of its representation.
The covariance formulation makes this result independent of the chosen operator basis and applies equally to Hermitian and non-Hermitian generators, while restriction to a fixed weight space sharpens the uncertainty inequality into a conservation law.
Across the full irreducible representation, specifying the Lie algebra generated by the operators and its representation fixes the optimal state-independent uncertainty exactly, without optimization, establishing Lie algebra and representation theory as a direct route to the fundamental limits imposed by quantum incompatibility.

Moving to continuous-variable systems poses a fundamental question: 
does the same algebraic picture persist beyond finite-dimensional quantum systems?
This question is especially relevant to quantum optics, where continuous variables provide the natural setting for bosonic degrees of freedom.
Lie algebraic methods already play a central role elsewhere in quantum technologies: 
dynamical Lie algebras characterize controllability and accessible evolutions in quantum control, while in quantum machine learning they organize expressibility and the scaling of gradient underlying trainability~\cite{Fontana2024,Ragone2024}.
The recurrence of the same mathematical structure in these settings and in the uncertainty limits derived here raises the possibility of a deeper physical connection.
In particular, operator incompatibility may constrain optimization landscapes, linking fundamental quantum uncertainty to the trainability of parameterized quantum circuits.
Uncovering such connections could reveal a common origin of intrinsic limitations across quantum foundations, quantum control and quantum machine learning, pointing towards a more unified algebraic understanding of quantum systems.


\section*{Acknowledgments}
Y. X. thanks Bing Yu and Xiaoyong Yang for valuable discussions, and gratefully acknowledges Yuxiang Yang for his hospitality and insightful feedback during Y. X.'s visit to the University of Hong Kong in September 2026.
This research is supported by A*STAR under its Career Development Fund (C243512002).
X. H. acknowledges support from the China Scholarship Council and the Research Fund of Jianghan University (Grant No.~2023JCYJ08). 
N. J. is funded by the Simons Foundation (Grant No.~MP-TSM-00002518).
J. Z. acknowledges financial support from the Natural Science Foundation of China (Grant No.~12571026).


\section*{AI USAGE}
ChatGPT was used to refine the writing, assist with the preparation of Dynkin diagrams, and support calculations for the entanglement detection analysis in the Appendices; 
all intellectual contributions -- the conceptual framework, main ideas, and proofs -- are entirely the authors' own, developed without artificial intelligence, or whatever future generations may call superintelligence.


\bibliography{Bib}


\clearpage
\onecolumngrid
\appendix

\section*{Appendices}
These appendices develop the mathematical framework underlying the results in the main text, from the structure of Lie algebras to explicit uncertainty bounds and their applications. 
Appendix~\ref{sec:Lie_Theory} reviews the necessary Lie theory, including root decompositions, the Killing form, observable bases and highest weight representations.
Appendix~\ref{sec:Lie_Algebraic_Quantum_Uncertainty} establishes a basis independent formulation of quantum uncertainty relation, derives its conservation within weight spaces and its optimal lower bound for arbitrary states, and expresses the bound in terms of Dynkin coefficients and the Cartan matrix. 
Appendix~\ref{sec:Simple_Lie_Algebraic_Quantum_Uncertainty} gives explicit bounds for all classical and exceptional simple Lie algebras. 
Appendix~\ref{sec:entanglement} extends the framework to semisimple algebras and develops an entanglement criterion from local uncertainty bounds, demonstrating the detection of entanglement in a state with positive partial transpose.
Figure~\ref{fig:Outline} outlines the structure of the appendices.

\addtocontents{toc}{\protect\AppendixTocStart}
\appendixtableofcontents

\begin{figure}[t]
    \centering
    \includegraphics[width=1\columnwidth]{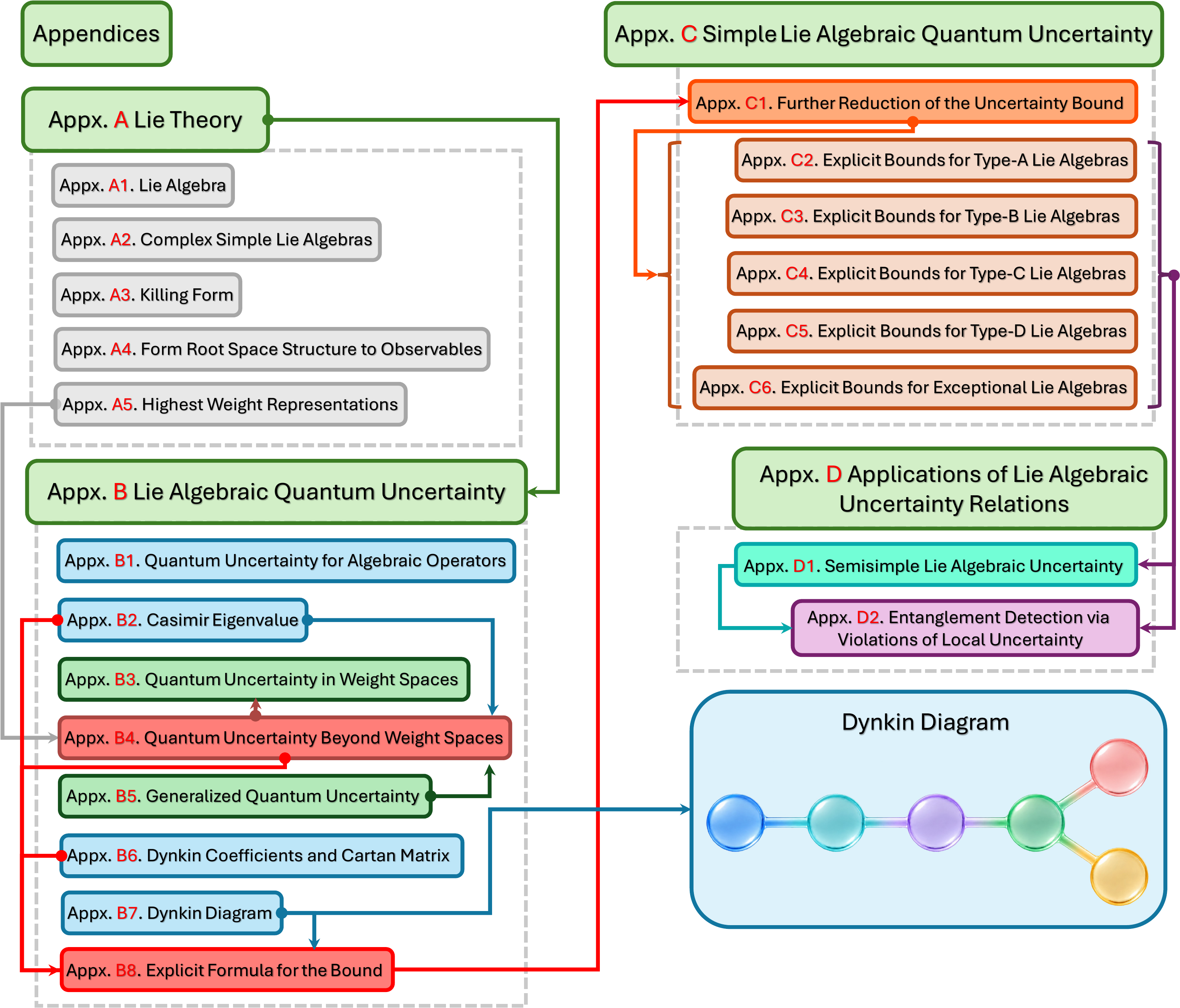}
    \caption{
    \textbf{Structure of Appendices.}
    Organization and logical dependencies of the appendices.
    Appendix~\ref{sec:Lie_Theory} introduces the Lie theory underlying the uncertainty framework developed in Appendix~\ref{sec:Lie_Algebraic_Quantum_Uncertainty}, including exact uncertainty relations within weight spaces and optimal bounds for arbitrary states. 
    Appendix~\ref{sec:Simple_Lie_Algebraic_Quantum_Uncertainty} evaluates these bounds for all classical and exceptional simple Lie algebras. Appendix~\ref{sec:entanglement} extends the framework to semisimple algebras and derives collective uncertainty criteria for entanglement detection. Arrows indicate connections between the mathematical ingredients and results.
    }
    \label{fig:Outline}
\end{figure}


\section{Lie Theory}\label{sec:Lie_Theory}
The derivation of optimal Lie algebraic uncertainty relations rests on representation-theoretic structures, which we review first. 
We begin with the basic notion of a Lie algebra, and provide the algebraic language in Subsec.~\ref{subsec:Lie_Algebra}. 
In Subsec.~\ref{subsec:Classification_Lie}, we then specialize to complex simple Lie algebras and recall their classification. 
The Killing form is introduced in Subsec.~\ref{subsec:Killing_form} as the canonical invariant bilinear form, allowing roots, weights and Cartan generators to be compared through an intrinsic geometry. 
This geometry is then connected to physics by passing from the root-space decomposition to observables in Subsec.~\ref{subsec:Root_Space_Observables}. 
Finally, in Subsec.~\ref{subsec:Highest_Weight_Representations}, we review highest weight representations, where the highest weight provides the geometric structure that later determines both the optimal uncertainty bounds and the states that attain them. 


\subsection{Lie Algebra}\label{subsec:Lie_Algebra}

Continuous symmetries and non-commutative structures in quantum theory are naturally encoded in the language of Lie algebras~\cite{woit2017quantum}.
Beyond their foundational role in describing symmetry generators and dynamical transformations, Lie algebras provide the intrinsic algebraic framework for characterizing incompatibility relations among observables, thereby offering a structural route to the formulation of uncertainty relations.
The discussions that follow therefore rely heavily on several fundamental concepts and structural features of Lie algebras, together with the associated notation and mathematical tools introduced in this section.
Unless otherwise specified, our treatment follows the standard theory of semisimple Lie algebras and their representations in the classical textbook of Humphreys~\cite{Humphreys1972}.
Meanwhile, all Lie algebras and representations considered here are finite dimensional, and the ground field is $\mathbb R$ or $\mathbb C$. Root and highest-weight theory will be formulated for complex semisimple Lie algebras.
We begin by recalling the definition of a Lie algebra.

\begin{mydef}{Lie Algebra}{Lie_Algebra}
    A Lie algebra $\g$ over a field $\mathbb{F}= \mathbb{R}$ or  $\mathbb{C}$ is a vector space equipped with a bilinear operation $[\cdot,\cdot]:\g\times\g\to\g$ referred to as the Lie bracket, which satisfies the following properties:
    \begin{itemize}
      \item \textbf{Anti-symmetry}. $[x,y]=-[y,x]$,
      \item \textbf{Jacobi Identity}. $[x,[y,z]]+[y,[z,x]]+[z,[x,y]]=0$,
    \end{itemize}
    for all $x$, $y$, and $z\in\g$.
\end{mydef}

Arising as the infinitesimal counterparts of continuous symmetry groups, Lie algebras form one of the foundational mathematical languages of modern theoretical physics.
Canonical examples include $\mathfrak{su}(2)$, which underlies angular momentum and spin, and $\mathfrak{su}(3)$, which captures the color symmetry structure of quantum chromodynamics.
More broadly, through Noether's theorem, continuous symmetries represented by Lie groups and their associated Lie algebras are fundamentally connected to conserved quantities, furnishing a unifying structural framework that permeates quantum mechanics, quantum field theory, condensed matter physics, and quantum information theory.
Among the simplest and most fundamental examples is the special linear Lie algebra $\mathfrak{sl}(2,\mathbb C)$, defined by
\begin{align}\label{eq:sl_2_C}
    \mathfrak{sl}(2,\mathbb C)\coloneqq
    \left\{
    X\in M_2(\mathbb C)
    \;\middle|\;
    \Tr[X]=0
    \right\}.
\end{align}
Here, $M_2(\mathbb C)$ denotes the set of all $2\times2$ matrices with complex entries, and $\mathfrak{sl}(2,\mathbb C)$ consists of those matrices with vanishing trace.

A complex Lie algebra encodes the underlying algebraic structure, including its roots, weights and representations, but does not by itself specify the generators relevant to a physical realization. 
This additional structure is supplied by a real form, namely a real Lie subalgebra whose complexification recovers the original algebra. 
Different real forms of the same complex Lie algebra can exhibit markedly different properties, most notably compactness or non-compactness, leading to distinct representation-theoretic and physical interpretations. 
In quantum mechanics, the compact real form is particularly natural because its finite-dimensional representations can be chosen unitary, with its anti-Hermitian generators corresponding, after multiplication by $\mathrm{i}$, to Hermitian observables. 
The complex Lie algebra $\g$ is therefore well suited to describing the abstract root and representation structure, whereas the compact real form $\g_{\mathbb R}$ provides the natural setting for connecting this structure to unitary symmetries and physical quantities such as quantum uncertainty.
For example, the compact real form of $\mathfrak{sl}(2,\mathbb C)$ is the special unitary Lie algebra $\mathfrak{su}(2)$, given by
\begin{align}\label{eq:su_2}
    \mathfrak{su}(2)
    \coloneqq
    \left\{
    X\in M_2(\mathbb C)
    \;\middle|\;
    X^\dagger=-X,
    \ \Tr[X]=0
    \right\}.
\end{align}
A canonical basis is provided by the anti-Hermitian matrices
\begin{align}\label{eq:T_xyz}
    T_x\coloneqq-\frac{\mathrm{i}}{2}\sigma_x
    =
    -\frac{\mathrm{i}}{2}
    \begin{pmatrix}
    0 & 1 \\
    1 & 0
    \end{pmatrix},
    \quad
    T_y\coloneqq-\frac{\mathrm{i}}{2}\sigma_y
    =
    \frac{1}{2}
    \begin{pmatrix}
    0 & -1 \\
    1 & 0
    \end{pmatrix},
    \quad
    T_z\coloneqq-\frac{\mathrm{i}}{2}\sigma_z
    =
    -\frac{\mathrm{i}}{2}
    \begin{pmatrix}
    1 & 0 \\
    0 & -1
    \end{pmatrix},
\end{align}
such that
\begin{align}\label{eq:su_2_T_xyz}
\mathfrak{su}(2)=\mathrm{span}\{T_x, T_y, T_z\}.
\end{align}
These generators obey the commutation relations
\begin{align}
    [T_x,T_y]=T_z,
    \qquad
    [T_y,T_z]=T_x,
    \qquad
    [T_z,T_x]=T_y,
\end{align}
with $\sigma_x$, $\sigma_y$, $\sigma_z$ denoting the Pauli matrices given by
\begin{align}
\sigma_x \coloneqq
\begin{pmatrix}
0 & 1 \\
1 & 0
\end{pmatrix},
\qquad
\sigma_y \coloneqq
\begin{pmatrix}
0 & -\mathrm{i} \\
\mathrm{i} & 0
\end{pmatrix},
\qquad
\sigma_z \coloneqq
\begin{pmatrix}
1 & 0 \\
0 & -1
\end{pmatrix}.
\end{align}
The physical spin operators are obtained by restoring the angular momentum scale,
\begin{align}\label{eq:su2_J_xyz}
    J_i
    \coloneqq
    \frac{\hbar}{2}\sigma_i,
    \quad
    i\in\{x,y,z\}.
\end{align}
In units with $\hbar=1$, one has $J_i=\sigma_i/2=\mathrm{i}T_i$. 
Here and throughout, $\mathrm{i}$ denotes the imaginary unit, $\mathrm{i}^2=-1$, whereas italic $i$ is an index.
With this normalization, the Pauli relation
\begin{align}
    [\sigma_i,\sigma_j]=
    2\mathrm{i}\epsilon_{ijk}\sigma_k
\end{align}
is precisely the spin-$1/2$ realization of the angular momentum commutation rule
\begin{align}
    [J_i,J_j]=
    \mathrm{i}\epsilon_{ijk}J_k.
\end{align}
Here, $\epsilon_{ijk}$ denotes the Levi-Civita symbol, the totally antisymmetric tensor that encodes the cyclic structure of angular momentum.
The connection between $\mathfrak{su}(2)$ and $\mathfrak{sl}(2,\mathbb C)$ is expressed through the complexification isomorphism:
\begin{align}\label{eq:Complexification_su_2}
    \mathfrak{su}(2)_{\mathbb C}
    \coloneqq
    \mathfrak{su}(2)\otimes_{\mathbb R}\mathbb C
    \cong
    \mathfrak{sl}(2,\mathbb C).
\end{align}

With the basic structure of Lie algebras in place, we now turn to the internal substructures that govern their decomposition. 
The first is an ideal, the Lie algebraic counterpart of a normal subgroup in group theory.
An ideal specifies a subspace stable under the adjoint action, thereby providing the fundamental mechanism through which the internal structure of a Lie algebra can be analyzed and decomposed.
Indeed, the classification into simple, semisimple, and solvable Lie algebras is formulated precisely in terms of the existence and arrangement of ideals, making them indispensable to the structural and representation-theoretic study of Lie algebras.

\begin{mydef}{Ideals}{Ideal}
A subspace $I\subseteq\g$ is called an ideal of the Lie algebra $\g$ if it is invariant under the adjoint action of $\g$, namely, if
\begin{align}
\mathrm{ad}_x[y]\coloneqq[x,y]\in I
\end{align}
for all $x\in\g$ and $y\in I$.
In compact form, this condition is expressed as
\begin{align}
[\g,I]\subseteq I.
\end{align}
Thus, ideals characterize those subspaces preserved by the intrinsic adjoint structure of the Lie algebra.
\end{mydef}

The Lie bracket (see Def.~\ref{def:Lie_Algebra}) induces a natural hierarchy of nested subalgebras, giving rise to the derived series, which characterizes the intrinsic non-commutative structure of the Lie algebra.
\begin{align}
    \g^{(0)}&\coloneqq\g,\label{eq:derived_algebra_g_0}\\
    \g^{(1)}&\coloneqq[\g,\g]
    =
    \mathrm{span}
    \left\{
    [x,y]
    \mid
    x,y\in\g
    \right\},\label{eq:derived_algebra_g_1}\\
    &\cdots\\
    \g^{(k+1)}&\coloneqq[\g^{(k)},\g^{(k)}].
\end{align}
The Lie algebra $\g$ is said to be abelian if the derived algebra $\g^{(1)}$ (see Eq.~\eqref{eq:derived_algebra_g_1}) vanishes, i.e., $\g^{(1)}=0$.
Equivalently, all elements of $\g$ commute pairwise, so that the algebra carries no intrinsic non-commutative structure.
The derived series captures how non-commutativity propagates through successive commutators, providing the foundation for the notion of solvability.

\begin{mydef}{Solvable Lie algebras}{Solvable}
A Lie algebra $\g$ is called solvable if its derived series terminates after finitely many steps, namely, if there exists an integer $m\geqslant 1$ such that
\begin{align}
\g^{(m)}=0.
\end{align}
\end{mydef}

A fundamental structural property is that the sum of two solvable ideals $I_1$ and $I_2$ of a Lie algebra $\g$ is itself solvable.
This stability under addition ensures the existence of a unique maximal solvable ideal, called the radical of $\g$, denoted by $\mathrm{Rad}(\g)$.
The radical therefore encapsulates the entire solvable sector inherent to the algebra and thereby provides the natural starting point for the definition of semisimple Lie algebras.

\begin{mydef}{Simple Lie algebra}{Simple_Lie_Algebra}
A Lie algebra $\g$ is called simple if it is non-abelian and has no nontrivial ideals; that is, its only ideals are ${0}$ and $\g$ itself.
\end{mydef}

The notion of simplicity identifies Lie algebras that cannot be decomposed further through ideals. 
Semisimplicity is the natural extension of this idea: 
it allows several simple components to appear simultaneously, while excluding abelian or solvable parts. 
Thus, semisimple Lie algebras are precisely those that can be built from simple Lie algebras as independent summands.

\begin{mydef}{Semisimple Lie algebra}{Semisimple_Lie_Algebra}
A Lie algebra
$\g$
is called semisimple if it decomposes into a direct sum of simple Lie algebras,
\begin{align}
    \g
    =
    \bigoplus_k
    \g_k ,
\end{align}
where each
$\g_k$
is simple.
\end{mydef}

A convenient way to test whether a Lie algebra $\g$ is semisimple is to examine its radical. The Lie algebra $\g$ is semisimple if and only if its radical is trivial, $\mathrm{Rad}(\g)=0$. 
In other words, $\g$ contains no nontrivial solvable ideals (see Def.~\ref{def:Solvable}).
Semisimple Lie algebras exhibit a highly symmetric internal structure, including Cartan subalgebras and root-space decompositions, root systems, invariant bilinear forms, and highest weight representation theory. These structures provide the algebraic and geometric framework underlying the Lie algebraic uncertainty relations developed in this work.

Since the construction of Lie algebraic uncertainty relations in this work relies fundamentally on representation theory, we now turn to the notion of a representation of a Lie algebra.

\begin{mydef}{Representation of a Lie Algebra}{repn_Lie_alg}
Let $\g$ be a Lie algebra over a field $\mathbb F$, with $V$ a vector space over $\mathbb F$. 
A representation of $\g$ on $V$ is a linear, bracket-preserving map
\begin{align}
    \pi:\g\to\mathfrak{gl}(V),
\end{align}
where $\mathfrak{gl}(V)$ is the Lie algebra of linear endomorphisms of $V$; equivalently,
\begin{align}
    \pi\left([x,y]\right)
    =
    [\pi(x),\pi(y)]
\end{align}
for all $x,y\in\g$.
\end{mydef}
For a real Lie algebra acting on a complex vector space, the representation map is real-linear and its values are complex-linear endomorphisms.
The simplest example is the defining representation $\mathrm{def}$, arising from the standard matrix realization of a Lie algebra. 
When $\g$ is realized as a matrix Lie algebra acting on a vector space $V$, its defining representation is the natural action
\begin{align}\label{eq:def_rep}
    \pi_{\mathrm{def}}(x)\coloneqq x,
\end{align}
for all $x\in\g$. 
Thus, the notion of a defining representation refers to a particular matrix realization of $\g$, rather than to the abstract Lie algebra itself.

By contrast, every Lie algebra carries a canonical representation on itself: 
the adjoint representation $\mathrm{ad}$. 
For each $x\in\g$, define
\begin{align}\label{eq:ad_rep}
    \mathrm{ad}_x(y)\coloneqq[x,y].
\end{align}
The Jacobi identity (see Def.~\ref{def:Lie_Algebra}) implies
\begin{align}
    [\mathrm{ad}_x,\mathrm{ad}_y]
    =
    \mathrm{ad}_{[x,y]},
\end{align}
and hence the map 
\begin{align}
    \mathrm{ad}:\g\to\mathrm{End}(\g)
\end{align}
defines a representation of $\g$ on its underlying vector space.
This internal representation singles out an important class of subalgebras. A subalgebra $\mathfrak h\subseteq\g$ is called toral if each operator $\mathrm{ad}_h$, with $h\in\mathfrak h$, is semisimple, equivalently diagonalizable over $\mathbb C$. 
A maximal toral subalgebra is one that is maximal among toral subalgebras with respect to inclusion.

For semisimple Lie algebras over $\mathbb C$, maximal toral subalgebras are exactly Cartan subalgebras. 
Their elements, often called Cartan generators in physics, form a maximal toral subalgebra of the algebra. 
This sector provides the simultaneous eigenspace structure that organizes the Lie algebra into its root-space decomposition. 
The typical example is $\mathfrak{su}(2)$ (see Eq.~\eqref{eq:su_2}): 
the Cartan subalgebra is $\mathbb R(-\mathrm{i}J_z)=\mathrm {span}_{\mathbb R} \{-\mathrm{i}J_z\}$; its complexification is the Cartan subalgebra $\mathbb C J_z$ of $\mathfrak{sl}(2,\mathbb C)$  (with $\hbar=1$).

\begin{mydef}{Roots}{Roots}
Let $\g$ be a complex semisimple Lie algebra with Cartan subalgebra $\mathfrak h$, and let $\mathfrak h^*$ denote the dual space of $\mathfrak h$. The adjoint action of $\mathfrak h$ decomposes $\g$ into simultaneous eigenspaces. For each $\alpha\in\mathfrak h^*$, define
\begin{align}
    \g_\alpha
    = \left\{x\in\g
    \;\middle|\;
    \mathrm{ad}_h(x)=[h,x]=\alpha(h)x,
    \ \forall h\in\mathfrak h
    \right\}.
\end{align}
A nonzero functional $\alpha\in\mathfrak h^*$ for which $\g_\alpha\neq\{0\}$ is called a root, and the collection of all such roots is denoted by $\Phi\subset\mathfrak h^*$. 
This gives the canonical root space decomposition
\begin{align}\label{eq:Root_Space_Decomposition}
    \g
    =
    \underbrace{\mathfrak h}_{\text{Cartan Subalgebra}}
    \oplus
    \underbrace{
    \bigoplus_{\alpha\in\Phi}
    \g_\alpha
    }_{\text{Root Space}}.
\end{align}
\end{mydef}

The root-space decomposition, once a set of positive roots $\Phi^+$ has been chosen, separates the generators of $\g$ into three sectors: 
(\romannumeral1) the Cartan subalgebra $\mathfrak h$, whose elements commute with one another, 
(\romannumeral2) the positive-root spaces $\g_{\alpha}$ with $\alpha\in\Phi^+$, and (\romannumeral3) the negative-root spaces $\g_{-\alpha}$ with $\alpha\in\Phi^+$. 
The latter two sectors play the role of raising and lowering operators relative to the Cartan subalgebra. 
This structure will be central to our derivation of Lie algebraic uncertainty relations. 
To build intuition, we first consider the decomposition of $\mathfrak{sl}(2,\mathbb C)$, the complexification of $\mathfrak{su}(2)$ (see Eq.~\eqref{eq:su_2}), where the raising and lowering operators are defined as follows:
\begin{align}
\label{eq:example_su21}
    J_{+}
    &\coloneqq 
    J_x+ \mathrm{i}J_y
    =
    \frac12 \begin{pmatrix}
    0 & 1 \\
    1 & 0
    \end{pmatrix}
    +
    \mathrm{i}\cdot
    \frac12 \begin{pmatrix}
    0 & -\mathrm{i} \\
    \mathrm{i} & 0
    \end{pmatrix}
    =
    \begin{pmatrix}
    0 & 1 \\
    0 & 0
    \end{pmatrix},\\
    J_{-}
    &\coloneqq 
    J_x- \mathrm{i}J_y
    =
    \frac12 \begin{pmatrix}
    0 & 1 \\
    1 & 0
    \end{pmatrix}
    -
    \mathrm{i}\cdot
    \frac12 \begin{pmatrix}
    0 & -\mathrm{i} \\
    \mathrm{i} & 0
    \end{pmatrix}
    =
    \begin{pmatrix}
    0 & 0 \\
    1 & 0
    \end{pmatrix},
\end{align}
where $J_x$ and $J_y$ are defined in Eq.~\eqref{eq:su2_J_xyz}.
The corresponding root system is given by
\begin{align}\label{eq:sl_2_root_system}
    \Phi=\{\alpha,-\alpha\},
\end{align}
with
\begin{align}
    \alpha(J_z)=1.
\end{align}
The root space decomposition then takes the form
\begin{align}
    \mathfrak{sl}(2,\mathbb C)
    =
    \mathfrak h
    \oplus
    \g_\alpha
    \oplus
    \g_{-\alpha},
\end{align}
where the three spaces are given by
\begin{alignat}{2}
    &\text{Cartan Subalgebra}\qquad
    &\mathfrak h
    &=
    \mathrm{span}\{J_z\},\\
    &\text{Positive Root Space}\qquad
    &\g_\alpha
    &=
    \mathrm{span}\{J_+\},\\
    &\text{Negative Root Space}\qquad
    &\g_{-\alpha}
    &=
    \mathrm{span}\{J_{-}\}.
\end{alignat}

So far, we have illustrated the structure using the spin-$1/2$ representation of $\mathfrak{su}(2)$. 
The same picture extends to the irreducible spin-$j$ representation, where the ladder of weights becomes fully visible. 
Let $\pi_j$ denote the representation map that realizes the abstract generators of $\mathfrak{sl}(2, \mathbb C)$ as linear operators on the spin-$j$ representation space $V_j$,
\begin{align}
    \pi_j:\mathfrak{sl}(2,\mathbb C)
    \to 
    \mathrm{End}(V_j),
\end{align}
with $\dim V_j=2j+1$.
Under this representation, the Cartan generator $J_z$ is then represented by the operator $\pi_j(J_z)$, whose eigenvalues are given by
\begin{align}
    m\in\left\{-j,-j+1,\cdots,j-1,j\right\}.
\end{align}
These eigenvalues are the possible values of the angular momentum along the $z$ direction. 
Physically, $m$ is the magnetic quantum number; 
in representation-theoretic language, the same number is called a weight. 
Each weight labels the subspace of states on which $J_z$ takes that fixed value.
This subspace is called the corresponding weight space, defined explicitly by
\begin{align}\label{eq:V_j_m}
    V_{j,m} &\coloneqq \left\{ \ket{\psi} \in V_j \,\,\middle|\,\, \pi_j(J_z) \ket{\psi} = m \ket{\psi} \right\}.
\end{align}
For $\mathfrak{su}(2)$, each such weight space is one-dimensional and is spanned by the standard angular momentum state $\ket{j,m}$,
\begin{align}
    V_{j,m}=\mathrm{span}_{\mathbb C}
    \{\ket{j,m}\}.
\end{align}
Thus the full representation space decomposes as
\begin{align}
    V_{j}
    =
    \bigoplus_{m=-j}^{j}
    V_{j,m}.
\end{align}

For example, in the spin-$1/2$ representation, the space is two-dimensional, namely $\dim V_{1/2}=2\cdot(1/2)+1=2$, and the two possible weights are $\pm1/2$.
Accordingly,
\begin{align}
    V_{\tfrac{1}{2}}
    =
    V_{\tfrac{1}{2},\tfrac{1}{2}}
    \oplus
    V_{\tfrac{1}{2},-\tfrac{1}{2}}
\end{align}
where
\begin{align}
    V_{\tfrac{1}{2},\tfrac{1}{2}}=\mathrm{span}_{\mathbb C}\{\ket{\tfrac{1}{2},\tfrac{1}{2}}\},
\end{align}
and
\begin{align}
    V_{\tfrac{1}{2},-\tfrac{1}{2}}=\mathrm{span}_{\mathbb C}\{\ket{\tfrac{1}{2},-\tfrac{1}{2}}\}.
\end{align}
In the computational basis, these two states may be written as
\begin{align}
    \ket{\tfrac{1}{2},\tfrac{1}{2}}=
    \begin{pmatrix}
    1\\
    0
    \end{pmatrix},
\end{align}
and
\begin{align}
    \ket{\tfrac{1}{2},-\tfrac{1}{2}}=
    \begin{pmatrix}
    0\\
    1
    \end{pmatrix}.
\end{align}
Meanwhile, we have
\begin{align}
    \pi_{\tfrac{1}{2}}(J_z)
    =
    \frac{1}{2}
    \begin{pmatrix}
    1&0\\
    0&-1
    \end{pmatrix}.
\end{align}
Thus the decomposition into weight spaces is simply the decomposition of a two-dimensional vector into its upper and lower spin components.

As a second example, the spin-$1$ representation is three-dimensional, i.e., $\dim V_1=2\cdot1+1=3$. 
The possible weights are
\begin{align}
    \left\{-1,0,1\right\},
\end{align}
so that
\begin{align}
    V_1
    =
    V_{1, 1}
    \oplus
    V_{1, 0}
    \oplus
    V_{1, -1},
\end{align}
with
\begin{align}
    V_{1, 1}&=
    \mathrm{span}\{\ket{1,1}\},\\
    V_{1, 0}&=
    \mathrm{span}\{\ket{1,0}\},\\
    V_{1, -1}&=
    \mathrm{span}\{\ket{1,-1}\}.
\end{align}

The Cartan generator $J_z$ acts diagonally on the weight basis:
\begin{align}
    J_z\ket{1,m}
    =
    m\ket{1,m},
    \quad
    m\in\{-1,0,1\}.
\end{align}
Its matrix representation is therefore given by
\begin{align}
    \pi_1(J_z)
    =
    \begin{pmatrix}
    1&0&0\\
    0&0&0\\
    0&0&-1
    \end{pmatrix}.
\end{align}

The ladder operators act as
\begin{align}
    J_\pm\ket{1,m}
    =
    \sqrt{(1\mp m)(1\pm m+1)}
    \ket{1,m\pm1},
\end{align}
and therefore connect adjacent weight spaces. 
In particular,
\begin{align}
    J_+\ket{1,-1}
    &=
    \sqrt2\ket{1,0},
    \\
    J_+\ket{1,0}
    &=
    \sqrt2\ket{1,1},
    \\
    J_+\ket{1,1}
    &=
    0.
\end{align}
This means that $J_+=J_x+ \mathrm{i}J_y$ moves states upward along
\begin{align}
    \ket{1,-1}\to \ket{1,0}\to \ket{1,1}.
\end{align}
In the ordered weight basis $\{\ket{1,1}, \ket{1,0}, \ket{1,-1} \}$, $J_+$ therefore takes the matrix form
\begin{align}
    \pi_1(J_+)
    =
    \begin{pmatrix}
    0&\sqrt2&0\\
    0&0&\sqrt2\\
    0&0&0
    \end{pmatrix}.
\end{align}
Conversely, the lowering operator acts as
\begin{align}
    J_-\ket{1,1}
    &=
    \sqrt2\ket{1,0},
    \\
    J_-\ket{1,0}
    &=
    \sqrt2\ket{1,-1},
    \\
    J_-\ket{1,-1}
    &=
    0.
\end{align}
The operator $J_-=J_x- \mathrm{i}J_y$ therefore moves through the weight spaces in the opposite direction to $J_+$.
\begin{align}
    \ket{1,1}\to \ket{1,0}\to \ket{1,-1}.
\end{align}
In the ordered weight basis, its matrix representation is
\begin{align}
    \pi_1(J_-)
    =
    \begin{pmatrix}
    0&0&0\\
    \sqrt2&0&0\\
    0&\sqrt2&0
    \end{pmatrix}.
\end{align}

To conclude, the weight structure records the spectrum of the Cartan generator $J_z$, while the root structure describes how the operators $J_+$ and $J_-$ move between neighboring weight spaces. 
In this familiar angular momentum setting, roots and weights therefore provide a precise algebraic language for the organization of quantum states.

It remains to single out the elementary building blocks of a positive root system: 
the simple roots.

\begin{mydef}{Simple Roots}{Simple_Roots}
Let $\Phi$ be the root system of a complex semisimple Lie algebra $\g$, and fix a choice of positive roots $\Phi^+$. 
A root $\alpha\in\Phi^+$ is called simple if it cannot be decomposed as a sum of two positive roots, i.e.,
\begin{align}
    \alpha\neq \beta+\gamma,
\end{align}
for all $\beta,\gamma\in\Phi^+$.
The set of all simple roots is denoted by
\begin{align}\label{eq:simple_roots_set}
    \Pi=\{\alpha_1,\ldots,\alpha_r\},
\end{align}
where $r=\mathrm{rank}(\g)=\dim\mathfrak h$. 
Equivalently, $\Pi$ is the unique linearly independent subset of $\Phi^+$ such that every positive root is a non-negative integer linear combination of elements of $\Pi$.
\end{mydef}

The simplest examples are provided by $\mathfrak{sl}(2,\mathbb C)$ and $\mathfrak{sl}(3,\mathbb C)$. 
For $\mathfrak{sl}(2,\mathbb C)$, whose root system is $\Phi=\{\alpha,-\alpha\}$ (see Eq.~\eqref{eq:sl_2_root_system}), there is only one positive root $\alpha$, and hence one simple root $\alpha$, namely 
\begin{align}
    \Pi=\Phi^{+}=\{\alpha\}.
\end{align}
For $\mathfrak{sl}(3,\mathbb C)$, whose root system is $\Phi=\{\pm\alpha_1,\pm\alpha_2,\pm(\alpha_1+\alpha_2)\}$ with $\Phi^+=\{\alpha_1, \alpha_2, \alpha_1+\alpha_2\}$, one may choose $\alpha_1$ and $\alpha_2$ as simple roots, i.e.,
\begin{align}
    \Pi=\{\alpha_1, \alpha_2\}\subsetneq\Phi^+.
\end{align}
The remaining positive root is $\alpha_1+\alpha_2$, which is not simple because it decomposes as a sum of two positive roots.


\subsection{Complex Simple Lie Algebras}
\label{subsec:Classification_Lie}

In this subsection, we briefly review the classification of complex simple Lie algebras, focusing only on the structural elements needed in what follows. 
Every complex semisimple Lie algebra decomposes as a direct sum of simple Lie algebras, each characterized by its root system, or equivalently by its Dynkin diagram.
The allowed connected Dynkin diagrams are highly constrained, leaving 
the four infinite classical families $A_n$, $B_n$, $C_n$ and $D_n$, and the five exceptional cases $E_6$, $E_7$, $E_8$, $F_4$ and $G_2$. 
These nine building blocks provide a unified algebraic language for organizing generators, representations and commutation relations, and will serve as the foundation for formulating uncertainty relations beyond the elementary $\mathfrak{su}(2)$ setting (see Eq.~\eqref{eq:su_2}).

We begin our discussion of simple Lie algebras with the classical families, starting with the special linear Lie algebra.
This family is denoted by type $A_n$ and is defined as follows.

\begin{mydef}{Special Linear Lie Algebra}{A_n}
The special linear Lie algebra of order $n$ over the field $\mathbb C$, denoted by
$\mathfrak{sl}_n(\mathbb C)$ or $\mathfrak{sl}(n,\mathbb C)$, is the Lie algebra consisting of all
$n\times n$ matrices with entries in $\mathbb C$ and vanishing trace. 
Explicitly,
\begin{align}
    \mathfrak{sl}(n,\mathbb C)
    \coloneqq
    \left\{
    X\in M_n(\mathbb C)
    \;\middle|\;
    \Tr[X]=0
    \right\},
\end{align}
equipped with the Lie bracket given by the matrix commutator; that is
\begin{align}
    [X,Y]=XY-YX.
\end{align}
\end{mydef}

The algebra $\mathfrak{sl}(n,\mathbb{C})$ has rank $n-1$ and corresponds to type $A_{n-1}$.
The next classical family is given by the odd-dimensional special orthogonal Lie algebras. 
We denote this family by type $B_n$ and define it as

\begin{mydef}{Odd-Dimensional Special Orthogonal Lie Algebra $B_n$}{B_n}
The odd-dimensional special orthogonal Lie algebra is the Lie algebra associated with the special orthogonal group in odd dimension.
Over the field $\mathbb C$, it is usually denoted by $\mathfrak{so}_{2n+1}(\mathbb C)$ or $\mathfrak{so}(2n+1,\mathbb C)$ and is defined as the set of all $(2n+1)\times(2n+1)$ matrices that infinitesimally preserve a non-degenerate symmetric bilinear form. 
Equivalently, it can be written as
\begin{align}
\mathfrak{so}(2n+1,\mathbb C)
=
\left\{
X\in M_{2n+1}(\mathbb C)
\;\middle|\;
X^{T}S+SX=0
\right\},
\end{align}
where $S$ is the matrix representing the chosen symmetric bilinear form
\begin{align}
S=
\begin{bmatrix}
1 & 0 & 0\\
0 & 0 & I_n\\
0 & I_n & 0
\end{bmatrix}.
\end{align}
The Lie bracket is the usual matrix commutator.
\end{mydef}

The symplectic family, corresponding to type $C_n$, is realized by the Lie algebra below

\begin{mydef}{Symplectic Lie Algebra $C_n$}{C_n}
The symplectic Lie algebra is the Lie algebra consisting of all linear transformations that infinitesimally preserve a symplectic form.
Over the field $\mathbb C$, the symplectic Lie algebra is usually denoted by $\mathfrak{sp}_{2n}(\mathbb C)$ or $\mathfrak{sp}(2n,\mathbb C)$.
It is defined as
\begin{align}
\mathfrak{sp}(2n,\mathbb C)
=
\left\{
X\in M_{2n}(\mathbb C)
\;\middle|\;
X^{T}S+SX=0
\right\},
\end{align}
where $S$ represents the standard symplectic form
\begin{align}
S=
\begin{bmatrix}
0 & I_n\\
-I_n & 0
\end{bmatrix}.
\end{align}
\end{mydef}

The remaining classical family is the even-dimensional special orthogonal Lie algebra, corresponding to type $D_n$, which is given by

\begin{mydef}{Even-Dimensional Special Orthogonal Lie Algebra $D_n$}{D_n}
The even-dimensional special orthogonal Lie algebra is the Lie algebra associated with the special orthogonal group in even dimension.
Over the field $\mathbb C$, it is usually denoted by $\mathfrak{so}_{2n}(\mathbb C)$ or $\mathfrak{so}(2n,\mathbb C)$ and is defined as 
\begin{align}
\mathfrak{so}(2n,\mathbb C)
=
\left\{
X\in M_{2n}(\mathbb C)
\;\middle|\;
X^{T}S+SX=0
\right\},
\end{align}
where $S$ stands for the following symmetric bilinear form
\begin{align}
S=
\begin{bmatrix}
0 & I_n\\
I_n & 0
\end{bmatrix}.
\end{align}
\end{mydef}

The four classical families, together with the five exceptional Lie algebras $E_6$, $E_7$, $E_8$, $F_4$ and $G_2$, give the complete classification of simple Lie algebras over $\mathbb C$. 
The corresponding structural data are summarized in Tab.~\ref{tab:Simple_Lie_Algebra}.

\begin{table}[ht]
    \centering
    \begin{tblr}{
      width = \textwidth,
      colspec = {l || c | l | X[l] | X[c]},
      row{1} = {bg=gray!50, font=\bfseries, c},
      row{2} = {bg=gray!20, font=\bfseries, c},
      hlines,                                   
      vlines,                                   
      cells = {m},                             
      cell{3-6}{1-4} = {bg=mGreen},
      cell{7-11}{1-4} = {bg=mRed},
    }
      \SetCell[c=5]{c} Complex Simple Lie Algebras & & & & \\
      Category & Type & Lie Algebra & Standard Name & Dynkin Diagram \\
      \SetCell[r=4]{c} Classical & 
      $A_n$ & 
      $\mathfrak{sl}(n+1, \mathbb{C})$ & 
      Special Linear Lie Algebra (See Def.~\ref{def:A_n}) & 
        \raisebox{-0.4ex}{\includegraphics[scale=0.6]{DK_A_n.pdf}} \\
                & $B_n$ & $\mathfrak{so}(2n+1, \mathbb{C})$ & Odd-Dimensional Special Orthogonal Lie Algebra (See Def.~\ref{def:B_n}) & 
        \raisebox{-1ex}{\includegraphics[scale=0.6]{DK_B_n.pdf}}  \\
                & $C_n$ & $\mathfrak{sp}(2n, \mathbb{C})$ & Symplectic Lie Algebra (See Def.~\ref{def:C_n}) & 
         \raisebox{-1.4ex}{\includegraphics[scale=0.6]{DK_C_n.pdf}}  \\[1.2ex]
                & $D_n$ & $\mathfrak{so}(2n, \mathbb{C})$ & Even-Dimensional Special Orthogonal Lie Algebra (See Def.~\ref{def:D_n}) & 
         \raisebox{-4.5ex}{\includegraphics[scale=0.6]{DK_D_n.pdf}} \\
      \SetCell[r=5]{c} Exceptional & $G_2$ & $\mathfrak{g}_2$ & Exceptional Lie algebra of type $G_2$ & 
            \raisebox{-1.2ex}{\includegraphics[scale=0.75]{DK_G_2.pdf}} \\[1.2ex]
                & $F_4$ & $\mathfrak{f}_4$ & Exceptional Lie algebra of type $F_4$ & \raisebox{-1.4ex}{\includegraphics[scale=0.6]{DK_F_4.pdf}}
          \\[1.2ex]
                & $E_6$ & $\mathfrak{e}_6$ & Exceptional Lie algebra of type $E_6$ & 
         \raisebox{-3ex}{\includegraphics[scale=0.6]{DK_E_6.pdf}}\\
                & $E_7$ & $\mathfrak{e}_7$ & Exceptional Lie algebra of type $E_7$ & \raisebox{-3ex}{\includegraphics[scale=0.6]{DK_E_7.pdf}} \\
                & $E_8$ & $\mathfrak{e}_8$ & Exceptional Lie algebra of type $E_8$ & \raisebox{-3ex}{\includegraphics[scale=0.6]{DK_E_8.pdf}} \\
    \end{tblr}
    \caption{\textbf{Classification of Complex Simple Lie Algebras}. 
The table lists the four infinite classical families, $A_n$, $B_n$, $C_n$ and $D_n$, together with the five exceptional types, $G_2$, $F_4$, $E_6$, $E_7$ and $E_8$. 
The classical ranges used here are
$A_n$ ($n\geqslant1$), $B_n$ ($n\geqslant2$),
$C_n$ ($n\geqslant2$), and $D_n$ ($n\geqslant4$),
with the low-rank identification $B_1\simeq C_1\simeq A_1$, $B_2\simeq C_2$, $D_3\simeq A_3$, and $D_2\simeq A_1\oplus A_1$, the latter being semisimple rather than simple.
For each type, we give its conventional Lie algebra realization, standard name, and Dynkin diagram.
The classical families arise from special linear, orthogonal and symplectic algebras, whereas the exceptional cases occur as isolated structures not belonging to any infinite series.}
    \label{tab:Simple_Lie_Algebra}
\end{table}


\subsection{Killing Form}
\label{subsec:Killing_form}

The Killing form provides a central bridge between the algebraic structure of a Lie algebra and the geometry naturally associated with its symmetries. 
Algebraically, it is a canonical bilinear form defined entirely through the adjoint action $\mathrm{ad}$ (see Eq.~\eqref{eq:ad_rep}), and therefore reflects the internal commutation structure of the Lie algebra itself. 
Its non-degeneracy characterizes semisimplicity, while its restriction to a Cartan subalgebra induces the natural inner product on the root space. 
It is therefore fundamental to the classification of simple Lie algebras, the geometry of root systems, and the construction of weights and Weyl vectors. 
In physics, the same structure gives a natural metric on compact Lie algebras, fixes generator normalizations and the corresponding Casimir operators. 
The Killing form is thus not merely a technical definition, but the canonical geometric object through which algebraic symmetries acquire lengths and invariant inner products. 
We begin with its formal definition.

\begin{mydef}{Killing Form}{Killing_Form}
For a finite-dimensional Lie algebra $\g$ over $\mathbb C$, the Killing form $\kappa_{\mathrm{Kill}}$ is the symmetric bilinear form 
\begin{align}
   \kappa_{\mathrm{Kill}}: \g\times\g\to\mathbb C,
\end{align}
defined by
\begin{align} \label{eq:Killing_form}
    \kappa_{\mathrm{Kill}}(x,y) \coloneqq 
    \Tr[\mathrm{ad}_x \mathrm{ad}_y], \qquad x,y\in\g
\end{align}
where $\mathrm{ad}$ denotes the adjoint representation specified in Eq.~\eqref{eq:ad_rep}. 
\end{mydef}

Equation~\eqref{eq:Killing_form} gives the conventional, unnormalized Killing form $\kappa_{\mathrm{Kill}}$. 
Its direct use in Lie algebraic uncertainty relations introduces algebra dependent normalization factors. 
To maintain a uniform convention across different Lie algebras, we henceforth let $\kappa$ denote the normalized invariant symmetric bilinear form specified below.
Let
\begin{align}
    \g=\bigoplus_{s}\g^{(s)}
\end{align}
be the decomposition of the complex semisimple Lie algebra $\g$ into simple ideals $\{\g^{(s)}\}_{s}$. 
Each simple component $\g^{(s)}$ is associated with a dual Coxeter number $h_s^\vee$, determined by its Lie type. 
For the classical and exceptional algebras (see Tab.~\ref{tab:Simple_Lie_Algebra}), these values are given in Tab.~\ref{tab:Dual_Coxeter_Numbers} 

\begin{table}[ht]
    \centering
    \begin{tblr}{
      width = \textwidth,
      colspec = {l || c | c | c | c | c | c | c | c | c},
      row{1} = {bg=gray!50, font=\bfseries, c}, 
      row{2} = {font=\bfseries, c}, 
      column{1} = {bg=gray!20, font=\bfseries, l}, 
      hlines,                                   
      vlines,                                   
      cells = {m},                             
      cell{2-5}{2-5} = {bg=mGreen},
      cell{2-5}{6-10} = {bg=mRed},
    }
      \SetCell[c=10]{c} Dual Coxeter Numbers & & & & & & & & & \\
      Category & \SetCell[c=4]{c} Classical & & & & \SetCell[c=5]{c} Exceptional & & & & \\
      Type & $A_n$ & $B_n$ & $C_n$ & $D_n$ & $G_2$ & $F_4$ & $E_6$ & $E_7$ & $E_8$ \\
      Dual Coxeter & $n+1$ & $2n-1$ & $n+1$ & $2n-2$ & $4$ & $9$ & $12$ & $18$ & $30$ \\
      Dynkin & 
      \raisebox{-0.5ex}{\includegraphics[scale=0.2]{DK_A_n.pdf}} & 
      \raisebox{-0.5ex}{\includegraphics[scale=0.2]{DK_B_n.pdf}}  &  
      \raisebox{-0.5ex}{\includegraphics[scale=0.2]{DK_C_n.pdf}}  & 
      \raisebox{-1.4ex}{\includegraphics[scale=0.2]{DK_D_n.pdf}}  &  
      \raisebox{-0.5ex}{\includegraphics[scale=0.2]{DK_G_2.pdf}}  & 
      \raisebox{-0.5ex}{\includegraphics[scale=0.2]{DK_F_4.pdf}} & 
      \raisebox{-0.5ex}{\includegraphics[scale=0.2]{DK_E_6.pdf}} & 
      \raisebox{-0.5ex}{\includegraphics[scale=0.2]{DK_E_7.pdf}} & 
      \raisebox{-0.5ex}{\includegraphics[scale=0.2]{DK_E_8.pdf}} \\
    \end{tblr}
    \caption{\textbf{Dual Coxeter Numbers $h^\vee$}. 
    The four classical families and five exceptional types are shown together with their dual Coxeter numbers and Dynkin diagrams.}
    \label{tab:Dual_Coxeter_Numbers}
\end{table}

\noindent
Here $n$ is the rank, with the standard ranges $n\geqslant1$ for $A_n$, $n\geqslant2$ for $B_n$ and $C_n$, and $n\geqslant4$ for $D_n$. 
Since $\mathfrak{sl}(n,\mathbb C)$ is of type $A_{n-1}$, its dual Coxeter number is $h^\vee=n$ (see Tab.~\ref{tab:Dual_Coxeter_Numbers}).

For each simple ideal $\g^{(s)}$, we set
\begin{align}\label{eq:normal_Killing}
    \kappa(x,y)
    \coloneqq
    \frac{1}{2h_{s}^{\vee}}
    \Tr_{\g^{(s)}}\!\left[\mathrm{ad}_x \mathrm{ad}_y\right],
    \qquad x,y\in\g^{(s)}.
\end{align}
Distinct simple ideals are mutually orthogonal under $\kappa$. 
This normalization ensures that the induced bilinear form on the dual Cartan space assigns squared length $2$ to every long root. 
Unless stated otherwise, all dual bases, Gram matrices and Casimir elements introduced subsequently are defined with respect to $\kappa$ (see Eq.~\eqref{eq:normal_Killing}).

The normalized Killing form $\kappa(\cdot,\cdot)$ (see Eq.~\eqref{eq:normal_Killing}) provides the natural metric structure supporting the theory.
In particular, it induces a positive-definite inner product on the compact real form $\g_{\mathbb R}$ of the Lie algebra $\g$ after a change of sign, and an invariant bilinear form on the dual of a Cartan subalgebra $\mathfrak h$. 
A real Lie algebra $\g_{\mathbb R}\subset \g$ is a real form if its complexification is $\g$, that is,
\begin{align}\label{eq:complexification}
    \g_{\mathbb R}\otimes_{\mathbb R}\mathbb C\cong\g.
\end{align}
It is called compact if the restriction of the Killing form to $\g_{\mathbb R}$ is negative definite. 
Typical examples include
$    \mathfrak{su}(n)
    \subset
    \mathfrak{sl}(n,\mathbb C),
    \mathfrak{so}(n)
    \subset
    \mathfrak{so}(n,\mathbb C)$,
and $\mathfrak{sp}(n)
    \subset
    \mathfrak{sp}(2n,\mathbb C)$.
Here the terminology ``real'' refers to the underlying vector space structure rather than to the matrix entries themselves.
For example, $\mathfrak{su}(2)$ is a real Lie algebra spanned over $\mathbb R$ by generators $\{T_x, T_y, T_z\}$ whose matrix entries may be complex, as in Eq.~\eqref{eq:T_xyz}.
Consequently, compactness implies that the Killing form $\kappa(\cdot,\cdot)$ (see Eq.~\eqref{eq:normal_Killing}) is negative definite on $\g_{\mathbb R}$.
Consequently, $-\kappa(\cdot,\cdot)$ defines a positive-definite inner product, equipping the compact real form $\g_{\mathbb R}$ with a natural Euclidean structure.

\begin{mydef}{Inner Product on $\g_{\mathbb R}$}{Inner_product_alg}
Let $\g$ be a complex semisimple Lie algebra with compact real form $\g_{\mathbb R}$, and let $\kappa$ be the normalized invariant bilinear form introduced in Eq.~\eqref{eq:normal_Killing}. 
Since $\kappa$ is negative definite on $\g_{\mathbb R}$, it defines a positive-definite inner product by
\begin{align}\label{eq:Inner_product_algebra}
    (x,y)\coloneqq
    -\kappa(x,y), \qquad x, y \in \g_{\mathbb R}.
\end{align}
\end{mydef}

To relate this inner product to an operator realization of $\g_{\mathbb R}$, we consider a finite-dimensional unitary representation
\begin{align}\label{eq:pizero}
\pi_0:\g_{\mathbb R}\to \mathfrak{gl}(V_0)
\end{align}
where $V_0$ is a finite-dimensional complex Hilbert space.
Such a representation naturally induces an invariant symmetric bilinear form,
\begin{align}\label{inner_product_repn}
    (x,y)_{\pi_0}
    \coloneqq
    -\Tr_{V_0}\bigl(\pi_0(x)\pi_0(y)\bigr),
    \qquad
    x,y\in\g_{\mathbb R}.
\end{align}
This form is positive semidefinite in general and is positive definite if and only if $\pi_0$ is faithful.

If $\g_{\mathbb R}$ is simple and $\pi_0$ is nontrivial, then $\pi_0$ is faithful. 
By the uniqueness, up to scaling, of invariant symmetric bilinear forms on a compact simple Lie algebra, there exists a constant $c_{\pi_0}>0$ such that
\begin{align}\label{eq:inner_repn}
    (x,y)_{\pi_0}
    =
    c_{\pi_0} (x,y), \qquad x,y\in\g_{\mathbb R},
\end{align}
where $(\cdot,\cdot)$ is the inner product defined in 
Eq.~\eqref{eq:Inner_product_algebra}.
More generally, the compact real form decomposes as
\begin{align}
\g_{\mathbb R}
=
\bigoplus_{s=1}^{m}\g_{\mathbb R}^{(s)},
\qquad
\g_{\mathbb R}^{(s)}
\coloneqq
\g_{\mathbb R}\cap\g^{(s)}.
\end{align}
Distinct simple ideals are orthogonal with respect to both forms.
Thus, writing
$x=\sum_{s=1}^{m}x_s$ and $y=\sum_{s=1}^{m}y_s$,
with $x_s,y_s\in\g_{\mathbb R}^{(s)}$, we have
\begin{align}
(x,y)_{\pi_0}
=
\sum_{s=1}^{m}c_{\pi_0,s}(x_s,y_s).
\end{align}
Here $c_{\pi_0,s}\geqslant 0$, with $c_{\pi_0,s}>0$ if and only if $\pi_0|_{\g_{\mathbb R}^{(s)}}$ is nontrivial.

For the remainder of this work, we fix a faithful reference representation $\pi_0$, so that $c_{\pi_0,s}>0$ for every $s$.
The inner product $(x,y)=-\kappa(x,y)$ on $\g_{\mathbb R}$
defined in Eq.~\eqref{eq:Inner_product_algebra} can then be
expressed as
\begin{align}
  (x,y)
=
\sum_{s=1}^{m}
\frac{1}{c_{\pi_0,s}}(x_s,y_s)_{\pi_0}.
\end{align}
This inner product remains fixed when representations on other
spaces $V(\Lambda)$ are considered; it is not redefined using
the trace on $V(\Lambda)$.
For complex generators, we use the unique Hermitian extension
of this inner product to $\g$, defined in
Eq.~\eqref{eq:inner_product_g}.

Furthermore, the inner product $(\cdot,\cdot)$ defined in Eq.~\eqref{eq:Inner_product_algebra} satisfies the following invariance property:

\begin{mylem}{Invariance of the Inner Product}{Invariance_Inner_Product}
The inner product $(\cdot,\cdot)$ defined in Eq.~\eqref{eq:Inner_product_algebra} satisfies
\begin{align}
    ([x,y],z)=(x,[y,z])
\end{align}
for all $x,y,z\in\g_{\mathbb R}$.
\end{mylem}

Before introducing the induced bilinear form $\langle\cdot,\cdot\rangle$ on the dual space $\mathfrak h^*$ of the Cartan subalgebra, we first recall a simple but important Lemma (see Lem.~\ref{lem:Killing_Form_Identification}).
It identifies each linear functional on $\mathfrak h$ with a unique element of $\mathfrak h$.

\begin{mylem}{Killing Form Identification}{Killing_Form_Identification}
Let $\g$ be a complex semisimple Lie algebra with Cartan subalgebra $\mathfrak h$, and let $\kappa$ be the normalized invariant bilinear form in Eq.~\eqref{eq:normal_Killing} on $\g$ whose restriction to $\mathfrak h$ is nondegenerate. 
Then each functional
$\alpha\in\mathfrak h^*$ determines a unique element
$t_\alpha\in\mathfrak h$ such that
\begin{align}\label{eq:roots_action_KF}
    \alpha(h)
    =
    \kappa(t_\alpha,h),
    \qquad
    \forall\,h\in\mathfrak h .
\end{align}
\end{mylem}

Using Lem.~\ref{lem:Killing_Form_Identification}, we can now define the induced bilinear form on the dual space $\mathfrak h^*$, the space of linear functionals on the Cartan subalgebra $\mathfrak h$.

\begin{mydef}{Induced Bilinear Form and Root-Space Inner Product}{Inner_Product_Root_Space}
For elements $\alpha,\beta\in\mathfrak h^*$, we define the induced bilinear form by
\begin{align} \label{eq:Inner_root}
\langle\alpha,\beta\rangle
\coloneqq
\kappa(t_\alpha,t_\beta)
\end{align}
where $t_\alpha,t_\beta\in\mathfrak h$ are the unique elements associated with $\alpha$ and $\beta$ (see Lem.~\ref{lem:Killing_Form_Identification}) through the normalized invariant bilinear form in Eq.~\eqref{eq:normal_Killing}.
\end{mydef}

The form in Eq.~\eqref{eq:Inner_root} is complex bilinear on $\mathfrak h^*$. Its restriction to the real root space $E=\mathrm{Span}_{\mathbb R} \Phi$ is positive definite. Throughout the geometric and uncertainty arguments, norms and inequalities refer to this real space, which contains all weights of finite-dimensional representations.

The root system $\Phi\subset\mathfrak h^*$ inherits the inner product induced by the Killing form.
For each root $\alpha\in\Phi$, the associated coroot is defined by
\begin{align}\label{eq:coroots}
    h_\alpha
    \coloneqq
    \frac{2\,t_\alpha}
    {\langle\alpha,\alpha\rangle}.
\end{align}
The resulting root-coroot pairing is therefore given by
\begin{align}\label{eq:root_coroot_pairing}
    \mu(h_\alpha)
    =
    \kappa(t_\mu,h_\alpha)
    =
    \frac{2}
    {\langle\alpha,\alpha\rangle}
    \kappa(t_\mu,t_\alpha)
    =
    \frac{
    2\langle\mu,\alpha\rangle
    }{
    \langle\alpha,\alpha\rangle},
    \qquad
    \mu\in\mathfrak h^*.
\end{align}

Given a basis of the Cartan subalgebra $\mathfrak h$ (see Def.~\ref{def:Roots}) and the normalized invariant form $\kappa(\cdot,\cdot)$ introduced in Eq.~\eqref{eq:normal_Killing}, we can define the corresponding dual basis.

\begin{mydef}{Dual Basis of the Cartan Subalgebra $\mathfrak h$}{Dual_Basis_Cartan_Subalgebra}
Let $\mathfrak h$ be a Cartan subalgebra equipped with the normalized invariant bilinear form $\kappa(\cdot,\cdot)$ defined in Eq.~\eqref{eq:normal_Killing}.
Given a basis $\{h_i\}_{i=1}^r$ of $\mathfrak h$, its $\kappa$-dual basis $\{h^i\}_{i=1}^r$ is defined by
\begin{align}
    \kappa(h_i,h^j)
    =
    \delta_{ij}.
\end{align}
\end{mydef}

For $\alpha,\beta\in\mathfrak h^*$, the induced bilinear form $\langle\alpha,\beta\rangle$ (see Def.~\ref{def:Inner_Product_Root_Space}) can be expressed in terms of a basis of the Cartan subalgebra and its $\kappa$-dual basis (see Def.~\ref{def:Dual_Basis_Cartan_Subalgebra}). 
This representation is captured by the following lemma.

\begin{mylem}{Dual Basis Expression of the Induced Bilinear Form}{Dual_Basis_Contraction}
For $\alpha,\beta\in\mathfrak h^*$, the induced bilinear form $\langle\alpha,\beta\rangle$ admits the following representation
\begin{align}\label{eq:dual_basis_general}
    \langle\alpha,\beta\rangle
    =
    \sum_{i=1}^r
    \alpha(h_i)\beta(h^i).
\end{align}
Here $\{h_i\}_{i=1}^r$ is a Cartan basis
and $\{h^i\}_{i=1}^r$ is $\kappa$-dual basis (see Def.~\ref{def:Dual_Basis_Cartan_Subalgebra}),  with $\kappa$ normalized as in Eq.~\eqref{eq:normal_Killing}.
On the real root space, setting $\alpha=\beta\in E$, Eq.~\eqref{eq:dual_basis_general} induces the associated squared norm
\begin{align}
\label{eq:dual_basis_norm}
    \|\alpha\|^2
    \coloneqq
    \langle\alpha,\alpha\rangle
    =
    \sum_{i=1}^r
    \alpha(h_i)\alpha(h^i).
\end{align}
\end{mylem}

\begin{proof}
Using the Killing form identification of Lem.~\ref{lem:Killing_Form_Identification}, we have
\begin{align}
    \alpha(h)&=\kappa(t_\alpha,h),\\
    \beta(h)&=\kappa(t_\beta,h),
\end{align}
for every $h\in\mathfrak h$.
This allows the right-hand side of Eq.~\eqref{eq:dual_basis_general} to be expressed as
\begin{align}\label{eq:DBP_RHS}
    \sum_{i=1}^r\alpha(h_i)\beta(h^i)
    =
    \sum_{i=1}^r\kappa(t_\alpha,h_i)\kappa(t_\beta,h^i)
    =
    \kappa\left(t_\alpha,\sum_{i=1}^r\kappa(t_\beta,h^i)h_i\right).
\end{align}
Since the $\kappa$-dual basis $\{h^i\}_{i=1}^r$ also forms a basis of the Cartan subalgebra $\mathfrak h$, with expansion coefficients determined by the Killing form, we further have
\begin{align}
    t_{\beta}=\sum_{i=1}^{r}\kappa(t_\beta,h^i)h_i.
\end{align}
Substituting this relation into Eq.~\eqref{eq:DBP_RHS}, we arrive at
\begin{align}
    \sum_{i=1}^r\alpha(h_i)\beta(h^i)
    =
    \kappa\left(t_\alpha,t_\beta\right)
    =
    \langle\alpha,\beta\rangle,
\end{align}
which completes the proof.
\end{proof}

The final concept we introduce is the Weyl vector $\delta$, which plays a central role in formulating state-independent bounds for Lie algebraic uncertainty relations.

\begin{mydef}{Weyl Vector}{Weyl_Vector}
Let $\g$ be a complex semisimple Lie algebra, and fix a choice of positive roots $\Phi^+$. 
The Weyl vector is defined as the half-sum of all positive roots:
\begin{align}\label{eq:weyl_vector1}
\delta 
\coloneqq 
\frac12 \sum_{\alpha \in \Phi^+} \alpha \in \mathfrak h^*.
\end{align}
\end{mydef}

By linearity of the invariant-form  identification, the Weyl vector satisfies
\begin{align}
    \sum_{\alpha\in \Phi^{+}}t_{\alpha}=t_{2\delta}=2t_{\delta}
\end{align}
or equivalently,
\begin{align}
 \sum_{\alpha\in \Phi^{+}}\frac{\langle \alpha,\alpha\rangle}{2}h_{\alpha}=2t_{\delta}
\end{align}
The Weyl vector $\delta$ in Eq.~\eqref{eq:weyl_vector1} encodes the geometry of the representation and determines the optimal bound appearing in our uncertainty relations.


\subsection{From Root Space Structure to Observables}
\label{subsec:Root_Space_Observables}

This subsection introduces the sequence of bases that connects the algebraic structure of a semisimple Lie algebra $\g$ to the operator language used in physics. 
We begin with the Chevalley basis, which is adapted to the root space decomposition and makes the underlying algebraic relations explicit. 
We then pass to a real anti-Hermitian basis, suited to the compact real form of the algebra, and finally arrive at the Hermitian basis used in physics, whose elements are interpreted as observables.

The natural starting point is the Chevalley basis, whose generators reflect the root space decomposition of a complex semisimple Lie algebra $\g$. 
As we shall see, this basis forms the foundation for constructing the anti-Hermitian generators of the compact real form $\g_{\mathbb R}$.

\begin{mydef}{Chevalley Basis}{Chevalley_Basis}
Let $\g$ be a complex semisimple Lie algebra of rank $r$, with Cartan subalgebra $\mathfrak h$
and root system $\Phi$.
After fixing a set of positive roots $\Phi^+\subset\Phi$,
the associated Chevalley basis takes the form
\begin{align}
    \{ h_i,
    e_\alpha,
    f_\alpha
    \mid
    i=1,\dots,r,\,
    \alpha\in\Phi^+\},
\end{align}
Here the rank is $r=\dim\mathfrak h$, the elements $h_i=h_{\alpha_i}$ for the simple roots $\alpha_i$ form a basis of $\mathfrak h$, while $e_{\alpha}$  and $f_{\alpha}$ span the root spaces associated with $\alpha$ and $-\alpha$ respectively.
In particular, they satisfy
\begin{align}
    [h_i,e_\alpha]= \alpha(h_i)e_\alpha,
    \quad     
    [h_i,f_\alpha]= -\alpha(h_i)f_\alpha,
    \quad 
    [e_\alpha,f_\alpha]=h_\alpha,
\end{align}
where $h_\alpha\in\mathfrak h$
denotes the coroot associated with
$\alpha$ (see Eq.~\eqref{eq:coroots}).
\end{mydef}

The Cartan generators $\{h_i\}_{i=1}^{r}$
span the Cartan sector of the Lie algebra, whereas $e_\alpha$ and $f_\alpha$ are root vectors associated with the positive and negative roots, respectively. 
We choose the Cartan subalgebra and Chevalley basis compatibly with the fixed compact real form $\g_{\mathbb R}$.
Let $\tau$ denote the conjugation with respect to $\g_{\mathbb R}$, so that
\begin{align}\label{eq:tau_conjugation}
    \tau(h_i)=-h_i,\qquad \tau(e_\alpha)=-f_\alpha.
\end{align}
After extending a finite-dimensional unitary representation of 
$\g_{\mathbb R}$ complex-linearly to $\g$, the represented operators satisfy
\begin{equation}
e_\alpha^\dagger = f_\alpha,\quad f_\alpha^\dagger = e_\alpha,\quad 
h_i^\dagger = h_i,
\end{equation}
where we suppress the representation map in operator expressions.
As a simple example, consider $\mathfrak{sl}(2,\mathbb C)$ (see Eq.~\eqref{eq:sl_2_C}), whose Chevalley basis is given by
\begin{align}\label{sl2_chevellry_basis}
    h
    &=
    \begin{pmatrix}
        1 & 0\\
        0 & -1
    \end{pmatrix}
    =\sigma_z,\\
    e_\alpha
    &=
    \begin{pmatrix}
        0 & 1\\
        0 & 0
    \end{pmatrix}
    =\frac{1}{2}
    \left(\sigma_x+ \mathrm{i}\sigma_y\right),
    \\
    f_\alpha
    &=
    \begin{pmatrix}
        0 & 0\\
        1 & 0
    \end{pmatrix}
    =\frac{1}{2}
    \left(\sigma_x- \mathrm{i}\sigma_y\right).
\end{align}
These generators obey
\begin{align}\label{eq:chevalley_relations}   
     [h,e_\alpha]=2e_\alpha,\quad
    [h,f_\alpha]=-2f_\alpha,\quad
    [e_\alpha,f_\alpha]=h.
\end{align}
Thus $e_\alpha$ and $f_\alpha$ span the positive and negative root spaces, while $h$ spans the Cartan subalgebra. 
This example gives the basic model of the root space decomposition and serves as the prototype for the semisimple structure.

Passing to the compact real form $\g_{\mathbb R}$ (see Eq.~\eqref{eq:complexification}), the Chevalley basis may be recombined into an anti-Hermitian basis. 
The Cartan generators, however, are not generally orthogonal under the chosen invariant inner product and therefore require an additional orthogonalization step. 
To formulate this construction systematically, we first introduce the Cartan Gram matrix $G$ associated with the Cartan generators $\{h_i\}_{i=1}^{r}$ (see Def.~\ref{def:Chevalley_Basis}).

\begin{mydef}{Cartan Gram Matrix}{Cartan_Gram_Matrix}
Let $\g$ be a complex semisimple Lie algebra with Cartan subalgebra $\mathfrak h$, and let
$\{h_1,\ldots,h_r\}$ be the generators of $\mathfrak h$ (see Def.~\ref{def:Chevalley_Basis}), with
$r=\mathrm{rank}(\g)$. 
Since the invariant bilinear form $\kappa$ is positive definite when restricted to the real span $\mathfrak{h}_{\mathbb R} = \mathrm{span}_{\mathbb R}\{h_1,\ldots,h_r\}$, the associated Cartan Gram matrix 
\begin{align}
    G\coloneqq (G_{ij})_{1\leqslant i,j\leqslant r},
\end{align}
with
\begin{align}\label{eq:Cartan_Gram_Matrix}
G_{ij}
\coloneqq
\kappa(h_i,h_j),
\end{align}
 is real, symmetric, and positive definite.
\end{mydef}

Choose a real invertible matrix $U=(U_{ij})_{1\leqslant i,j \leqslant r}\in \mathrm{GL}_r(\mathbb R)$
satisfying
\begin{align}\label{eq:Cartan_orthonormalization_matrix}
UGU^{\T}
=
2\1_r.
\end{align}
Since the Cartan Gram matrix $G$ is positive definite (see Def.~\ref{def:Cartan_Gram_Matrix}), such a matrix always exists. 
One convenient choice is
\begin{align}\label{eq:one_choice_of_U}
U
=
\sqrt{2}\,G^{-\frac{1}{2}}.
\end{align}
The choice of $U$ is, however, not unique.
Indeed, if $U$ is one solution of Eq.~\eqref{eq:Cartan_orthonormalization_matrix}, then every matrix obtained by left multiplication with an orthogonal transformation $O\in \mathrm{O}(r)$, i.e., $OU$ is also a solution.

Using $U$, we define new Cartan generators by
\begin{align}\label{eq:orthogonalized_H_i}
H_i
\coloneqq
-\mathrm{i}\sum_{j=1}^{r}U_{ij}h_j,
\quad
i=1,\ldots,r.
\end{align}
These generators form an orthogonal anti-Hermitian basis of the compact Cartan subalgebra.
Let $(\cdot,\cdot)$ denote the positive invariant inner product on the compact real form $\g_{\mathbb R}$ (see Eq.~\eqref{eq:complexification}), $(x,y)=-\kappa(x,y)$, as defined in Eq.~\eqref{eq:Inner_product_algebra}. 
Then
\begin{align}
(H_i,H_j)
&=
-\kappa\left(
-\mathrm{i}\sum_{k=1}^{r}U_{ik}h_k,
-\mathrm{i}\sum_{l=1}^{r}U_{jl}h_l
\right)
\\
&=
\sum_{k,l=1}^{r}
U_{ik}U_{jl}\,\kappa(h_k,h_l)
\\
&=
(UGU^{\T})_{ij}
\\
&=
2\delta_{ij},
\label{eq:Cartan_Ha_orthonormal_proof}
\end{align}
where the last equality follows from
Eq.~\eqref{eq:Cartan_orthonormalization_matrix}.

With the orthogonalized Cartan generators $H_i$ introduced in Eq.~\eqref{eq:orthogonalized_H_i}, all ingredients are now in place to construct an anti-Hermitian basis of the compact real form $\g_{\mathbb R}$ (see Eq.~\eqref{eq:complexification}). The root-length factors normalize all root directions to squared norm 2.

\begin{mydef}{Anti-Hermitian Basis}{Anti_Hermitian_Basis}
Let $\g$ be a complex semisimple Lie algebra with compact real form $\g_{\mathbb R}$, and let $\{h_i,e_\alpha,f_\alpha\}$
be a Chevalley basis of $\g$
(see Def.~\ref{def:Chevalley_Basis}).
The anti-Hermitian basis associated with this Chevalley basis is defined by
\begin{align}
    \{H_i, E_\alpha, F_\alpha
    \mid
    i=1,\dots,r,\,
    \alpha\in\Phi^+\},
\end{align}
where the Cartan generators $\{H_i\}_{i=1}^{r}$ are given by Eq.~\eqref{eq:orthogonalized_H_i}, and the generators of the root spaces are
\begin{align}
E_\alpha
&\coloneqq
-\mathrm{i}\sqrt{\frac{\langle \alpha,\alpha\rangle}{2}}(e_\alpha+f_\alpha),
\label{eq:anti_hermitian_root_generators_E}
\\
F_\alpha
&\coloneqq
-\sqrt{\frac{\langle \alpha,\alpha\rangle}{2}} (e_\alpha-f_\alpha)
.\label{eq:anti_hermitian_root_generators_F}
\end{align}
\end{mydef}

With respect to the compact invariant inner product $(\cdot,\cdot)$ introduced in Eq.~\eqref{eq:Inner_product_algebra}, the anti-Hermitian generators $\{H_i,E_\alpha,F_\alpha\}$ (see Def.~\ref{def:Anti_Hermitian_Basis}) form an orthogonal basis of $\g_{\mathbb R}$ satisfying
\begin{align}
(H_i,H_j)
&=
2\delta_{ij},\label{eq:Inner}\\
(E_\alpha,E_\beta)
&=
2\delta_{\alpha\beta},\label{eq:Inner2}\\
(F_\alpha,F_\beta)
&=
2\delta_{\alpha\beta},\label{eq:Inner3}
\end{align}
while all mixed inner products vanish.

A useful property of the generators $\{H_i, E_{\alpha}, F_{\alpha}\}$ is that all their elements are anti-Hermitian. We use the term anti-Hermitian, as is common in the physics literature; in the mathematical literature, the same property is often called skew-Hermitian. To avoid ambiguity, we recall that an operator $X$ is anti-Hermitian if $X^{\dagger}=-X$. We now verify this property explicitly for the generators above.

\begin{mylem}{Anti-Hermitian Property}{Anti_Hermitian}
The generators
$\{H_i,E_\alpha,F_\alpha\}$ (see Def.~\ref{def:Anti_Hermitian_Basis})
form an anti-Hermitian basis of the compact real form
$\g_{\mathbb R}$.
Explicitly, each generator satisfies the condition
\begin{align}
    H_i^\dagger
    =-H_i, 
    \quad
    E_\alpha^\dagger
    =-E_\alpha,
    \quad
    F_\alpha^\dagger
    =-F_\alpha.
\end{align}
\end{mylem}

\begin{proof}
Using
\begin{align}
    h_i^\dagger=h_i,
    \quad
    e_\alpha^\dagger=f_\alpha,
    \quad
    f_\alpha^\dagger=e_\alpha,
\end{align}
together with the fact that $U_{ij}\in\mathbb R$, the anti-Hermiticity of the generators follows directly. 
For the Cartan generators $\{H_i\}_{i=1}^{r}$ (see Eq.~\eqref{eq:orthogonalized_H_i}),
\begin{align}
    H_i^\dagger
    &=
    \left(\sum_{j=1}^{r} -\mathrm{i}U_{ij}h_j\right)^\dagger
    =
    \sum_{j=1}^{r} \mathrm{i}U_{ij}h_j
    =
    -H_i,
\end{align}
Here $U\in\mathrm{GL}_r(\mathbb R)$ denotes the real invertible matrix introduced in Eq.~\eqref{eq:Cartan_orthonormalization_matrix}.
Similarly, for each positive root $\alpha\in\Phi^+$,
\begin{align}
   E_\alpha^\dagger
    &=
    \bigl(-\mathrm{i}\sqrt{\langle\alpha,\alpha\rangle/2}(e_\alpha+f_\alpha)\bigr)^\dagger
    =
    \mathrm{i}\sqrt{\langle\alpha,\alpha\rangle/2}(f_\alpha+e_\alpha)
    =
    -E_\alpha ,\\
    F_\alpha^\dagger
    &=
    \big(-\sqrt{\langle\alpha,\alpha\rangle/2}(e_\alpha-f_\alpha)\big)^\dagger
    =
    \sqrt{\langle\alpha,\alpha\rangle/2}(e_\alpha-f_\alpha)
    =
    -F_\alpha.
\end{align}
Thus all generators are anti-Hermitian, completing the proof.
\end{proof}

We now arrive at the standard convention of quantum mechanics. 
Multiplying each anti-Hermitian generator (see Def.~\ref{def:Anti_Hermitian_Basis}) by the imaginary unit produces the Hermitian basis
\vspace{1.5cm}
\begin{align}
\label{eq:hermitian_base}
    \{ 
    \tikzmarknode{i1}{\mathrm{i}}H_{\tikzmarknode{idx}{i}}, 
    \tikzmarknode{i2}{\mathrm{i}}E_\alpha, 
    \tikzmarknode{i3}{\mathrm{i}}F_\alpha 
    \},
\end{align}
\vspace{1.0cm} 

\begin{tikzpicture}[overlay, remember picture, >=stealth, cyan!70!black, thick]
    \node (imaginary_text) [above=1.0cm of i2, text=black, font=\small\sffamily] {imaginary unit};
    
    \draw[->] ([yshift=3pt]i1.north) to [bend left=10] ([xshift=-6pt]imaginary_text.south);
    \draw[->] ([yshift=3pt]i2.north) -- (imaginary_text.south);
    \draw[->] ([yshift=3pt]i3.north) to [bend right=10] ([xshift=6pt]imaginary_text.south);

    \node (index_text) [below=0.7cm of idx, text=black, font=\small\sffamily] {index};
    \draw[->] ([yshift=-2pt]idx.south) -- (index_text.north);
\end{tikzpicture}

\noindent
whose elements are naturally identified with quantum observables. 
This final step bridges the abstract Lie algebraic framework and the operator formalism of quantum mechanics. 
The complete chain of basis transformations is summarized below.
\begin{align}
    \underbrace{\{h_i, e_{\alpha}, f_{\alpha}\}}_{\text{Chevalley Basis}}
    \longrightarrow
    \underbrace{\{H_i,E_\alpha,F_\alpha\}}_{\text{Anti-Hermitian Basis}}
    \longrightarrow
    \underbrace{\{\mathrm{i}H_i,\mathrm{i}E_\alpha,\mathrm{i}F_\alpha\}}_{\text{Observable Basis}}.
\end{align}

This subsection concludes with representative examples that make the abstract basis construction explicit in concrete Lie algebras.
The simplest case is the type-A Lie algebra $\mathfrak{sl}(2,\mathbb C)$ (see Eq.~\eqref{eq:sl_2_C}), for which a Chevalley basis $\{h, e_{\alpha}, f_{\alpha}\}$
(see Def.~\ref{def:Chevalley_Basis}) can be chosen as
\begin{align}
\left\{h=\begin{pmatrix}
    1&0\\0&-1
\end{pmatrix},
e_\alpha=\begin{pmatrix}
    0&1\\0&0
\end{pmatrix}, 
f_\alpha=\begin{pmatrix}
    0&0\\1&0
\end{pmatrix}
\right\}.
\end{align}
In this case, the Cartan Gram matrix reduces to $G=(2)$ (see Def.~\ref{def:Cartan_Gram_Matrix}), so the orthogonalization is trivial and one may simply choose $U= (1)$ (see Eq.~\eqref{eq:Cartan_orthonormalization_matrix}).
The corresponding anti-Hermitian basis
$\{H,E_\alpha,F_\alpha\}$ (see Def.~\ref{def:Anti_Hermitian_Basis})
for the compact real form
$\mathfrak{su}(2)$ (see Eq.~\eqref{eq:su_2})
is obtained by
\begin{align}
H
&=
-\mathrm{i}h
=
\begin{pmatrix}
-\mathrm{i}&0\\
0&\mathrm{i}
\end{pmatrix},
\\
E_\alpha
&=
-\mathrm{i}(e_\alpha+f_\alpha)
=
\begin{pmatrix}
0&-\mathrm{i}\\
-\mathrm{i}&0
\end{pmatrix},\\
F_\alpha
&=
-(e_\alpha-f_\alpha)
=
\begin{pmatrix}
0&-1\\
1&0
\end{pmatrix}.
\end{align}
Thus, the corresponding anti-Hermitian basis (see Def.~\ref{def:Anti_Hermitian_Basis}) takes the form
\begin{align}
\left\{
H=\begin{pmatrix}
    -\mathrm{i}&0\\0&\mathrm{i}
\end{pmatrix},
E_\alpha=\begin{pmatrix}
    0&-\mathrm{i}\\-\mathrm{i}&0
\end{pmatrix}, 
F_\alpha=\begin{pmatrix}
    0&-1\\1&0
\end{pmatrix}
\right\}.
\end{align}
Their normalization is
\begin{align}
-\Tr[H^2]
=
-\Tr[E_\alpha^2]
=
-\Tr[F_\alpha^2]  
=
2,
\end{align}
while all mixed inner products vanish.

The passage to the quantum mechanical convention is completed by multiplying each anti-Hermitian generator by the imaginary unit $\mathrm{i}$, thereby obtaining the corresponding Hermitian basis of observables,
\begin{align}
\mathrm{i}H
&=
h
=
\begin{pmatrix}
1&0\\
0&-1
\end{pmatrix}
= 2J_z,
\\
\mathrm{i}E_\alpha
&=
e_\alpha+f_\alpha
=
\begin{pmatrix}
0&1\\
1&0
\end{pmatrix}
= 2J_x,\\
\mathrm{i}F_\alpha
&=
\mathrm{i}(f_\alpha-e_\alpha)
=
\begin{pmatrix}
0&-\mathrm{i}\\
\mathrm{i}&0
\end{pmatrix}
= 2J_y.
\end{align}
Thus, the resulting Hermitian basis of observables is
\begin{align}
\left\{
\mathrm{i}H=\begin{pmatrix}
    1&0\\0&-1
\end{pmatrix},
\mathrm{i}E_\alpha=\begin{pmatrix}
    0&1\\1&0
\end{pmatrix}, 
\mathrm{i}F_\alpha=\begin{pmatrix}
    0&-\mathrm{i}\\ \mathrm{i}&0
\end{pmatrix}
\right\}.
\end{align}
Equivalently, up to the overall normalization factor of $2$, this construction recovers the familiar spin-$1/2$ angular momentum.

For $\mathfrak{sl}(2,\mathbb C)$ and its compact real form $\mathfrak{su}(2)$, the Chevalley (see Def.~\ref{def:Chevalley_Basis}), anti-Hermitian (see Def.~\ref{def:Anti_Hermitian_Basis}) and Hermitian bases are therefore linked by the transformation chain
\begin{align}\label{eq:su2_three_basis}
\underbrace{\{h,e_\alpha,f_\alpha\}}_{\text{Chevalley Basis}}
\longrightarrow
\underbrace{\{H,E_\alpha,F_\alpha\}}_{\text{Anti-Hermitian Basis}}
\longrightarrow
\underbrace{\{\mathrm{i}H,\mathrm{i}E_\alpha,\mathrm{i}F_\alpha\}}_{\text{Hermitian Basis}}
=
\underbrace{\{2J_z,2J_x,2J_y\}\, .}_{\text{Angular Momentum}}
\end{align}

The next example moves to the rank two type-A algebra $\mathfrak{sl}(3,\mathbb C)$ and its compact real form $\mathfrak{su}(3)$.
Denoting by $\mathbf E_{ij}$ the $3\times3$ matrix unit with a single nonzero entry in position $(i,j)$, a standard Chevalley basis (see Def.~\ref{def:Chevalley_Basis})
\begin{align}
    \{h_1, h_2, e_{\alpha_1}, e_{\alpha_2}, e_{\alpha_1+\alpha_2}, f_{\alpha_1}, f_{\alpha_2}, f_{\alpha_1+\alpha_2}\}
\end{align}
of $\mathfrak{sl}(3,\mathbb C)$ takes the form
\begin{gather}
h_1
=
 E_{11}- E_{22}
=
\begin{pmatrix}
1&0&0\\
0&-1&0\\
0&0&0
\end{pmatrix},
\quad
h_2
=
 E_{22}- E_{33}
=
\begin{pmatrix}
0&0&0\\
0&1&0\\
0&0&-1
\end{pmatrix},\\
e_{\alpha_1}
=
 E_{12}
=
\begin{pmatrix}
0&1&0\\
0&0&0\\
0&0&0
\end{pmatrix},
\quad
e_{\alpha_2}
=
 E_{23}
=
\begin{pmatrix}
0&0&0\\
0&0&1\\
0&0&0
\end{pmatrix},
\quad
e_{\alpha_1+\alpha_2}
=
 E_{13}
=
\begin{pmatrix}
0&0&1\\
0&0&0\\
0&0&0
\end{pmatrix},\\
f_{\alpha_1}
=
E_{21}
=
\begin{pmatrix}
0&0&0\\
1&0&0\\
0&0&0
\end{pmatrix},
\quad
f_{\alpha_2}
=
 E_{32}
=
\begin{pmatrix}
0&0&0\\
0&0&0\\
0&1&0
\end{pmatrix},
\quad
f_{\alpha_1+\alpha_2}
=
E_{31}
=
\begin{pmatrix}
0&0&0\\
0&0&0\\
1&0&0
\end{pmatrix}.
\label{eq:sl3_chevalley_basis}
\end{gather}
The corresponding set of positive roots is $\Phi^+=\{\alpha_1,\alpha_2,\alpha_1+\alpha_2\}$.

For this algebra, one convenient choice of U satisfying Eq.~\eqref{eq:Cartan_orthonormalization_matrix} is
\begin{align}
    U=
    \begin{pmatrix}
    1&0\\
    \frac{\sqrt{3}}{3}&\frac{2\sqrt{3}}{3}
    \end{pmatrix},
\end{align}
from which the corresponding anti-Hermitian Cartan generators follow as
\begin{align}\label{eq:su3_anti_hermitian_cartan_chevalley}
H_1
=
-\mathrm{i}h_1
=
\begin{pmatrix}
-\mathrm{i}&0&0\\
0&\mathrm{i}&0\\
0&0&0
\end{pmatrix},
\quad
H_2
=
-\frac{\sqrt{3}}{3}\mathrm{i}(h_1+2h_2)
=
-\frac{\sqrt{3}}{3}\mathrm{i}
\begin{pmatrix}
1 & 0 & 0 \\
0 & 1 & 0 \\
0 & 0 & -2
\end{pmatrix}.
\end{align}
The root space sector is obtained analogously, giving the anti-Hermitian generators
\begin{gather}
E_{\alpha_1}
=
-\mathrm{i}(e_{\alpha_1}+f_{\alpha_1})
=
\begin{pmatrix}
0&-\mathrm{i}&0\\
-\mathrm{i}&0&0\\
0&0&0
\end{pmatrix},
\quad
E_{\alpha_2}
=
-\mathrm{i}(e_{\alpha_2}+f_{\alpha_2})
=
\begin{pmatrix}
0&0&0\\
0&0&-\mathrm{i}\\
0&-\mathrm{i}&0
\end{pmatrix},\\
E_{\alpha_1+\alpha_2}
=
-\mathrm{i}(e_{\alpha_1+\alpha_2}
+
f_{\alpha_1+\alpha_2})
=
\begin{pmatrix}
0&0&-\mathrm{i}\\
0&0&0\\
-\mathrm{i}&0&0
\end{pmatrix}\\
F_{\alpha_1}
=
-(e_{\alpha_1}-f_{\alpha_1})
=
\begin{pmatrix}
0&-1&0\\
1&0&0\\
0&0&0
\end{pmatrix},
\quad
F_{\alpha_2}
=
-(e_{\alpha_2}-f_{\alpha_2})
=
\begin{pmatrix}
0&0&0\\
0&0&-1\\
0&1&0
\end{pmatrix},\\
F_{\alpha_1+\alpha_2}
=
-(e_{\alpha_1+\alpha_2}
-
f_{\alpha_1+\alpha_2})
=
\begin{pmatrix}
0&0&-1\\
0&0&0\\
1&0&0
\end{pmatrix}.
\label{eq:su3_anti_hermitian_root_chevalley}
\end{gather}
These eight generators are anti-Hermitian and mutually orthogonal with respect to the chosen invariant inner product. 
They therefore constitute the basis (see Def.~\ref{def:Anti_Hermitian_Basis})
\begin{align}
\left\{
H_1,H_2,
E_{\alpha_1},
E_{\alpha_2}, E_{\alpha_1+\alpha_2}, F_{\alpha_1}, F_{\alpha_2}, F_{\alpha_1+\alpha_2}
\right\},
\end{align}
of the compact real form $\mathfrak{su}(3)$ induced by the Chevalley basis above (see Def.~\ref{def:Chevalley_Basis}).

The corresponding Hermitian basis 
\begin{align}
\left\{
\mathrm{i}H_1, \mathrm{i}H_2, \mathrm{i}E_{\alpha_1}, \mathrm{i}E_{\alpha_2}, \mathrm{i}E_{\alpha_1+\alpha_2}, \mathrm{i}F_{\alpha_1}, \mathrm{i}F_{\alpha_2}, \mathrm{i}F_{\alpha_1+\alpha_2} \right\}
\end{align}
is
\begin{gather}
\mathrm{i}H_1
=
\begin{pmatrix}
1&0&0\\
0&-1&0\\
0&0&0
\end{pmatrix},
\quad
\mathrm{i}H_2
=
\frac{\sqrt{3}}{3}
\begin{pmatrix}
1 & 0 & 0 \\
0 & 1 & 0 \\
0 & 0 & -2
\end{pmatrix},\label{eq:Gell_Cartan}\\
\mathrm{i}E_{\alpha_1}
=
\begin{pmatrix}
0&1&0\\
1&0&0\\
0&0&0
\end{pmatrix},
\quad
\mathrm{i}E_{\alpha_2}
=
\begin{pmatrix}
0&0&0\\
0&0&1\\
0&1&0
\end{pmatrix},
\quad
\mathrm{i}E_{\alpha_1+\alpha_2}
=
\begin{pmatrix}
0&0&1\\
0&0&0\\
1&0&0
\end{pmatrix},\label{eq:Gell_E}\\
\mathrm{i}F_{\alpha_1}
=
\begin{pmatrix}
0&-\mathrm{i}&0\\
\mathrm{i}&0&0\\
0&0&0
\end{pmatrix},
\quad
\mathrm{i}F_{\alpha_2}
=
\begin{pmatrix}
0&0&0\\
0&0&-\mathrm{i}\\
0&\mathrm{i}&0
\end{pmatrix},
\quad
\mathrm{i}F_{\alpha_1+\alpha_2}
=
\begin{pmatrix}
0&0&-\mathrm{i}\\
0&0&0\\
\mathrm{i}&0&0
\end{pmatrix}.
\label{eq:su3_hermitian_chevalley_basis}
\end{gather}
These generators are precisely the Gell-Mann matrices~\cite{Gell-Mann}, up to a permutation of their ordering:

\begin{align}
     \raisebox{-8ex}{\includegraphics[scale=1]{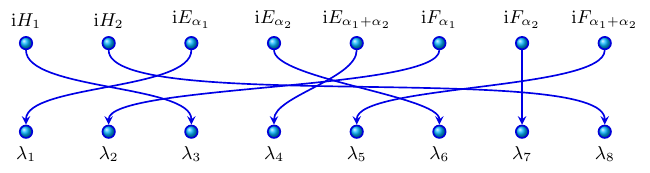}} \, .
\end{align}

This identification shows that the familiar $\mathfrak{su}(3)$ observables emerge naturally from the Chevalley basis (see Def.~\ref{def:Chevalley_Basis}), with the root space decomposition providing the underlying algebraic structure. 
The complete correspondence is summarized by the following sequence of basis transformations:
\begin{gather}\label{eq:su3_three_basis}
\underbrace{\{h_1,h_2, e_{\alpha_1},e_{\alpha_2},e_{\alpha_1+\alpha_2},f_{\alpha_1}, f_{\alpha_2}, f_{\alpha_1+\alpha_2}\}}_{\text{Chevalley Basis}} \\
\downarrow \notag \\
\underbrace{\{H_1,H_2,
E_{\alpha_1},E_{\alpha_2}, E_{\alpha_1+\alpha_2}, F_{\alpha_1}, F_{\alpha_2}, F_{\alpha_1+\alpha_2}\}}_{\text{Anti-Hermitian Basis}}  \\
\downarrow \notag\\
\underbrace{\{\mathrm{i}H_1, \mathrm{i}H_2, \mathrm{i}E_{\alpha_1}, \mathrm{i}E_{\alpha_2}, \mathrm{i}E_{\alpha_1+\alpha_2}, \mathrm{i}F_{\alpha_1}, \mathrm{i}F_{\alpha_2}, \mathrm{i}F_{\alpha_1+\alpha_2}\}
}_{\text{Hermitian Basis}}\\
\rotatebox[origin=c]{90}{$=$} \notag\\
\underbrace{\{\lambda_3, \lambda_8, \lambda_1, \lambda_6, \lambda_4, \lambda_2, \lambda_7, \lambda_5\}}_{
\text{Gell-Mann Matrices}}\label{eq:Gell_mann}.
\end{gather}

These examples make explicit how the abstract root space decomposition is converted into the operator language of quantum mechanics. 
Starting from a Chevalley basis (see Def.~\ref{def:Chevalley_Basis}), orthogonalization of the Cartan generators and recombination of opposite root spaces produce an orthogonal anti-Hermitian basis (see Def.~\ref{def:Anti_Hermitian_Basis}) of the compact real form, and hence, after multiplication by $\mathrm{i}$, a corresponding orthogonal set of Hermitian observables. 
For $\mathfrak{su}(2)$ and $\mathfrak{su}(3)$, this construction reproduces the familiar Pauli and Gell-Mann operators, up to normalization, ordering and the choice of orthogonal Cartan basis. 
This correspondence will allow the root geometry and representation-theoretic structure of the Lie algebra to be carried directly into the formulation of uncertainty relations.


\subsection{Highest Weight Representations}
\label{subsec:Highest_Weight_Representations}

Highest weight representations provide the core organizing principle for the finite-dimensional representation theory of semisimple Lie algebras. 
Once a Cartan subalgebra $\mathfrak h$ and a choice of positive roots $\Phi^+$ have been fixed, a representation decomposes into weight spaces, and the root generators act as raising and lowering operators between them. 
A highest weight vector is a  nonzero weight vector annihilated by all raising operators; 
from this single vector, the entire irreducible representation is generated by successive applications of lowering operators. 
In this way, the highest weight serves as a compact label for the representation and determines its full weight space structure, reducing the classification of finite-dimensional irreducible representations to the classification of highest weights. 
The allowed highest weights are precisely the dominant integral weights.
In this section, we introduce the weight space decomposition and the notion of highest weight within this general framework.

Let $\pi:\g\to \mathrm{End}(V)$ be a finite-dimensional representation of a complex semisimple Lie algebra $\g$ (see Def.~\ref{def:repn_Lie_alg}) with Cartan subalgebra $\mathfrak h$. 
In a finite-dimensional representation of a complex semisimple Lie algebra, the Cartan generators act by commuting semisimple operators and are therefore simultaneously diagonalizable.
This leads to the following weight-space decomposition theorem.

\begin{mythm}{Weight Space Decomposition}{Weight_Space_Decomposition}
Every finite-dimensional representation $V$ decomposes into weight spaces
\begin{align}
    V
    =
    \bigoplus_{\mu\in\mathfrak h^*}
    V_\mu,
\end{align}
where
\begin{align}
    V_\mu
    \coloneqq
    \{
    \psi\in V
    \mid
    \pi(h)\ket{\psi}
    =
    \mu(h)\ket{\psi},
    \ \forall\,h\in\mathfrak h
    \}
\end{align}
is the simultaneous eigenspace of the Cartan subalgebra $\mathfrak h$ with eigenvalue functional $\mu\in\mathfrak h^*$.
Whenever $V_\mu\neq\{0\}$, the corresponding functional $\mu$ is called a weight of the representation.
\end{mythm}

The root generators act as raising and lowering operators with respect to the weight-space decomposition. More precisely, if $|\psi_{\mu}\rangle\in V_\mu$,
then
\begin{align}
    \pi(e_\alpha)\ket{\psi_{\mu}}
    \in
    V_{\mu+\alpha},
    \quad
    \pi(f_\alpha)\ket{\psi_{\mu}}
    \in
    V_{\mu-\alpha}.
\end{align}
Thus the root system determines how the representation is connected internally, with each root specifying an allowed transition between weight sectors. 

\begin{mydef}{Highest Weight Vector}{Highest_Weight_Vector}
A nonzero vector $\ket{\psi_{\Lambda}}\in V$ is called a highest weight vector of weight $\Lambda\in\mathfrak h^*$ if it satisfies the following conditions:
\begin{alignat}{2}
    \pi(h)\ket{\psi_{\Lambda}}
    &= \Lambda(h)\ket{\psi_{\Lambda}},
    &&\qquad \forall\,h\in\mathfrak h,
    \\
    \pi(e_\alpha)\ket{\psi_{\Lambda}}
    &= 0,
    &&\qquad \forall\,\alpha\in\Phi^+.
\end{alignat}
\end{mydef}

\noindent
An irreducible representation generated by a highest weight vector $\ket{\psi_{\Lambda}}$ is denoted by $V(\Lambda)$ and is called the highest weight representation with highest weight $\Lambda$.

To illustrate the abstract highest weight construction in a familiar setting, we begin with $\mathfrak{sl}(2, \mathbb C)$.
Its spin-$j$ irreducible representation admits the standard weight basis
\begin{align}
\{\ket{j,m} \,|\, m=-j,-j+1,\ldots,j-1,j\}.
\end{align}
The Cartan generator acts diagonally,
\begin{align}
h\ket{j,m}
=
2m\ket{j,m},
\end{align}
so that $\ket{j,m}$ is a weight vector with weight $\mu_m$ satisfying
\begin{align}
    \mu_m(h)=2m.
\end{align}
The positive root generator $e_\alpha$ is identified with the angular momentum raising operator $J_+$,
\begin{align}
e_\alpha\ket{j,m}
=
J_+\ket{j,m}
=
\sqrt{(j-m)(j+m+1)}\,\ket{j,m+1}.
\end{align}
At the upper end of the weight string,
\begin{align}
e_\alpha\ket{j,j}
=
J_+\ket{j,j}
=
0,
\end{align}
which identifies $\ket{j,j}$ as the highest weight vector. 
The corresponding highest weight $\Lambda$ is therefore specified by
\begin{align}
\Lambda(h)=2j.
\end{align}
Thus, the spin-$j$ irreducible representation is the highest weight representation
\begin{align}
    V_j\coloneqq V(\Lambda)=V(2j),
\end{align}
with
\begin{align}
\dim V_j=2j+1.
\end{align}
Here $V(2j)$ uses the eigenvalue of the Chevalley generator $h=2J_z$ as the highest-weight label.

For finite-dimensional representations of complex semisimple Lie algebras, irreducible representations are uniquely determined by their highest weights.
Every weight $\mu$ of $V(\Lambda)$ is obtained from the highest weight by subtracting a nonnegative integer combination of positive roots:
\begin{align}
    \mu
    =
    \Lambda
    -
    \sum_{\alpha\in\Phi^+}
    m_\alpha\alpha,
    \quad
    m_\alpha\in\mathbb Z_{\geqslant0}.
\end{align}
This structure reflects the successive action of the lowering operators
$f_\alpha$ in the Chevalley basis (see Def.~\ref{def:Chevalley_Basis}).
Thus the highest weight $\Lambda$ fixes the full weight space structure of the representation.
A useful property of the highest weight is stated below.

\begin{mylem}{Extremal Characterization of the Highest Weight}{Extremal_Characterization_Highest_Weight}
Let $V(\Lambda)$ be the finite-dimensional irreducible highest weight representation of $\g$ associated with the positive root system $\Phi^+$, and let $\Lambda$ be its highest weight.
Then every weight $\mu\in V(\Lambda)$ satisfies
\begin{align}
\label{eq:norm_inequality}
    \langle\mu,\mu\rangle 
    \leqslant
    \langle\Lambda,\Lambda\rangle,
\end{align}
where $\langle\cdot,\cdot\rangle$ denotes the inner product on the real root space $E$ (see Def.~\ref{def:Inner_Product_Root_Space}).
Moreover, equality holds precisely for the extremal weights, namely those lying in the Weyl group orbit of $\Lambda$.
\end{mylem}

This is a standard consequence of the weight-polytope
description; see \cite{Hall2015}.
The choice of $\Phi^+$ singles out $\Lambda$ as the unique
dominant weight among the weights of maximal norm.
This extremal property gives the state-independent uncertainty bound below.
The compact group orbit of a highest weight vector consists of states attaining that bound.


\section{Lie Algebraic Quantum Uncertainty}
\label{sec:Lie_Algebraic_Quantum_Uncertainty}

Uncertainty is not only a property of individual observables, but a structural feature of the algebra that organizes them. 
This section develops a general framework for Lie algebraic uncertainty relations, showing how representation theory turns quantum uncertainties into an intrinsic algebraic quantity. 
In Subsec.~\ref{subsec:Quantum_Uncertainty_Algebraic_Operators}, we introduce the definition of Lie algebraic quantum uncertainty for a compact semisimple Lie algebra. 
In Subsec.~\ref{subsec:Casimir_Eigenvalue}, we show that the quadratic Casimir element determines the state-independent contribution to this uncertainty, thereby isolating the part fixed entirely by the representation. 
In Subsec.~\ref{subsec:Quantum_Uncertainty_Weight_Spaces}, we specialize the framework to weight spaces of an irreducible representation and recover the angular momentum uncertainty relation as the simplest nontrivial instance, where the spin quantum number and magnetic quantum number together determine the uncertainty in each weight state. 
In Subsec.~\ref{subsec:Quantum_Uncertainty_Beyond_Weight_Spaces}, we extend the analysis from weight states to arbitrary states in the full representation space by introducing the moment element, whose norm captures the state-dependent contribution and reveals the underlying structure of the bound. 
In Subsec.~\ref{subsec:Three_Bases_One_Uncertainty}, we prove that the quantum uncertainty is independent of the chosen presentation of the algebra: 
the Chevalley basis, the anti-Hermitian basis of the compact real form, and the corresponding Hermitian observable basis all lead to the same quantity. 
In Subsec.~\ref{subsec:Dynkin_Cartan}, we introduce Dynkin coefficients and the Cartan matrix, which provide intrinsic coordinates for the highest weight and the root system geometry.
In Subsec.~\ref{subsec:Dynkin_Diagram}, we clarify the correspondence between the Dynkin diagram and the Cartan matrix, whose entries encode the diagram's bonds and arrow orientations. 
Finally, in Subsec.~\ref{subsubsec:exceptional_Formula_Bound}, these ingredients are combined to obtain an explicit and computable formula for the optimal Lie algebraic uncertainty bound.


\subsection{Quantum Uncertainty for Algebraic Operators}
\label{subsec:Quantum_Uncertainty_Algebraic_Operators}

For a physical observable, represented by a Hermitian operator, uncertainty is usually quantified by its variance, leading to the Robertson uncertainty relation in quantum mechanics~\cite{PhysRev.34.163}.
The situation is less immediate for non-Hermitian operators, such as the anti-Hermitian generators (see Def.~\ref{def:Anti_Hermitian_Basis}) forming a basis of the compact real form of a Lie algebra $\g_{\mathbb R}$ (see Eq.~\eqref{eq:complexification}). 
In that case, one must first clarify what it means to assign an uncertainty to an abstract algebraic generator, and how this notion relates to the usual quantum-mechanical variance of observables. 
This subsection addresses this point by extending the variance from Hermitian observables to anti-Hermitian operators, thereby connecting the algebraic structure of $\g_{\mathbb R}$ with measurable quantum uncertainty.

Let $\pi:\g_{\mathbb R}\to\mathfrak{gl}(V)$ be a finite-dimensional unitary representation of the compact real Lie algebra $\g_{\mathbb R}$, with representation space $V$. 
A pure quantum state is represented by a normalized vector $\ket{\psi}\in V$. 
For any generator $X\in\g_{\mathbb R}$, the corresponding representation $\pi(X)$ is anti-Hermitian (see Lem.~\ref{lem:Anti_Hermitian}), and its expectation value in the state $\psi$ is defined as
\begin{align}
    \langle X\rangle_{\psi}\coloneqq
    \bra{\psi}\pi(X)\ket{\psi}.
\end{align}
Since $\pi(X)^\dagger=-\pi(X)$, this expectation value is purely imaginary.
For notational simplicity, we suppress the representation map $\pi$ inside expectation values: whenever an element $X\in\g_{\mathbb R}$ appears between state vectors, it is understood as its representation $\pi(X)$. 
This convention will be used throughout and should cause no ambiguity.
The variance associated with the compact real Lie algebra element $X$ is therefore defined by taking the positive quadratic uncertainty of the anti-Hermitian generator.

\begin{mydef}{Variance for Anti-Hermitian Operators}{Variance_Anti_Hermitian_Operators}
For an anti-Hermitian operator $X\in\g_{\mathbb R}$, its variance in the quantum state $\psi$ is defined by 
\begin{align}\label{eq:non_hermitian_variance}
    \Delta_{\psi}^2(X)
    \coloneqq
    \langle X^\dagger X\rangle_{\psi}
    -
    \langle X^\dagger\rangle_{\psi}
    \langle X\rangle_{\psi}.
\end{align}
\end{mydef}

A key property supporting its interpretation as a measure of quantum uncertainty is its non-negativity.

\begin{mylem}{Non-Negativity}{Non_Negativity}
The variance introduced in Def.~\ref{def:Variance_Anti_Hermitian_Operators} is non-negative for any anti-Hermitian operator $X\in\g_{\mathbb R}$ and any quantum state $\psi$:
\begin{align}
    \Delta_{\psi}^2(X)\geqslant0.
\end{align}
\end{mylem}

\begin{proof}
We begin the proof by rewriting the variance defined in Def.~\ref{def:Variance_Anti_Hermitian_Operators} as follows 
\begin{align}
    \Delta_{\psi}^2(X)
    &=
    \langle X^\dagger X\rangle_{\psi}-\langle X^\dagger\rangle_{\psi}\langle X\rangle_{\psi}\\
    &=
    \bra{\psi}\bigl(\pi(X)-\langle X\rangle_{\psi}\bigr)^\dagger\bigl(\pi(X)-\langle X\rangle_{\psi}\bigr)\ket{\psi}\\
    &=
    \left\|\bigl(\pi(X)-\langle X\rangle_{\psi}\bigr)\ket{\psi}\right\|^2\\
    &\geqslant0.
\end{align}
The variance $\Delta_{\psi}^2(X)$ is therefore non-negative, as claimed.
\end{proof}

Another useful property clarifies why this generalized variance for anti-Hermitian operators is experimentally accessible.

\begin{mylem}{Variance Equivalence Under Hermitianization}{Variance_Equivalence_Under_Hermitianization}
For any anti-Hermitian generator $\pi(X)$, multiplication by $\mathrm{i}$ produces the corresponding Hermitian observable $\mathrm{i}\pi(X)$.
The generalized variance (see Def.~\ref{def:Variance_Anti_Hermitian_Operators}) of $\pi(X)$ therefore coincides exactly with the ordinary quantum variance of the Hermitian observable $\mathrm{i}\pi(X)$; namely
\begin{align}
    \Delta_{\psi}^2(X)=
    \Delta_{\psi}^2(\pi(X))=
    \Delta_{\psi}^2(\mathrm{i}\pi(X))=
    \Delta_{\psi}^2(\mathrm{i}X).
\end{align}
\end{mylem}

Consequently, the Lie algebraic uncertainty relations introduced below have a direct physical interpretation: 
they are equivalent to uncertainty relations for the Hermitian observables obtained by Hermitianizing the corresponding anti-Hermitian generators $\{H_i,E_\alpha,F_\alpha\}$ (see Def.~\ref{def:Anti_Hermitian_Basis}).

\begin{mydef}{Lie Algebraic Quantum Uncertainty}{Lie_Algebraic_Uncertainty_Relations}
Let $\g_{\mathbb R}$ be a compact real semisimple Lie algebra (see Eq.~\eqref{eq:complexification}) equipped with an anti-Hermitian basis $\{H_i,E_\alpha,F_\alpha\}$ (see Def.~\ref{def:Anti_Hermitian_Basis}), and let $\ket{\psi}$ be a quantum state in a unitary representation space $V$ of $\pi:\g_{\mathbb R}\to\mathfrak{gl}(V)$. 
We fix the normalization of the basis by requiring that, for any $Q_a, Q_b\in\{H_i,E_\alpha,F_\alpha\}$, $(Q_a,Q_b)=2\delta_{ab}$.
Here $(\cdot,\cdot)$ is the inner product defined in Def.~\ref{def:Inner_product_alg}.
The Lie algebra quantum uncertainty associated with $\g_{\mathbb R}$ is defined as follows
\begin{align}
\label{eq:QUncertainty_Lie_Algebra}
    \Delta_{\psi}^2(\g) \coloneqq \frac12\Biggl(\sum_{i=1}^r\Delta_{\psi}^2(H_i)
    +\sum_{\alpha\in\Phi^+}
    \bigl(\Delta_{\psi}^2(E_\alpha)
    +\Delta_{\psi}^2(F_\alpha)\bigr)\Biggr).
\end{align}
Here, $r$ denotes the dimension of the Cartan subalgebra $\mathfrak h$, equivalently the rank of $\g_{\mathbb R}$, i.e., $r=\dim_{\mathbb C}{\mathfrak h}=\mathrm{rank} \g$, and $\Phi^+$ denotes a fixed choice of positive roots. 
The variances are not properties of the abstract Lie algebra alone. They are computed only after a representation $\pi$ has been chosen, and hence depend on the concrete action of the generators on the space $V$.
This formulation contains, as its simplest instance, $\mathfrak{sl}(2,\mathbb C)$, the complexification of $\mathfrak{su}(2)$ (see Eq.~\eqref{eq:Complexification_su_2}), where it recovers the conventional uncertainty relations of angular momentum.
When no ambiguity can arise, we will also write the uncertainty $\Delta_{\psi}^2(\g)$ as $\Delta_{\psi}^2(\g_{\mathbb R})$, to emphasize its association with the compact real form;
for example, $\Delta_{\psi}^2(\mathfrak{sl}(2,\mathbb C))=\Delta_{\psi}^2(\mathfrak{su}(2))$.
\end{mydef}

To connect Lie algebraic uncertainty (see Def.~\ref{def:Lie_Algebraic_Uncertainty_Relations}) with representation-theoretic invariants, we introduce the quadratic Casimir element of the compact real form $\g_{\mathbb R}$, denoted by $C_2$.
Let
\begin{align}\label{eq:Set_X_a}
\{X_a\}
\coloneqq
\{
H_i,E_\alpha,F_\alpha 
\,|\,
i=1,\ldots,r,\ \alpha\in\Phi^+
\},
\end{align}
be a basis of $\g_{\mathbb R}$, and let $\{X^a\}$ denote its dual with respect to the normalized form $\kappa(\cdot,\cdot)$ (see Eq.~\eqref{eq:normal_Killing}), defined by
\begin{align}\label{eq:Dual_Basis}
    \kappa(X_a,X^b)=\delta_{ab}.
\end{align}
Because  $\kappa$ is negative definite on $\g_{\mathbb R}$, we use throughout the positive definite inner product $(x,y)=-\kappa(x,y)$, as specified in Def.~\ref{def:Inner_product_alg}.
The unit normalization in Eq.~\eqref{eq:Dual_Basis} is introduced solely to define the $\kappa$-dual basis entering $C_2$. 
It is therefore distinct from the normalization used for the generators in the Lie algebraic uncertainty relation, for which the normalization is fixed to 2 throughout. 
With these conventions in place, the quadratic Casimir element $C_2$ takes the following form.

\begin{mydef}{Quadratic Casimir Element}{Quadratic_Casimir_Element}
The quadratic Casimir element $C_2$ associated with the chosen generators $\{X_a\}$ of $\g_{\mathbb R}$ and the normalized invariant bilinear form $\kappa$ in Eq.~\eqref{eq:normal_Killing} is defined by
\begin{align}\label{eq:Quadratic_Casimir_Element}
    C_2\coloneqq \sum_a X_a X^a,
\end{align}
where $\{X^a\}$ denotes the dual basis determined by Eq.~\eqref{eq:Dual_Basis}.
\end{mydef}

The basis independence of the Casimir element $C_2$, established below, will be central to its explicit construction.

\begin{mylem}{Basis Independence of the Quadratic Casimir Element}{Basis_Independence_Casimir_Element}
The construction of the quadratic Casimir element $C_2$ defined in Eq.~\eqref{eq:Quadratic_Casimir_Element} is independent of the choice of basis $\{X_a\}$ of $\g_{\mathbb R}$.
\end{mylem}

We now specialize to the anti-Hermitian basis
$\{H_i,E_\alpha,F_\alpha\}$ of the compact real Lie algebra $\g_{\mathbb R}$, introduced in Def.~\ref{def:Anti_Hermitian_Basis}.
Consistent with the normalization adopted for the Lie algebraic uncertainty relation, these generators form an orthogonal basis with squared norm 2 under the positive definite invariant inner product $(\cdot, \cdot)$, as established in Eqs.~\eqref{eq:Inner},~\eqref{eq:Inner2}, and~\eqref{eq:Inner3}.
Since $(x,y)=-\kappa(x,y)$, their normalization with respect to $\kappa$ is
\begin{align}
\kappa(H_i,H_j)&=-2\delta_{ij},\\
\kappa(E_\alpha,E_\beta)&=-2\delta_{\alpha\beta},\\
\kappa(F_\alpha,F_\beta)&=-2\delta_{\alpha\beta}.
\end{align}
The associated dual generators therefore take the particularly simple form
\begin{align}\label{eq:Killing_dual_generators}
H^i
=
-\frac12H_i,
\quad
E^\alpha
=
-\frac12E_\alpha,
\quad
F^\alpha
=
-\frac12F_\alpha.
\end{align}
Substitution into the quadratic Casimir element then gives
\begin{align}
    C_2&
    =\sum_{i=1}^rH_i\cdot H^i+\sum_{\alpha\in \Phi^{+}}\left(E_\alpha \cdot E^{\alpha}+F_{\alpha}\cdot F^{\alpha}\right)\\
    &=\sum_{i=1}^rH_i\cdot \left(-\frac12H_i\right)+\sum_{\alpha\in \Phi^{+}}\left(E_\alpha \cdot \left(-\frac12E_\alpha\right)+F_{\alpha}\cdot \left(-\frac12F_\alpha\right)\right)\\
    &=-\frac12 \left( \sum_{i=1}^r H_i^2
    +\sum_{\alpha\in\Phi^+} \left(E_\alpha^2
    + F_\alpha^2 \right)\right),
\end{align}
The resulting expression directly relates the quadratic Casimir element to the complete set of generators entering the Lie algebraic uncertainty relation, leading to the following theorem.

\begin{mythm}{Explicit Form of the Quadratic Casimir Element}{Explicit_Form_Casimir_Element}
The quadratic Casimir element $C_2$ defined in Eq.~\eqref{eq:Quadratic_Casimir_Element}, using the normalized form $\kappa$ of Eq.~\eqref{eq:normal_Killing}, takes the form

\begin{align}\label{eq:Quadratic_Casimir_Element_Explicit}
    C_2
    =-\frac12 \Biggl( \sum_{i=1}^r H_i^2
    +\sum_{\alpha\in\Phi^+} \bigl(E_\alpha^2
    + F_\alpha^2 \bigr)\Biggr),
\end{align}
where $\{H_i,E_\alpha,F_\alpha\}$ denotes the anti-Hermitian basis of $\g_{\mathbb R}$ (see Def.~\ref{def:Anti_Hermitian_Basis}).
\end{mythm}

In the associated Hermitian observable basis
$\{\mathrm{i}H_i,\, \mathrm{i}E_\alpha,\, \mathrm{i}F_\alpha\}$, this becomes
\begin{align}\label{eq:Quadratic_Casimir_Element_hermitian}
C_2
=
\frac12
\left[
\sum_{i=1}^{r}\left(\mathrm{i}H_i\right)^2
+
\sum_{\alpha\in\Phi^+}
\left(
\left(\mathrm{i}E_\alpha\right)^2+\left(\mathrm{i}F_\alpha\right)^2
\right)
\right].
\end{align}

By the definition of the variance for anti-Hermitian operators
(see Def.~\ref{def:Variance_Anti_Hermitian_Operators}) and the $\kappa$-form
definition of the quadratic Casimir element $C_2$
(see Def.~\ref{def:Quadratic_Casimir_Element}), a  straightforward calculation then shows that the total quantum uncertainty $\Delta_{\psi}^2(\g)$ (see Def.~\ref{def:Lie_Algebraic_Uncertainty_Relations}) associated with the Lie algebra $\g_{\mathbb R}$ can be expressed in terms of $C_2$ as follows.

\begin{mycor}{Quantum Uncertainty in Terms of the Casimir Element}{Quantum_Uncertainty_Casimir_Element}
The quantum uncertainty (see Def.~\ref{def:Lie_Algebraic_Uncertainty_Relations}) associated with the Lie algebra $\g_{\mathbb R}$ can be equivalently expressed as
\begin{align}\label{eq:casimir_variance}
    \Delta_{\psi}^2(\g)
    =\langle C_2\rangle_{\psi}
    + \frac12\Biggl(\sum_{i=1}^r\langle H_i\rangle_{\psi}^2
    +\sum_{\alpha\in\Phi^+} \bigl(\langle E_\alpha\rangle_{\psi}^2
    +\langle F_\alpha\rangle_{\psi}^2
    \bigr)
    \Biggr),
\end{align}
where $C_2$ denotes the Casimir element defined in Eq.~\eqref{eq:Quadratic_Casimir_Element}, whose explicit form is given in Eq.~\eqref{eq:Quadratic_Casimir_Element_Explicit}.
\end{mycor}

We conclude this subsection by noting that the Lie algebraic quantum uncertainty (see Def.~\ref{def:Lie_Algebraic_Uncertainty_Relations}) has so far been formulated with respect to an orthogonal anti-Hermitian basis $\{H_i,E_\alpha,F_\alpha\}$ (see Def.~\ref{def:Anti_Hermitian_Basis}) of the compact real form $\g_{\mathbb R}$. 
Lemma~\ref{lem:Variance_Equivalence_Under_Hermitianization} establishes the equivalent formulation for the associated Hermitian observables $\{\mathrm{i}H_i, \mathrm{i}E_\alpha, \mathrm{i}F_\alpha\}$. 
This construction does not, however, apply directly to a Chevalley basis $\{ h_i, e_\alpha, f_\alpha\}$ (see Def.~\ref{def:Chevalley_Basis}) of $\g$, because the Chevalley basis need not be orthogonal for the normalized form $\kappa(\cdot,\cdot)$ in Eq.~\eqref{eq:normal_Killing}. 
Subsec.~\ref{subsec:Three_Bases_One_Uncertainty} expresses the same uncertainty in those coordinates.


\subsection{Casimir Eigenvalue}
\label{subsec:Casimir_Eigenvalue}

Although the Casimir element has a nontrivial algebraic construction, its role becomes remarkably simple in an irreducible highest weight representation:
it acts as a scalar.
This subsection identifies that scalar explicitly and prepares the ground for the uncertainty bounds developed in this work.

Let $V(\Lambda)$ be an irreducible highest weight representation (see Subsec.~\ref{subsec:Highest_Weight_Representations}) with highest weight vector $\ket{\psi_\Lambda}$. 
Then, for any element of the Cartan subalgebra $h_i\in\mathfrak h$, one has
\begin{align}
    h_i \ket{\psi_\Lambda}
    =
    \Lambda(h_i)\ket{\psi_\Lambda},
    \quad
    i\in\{1,\ldots,r\}.
\end{align}
For every positive root $\alpha\in\Phi^+$, the corresponding raising operator $e_\alpha$ (see Def.~\ref{def:Chevalley_Basis}) annihilates the highest weight vector $\ket{\psi_\Lambda}$.
\begin{align}
    e_\alpha \ket{\psi_\Lambda}=0,
    \quad
    \alpha\in\Phi^+.
\end{align}

For the Chevalley generators
$\{h_i,e_\alpha,f_\alpha\}$ (see Def.~\ref{def:Chevalley_Basis}),
the generators are generally not mutually orthogonal for the normalized invariant form $\kappa(\cdot,\cdot)$ in Eq.~\eqref{eq:normal_Killing}.
Indeed, using
 $[e_\alpha,f_\alpha]=h_\alpha$,
we have
\begin{align}
    \kappa(h_\alpha,h)
    =
    \kappa([e_\alpha,f_\alpha],h)
    =
    \kappa(e_\alpha,[f_\alpha,h]).
\end{align}
Using
\begin{align}
    [h,f_\alpha]
    =
    -\alpha(h)f_\alpha,
\end{align}
or equivalently
\begin{align}
    [f_\alpha,h]
    =
    \alpha(h)f_\alpha,
\end{align}
we obtain
\begin{align}
    \kappa(h_\alpha,h)
    =
    \alpha(h)\kappa(e_\alpha,f_\alpha).
\label{eq:ha_h_relation_1}
\end{align}

On the other hand, according to the definition of the coroot in Eq.~\eqref{eq:coroots}, i.e., $h_\alpha=2t_\alpha/\langle\alpha,\alpha\rangle$, the pairing $\kappa(h_\alpha,h)$ can be rewritten as
\begin{align}
    \kappa(h_\alpha,h)
    =
    \frac{2}{\langle\alpha,\alpha\rangle}
    \kappa(t_\alpha,h)
    =
    \frac{2}{\langle\alpha,\alpha\rangle}
    \alpha(h),
\label{eq:ha_h_relation_2}
\end{align}
where the last equality follows from Eq.~\eqref{eq:roots_action_KF}.

Comparing Eq.~\eqref{eq:ha_h_relation_1} and Eq.~\eqref{eq:ha_h_relation_2}, and choosing $h\in\mathfrak h$ with $\alpha(h)\neq0$, gives
\begin{align}
    \kappa(e_\alpha,f_\alpha)
    =
    \frac{2}{\langle\alpha,\alpha\rangle}.
\end{align}
Thus the $\kappa$-dual of $e_\alpha$ is
\begin{align}
    e^\alpha
    =
    \frac{\langle\alpha,\alpha\rangle}{2}    f_\alpha.
\end{align}

By Lem.~\ref{lem:Basis_Independence_Casimir_Element}, the quadratic Casimir element $C_2$, defined in Def.~\ref{def:Quadratic_Casimir_Element}, is basis independent. 
Evaluating it in the Chevalley basis of Def.~\ref{def:Chevalley_Basis} yields

\begin{mythm}{Alternative Form of the Quadratic Casimir Element}{Alternative_Form_Casimir_Element}
Let $\g$ be a complex semisimple Lie algebra of rank $r$, with Cartan
subalgebra $\mathfrak h$, positive root system $\Phi^+$, and normalized invariant form $\kappa$ (Eq.~\eqref{eq:normal_Killing}). 
For the Chevalley basis $\{h_i, e_{\alpha}, f_{\alpha}\}$ of Def.~\ref{def:Chevalley_Basis}, the quadratic Casimir element takes the form
\begin{align}
    C_2
    =
    \sum_{i=1}^r h_ih^i
    +
    \sum_{\alpha\in\Phi^+} \frac{
    \langle\alpha,\alpha\rangle}{2}
    \bigl(e_\alpha f_\alpha+f_\alpha e_\alpha\bigr),
\end{align}
where $\{h^1,\dots,h^r\}$ is the $\kappa$-dual basis of $\{h_1,\dots,h_r\}$. 

\end{mythm}

We now evaluate the quadratic Casimir element $C_2$
on the highest weight vector $\ket{\psi_\Lambda}\in V_\Lambda$.
Since $C_2$ belongs to the center of the universal enveloping algebra, Schur’s lemma implies that it acts as a scalar on the irreducible representation $V(\Lambda)$.
It is therefore enough to compute this scalar on the highest weight state.
Separating $C_2$ into Cartan and root contributions, the Cartan part acts as 
\begin{align}
    \sum_{i=1}^r h_i h^i \ket{\psi_\Lambda}
    =
    \sum_{i=1}^r \Lambda(h^i)\Lambda(h_i)\ket{\psi_\Lambda}
    =
    \langle\Lambda,\Lambda\rangle\ket{\psi_\Lambda},
\end{align}
where the final equality follows from Lem.~\ref{lem:Dual_Basis_Contraction}. This completely determines the Cartan contribution $\sum_{i=1}^r h_ih^i$ to the quadratic Casimir element $C_2$ (see Thm.~\ref{thm:Alternative_Form_Casimir_Element}).

For the root contribution, the highest weight state is annihilated by all positive root generators, i.e., 
\begin{align}
    e_\alpha \ket{\psi_\Lambda}=0.
\end{align}
Hence, $f_\alpha e_\alpha \ket{\psi_\Lambda}=0$.
To evaluate the action of $e_\alpha f_\alpha$ on the highest weight state $\ket{\psi_\Lambda}$, we use the commutation relation $[e_\alpha,f_\alpha]=h_\alpha$ together with the following identity
\begin{align}
    e_\alpha f_\alpha \ket{\psi_\Lambda}
    =
    \bigl([e_\alpha,f_\alpha]+f_\alpha e_\alpha\bigr)
    \ket{\psi_\Lambda}
    =
    h_\alpha \ket{\psi_\Lambda}
    =
    \Lambda(h_\alpha)\ket{\psi_\Lambda}.
\end{align}
This leads to
\begin{align}
\sum_{\alpha\in\Phi^+}\frac{\langle\alpha,\alpha\rangle}{2}
\bigl(e_\alpha f_\alpha+f_\alpha e_\alpha\bigr)
\ket{\psi_\Lambda}
=
\left(\sum_{\alpha\in\Phi^+}\frac{\langle\alpha,\alpha\rangle}{2}\Lambda(h_\alpha)\right)
\ket{\psi_\Lambda}.
\end{align}
Using the identification $\Lambda(h_\alpha)=2\langle\Lambda,\alpha\rangle/\langle\alpha,\alpha\rangle$ and the Weyl vector $\delta=\sum_{\alpha\in\Phi^+}\alpha/2$ (see Def.~\ref{def:Weyl_Vector}),
we obtain
\begin{align}
    \sum_{\alpha\in\Phi^+} \frac{  \langle\alpha,\alpha\rangle}{2   }\Lambda(h_\alpha)
    = 
    \sum_{\alpha\in\Phi^+}\langle\Lambda,\alpha\rangle
    =
    2\langle\Lambda,\delta\rangle.
\end{align}
Combining the Cartan and root contributions, the quadratic Casimir element $C_2$ acts on the highest weight state $\ket{\psi_\Lambda}$ as
\begin{align}
    C_2 \ket{\psi_\Lambda}
    =
    \bigl(\langle\Lambda,\Lambda\rangle
    + 2\langle\Lambda,\delta\rangle\bigr)
    \ket{\psi_\Lambda}
    = 
    \langle\Lambda,\Lambda+2\delta\rangle \ket{\psi_\Lambda}.
\end{align}
The result can be summarized as follows.

\begin{mythm}{Value of the Quadratic Casimir Element}{Value_Casimir_Element}
Let $V(\Lambda)$ be an irreducible highest weight representation of a semisimple Lie algebra $\g$, with highest weight $\Lambda$. 
The quadratic Casimir element $C_2$ (see Eq.~\eqref{eq:Quadratic_Casimir_Element}) acts on $V(\Lambda)$ by scalar multiplication; that is
\begin{align}
    C_2
    =
    \langle\Lambda,\Lambda+2\delta\rangle\1_{V(\Lambda)},
\end{align}
where $\delta=\sum_{\alpha\in\Phi^+}\alpha/2$ is the Weyl vector (see Def.~\ref{def:Weyl_Vector}). 
Consequently, for any quantum state $\ket{\psi}\in V(\Lambda)$, the expectation value of $C_2$ is
\begin{align}
\label{eq:exp_Casimir}
    \langle C_2\rangle_{\psi}
    =
    \langle\Lambda,\Lambda+2\delta\rangle.
\end{align}
\end{mythm}

Although the computation is most transparent in the Chevalley basis (see Def.~\ref{def:Chevalley_Basis}), the resulting eigenvalue is intrinsic.
The Chevalley basis (see Thm.~\ref{thm:Alternative_Form_Casimir_Element}) and anti-Hermitian (see Thm.~\ref{thm:Explicit_Form_Casimir_Element}) representations are simply two descriptions of the same quadratic Casimir element $C_2$. 
Accordingly, the scalar $\langle\Lambda,\Lambda+2\delta\rangle$ is independent of the chosen generators.


\subsection{Quantum Uncertainty in Weight Spaces}\label{subsec:Quantum_Uncertainty_Weight_Spaces}

Given a highest weight representation $V(\Lambda)$ of the Lie algebra $\g$, the representation decomposes into weight spaces
\begin{align}\label{eq:weight_space_decomposition}
    V(\Lambda)
    =
    \bigoplus_{\mu\in\Omega(\Lambda)}V_{\Lambda,\mu},
\end{align}
where $\Omega(\Lambda)\subset\mathfrak h^*$ denotes the set of weights appearing in the representation $V(\Lambda)$.
This subsection develops uncertainty relations for general highest weight representations, with particular emphasis on the structure of states belonging to a fixed weight space $V_{\Lambda,\mu}$.

For any state $\ket{\psi}$ in the representation space $V(\Lambda)$, Cor.~\ref{cor:Quantum_Uncertainty_Casimir_Element} gives
\begin{align}
\label{eq:decomposition}
    \Delta_{\psi}^2(\g)
    = 
    \langle C_2\rangle_{\psi}
    + 
    \frac12 
    \left( \sum_{i=1}^r\langle H_i\rangle_{\psi}^2
    + \sum_{\alpha\in\Phi^+}\bigl(\langle E_\alpha\rangle_{\psi}^2
    + \langle F_\alpha\rangle_{\psi}^2\bigr)\right).
\end{align}
The root sector contribution can be rewritten directly in the Chevalley basis (see Def.~\ref{def:Chevalley_Basis}). 
Using
\begin{align}
    E_\alpha
    =
    -\mathrm{i}\sqrt{\frac{\langle \alpha,\alpha\rangle}{2}}(e_\alpha+f_\alpha),\qquad
    F_\alpha
    =
    -\sqrt{\frac{\langle \alpha,\alpha\rangle}{2}}(e_\alpha-f_\alpha),
\end{align}
together with $e_\alpha^\dagger=f_\alpha$, define 
\begin{align}
    z_\alpha
    \coloneqq
    \langle e_\alpha\rangle_{\psi}.
\end{align}
Then
\begin{align}
    \langle f_\alpha\rangle_{\psi}
    =
    z_\alpha^*,
\end{align}
so that the expectation values of the generators are encoded by the real and imaginary parts of $z_\alpha$. 
Thus,
\begin{align}
\label{eq:EF_to_e}
\langle E_\alpha\rangle_{\psi}
    &= -\mathrm{i}\sqrt{\frac{\langle \alpha,\alpha\rangle}{2}}(z_\alpha+z_\alpha^*)
    =
   - 2\mathrm{i}\sqrt{\frac{\langle \alpha,\alpha\rangle}{2}}~\mathrm{Re}(z_\alpha),\\
\langle F_\alpha\rangle_{\psi}
    &=
   -\sqrt{\frac{\langle \alpha,\alpha\rangle}{2}}(z_\alpha-z_\alpha^*)
    =
    -2\mathrm{i}\sqrt{\frac{\langle \alpha,\alpha\rangle}{2}}~\mathrm{Im}(z_\alpha).
\end{align}
It follows that
\begin{align}\label{eq:Sum_E_F_alpha}
    \langle E_\alpha\rangle_{\psi}^2
    +\langle F_\alpha\rangle_{\psi}^2
    =
    -2\langle \alpha,\alpha\rangle\bigl| \langle e_\alpha\rangle_{\psi} \bigr|^2.
\end{align}

Substituting Eq.~\eqref{eq:Sum_E_F_alpha} into Eq.~\eqref{eq:decomposition} gives the following identity.

\begin{mylem}{Lie Algebraic Uncertainty Identity}{Lie_Algebraic_Uncertainty_Identity}
Let $V(\Lambda)$ be an irreducible highest weight representation (see Subsec.~\ref{subsec:Highest_Weight_Representations}) of Lie algebra $\g$. 
Then every state $\ket{\psi}\in V(\Lambda)$ satisfies the Lie algebraic uncertainty relation
\begin{align}
\label{eq:decomp_chevalley1}
    \Delta^2_{\psi}(\g)
    = 
    \underbrace{\langle C_2\rangle_{\psi}}_{\text{State-Independent}}
    + 
    \underbrace{\frac12 \sum_{i=1}^r \langle H_i\rangle_{\psi}^2
    - \sum_{\alpha\in\Phi^+} \langle \alpha,\alpha\rangle\bigl| \langle e_\alpha\rangle_{\psi}\bigr|^2}_{\text{State-Dependent}},
\end{align}
where $C_2$ is the quadratic Casimir element (see Def.~\ref{def:Quadratic_Casimir_Element}).
\end{mylem}

Equation~\eqref{eq:decomp_chevalley1} decomposes the Lie algebraic variance into an intrinsic, state-independent Casimir term and state-dependent contributions from the Cartan and positive root sectors. 
This separation exposes the representation-theoretic structure of the quantum uncertainty and provides the basis for deriving the uncertainty bound.

We now restrict to a fixed weight space. 
Let $\ket{\psi_{\mu}}\in V_{\Lambda,\mu}$ denote a simultaneous eigenstate of the Cartan generators, satisfying
\begin{align}
\label{eq:weight_state}
    H_i\ket{\psi_{\mu}}
    = 
    \mu(H_i)\ket{\psi_{\mu}},
    \quad
    i\in\{1,\ldots,r\}.
\end{align}
The analysis simplifies because each root generator connects distinct weight spaces. 
For any positive root $\alpha\in\Phi^+$ and any $h\in\mathfrak h$, the commutation relation $[h,e_\alpha]=\alpha(h)e_\alpha$ gives
\begin{align}
h(e_\alpha\ket{\psi_\mu})
&=
e_\alpha h\ket{\psi_\mu}
+
[h,e_\alpha] \ket{\psi_\mu}\\
&=
(\mu(h)+\alpha(h))e_\alpha \, \ket{\psi_\mu}.
\end{align}
Hence, whenever $e_\alpha\ket{\psi_\mu}\neq0$, the resulting state belongs to the shifted weight space
\begin{align}
    e_\alpha\ket{\psi_\mu}
    \in 
    V_{\Lambda,\mu+\alpha}.
\end{align}
Because distinct weight spaces are orthogonal in a unitary representation, $e_\alpha|\psi_{\mu}\rangle$ is orthogonal to $|\psi_{\mu}\rangle$.
It follows that
\begin{align}
\label{eq:weight_vanish}
    \langle\psi_{\mu}|e_\alpha|\psi_{\mu}\rangle=0,
    \quad
    \forall\,\alpha\in\Phi^+.
\end{align}
Thus, all positive root expectation values, i.e., $\langle e_\alpha\rangle_{\psi_{\mu}}$, vanish on weight states.
Substituting Eq.~\eqref{eq:weight_vanish} into the Lie algebraic quantum uncertainty decomposition in Eq.~\eqref{eq:decomp_chevalley1} yields

\begin{mylem}{Lie Algebraic Uncertainty Identity in Weight Spaces}{Lie_Algebraic_Uncertainty_Identity_Weight}
Let $V_{\Lambda,\mu}$ be a weight space of the Lie algebra $\g$ (see Thm.~\ref{thm:Weight_Space_Decomposition}). 
The Lie algebraic quantum uncertainty of any state $\ket{\psi_{\mu}}\in V_{\Lambda,\mu}$ separates into a weight-independent Casimir contribution and weight-dependent Cartan terms determined by $\mu$,
\begin{align}
\label{eq:decomp_chevalley2}
    \Delta^2_{\psi_{\mu}}(\g)
    = 
    \underbrace{\langle C_2\rangle_{\psi_{\mu}}}_{\text{Weight-Independent}}
    + 
    \underbrace{\frac12 \sum_{i=1}^r \mu(H_i)^2}_{\text{Weight-Dependent}},
\end{align}
Here $C_2$ denotes the quadratic Casimir element of $\g$, as defined in Def.~\ref{def:Quadratic_Casimir_Element}.
\end{mylem}

We choose the anti-Hermitian Cartan generators
$\{H_i\}_{i=1}^r$
to be orthogonal with respect to the compact inner product
$(x,y)=-\kappa(x,y)$,
with normalization
$(H_i,H_j) = 2\delta_{ij}$.
Equivalently,
\begin{align}
    \kappa(H_i,H_j)
    =
    -2\delta_{ij}.
\end{align}
It follows that the $\kappa$-dual basis is
\begin{align}
    H^i
    =
    -\frac12 H_i.
\end{align}
Then 
\begin{align}
    \kappa(H_i,H^j)
    =
    \kappa\left(H_i,-\frac12 H_j\right)
    =
    \delta_{ij}.
\end{align}

Applying Lem.~\ref{lem:Dual_Basis_Contraction} to the Cartan basis $\{H_i\}_{i=1}^r$
and its Killing-dual basis $\{H^i\}_{i=1}^r$, we obtain
\begin{align}
    \langle\mu,\mu\rangle
    =
    \sum_{i=1}^r
    \mu(H_i)\mu(H^i),
\end{align}
and hence
\begin{align}
       \frac12 \sum_{i=1}^r \mu(H_i)^2 
    =
    -\sum_{i=1}^r \mu(H_i)\mu\!\left(-\frac12H_i\right) 
    = 
    -
    \sum_{i=1}^r
    \mu(H_i)\mu(H^i)
    =
     -\langle\mu,\mu\rangle.
\end{align}
Therefore, for any quantum state $\ket{\psi_\mu}\in V_{\Lambda,\mu}$, the Lie algebraic uncertainty relation reduces to

\begin{mythm}{Lie Algebraic Uncertainty Identity in Weight Spaces}{Lie_Algebraic_Uncertainty_Identity_Weight}
Let $V_{\Lambda,\mu}$ be a weight space of the irreducible highest weight representation $V(\Lambda)$ of the Lie algebra $\g$, as in Thm.~\ref{thm:Weight_Space_Decomposition}.
Then, for any state $\ket{\psi_{\mu}}\in V_{\Lambda,\mu}$, the Lie algebraic quantum uncertainty reduces to a Casimir contribution and a weight dependent term determined by the inner product of $\mu$,
\begin{align}\label{eq:decomp_weight_final}
    \Delta_{\psi_{\mu}}^2(\g)
    = 
    \underbrace{\langle C_2\rangle_{\psi_\mu}}_{\text{Weight-Independent}}
    - 
    \underbrace{\langle\mu,\mu\rangle\, ,}_{\text{Weight-Dependent}}
\end{align}
where $C_2$ is the quadratic Casimir element of $\g$, defined in Def.~\ref{def:Quadratic_Casimir_Element}.
\end{mythm}

Every weight $\mu$ appearing in the weight space decomposition of $V(\Lambda)$ obeys the inequality $\langle\mu,\mu\rangle\leqslant\langle\Lambda,\Lambda\rangle$ (see Lem.~\ref{lem:Extremal_Characterization_Highest_Weight}), with equality if and only if $\mu$ lies in the Weyl group orbit of the highest weight $\Lambda$; 
equivalently, $\Lambda$ is the unique dominant weight in the Weyl orbit $W\Lambda$. 
Combining this bound with the weight space uncertainty identity in Eq.~\eqref{eq:decomp_weight_final} and the Casimir eigenvalue relation in Eq.~\eqref{eq:exp_Casimir} gives the following uncertainty bound for weight states.

\begin{mycor}{Lie Algebraic Uncertainty Bound in Weight Spaces}{Weight_Bound}
Let $V_{\Lambda, \mu}\subset V(\Lambda)$ be a weight space of an irreducible highest weight representation with highest weight $\Lambda$
(see Thm.~\ref{thm:Weight_Space_Decomposition}).
For every state $\ket{\psi_{\mu}}\in V_{\Lambda,\mu}$,
\begin{align}
    \Delta_{\psi_{\mu}}^2(\g)
    \geqslant
    \underbrace{2\langle\Lambda,\delta\rangle \, .}_{\text{Weight-Independent}}
\end{align}
Equality holds if and only if the weight $\mu$ lies in the Weyl group orbit of the highest weight $\Lambda$.
\end{mycor}

A careful reader may observe that Cor.~\ref{cor:Weight_Bound} already provides an uncertainty relation with a state-independent lower bound. 
This naturally raises the question of whether the bound $2\langle\Lambda,\delta\rangle$ derived for individual weight spaces $V_{\Lambda, \mu}$ can be promoted directly to the full representation space $V(\Lambda)$. 
In general, the weight-space argument alone does not establish this extension.
Cor.~\ref{cor:Weight_Bound} applies only to states within a single weight space. 
If $\psi_1 \in V_{\Lambda, \mu_1}$ and $\psi_2 \in V_{\Lambda, \mu_2}$ with $\mu_1\neq\mu_2$, then any nontrivial superposition of the two is not itself a weight state and lies in neither weight space individually. 
Thus, the weight space bound alone does not establish a state-independent Lie algebraic uncertainty relation on the full representation space, or equivalently, for all quantum states. 
We therefore provide an explicit moment element argument for its extension to all states, complementing the general bound discussed in~\cite{deGuise2018}.

\subsection{Quantum Uncertainty Beyond Weight Spaces}\label{subsec:Quantum_Uncertainty_Beyond_Weight_Spaces}
The uncertainty relations derived in Subsec.~\ref{subsec:Quantum_Uncertainty_Weight_Spaces} apply directly to states supported within a single weight space $V_{\Lambda, \mu}$, where the Cartan generators are diagonal and the Lie algebraic uncertainty reduces to a simply expression. 
A generic state in an irreducible representation, however, may coherently combine different weight sectors, so the weight space argument no longer applies directly to the expectation value terms in the variance.
To obtain a bound on the full representation space, we introduce the moment element $M_\psi\in\g_{\mathbb R}$ (see Def.~\ref{def:Moment_Element}), which packages the expectation values of all Lie algebra generators into a single geometric object. 
Its norm is universally bounded by the norm of the highest weight $\Lambda$, imposing a representation-theoretic limit on the maximal value attainable by any state in $V(\Lambda)$. 
This highest weight constraint turns the variance decomposition into a state-independent Lie algebraic uncertainty bound for arbitrary quantum states.

Let $\g_{\mathbb R}$ be a compact real Lie algebra with anti-Hermitian generators $\{H_i,E_\alpha,F_\alpha\}$ as introduced in Def.~\ref{def:Anti_Hermitian_Basis}. 
We choose this basis such that, for any two generators $Q_a, Q_b\in \{H_i,E_\alpha,F_\alpha\}$,
\begin{align}\label{eq:normalization_Q}
    (Q_a,Q_b)=2\delta_{ab} .
\end{align}
Here $(\cdot,\cdot)$ denotes the compact inner product introduced in Def.~\ref{def:Inner_product_alg}, with equivalent formulations given in Eqs.~\eqref{eq:Inner}-\eqref{eq:Inner3}.
For a quantum state $\ket{\psi}\in V(\Lambda)$, the expectation values $ \langle Q_a\rangle_\psi$ are purely imaginary, since the generators $Q_a$ are anti-Hermitian. 
Hence
\begin{align}\label{eq:m_a_psi}
    m_a(\psi)
    \coloneqq
    \langle \mathrm{i}Q_a\rangle_\psi
    =
    \mathrm{i}\langle Q_a\rangle_\psi
\end{align}
are real numbers.
These expectation values define the moment element.

\begin{mydef}{Moment Element on $\g_{\mathbb R}$}{Moment_Element}
Let $\ket{\psi}\in V(\Lambda)$ be a quantum state, and let $\{Q_a\}_{a}=\{H_i,E_\alpha,F_\alpha\}$ denote the anti-Hermitian generator basis (see Def.~\ref{def:Anti_Hermitian_Basis}) of the compact real Lie algebra $\g_{\mathbb R}$.
Using the real coefficients $m_a(\psi)$ defined in Eq.~\eqref{eq:m_a_psi}, we define the associated moment element by
\begin{align}
    M_{\psi} \coloneqq  \frac12 \sum_a m_a(\psi)\,Q_a \in \g_{\mathbb R} .
\end{align}
\end{mydef}

The normalization in Eq.~\eqref{eq:normalization_Q} fixes the norm of $M_{\psi}$ as follows
\begin{align}
\label{eq:moment_norm}
    \|M_{\psi}\|^2
    =
    (M_{\psi}, M_{\psi})
    =
    \frac12 \sum_a m_a(\psi)^2
    =
    -\frac12 \sum_a\langle Q_a\rangle^2_{\psi} \,.
\end{align}
Equivalently, resolving the sum into Cartan and root sectors gives
\begin{align}
    \|M_{\psi}\|^2
    = 
    -\frac12 \Biggl[\sum_{i=1}^r \langle H_i\rangle_{\psi}^2
    + 
    \sum_{\alpha\in\Phi^+} \Bigl( \langle E_\alpha\rangle_{\psi}^2
    +
    \langle F_\alpha\rangle_{\psi}^2\Bigr) \Biggr] .
\end{align}
The Casimir eigenvalue of $C_2$ turns the Lie algebraic uncertainty relation in Eq.~\eqref{eq:casimir_variance} into the equivalent form

\begin{mylem}{Quantum Uncertainty in Terms of the Moment Element}{Quantum_Uncertainty_Moment_Element}
For the compact real Lie algebra $\g_{\mathbb R}$, the Lie algebraic quantum uncertainty of Def.~\ref{def:Lie_Algebraic_Uncertainty_Relations} admits the following equivalent expression
\begin{align}
\label{eq:variance_moment}
    \Delta_{\psi}^2(\g)
    = \langle\Lambda,\Lambda+2\delta\rangle
    - \|M_{\psi}\|^2 .
\end{align}
Here $M_{\psi}$ is the moment element associated with the state $\psi$ (see Def.~\ref{def:Moment_Element}), and $\delta$ denotes the Weyl vector (see Def.~\ref{def:Weyl_Vector}).
This identity holds for every quantum state $\ket{\psi}\in V(\Lambda)$.
\end{mylem}

Thus, the passage from weight-space uncertainty relations (see Cor.~\ref{cor:Weight_Bound}) to a bound valid for arbitrary states is reduced to the moment bound below.

\begin{mylem}{Moment-Norm Bound}{Moment_Norm_Bound}
Let $V(\Lambda)$ be a finite-dimensional highest weight representation of $\g$, with highest weight $\Lambda$.
Then, for every state $\ket{\psi}\in V(\Lambda)$, the associated moment element $M_{\psi}$ (see Def.~\ref{def:Moment_Element}) satisfies
\begin{align}
\label{eq:moment_norm_bound}
    \|M_{\psi}\|^2
      \leqslant
    \langle\Lambda,\Lambda\rangle .
\end{align}
\end{mylem}

\begin{proof}
For an arbitrary element
\begin{align}
    X=\sum_a x_a Q_a\in \g_{\mathbb R} ,
\end{align}
its inner product (see Def.~\ref{def:Inner_product_alg}) with the moment element $M_{\psi}$ is
\begin{align}
    (M_{\psi},X)
    &=\left(\frac12\sum_a m_a(\psi) Q_a,\sum_b x_b Q_b \right) \\
    &= \sum_a x_a m_a(\psi)\\
    &= \sum_a x_a\langle \mathrm{i} Q_a\rangle_\psi\\
    &=\langle \mathrm{i}X\rangle_{\psi} .
\end{align}
Thus, with respect to the inner product $(\cdot,\cdot)$ in Def.~\ref{def:Inner_product_alg}, $M_{\psi}$ represents the real linear functional
\begin{align}
     M_{\psi}: X\mapsto \langle \mathrm{i}X\rangle_\psi
\end{align}
on $\g_{\mathbb R}$. 
This allows us to express the vector norm $\|M_{\psi}\|$ as the norm of the linear functional represented by $M_{\psi}$ as
(we still use the same symbol)
\begin{align}
\label{eq:moment_dual_norm}
    \|M_{\psi}\|
    \coloneqq
    \frac{1}{\sqrt2}
    \sup_{(X,X)=2}
    |\langle \mathrm{i}X\rangle_{\psi}| ,
\end{align}
where the supremum is over normalized anti-Hermitian operators $X\in\g_{\mathbb R}$ with $(X,X)=2$.
Physically, $\sqrt{2}\|M_{\psi}\|$ represents the largest abolute expectation value attainable by the quantum state $\psi$ over all observables $\mathrm{i}X$ in the Lie algebra subject to $(X,X)=2$.

It remains to bound $|\langle \mathrm{i}X\rangle_\psi|$ for elements $X\in\g_{\mathbb R}$ with $(X,X)=2$. 
Let $G$ be the connected, simply connected compact Lie group with Lie algebra $\mathfrak g_{\mathbb R}$, and let $\mathfrak t=\operatorname{span}_{\mathbb R}\{H_i\}=-\mathrm{i}\mathfrak h_{\mathbb R}$ be its compact Cartan subalgebra. 
Every  $X\in \g_{\mathbb R}$ is conjugate under $G$ to some $A\in \mathfrak{t}$:  there exist $g\in G$ and $A\in \mathfrak{t}$ such that 
\begin{align}\label{eq:Conjugation}
    X=\mathrm{Ad}_g A=gAg^{-1} .
\end{align}
The adjoint invariance of the inner product $(\cdot,\cdot)$ ensures that the conjugation in Eq.~\eqref{eq:Conjugation} preserves the norm, namely
\begin{align}\label{eq:A_X_2}
    (A,A)=(X,X)=2 .
\end{align}
It therefore suffices to bound the expectation of normalized elements of $\mathfrak{t}=-\mathrm{i}\mathfrak{h}_{\mathbb R}$ uniformly over all states in $V(\Lambda)$.
  
In the representation $V(\Lambda)$, the adjoint conjugation $X=\mathrm{Ad}_g A$ is implemented by the corresponding unitary action of the compact group:
\begin{align}\label{eq:pi_A_X_commutative}
    \pi(X)
    =
    U_g\,\pi(A)\,U_g^\dagger .
\end{align}
where $U_g$ denotes the unitary operator representing $g$ on $V(\Lambda)$. 
This equivalence (see Eq.~\eqref{eq:pi_A_X_commutative}) is captured by the following commutative diagram:
\begin{equation}
\begin{tikzcd}[row sep=huge, column sep=huge]
V(\Lambda) \arrow[r, "\pi(A)"] \arrow[d, "U_g"'] & 
V(\Lambda) \arrow[d, "U_g"] \\
V(\Lambda) \arrow[r, "\pi(X)"'] & 
V(\Lambda)
\end{tikzcd}
\end{equation}
Hence $\pi(X)$ and $\pi(A)$ are unitarily equivalent and, in particular, have the same spectrum.
Equivalently, by introducing the transformed state
\begin{align}
    \ket{\phi}
    \coloneqq
    U_g^\dagger \ket{\psi},
\end{align}
we obtain
\begin{align}
    \langle \mathrm{i}X\rangle_\psi
    =
    \bra{\psi}\mathrm{i}\pi(X)\ket{\psi}
    =
    \bra{\phi}\mathrm{i}\pi(A)\ket{\phi}
    =
    \langle \mathrm{i}A\rangle_\phi .
\end{align}
We can therefore use the Cartan spectrum to bound the expectation in the transformed state $\phi$. Since $A\in \mathfrak{t}$, it acts on $V_{\Lambda,\mu}$ by $\mu(A)$, and $\mathrm{i}A$ has eigenvalue $\mathrm{i}\mu(A)$. 
The spectral bound on expectation values gives
\begin{align}
\label{eq:spectral_bound}
    |\langle \mathrm{i}X\rangle_\psi|
    =
    |\langle \mathrm{i}A\rangle_\phi|
    \leqslant
    \max_{\mu\in\Omega(\Lambda)}
    |\mu(A)| .
\end{align}
Write
\begin{align}
A=\sum_{i=1}^r a_iH_i
\end{align}
in a Cartan basis normalized as $(H_i,H_j)=2\delta_{ij}$.
The constraint $(A,A)=2$ (see Eq.~\eqref{eq:A_X_2}) is then equivalent to
\begin{align}
    \sum_{i=1}^r a_i^2=1.
\end{align}
Therefore, for any weight $\mu$,
\begin{align} \label{eq:maximal_mu}
    |\mu(A)|^2
    =
    \left| \sum_{i=1}^r a_i\mu(H_i)\right|^2
    \leqslant
    \sum_{i=1}^r |\mu(H_i)|^2 .
\end{align}
The last step is the Cauchy-Schwarz inequality.
The right-hand side of Eq.~\eqref{eq:maximal_mu} is twice the squared weight norm.   
For the anti-Hermitian Cartan generators,
\begin{align}
    \frac12\sum_{i=1}^r |\mu(H_i)|^2 
    &=
    \frac12\sum_{i=1}^r \overline{\mu(H_i)}\,\mu(H_i) \\
    &=
    \frac12 \sum_{i=1}^r \mu(H_i^\dagger)\mu(H_i) \\
    &=
    -\frac12\sum_{i=1}^r \mu(H_i)\mu(H_i) \\
    &=
    \sum_{i=1}^r\mu(H_i)\mu(H^i) \\
    &=
    \langle\mu,\mu\rangle .
\end{align}
Hence
\begin{align}
    |\mu(A)|
      \leqslant 
    \sqrt{2\langle\mu,\mu\rangle}\, .
\end{align}
For an irreducible highest weight representation $V(\Lambda)$, every weight $\mu$ satisfies $\langle\mu,\mu\rangle\leqslant\langle\Lambda,\Lambda\rangle$.
Combining these gives the expectation value bound
\begin{align}
    |\langle \mathrm{i}X\rangle_\psi| 
    =
    |\langle \mathrm{i}A\rangle_\phi|
      \leqslant 
    \sqrt{2\langle\Lambda,\Lambda\rangle} .
\end{align}
Taking the supremum over all normalized Lie algebra elements $X$ in Eq.~\eqref{eq:moment_dual_norm} therefore yields
\begin{align}
    \|M_\psi\|
      \leqslant 
    \sqrt{\langle\Lambda,\Lambda\rangle} \,,
\end{align}
or equivalently,
\begin{align}
    \|M_\psi\|^2
      \leqslant 
    \langle\Lambda,\Lambda\rangle\, .
\end{align}
This establishes the desired bound on the moment element $M_{\psi}$ and completes the proof.
\end{proof}

We now extend the uncertainty relations derived in Subsec.~\ref{subsec:Quantum_Uncertainty_Weight_Spaces} for individual weight spaces $V_{\Lambda,\mu}$ to the full representation space $V(\Lambda)$, thereby establishing a state-independent bound valid for all quantum states.

\begin{mythm}{Lie Algebraic Uncertainty Relations}{Lie_Algebraic_Uncertainty_Relations}
Let $V(\Lambda)$ be an irreducible highest weight representation of the compact semisimple Lie algebra $\g_{\mathbb R}$.
The Lie algebraic uncertainty is bounded from below across the entire representation space: 
for every quantum state $\ket{\psi}\in V(\Lambda)$ 
\begin{align}
\label{eq:general_uncertainty_bound}
    \Delta_\psi^2(\g)
     \geqslant 2\langle\Lambda,\delta\rangle \,,
\end{align}
where $\delta$ denotes the Weyl vector (see Def.~\ref{def:Weyl_Vector}), defined as half the sum of all positive roots.
\end{mythm}

\begin{proof}
Substituting the moment-norm estimate of Eq.~\eqref{eq:moment_norm_bound} into the uncertainty identity Eq.~\eqref{eq:variance_moment}, one finds
\begin{align}
    \Delta_\psi^2(\g)
    = 
    \langle\Lambda,\Lambda+2\delta\rangle
    - \|M_{\psi}\|^2 \geqslant
    \langle\Lambda,\Lambda+2\delta\rangle
    - \langle\Lambda,\Lambda\rangle
    =
    2\langle\Lambda,\delta\rangle .
\end{align}
The bound is independent of $\psi$, and therefore holds uniformly over the full representation space $V(\Lambda)$, completing the proof.
\end{proof}
The same lower bound holds for mixed states.
For $\rho=\sum_jp_j|\psi_j\rangle\langle\psi_j|$, variance concavity gives $\Delta_\rho^2(\g)\geqslant\sum_jp_j\Delta_{\psi_j}^2(\g)\geqslant2\langle\Lambda,\delta\rangle$. 
Here expectations and variances are evaluated using $\Tr[\rho\, \cdot]$.

Before closing this subsection, we show how the Lie algebraic uncertainty relation of Def.~\ref{def:Lie_Algebraic_Uncertainty_Relations} recovers the familiar angular momentum uncertainty relation as its simplest instance.
Consider the Lie algebra of type $A_1$, namely $\mathfrak{sl}(2,\mathbb C)$; 
see Tab.~\ref{tab:Simple_Lie_Algebra}.
Its compact real form is $\mathfrak{su}(2)$, generated by the anti-Hermitian operators $\{T_x, T_y, T_z\}$ introduced in Eq.~\eqref{eq:T_xyz}:
\begin{align}
    T_x=-\frac{\mathrm{i}}{2}\sigma_x=-\mathrm{i} J_x ,
    \quad
    T_y=-\frac{\mathrm{i}}{2}\sigma_y=-\mathrm{i} J_y ,
    \quad
    T_z=-\frac{\mathrm{i}}{2}\sigma_z=-\mathrm{i} J_z .
\end{align}
These generators follow the standard physical normalization of angular momentum operators. 
In the present work, however, the Lie algebraic quantum uncertainty (see Def.~\ref{def:Lie_Algebraic_Uncertainty_Relations}) is defined using generators normalized by the inner product $(\cdot, \cdot)$ (see Def.~\ref{def:Inner_product_alg}).
We therefore introduce the rescaled anti-Hermitian generators $\{Q_x, Q_y, Q_z\}$:
\begin{align}\label{eq:su2_anti_hermitian_generators}
    Q_x\coloneqq 2T_x=-2\mathrm{i}J_x ,
    \quad
    Q_y\coloneqq 2T_y=-2\mathrm{i}J_y ,
    \quad
    Q_z\coloneqq 2T_z=-2\mathrm{i}J_z ,
\end{align}
which satisfy the normalization condition $(Q_a,Q_b)=2\delta_{ab}$ in Eq.~\eqref{eq:normalization_Q}.
Under this normalization, the general Lie algebraic quantum uncertainty of Def.~\ref{def:Lie_Algebraic_Uncertainty_Relations} reduces exactly to the familiar angular momentum uncertainty relation,
\begin{align}
    \Delta_{\psi}^2(\mathfrak{su}(2))
    &=\frac{1}{2}\left(\Delta_{\psi}^2(Q_z)
    +
   \Delta_{\psi}^2(Q_x)
    +
    \Delta_{\psi}^2(Q_y)\right)  \\
    &=
    2\left(
    \underbrace{\Delta_{\psi}^2(-\mathrm{i} J_z)}_{\text{Cartan Sector}}
    +
    \underbrace{\Delta_{\psi}^2(-\mathrm{i} J_x)
    +
    \Delta_{\psi}^2(-\mathrm{i} J_y)}_{\text{Root Sector}}
    \right)    \label{eq:su_2_anti_Hermitian_basis}\\
    &=
   2\left(
    \Delta_{\psi}^2(\mathrm{i}\cdot(-\mathrm{i} J_z))
    +
    \Delta_{\psi}^2(\mathrm{i}\cdot(-\mathrm{i} J_x))
    +
    \Delta_{\psi}^2(\mathrm{i}\cdot(-\mathrm{i} J_y))
    \right)\label{eq:su_2_Hermitian_basis}\\
    &=
   2\left(
    \Delta_{\psi}^2(J_z)
    +
    \Delta_{\psi}^2(J_x)
    +
    \Delta_{\psi}^2(J_y)
    \right) ,
\end{align}
where, in going from the anti-Hermitian formulation in Eq.~\eqref{eq:su_2_anti_Hermitian_basis} to the Hermitian observable formulation in Eq.~\eqref{eq:su_2_Hermitian_basis}, we have used the Hermitianization identity established in Lem.~\ref{lem:Variance_Equivalence_Under_Hermitianization}.
This analysis shows that, up to the scalar factor $2$, angular momentum uncertainty in finite-dimensional quantum systems is determined by the Lie algebraic structure of $\mathfrak{sl}(2,\mathbb C)$.

We further illustrate the $\mathfrak{sl}(2,\mathbb C)$ case by evaluating explicitly the right-hand side of Eq.~\eqref{eq:variance_moment}, namely $\langle\Lambda,\Lambda+2\delta\rangle-\|M_{\psi}\|^2$.
This calculation exposes the precise decomposition of the Lie algebraic quantum uncertainty into a Casimir term, fixed by the representation, and a moment term, fixed by the state.
Recall that $\mathfrak{sl}(2,\mathbb C)$ has a single positive root, $\Phi^+=\{\alpha\}$, as in Eq.~\eqref{eq:sl_2_root_system}, with the standard normalization $\langle \alpha, \alpha\rangle=2$ introduced in Def.~\ref{def:Inner_Product_Root_Space}. 
Its Weyl vector is therefore $\delta=\alpha/2$, according to Def.~\ref{def:Weyl_Vector}. 
For the spin-$j$ representation $V_j$ of $\mathfrak{su}(2)$, the highest weight is
\begin{align}
    \Lambda=j\alpha.
\end{align} 
Accordingly, the Casimir term entering Eq.~\eqref{eq:variance_moment} takes the explicit form
\begin{align}
    \langle C_2\rangle_{\psi}
    =
    \langle\Lambda,\Lambda+2\delta\rangle
    =
    \left\langle j\alpha,j\alpha+\alpha \right\rangle
    =
    j(j+1)\langle\alpha,\alpha\rangle
    = 
    2j(j+1).
\end{align}
Thus, within the spin-$j$ representation, the Casimir term is fixed entirely by the representation label $j$ and is independent of the particular quantum state $\ket{\psi}$.

Before turning to the general moment contribution $\|M_{\psi}\|^2$ (see Eq.~\eqref{eq:variance_moment}), we first consider the case in which the relevant weight space is fixed (see Thm.~\ref{thm:Lie_Algebraic_Uncertainty_Identity_Weight}). 
Physically, this corresponds to knowing the magnetic quantum number of the state.
Let $|\psi_{\mu_m}\rangle$ be a state supported on the weight space $V_{j,m}$ defined in Eq.~\eqref{eq:V_j_m}, with $m\in\{-j,-j+1,\ldots,j\}$.
The associated weight is then $\mu_m=m\alpha$, and thus we have
\begin{align}
    \langle\mu_m,\mu_m\rangle
    =
    m^2 \langle\alpha,\alpha\rangle
    =
    2m^2.
\end{align}
At fixed weight, the total uncertainty is determined by an exact identity: 
the quantum uncertainty of the three angular momentum components is completely determined by the spin-$j$ and the magnetic quantum number $m$ (see Thm.~\ref{thm:Lie_Algebraic_Uncertainty_Identity_Weight}), as follows

\begin{mycor}{Uncertainty Identity for Spin-$j$ Uncertainty at Fixed Magnetic Quantum Number}{Conservation_Law_Magnetic_Quantum_Number}
In the spin-$j$ representation of $\mathfrak{su}(2)$, the total quantum uncertainty of a weight state $|\psi_{\mu_m}\rangle\in V_{j,m}$ is conserved at fixed magnetic quantum number $m$, and Eq.~\eqref{eq:decomp_weight_final} reduces to 
\begin{align}
    \Delta^2_{\psi_{\mu_m}}(\mathfrak{su}(2))
    =
    \underbrace{2j(j+1)}_{\text{Highest-Weight-Dependent}}
    -
    \underbrace{2m^2}_{\text{Weight-Dependent}},
\end{align}
Equivalently, in terms of the physical angular momentum observables, one obtains 
\begin{align}\label{eq:magnetic_uncertainty}
     \Delta_{\psi_{\mu_m}}^2(J_x) + \Delta_{\psi_{\mu_m}}^2(J_y) + \Delta_{\psi_{\mu_m}}^2(J_z)
    =
    \frac{1}{2}\Delta^2_{\psi_{\mu_m}}(\mathfrak{su}(2))
    =
    \underbrace{j(j+1)}_{\text{Highest-Weight-Dependent}}
    -
    \underbrace{m^2}_{\text{Weight-Dependent}}.
\end{align}
Thus, within each fixed weight sector, the total angular momentum uncertainty is fixed by the magnetic quantum number $m$. 
It attains its minimum at the extremal weights $m=\pm j$.
\end{mycor}

Returning to the identity in Eq.~\eqref{eq:variance_moment}, all state dependence is isolated in the moment element $M_{\psi}$ introduced in Def.~\ref{def:Moment_Element}. 
Its norm is characterized by the dual formulation in Eq.~\eqref{eq:moment_dual_norm}. To estimate this norm, we expand an arbitrary element of $\mathfrak{su}(2)$ in the normalized generator basis $\{Q_a\}_{a\in\{x,y,z\}}$ as
\begin{align}
    X
    =
    \sum_{k=x,y,z}
    x_kQ_k,
\end{align}
with $\sum_k x_k^2=1$.
Since $\mathrm{i}Q_z=2J_z$, the spectrum of $\mathrm{i}Q_z$ is $\{-2j,-2j+2,\ldots,2j-2,2j\}$.
For every state $\psi$, we now have
\begin{align}
    |\langle \mathrm{i}X\rangle_\psi|
    \leqslant
    2j,
\end{align}
and hence
\begin{align}
    \|M_\psi\|
    =
    \frac1{\sqrt2} \sup_{(X,X)=2}
    |\langle \mathrm{i}X\rangle_\psi|
    \leqslant
    \frac1{\sqrt2}\cdot  2j=\sqrt2 j.
\end{align}
Consequently,
\begin{align}
    \|M_\psi\|^2
    \leqslant
    2j^2=\langle\Lambda,\Lambda\rangle.
\end{align}
In this case, the right-hand side of Eq.~\eqref{eq:variance_moment} evaluates explicitly as
\begin{align}
    \langle\Lambda,\Lambda+2\delta\rangle
    - \|M_{\psi}\|^2
    \geqslant
    2j(j+1)-2j^2
    =
    2j.
\end{align}
This recovers the angular momentum uncertainty relation within the Lie algebraic framework.

\begin{mycor}{Angular Momentum Uncertainty Relation}{Angular_Momentum_Uncertainty_Relation}
In the spin-$j$ representation of $\mathfrak{su}(2)$, Eq.~\eqref{eq:variance_moment} gives the Lie algebraic uncertainty of a quantum state $\ket{\psi}\in V_{j}$ as
\begin{align}
    \Delta^2_{\psi}(\mathfrak{su}(2))
    \geqslant
    2j\,,
\end{align}
Equivalently, when expressed in terms of the physical angular momentum, this becomes 
\begin{align}\label{eq:angular_momentum_uncertainty}
     \Delta_{\psi}^2(J_x) + 
     \Delta_{\psi}^2(J_y) + 
     \Delta_{\psi}^2(J_z)
    =
    \frac{1}{2}\Delta^2_{\psi}(\mathfrak{su}(2))
    \geqslant
    j\,.
\end{align}
Thus, the optimal lower bound on the total angular momentum uncertainty is completely determined by the spin quantum number $j$.

\end{mycor}

Up to this point, we have seen that our general framework for Lie algebraic uncertainty relations, i.e., Eq.~\eqref{eq:general_uncertainty_bound}, recovers the angular momentum uncertainty relation (see Eq.~\eqref{eq:angular_momentum_uncertainty}) as the special case associated with representations of $\mathfrak{su}(2)$. 
Moreover, when the magnetic quantum number $m$ is specified, or equivalently the weight within the representation, the total uncertainty sum becomes a constant determined solely by the total spin and the magnetic quantum number (see Eq.~\eqref{eq:magnetic_uncertainty}). 

To gain further intuition and to illustrate how the general result extends beyond this familiar setting, we now turn to the next fundamental example: 
$\mathfrak{su}(3)$ (see Tab.~\ref{tab:Simple_Lie_Algebra}). This algebra occupies a central place in particle physics, both in the flavor classification of hadrons and in the color symmetry underlying the modern description of strong interactions.
Since $\mathfrak{su}(3)$ has rank two,  the real span of its weights is two-dimensional, with basis  $\omega_1$ and $\omega_2$ (see Def.~\ref{def:Fundamental_Weights}). 
Thus, for any irreducible representation, the highest weight $\Lambda$ can be expressed as
\begin{align}\label{eq:su_3_Highest_Weight}
    \Lambda=
    \lambda_1\omega_1+\lambda_2\omega_2.
\end{align}
Likewise, any weight $\mu$ appearing in this representation admits an expansion in the same basis,
\begin{align}\label{eq:su_3_Weight}
    \mu =m_1\omega_1+m_2\omega_2.
\end{align}

The discussion below uses the notion of fundamental weights, which is introduced formally only later in Def.~\ref{def:Fundamental_Weights} of Subsec.~\ref{subsec:Dynkin_Cartan}.
We nevertheless present the analysis at this point because the optimal bound for the $\mathfrak{su}(3)$ Lie algebraic uncertainty relation is most naturally understood as the next step beyond the $\mathfrak{su}(2)$ angular momentum case (see Cors.~\ref{cor:Conservation_Law_Magnetic_Quantum_Number} and~\ref{cor:Angular_Momentum_Uncertainty_Relation}).
The derivation itself does not require a detailed representation-theoretic understanding of fundamental weights. 
At this stage, they may be regarded simply as a convenient coordinate system for labeling highest weights. 
Their conceptual role will be developed in the later subsection, and readers may find it useful to revisit the present analysis after that discussion.

At fixed representation and weight, the Lie algebraic uncertainty is determined by the weight-space identity. 
This is summarized in the following corollary.

\begin{mycor}{Uncertainty Identity for $\mathfrak{su}(3)$ Uncertainty in Weight Spaces}{su3_Weight_Space_Uncertainty}
Let $V(\Lambda)$ be an irreducible $\mathfrak{su}(3)$ representation with highest weight $\Lambda$, and let $\mu$ denote a weight. 
With the expansions of $\Lambda$ and $\mu$ in the fundamental weight basis specified in Eqs.~\eqref{eq:su_3_Highest_Weight} and~\eqref{eq:su_3_Weight}, the corresponding Lie algebraic uncertainty relation on $V_{\Lambda, \mu}$ space reads
\begin{align}
\label{eq:su3_weight_conservation}
    \Delta_{\psi_\mu}^2(\mathfrak{su}(3))
    =
    \underbrace{
    \frac{2}{3}
    \left(\lambda_1^2+\lambda_1\lambda_2+\lambda_2^2 \right)
     +
    2(\lambda_1+\lambda_2)
    }_{\text{Highest-Weight-Dependent}}
    -
    \underbrace{
    \frac{2}{3} \left(m_1^2+m_1m_2+m_2^2\right)
    }_{\text{Weight-Dependent}}.
\end{align}
Thus, for a fixed highest weight $\Lambda$, the Lie algebraic uncertainty in each $\mathfrak{su}(3)$ weight sector $V_{\Lambda, \mu}$ is specified entirely by the two weight labels $m_1$ and $m_2$ (see Eq.~\eqref{eq:su_3_Weight}).
\end{mycor}

\begin{proof}
In the case of $\mathfrak{su}(3)$, the associated root system is of type $A_2$, with the corresponding Dynkin diagram shown below. 
\begin{equation}\label{eq:DynkinAtwonumber}
    \raisebox{-2.8ex}{\includegraphics[scale=1]{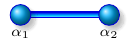}}\, .
\end{equation}
Denote the two simple roots by $\alpha_1$ and $\alpha_2$, as shown in Eq.~\eqref{eq:DynkinAtwonumber}, and choose the normalization 
\begin{align}\label{eq:}
    \langle\alpha_1,\alpha_1\rangle
    =
    \langle\alpha_2,\alpha_2\rangle
    =
    2.
\end{align}
With this convention, the inner product between distinct simple roots is given by
\begin{align}
    \langle\alpha_1,\alpha_2\rangle
    =
    -1.
\end{align}

The fundamental weights (see Def.~\ref{def:Fundamental_Weights}) are defined as the basis dual to the simple coroots. 
For $\mathfrak{su}(3)$, they are denoted by $\omega_1$, $\omega_2$, and satisfy the duality relations
\begin{align}
    \frac{
    2\langle\omega_i,\alpha_j\rangle
    }{
    \langle\alpha_j,\alpha_j\rangle
    }
    =
    \delta_{ij},
    \quad
    i,j\in\{1,2\}.
\end{align}
Under the normalization
\begin{align}
    \langle\alpha_j,\alpha_j\rangle=2,
\end{align}
the defining duality relations reduce to
\begin{align}
    \langle\omega_i,\alpha_j\rangle
    =
    \delta_{ij}.
\end{align}
To express the simple roots in the basis of fundamental weights, write
\begin{align}
    \alpha_1
    =
    a_{11}\omega_1+a_{12}\omega_2.
\end{align}
Pairing both sides with $\alpha_1$ and $\alpha_2$ then gives
\begin{align}
    a_{11}
    &=
    \langle\alpha_1,\alpha_1\rangle
    =
    2,\\
    a_{12}
    &=
    \langle\alpha_1,\alpha_2\rangle
    =
    -1.
\end{align}
Thus the first simple root $\alpha_1$ (see Eq.~\eqref{eq:DynkinAtwonumber}) has the expansion
\begin{align}
\label{eq:alpha1}
\alpha_1=2\omega_1-\omega_2,
\end{align}
in the basis of fundamental weights $\omega_1$ and $\omega_2$.

Likewise, expanding the second simple root $\alpha_2$ (see Eq.~\eqref{eq:DynkinAtwonumber}) as
\begin{align}
    \alpha_2
    =
    a_{21}\omega_1+a_{22}\omega_2,
\end{align}
and pairing with the two simple roots gives
\begin{align}
    a_{21}
    &=
    \langle\alpha_2,\alpha_1\rangle
    =-1,
    \\
    a_{22}
    &=
    \langle\alpha_2,\alpha_2\rangle
    =2.
\end{align}
Thus the second simple root $\alpha_2$ (see Eq.~\eqref{eq:DynkinAtwonumber}) has the expansion
\begin{align}
\label{eq:alpha2}
\alpha_2=-\omega_1+2\omega_2.
\end{align}

For $\mathfrak{su}(3)$, the positive roots are
\begin{align}
    \Phi^+
    =
    \{
    \alpha_1,
    \alpha_2,
    \alpha_1+\alpha_2
    \}.
\end{align}
The Weyl vector $\delta$ (see Def.~\ref{def:Weyl_Vector}) is therefore given by the half-sum of positive roots,
\begin{align}
    \delta
    =
    \frac12
    \sum_{\alpha\in\Phi^+}\alpha
    =
    \frac12
    \left(
    \alpha_1+\alpha_2+\left(\alpha_1+\alpha_2\right)
    \right)
    =
    \alpha_1+\alpha_2.
\end{align}
Using the expansions of the simple roots in the fundamental weights (see Eqs.~\eqref{eq:alpha1} and~\eqref{eq:alpha2}),
\begin{align}
\label{eq:alpha3}
    \alpha_1+\alpha_2
    =
    (2\omega_1-\omega_2)+(-\omega_1+2\omega_2)
    =
    \omega_1+\omega_2.
\end{align}
Hence the Weyl vector $\delta$ takes the simple form
\begin{align}
    \delta
    =
    \omega_1+\omega_2.
\end{align}

It remains to evaluate the inner products between the fundamental weights. 
The relations in Eqs.~\eqref{eq:alpha1} and \eqref{eq:alpha2} may be inverted to express the fundamental weights in the simple root basis:
\begin{align}
    \omega_1
    &=
    \frac23\alpha_1+\frac13\alpha_2,
    \\
    \omega_2
    &=
    \frac13\alpha_1+\frac23\alpha_2.
\end{align}
Using the root inner products, one obtains
\begin{align}
    \langle\omega_1,\omega_1\rangle
    =
    \langle\omega_2,\omega_2\rangle
    =
    \frac23,
\end{align}
and
\begin{align}
    \langle\omega_1,\omega_2\rangle
    =
    \frac13.
\end{align}

With these ingredients in place, the uncertainty identity in Thm.~\ref{thm:Lie_Algebraic_Uncertainty_Identity_Weight} can now be evaluated explicitly. 
The relevant quantity is $\langle C_2\rangle_{\psi_\mu}-\langle\mu,\mu\rangle$ (see Eq.~\eqref{eq:decomp_weight_final}).
The first term is fixed by the quadratic Casimir element $C_2$ of $\g$. Here $\g=\mathfrak{sl}(3,\mathbb C)$ with compact real form $\mathfrak{su}(3)$ (see Def.~\ref{def:Quadratic_Casimir_Element}). 
Let the highest weight $\Lambda$ of the irreducible representation be
\begin{align}
    \Lambda
    =
    \lambda_1\omega_1+\lambda_2\omega_2.
\end{align}
Then the quadratic Casimir element $C_2$ acts on $V(\Lambda)$ as multiplication by the scalar
\begin{align}
\label{eq:su3_casimir_value}
    \langle C_2\rangle\bigg|_{V(\Lambda)}
    =
    \langle\Lambda,\Lambda\rangle
    +
    2\langle\Lambda,\delta\rangle
    =
    \frac23
    \left(
    \lambda_1^2+\lambda_1\lambda_2+\lambda_2^2
    \right)
    +
    2(\lambda_1+\lambda_2).
\end{align}
Similarly, for a weight
\begin{align}
    \mu
    =
    m_1\omega_1+m_2\omega_2,
\end{align}
the norm is
\begin{align}
\label{eq:su3_weight_norm_explicit}
    \langle\mu,\mu\rangle
    =
    \frac23
    \left(
    m_1^2+m_1m_2+m_2^2
    \right).
\end{align}
Substituting Eqs.~\eqref{eq:su3_casimir_value}
and~\eqref{eq:su3_weight_norm_explicit} into Thm.~\ref{thm:Lie_Algebraic_Uncertainty_Identity_Weight} yields Eq.~\eqref{eq:su3_weight_conservation}, completing the proof.
\end{proof}

The corollary can now be evaluated in several standard classes of Lie algebraic representations. 
The first example is the defining representation (see Subsec.~\ref{subsec:Lie_Algebra}) of $\mathfrak{su}(3)$. 
In this representation, the highest weight is
\begin{align}
    \Lambda=\omega_1,
\end{align}
or equivalently
\begin{align}
    \lambda_1=1,\quad
    \lambda_2=0.
\end{align}
Substitution into Eq.~\eqref{eq:su3_casimir_value} gives
\begin{align}
    \langle C_2\rangle\bigg|_{V(\omega_1)}
    =
    \frac23+2
    =
    \frac83.
\end{align}

The weights of this representation are obtained by descending from the highest weight along the simple root directions. 
As illustrated below,
\begin{equation}\label{eq:WeightAtwoDef}
    \raisebox{-4ex}{\includegraphics[scale=1]{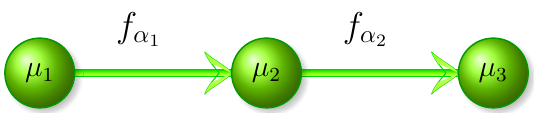}}\, ,
\end{equation}
starting from the highest weight node $\mu_1= \omega_1$ at the left end of the diagram, successive subtractions of the simple roots generate the full chain of weights:
\begin{align}
\mu_1&= \omega_1,
\\
\mu_2&=-\omega_1+\omega_2,
\\
\mu_3&= -\omega_2.
\end{align}
Therefore, in the fundamental weight coordinates, $\mu=m_1\omega_1+m_2\omega_2$,
the three weights of the defining representation are
\begin{align}
    &\mu_1: \quad (m_1,m_2)=(1,0),\\
    &\mu_2: \quad (m_1,m_2)=(-1,1),\\
    &\mu_3: \quad (m_1,m_2)=(0,-1).
\end{align}
For each of these weights, one has
\begin{align}
    m_1^2+m_1m_2+m_2^2
    =
    1.
\end{align}
Hence their norms are identical:
\begin{align}
    \langle\mu_i,\mu_i\rangle
    =
    \frac23,
    \quad
    \forall
    i\in\{1,2,3\}.
\end{align}
Substituting these values into Eq.~\eqref{eq:su3_weight_conservation} gives
\begin{align}
    \Delta_{\psi_{\mu_i}}^2(\mathfrak{su}(3))
    =
    \frac83-\frac23
    =
    2,
    \quad
    \forall
    i\in\{1,2,3\}.
\end{align}
Thus, all three weight-basis states $\psi_{\mu_i}$ in the defining representation of $\mathfrak{su}(3)$ carry the same Lie algebraic uncertainty. 
Equivalently, the fundamental qutrit representation exhibits a uniform uncertainty value across its entire weight basis.

\begin{mycor}{Uncertainty Identity in the Defining Representation of $\mathfrak{su}(3)$}{LAU_Conservation_def_su3}
For the defining representation of $\mathfrak{su}(3)$, whose weights $\mu_1$, $\mu_2$, $\mu_3$ are given in Eq.~\eqref{eq:WeightAtwoDef}, the Lie algebraic uncertainty relation restricted to each weight space $V_{\Lambda, \mu_i}$ takes the form
\begin{align}
\label{eq:LAU_Conservation_def_su3}
    \Delta_{\psi_{\mu_1}}^2(\mathfrak{su}(3))
    =
    \Delta_{\psi_{\mu_2}}^2(\mathfrak{su}(3))
    =
    \Delta_{\psi_{\mu_3}}^2(\mathfrak{su}(3))
    =
    2.
\end{align}
Here, the Lie algebraic uncertainty is conserved across the full weight basis state of the defining representation; 
it is independent of the weight label $i$.
\end{mycor}

Parallel to the angular momentum uncertainty relation obtained in Cor.~\ref{cor:Angular_Momentum_Uncertainty_Relation}, 
the $\mathfrak{su}(3)$ case also admits a representation-level uncertainty relation that goes beyond the weight space analysis. 
In particular, when specialized to the defining representation, Thm.~\ref{thm:Lie_Algebraic_Uncertainty_Relations} yields the following statement.

\begin{mycor}{Lie Algebraic Uncertainty in the Defining Representation of $\mathfrak{su}(3)$}{LAU_def_su3}
Consider the defining representation $V(\Lambda)$ of $\mathfrak{su}(3)$, whose highest weight is $\Lambda=\omega_1$.
For $\mathfrak{su}(3)$, the Weyl vector is $\delta=\omega_1+\omega_2$.
Then the Lie algebraic uncertainty is bounded from below over the entire representation space $V(\omega_1)$. 
Namely, for every normalized quantum state $\ket{\psi}\in V(\omega_1)$,
\begin{align}
\label{eq:LAU_def_su3}
    \Delta_{\psi}^2(\mathfrak{su}(3))
    \geqslant
    2.
\end{align}
\end{mycor}

\begin{proof}
In this case, the optimal lower bound takes the form
\begin{align}
    2\langle\Lambda,\delta\rangle
    &=
    2\langle\omega_1,\omega_1+\omega_2\rangle\\
    &=
    2\left(\langle\omega_1,\omega_1\rangle+\langle\omega_1,\omega_2\rangle\right)\\
    &=
    2\left(\frac{2}{3}+\frac{1}{3}\right)\\
    &=
    2.
\end{align}
This establishes the desired result.
\end{proof}
In the defining representation (see Eq.~\eqref{eq:def_rep}), all three weight states attain
the optimal value $2$.
In fact, every  pure state attains this value, since $SU(3)$ acts transitively on pure state rays.

The adjoint representation (see Eq.~\eqref{eq:ad_rep}) of $\mathfrak{su}(3)$ provides the next example. 
In this representation, the highest weight is the highest root,
\begin{align}
    \Lambda
    =
    \alpha_1+\alpha_2
    =
    \omega_1+\omega_2,
\end{align}
namely,
\begin{align}
    \lambda_1=\lambda_2=1.
\end{align}
Substituting these into Eq.~\eqref{eq:su3_casimir_value} yields
\begin{align}\label{eq:Casimir_su3_adj}
    \langle C_2\rangle\bigg|_{V(\omega_1+\omega_2)}
    =
    \frac23(1+1+1)+4
    =
    6.
\end{align}

The quadratic Casimir element $C_2$ of $\mathfrak{su}(3)$ above fixes the representation-dependent contribution to the uncertainty relation. 
It remains to identify the individual weight sectors and evaluate the corresponding weight norms. 
For the adjoint representation of $\mathfrak{su}(3)$, the weight diagram is particularly simple: 
the non-zero weights are precisely the six roots, while the zero weight occurs with multiplicity two. 
This structure is shown below, where each arrow denotes the action of a lowering operator. 
Acting with $f_{\alpha}$ on a vector of weight $\mu$ moves it to the weight sector of weight $\mu-\alpha$. 
\begin{equation}\label{eq:WeightAtwoAdj}
     \raisebox{-36ex}{\includegraphics[scale=1]{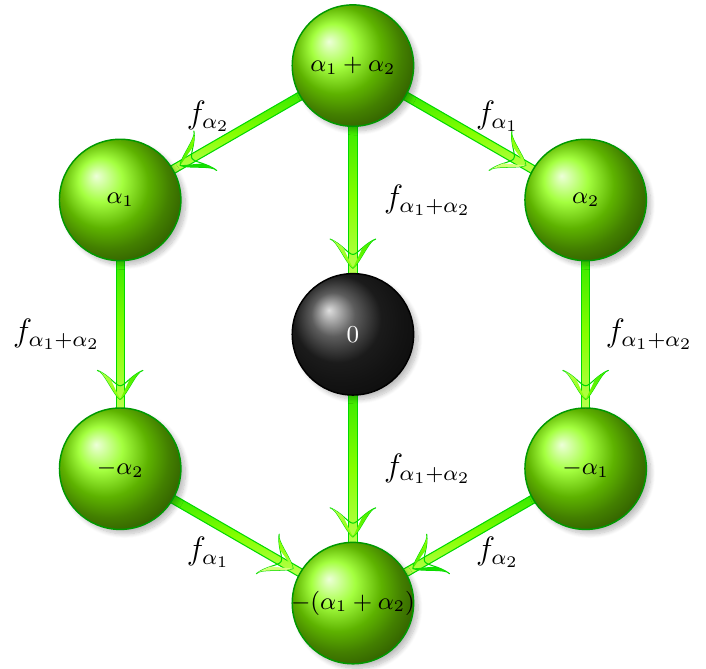}} \, .
\end{equation}
Starting from the highest weight
\begin{align}
    \Lambda=\alpha_1+\alpha_2,
\end{align}
the successive lowering operations generate the full weight diagram.
In summary, the adjoint representation of $\mathfrak{su}(3)$ has the weight space decomposition
\begin{align}
V(\omega_1+\omega_2)
=
\left(
\bigoplus_{\mu\in\Phi}
V_{\omega_1+\omega_2,\mu}
\right)
\oplus
V_{\omega_1+\omega_2,0}
=
\left(
\bigoplus_{\mu\in\Phi}
V_{\omega_1+\omega_2,\mu}
\right)
\oplus
\mathfrak h,
\end{align}
where the non-zero weights are precisely the roots
\begin{align}
\Phi
=
\{
\pm\alpha_1,
\pm\alpha_2,
\pm(\alpha_1+\alpha_2)
\}.
\end{align}

Here $\mathfrak h$ denotes the Cartan subalgebra of $\mathfrak{sl}(3,\mathbb C)$, identified with the zero-weight space of the complexified adjoint representation, which is two-dimensional:
\begin{align}
  \dim_{\mathbb C}\mathfrak h=2.
\end{align}
All remaining weight spaces, corresponding to the non-zero roots, are one-dimensional, namely
\begin{align}
    \dim V_{\omega_1+\omega_2,\mu}=1,\quad\forall\mu\in\Phi.
\end{align}

Any weight $\mu$ can be expressed in the basis of fundamental weights. 
For $\mathfrak{su}(3)$, this means writing $\mu=m_1\omega_1+m_2\omega_2$, where the coefficients $m_1$ and $m_2$ are obtained by expanding the corresponding root or weight in the fundamental weight basis, as shown below.
\begin{align}
    \Lambda=\alpha_1+\alpha_2=\omega_1+\omega_2
    \quad&\Rightarrow\quad
    (m_1,m_2)=(1,1),\\
    \alpha_1=2\omega_1-\omega_2
    \quad&\Rightarrow\quad
    (m_1,m_2)=(2,-1),\\
    \alpha_2=-\omega_1+2\omega_2
    \quad&\Rightarrow\quad
    (m_1,m_2)=(-1,2),\\
    -(\alpha_1+\alpha_2)=-(\omega_1+\omega_2)
    \quad&\Rightarrow\quad
    (m_1,m_2)=(-1,-1),\\
    -\alpha_1=-2\omega_1+\omega_2
    \quad&\Rightarrow\quad
    (m_1,m_2)=(-2,1),\\
    -\alpha_2=\omega_1-2\omega_2
    \quad&\Rightarrow\quad
    (m_1,m_2)=(1,-2),\\
    0=0\,\omega_1+0\,\omega_2
    \quad&\Rightarrow\quad
    (m_1,m_2)=(0,0).
\end{align}

It is now straightforward to verify that every non-zero weight $\mu$ of the adjoint representation satisfies
\begin{align}
    m_1^2+m_1m_2+m_2^2
    =
    3.
\end{align}
Equivalently,
\begin{align}
\langle\mu,\mu\rangle
=
\frac{2}{3}
\left(
m_1^2+m_1m_2+m_2^2
\right)
=
2,
\quad
\text{for}\,\,\mu\neq 0.
\end{align}
Since the quadratic Casimir element $C_2$ takes the value $6$ (see Eq.~\eqref{eq:Casimir_su3_adj}) on the adjoint representation, the Lie algebraic uncertainty in any non-zero weight sector is therefore (see Thm.~\ref{thm:Lie_Algebraic_Uncertainty_Identity_Weight})
\begin{align}
    \Delta_{\psi_\mu}^2(\mathfrak{su}(3))
    =
    \langle C_2\rangle_{\psi_\mu}
    -
    \langle\mu,\mu\rangle
    =
    6-2
    =
    4.
\end{align}
The preceding calculation yields the following corollary for the Lie algebraic uncertainty in every non-zero weight sector $V_{\omega_1+\omega_2,\mu}$.

\begin{mycor}{Uncertainty Identity in the Adjoint Representation of $\mathfrak{su}(3)$ ($\mu\neq0$)}{LAU_Conservation_adj_su3_n0}
Consider the adjoint representation $V(\omega_1+\omega_2)$ of $\mathfrak{su}(3)$, whose non-zero weights are the six roots $\mu\in\Phi=\{\pm\alpha_1, \pm\alpha_2, \pm(\alpha_1+\alpha_2)\}$. 
Then, for every non-zero weight sector $V_{\omega_1+\omega_2,\mu}$ with $\mu\neq0$, the Lie algebraic uncertainty takes the constant value
\begin{align}
    \Delta_{\psi_\mu}^2(\mathfrak{su}(3))=4.
\end{align}
Hence, the Lie algebraic uncertainty is conserved across all non-zero weight sectors of the adjoint representation; 
it is independent of the choice of weight $\mu\in\Phi$.
\end{mycor}

The zero-weight sector of the adjoint representation requires a separate treatment.
This sector is identified with the Cartan subalgebra $\mathfrak h$.
In fundamental weight coordinates, this corresponds to $m_1=m_2=0$, so that
\begin{align}
\langle\mu,\mu\rangle=0.
\end{align}
Together with the adjoint quadratic Casimir value $6$, this gives the Lie algebraic uncertainty in the zero-weight sector $\mathfrak h$ as
\begin{align}
    \Delta_{\psi_0}^2(\mathfrak{su}(3))=6.
\end{align}
This leads to the following corollary, which complements Cor.~\ref{cor:LAU_Conservation_adj_su3_n0} by isolating the contribution from the Cartan, or zero-weight, sector.

\begin{mycor}{Uncertainty Identity in the Adjoint Representation of $\mathfrak{su}(3)$ ($\mu=0$)}{LAU_Conservation_adj_su3_0}
Consider the adjoint representation $V(\omega_1+\omega_2)$ of $\mathfrak{su}(3)$.
Its zero-weight space is the Cartan subalgebra, i.e., $V_{\omega_1+\omega_2,0}=\mathfrak h$.
Then, for every state $\psi_0\in\mathfrak h$, the Lie algebraic uncertainty takes the value
\begin{align}
    \Delta_{\psi_0}^2(\mathfrak{su}(3))=6.
\end{align}
\end{mycor}

We close the discussion of the adjoint representation of $\mathfrak{su}(3)$ by moving beyond the weight space analysis (see Thm.~\ref{thm:Lie_Algebraic_Uncertainty_Identity_Weight}) and exploring the Lie algebraic uncertainty on the full representation space (see Thm.~\ref{thm:Lie_Algebraic_Uncertainty_Relations}).

\begin{mycor}{Lie Algebraic Uncertainty in the Adjoint Representation of $\mathfrak{su}(3)$}{LAU_adj_su3}
Let $V(\Lambda)$ be the adjoint representation of $\mathfrak{su}(3)$, with highest weight $\Lambda=\omega_1+\omega_2$ and Weyl vector $\delta=\omega_1+\omega_2$.
The Lie algebraic uncertainty is bounded from below on the entire representation space $V(\omega_1+\omega_2)$. 
That is, for every state $\ket{\psi}\in V(\omega_1+\omega_2)$
\begin{align}
\label{eq:LAU_adj_su3}
    \Delta_{\psi}^2(\mathfrak{su}(3))
    \geqslant
    4.
\end{align}
\end{mycor}

\begin{proof}
For the adjoint representation, the optimal lower bound is obtained as
\begin{align}
    2\langle\Lambda,\delta\rangle
    &=
    2\langle\omega_1+\omega_2,\omega_1+\omega_2\rangle \\
    &=
    2\Bigl(
        \langle\omega_1,\omega_1\rangle
        +
        \langle\omega_1,\omega_2\rangle
        +
        \langle\omega_2,\omega_1\rangle
        +
        \langle\omega_2,\omega_2\rangle
    \Bigr) \\
    &=
    2\left(
        \frac{2}{3}
        +
        \frac{1}{3}
        +
        \frac{1}{3}
        +
        \frac{2}{3}
    \right) \\
    &=
    4.
\end{align}
Hence the desired state-independent lower bound follows.
\end{proof}

The preceding analysis reveals that the optimal lower bound (see Cor.~\ref{cor:LAU_adj_su3}) is saturated on every non-zero weight space (see Cor.~\ref{cor:LAU_Conservation_adj_su3_n0}), whereas the Cartan subalgebra exhibits a strictly larger uncertainty (see Cor.~\ref{cor:LAU_Conservation_adj_su3_0}).
The roles of $\mathfrak{su}(2)$ in angular momentum (see Cors.~\ref{cor:Conservation_Law_Magnetic_Quantum_Number} and~\ref{cor:Angular_Momentum_Uncertainty_Relation}) and $\mathfrak{su}(3)$ in particle physics (see Cor.~\ref{cor:su3_Weight_Space_Uncertainty}) are generalized by $\mathfrak{su}(n)$ to arbitrary $n$-level quantum systems. 
After multiplication by $\mathrm{i}$, the generators of $\mathfrak{su}(n)$ form a complete basis for traceless Hermitian observables. 
Indeed, any Hermitian observable $O$ on a $n$-dimensional Hilbert space admits the decomposition
\begin{align}
    O
    =
    \frac{\Tr[O]}{n}\1
    +
    \left(O-\frac{\Tr[O]}{n}\1\right),
\end{align}
where the second term is traceless. 
The identity component shifts all measurement outcomes by the same constant and therefore does not affect commutators, incompatibility, dynamics, uncertainty, or information carried by relative outcomes. 
The physically relevant part of an observable can therefore be identified with its traceless component.
In the defining representation on $\mathbb C^n$, the
$\mathfrak{su}(n)$ relation constrains the total variance of a normalized complete basis of traceless Hermitian observables.  
In this sense, it extends angular momentum uncertainty from spin systems to general finite dimensional quantum mechanics.

         
\subsection{Generalized Quantum Uncertainty}
\label{subsec:Three_Bases_One_Uncertainty}

Optimal uncertainty bounds quantify the incompatibility of observables and, more generally, operators. 
An intrinsic description of this incompatibility should depend on their algebraic structure and representation, independently of the operator basis used to express it. 
The Lie algebraic uncertainty relation in Eq.~\eqref{eq:QUncertainty_Lie_Algebra} does not yet make this independence explicit. 
For anti-Hermitian generators $\{H_i,E_\alpha,F_\alpha\}$ (see Def.~\ref{def:Anti_Hermitian_Basis}), the uncertainty defined in Eq.~\eqref{eq:QUncertainty_Lie_Algebra} is unchanged upon multiplication by $i$, which converts them into Hermitian observables. 
Expressing the same algebra in a Chevalley basis (see Def.~\ref{def:Chevalley_Basis}), however, can yield a different optimal bound. 
This apparent discrepancy arises because the definition assigns equal weight to basis directions whose normalizations may differ. 
Resolving it requires a formulation that accounts for the geometry of the operator space. 
We therefore introduce a metric-based generalization that makes uncertainty independent of the choice and normalization of the operator basis.

Before introducing our generalized uncertainty relations for operators, we first explain how the formulation for anti-Hermitian generators extends to observables, then illustrate the basis dependence discussed above with an explicit example. 
We begin with the passage from generators to observables.
The quantum uncertainty in Eq.~\eqref{eq:QUncertainty_Lie_Algebra} is defined in terms of the orthogonal compact generators introduced in Def.~\ref{def:Anti_Hermitian_Basis}.
Denote this basis by $\{Q_a\}_{a=1}^{n}$, with $n=\dim{\mathbb C}\g$. 
In a unitary representation $\pi$, each $\pi(Q_a)$ is anti-Hermitian, so multiplication by $i$ yields the physical observable $i\pi(Q_a)$. 
The invariance of the variance under this conversion, established in Lem.~\ref{lem:Variance_Equivalence_Under_Hermitianization}, then gives
\begin{align}\label{eq:physical_variance_sum}
    \Delta_\psi^2(\g)
    =\frac12\sum_{a=1}^{n}
    \Delta_\psi^2\bigl(\mathrm{i}\pi(Q_a)\bigr).
\end{align}
The factor $1/2$ accounts for the squared norm of $2$ shared by all orthogonal generators in the chosen metric.

\begin{table}[htbp]
\centering
\scriptsize
\begin{tblr}{
      width = \textwidth,
      colspec = {
        Q[c,m,wd=0.08\textwidth]|
        X[c,m]||
        Q[c,m,wd=0.16\textwidth]|
        X[c,m]||
        Q[c,m,wd=0.16\textwidth]|
        X[c,m,wd=0.16\textwidth]
      },
      row{1} = {bg=gray!50,font=\bfseries},
      row{2} = {bg=gray!20,font=\bfseries},
      hlines,
      vlines,
      cells = {m,c},
      cell{3-10}{2} = {bg=mRed},
      cell{3-10}{4} = {bg=mGreen},
      cell{3-10}{6} = {bg=mBlue},
}
\SetCell[c=6]{m,c}
\textbf{Bases for $\mathfrak{sl}(3,\mathbb C)$, $\mathfrak{su}(3)$, and $\mathrm{i}\mathfrak{su}(3)$}
& & & & & \\

Chevalley
& Matrix
& Anti-Hermitian
& Matrix
& Observable
& Matrix \\

$h_1$
&
$\begin{pmatrix}
1&0&0\\
0&-1&0\\
0&0&0
\end{pmatrix}$
&
$H_1=-\mathrm{i}h_1$
&
$\begin{pmatrix}
-\mathrm{i}&0&0\\
0&\mathrm{i}&0\\
0&0&0
\end{pmatrix}$
&
$\mathrm{i}H_1=h_1$
&
$\begin{pmatrix}
1&0&0\\
0&-1&0\\
0&0&0
\end{pmatrix}
=\lambda_3$
\\

$h_2$
&
$\begin{pmatrix}
0&0&0\\
0&1&0\\
0&0&-1
\end{pmatrix}$
&
$H_2=-\dfrac{\mathrm{i}}{\sqrt3}(h_1+2h_2)$
&
$-\dfrac{\mathrm{i}}{\sqrt3}
\begin{pmatrix}
1&0&0\\
0&1&0\\
0&0&-2
\end{pmatrix}$
&
$\mathrm{i}H_2=\dfrac{1}{\sqrt3}(h_1+2h_2)$
&
$\dfrac{1}{\sqrt3}
\begin{pmatrix}
1&0&0\\
0&1&0\\
0&0&-2
\end{pmatrix}
=\lambda_8$
\\

$e_{\alpha_1}$
&
$\begin{pmatrix}
0&1&0\\
0&0&0\\
0&0&0
\end{pmatrix}$
&
$E_{\alpha_1}=-\mathrm{i}(e_{\alpha_1}+f_{\alpha_1})$
&
$\begin{pmatrix}
0&-\mathrm{i}&0\\
-\mathrm{i}&0&0\\
0&0&0
\end{pmatrix}$
&
$\mathrm{i}E_{\alpha_1}=e_{\alpha_1}+f_{\alpha_1}$
&
$\begin{pmatrix}
0&1&0\\
1&0&0\\
0&0&0
\end{pmatrix}
=\lambda_1$
\\

$f_{\alpha_1}$
&
$\begin{pmatrix}
0&0&0\\
1&0&0\\
0&0&0
\end{pmatrix}$
&
$F_{\alpha_1}=-(e_{\alpha_1}-f_{\alpha_1})$
&
$\begin{pmatrix}
0&-1&0\\
1&0&0\\
0&0&0
\end{pmatrix}$
&
$\mathrm{i}F_{\alpha_1}=-\mathrm{i}(e_{\alpha_1}-f_{\alpha_1})$
&
$\begin{pmatrix}
0&-\mathrm{i}&0\\
\mathrm{i}&0&0\\
0&0&0
\end{pmatrix}
=\lambda_2$
\\

$e_{\alpha_2}$
&
$\begin{pmatrix}
0&0&0\\
0&0&1\\
0&0&0
\end{pmatrix}$
&
$E_{\alpha_2}=-\mathrm{i}(e_{\alpha_2}+f_{\alpha_2})$
&
$\begin{pmatrix}
0&0&0\\
0&0&-\mathrm{i}\\
0&-\mathrm{i}&0
\end{pmatrix}$
&
$\mathrm{i}E_{\alpha_2}=e_{\alpha_2}+f_{\alpha_2}$
&
$\begin{pmatrix}
0&0&0\\
0&0&1\\
0&1&0
\end{pmatrix}
=\lambda_6$
\\

$f_{\alpha_2}$
&
$\begin{pmatrix}
0&0&0\\
0&0&0\\
0&1&0
\end{pmatrix}$
&
$F_{\alpha_2}=-(e_{\alpha_2}-f_{\alpha_2})$
&
$\begin{pmatrix}
0&0&0\\
0&0&-1\\
0&1&0
\end{pmatrix}$
&
$\mathrm{i}F_{\alpha_2}=-\mathrm{i}(e_{\alpha_2}-f_{\alpha_2})$
&
$\begin{pmatrix}
0&0&0\\
0&0&-\mathrm{i}\\
0&\mathrm{i}&0
\end{pmatrix}
=\lambda_7$
\\

$e_{\alpha_1+\alpha_2}$
&
$\begin{pmatrix}
0&0&1\\
0&0&0\\
0&0&0
\end{pmatrix}$
&
$E_{\alpha_1+\alpha_2}=-\mathrm{i}(e_{\alpha_1+\alpha_2}+f_{\alpha_1+\alpha_2})$
&
$\begin{pmatrix}
0&0&-\mathrm{i}\\
0&0&0\\
-\mathrm{i}&0&0
\end{pmatrix}$
&
$\mathrm{i}E_{\alpha_1+\alpha_2}=e_{\alpha_1+\alpha_2}+f_{\alpha_1+\alpha_2}$
&
$\begin{pmatrix}
0&0&1\\
0&0&0\\
1&0&0
\end{pmatrix}
=\lambda_4$
\\

$f_{\alpha_1+\alpha_2}$
&
$\begin{pmatrix}
0&0&0\\
0&0&0\\
1&0&0
\end{pmatrix}$
&
$F_{\alpha_1+\alpha_2}=-(e_{\alpha_1+\alpha_2}-f_{\alpha_1+\alpha_2})$
&
$\begin{pmatrix}
0&0&-1\\
0&0&0\\
1&0&0
\end{pmatrix}$
&
$\mathrm{i}F_{\alpha_1+\alpha_2}=-\mathrm{i}(e_{\alpha_1+\alpha_2}-f_{\alpha_1+\alpha_2})$
&
$\begin{pmatrix}
0&0&-\mathrm{i}\\
0&0&0\\
\mathrm{i}&0&0
\end{pmatrix}
=\lambda_5$
\\

\end{tblr}
\caption{
\textbf{Operator Bases for $\mathfrak{sl}(3,\mathbb C)$, $\mathfrak{su}(3)$, and $\mathrm{i}\mathfrak{su}(3)$.}
The Chevalley basis comprises two Cartan generators, $\{h_1,h_2\}$, and six root operators associated with the positive roots $\Phi^+=\{\alpha_1,\alpha_2,\alpha_1+\alpha_2\}$ and their negatives. 
The anti-Hermitian generators span the compact real form $\mathfrak{su}(3)$, with orthogonal Cartan generators $H_1=-\mathrm{i}h_1$ and $H_2=-\mathrm{i}(h_1+2h_2)/\sqrt3$. 
Multiplication by $i$ yields the corresponding Hermitian observables, recovering the eight Gell-Mann matrices $\{\lambda_1,\ldots,\lambda_8\}$. 
Both bases are orthogonal under the Hilbert-Schmidt inner product in the defining representation, with each generator having squared norm $2$.
}
\label{tab:su3_three_generator_bases}
\end{table}

The basis dependence is already apparent in the defining representation of $\mathfrak{sl}(3,\mathbb C)$. 
Let $\alpha_1,\alpha_2$ be the simple roots and $\Phi^+=\{\alpha_1,\alpha_2,\alpha_1+\alpha_2\}$ the positive roots. 
We compare the Chevalley and orthogonal anti-Hermitian bases listed in Tab.~\ref{tab:su3_three_generator_bases}, using the same prefactor $1/2$ in both variance sums. Each generator is identified with its matrix in the defining representation.
For an arbitrary normalized pure state $\ket{\psi}=(a,b,c)^{\T}$, define
\begin{align}
p_1\coloneqq|a|^2,\quad
p_2\coloneqq|b|^2,\quad
p_3\coloneqq|c|^2,\quad
p_1+p_2+p_3=1.
\end{align}
The Cartan generators in the Chevalley basis,
$h_1=\operatorname{diag}(1,-1,0)$ and
$h_2=\operatorname{diag}(0,1,-1)$, satisfy
\begin{align}
\langle h_1\rangle_\psi=p_1-p_2,
\quad
\langle h_1^2\rangle_\psi=p_1+p_2,
\\
\langle h_2\rangle_\psi=p_2-p_3,
\quad
\langle h_2^2\rangle_\psi=p_2+p_3.
\end{align}
Their uncertainty is therefore
\begin{align}
\Delta_\psi^2(h_1)+\Delta_\psi^2(h_2)
=
1+p_2-(p_1-p_2)^2-(p_2-p_3)^2.
\end{align}
The root generators are
$e_{\alpha_1}=E_{12}$,
$e_{\alpha_2}=E_{23}$,
$e_{\alpha_1+\alpha_2}=E_{13}$ and
$f_\alpha=e_\alpha^\dagger$, where
$E_{ij}=|i\rangle\langle j|$.
With $\psi_i\coloneqq \langle i|\psi\rangle$, we have
\begin{align}
\langle E_{ij}\rangle_\psi
=\overline{\psi_i}\psi_j,
\end{align}
and
\begin{align}
\langle E_{ij}^\dagger E_{ij}\rangle_\psi=p_j,
\end{align}
so that $\Delta_\psi^2(E_{ij})=p_j(1-p_i)$. 
The six root generators thus contribute
\begin{align}
\sum_{\alpha\in\Phi^+}
\bigl[\Delta_\psi^2(e_\alpha)+\Delta_\psi^2(f_\alpha)\bigr]
=
\sum_{i\ne j}p_j(1-p_i)
=
1+\sum_{i=1}^{3}p_i^2.
\end{align}
Combining the Cartan and root contributions yields
\begin{align}\label{eq:sl3_chevalley_variance_sum}
\Delta_{\mathrm{def}}^2(\mathfrak{sl}(3,\mathbb C))\coloneqq\frac12\Biggl[
\Delta_\psi^2(h_1)+\Delta_\psi^2(h_2)
+\sum_{\alpha\in\Phi^+}
\bigl[\Delta_\psi^2(e_\alpha)+\Delta_\psi^2(f_\alpha)\bigr]
\Biggr]
=
1+\frac32p_2(1-p_2)
\geqslant1,
\end{align}
which is optimal. 
Equality holds if and only if $p_2=0$ or $p_2=1$.
The bound can be achived, for example, by preparing state $\ket{1}$.
Thus, the quantum uncertainty in the Chevalley basis is
\begin{align}
    \min_{\psi}\Delta_{\mathrm{def}}^2(\mathfrak{sl}(3,\mathbb C))=1.
\end{align}

For the orthogonal anti-Hermitian basis, the generators are
\begin{align}
H_1=-ih_1,\quad
H_2=-\frac{i}{\sqrt3}(h_1+2h_2), \quad
E_\alpha=-i(e_\alpha+f_\alpha), \quad
F_\alpha=f_\alpha-e_\alpha.
\end{align}
Since $H_2=-i\operatorname{diag}(1,1,-2)/\sqrt3$, its variance is $3p_3(1-p_3)$. 
Together with $\Delta_\psi^2(H_1)=\Delta_\psi^2(h_1)$, this gives the Cartan uncertainty
\begin{align}
\Delta_\psi^2(H_1)+\Delta_\psi^2(H_2)
=
p_1+p_2-(p_1-p_2)^2+3p_3(1-p_3)
=
2\left(1-\sum_{i=1}^{3}p_i^2\right).
\end{align}
For each positive root, the cross terms cancel between the two variances, yielding
\begin{align}
\Delta_\psi^2(E_\alpha)+\Delta_\psi^2(F_\alpha)
=
2\left(\Delta_\psi^2(e_\alpha)+\Delta_\psi^2(f_\alpha)\right).
\end{align}
The total root uncertainty is therefore
\begin{align}
\sum_{\alpha\in\Phi^+}
\left(\Delta_\psi^2(E_\alpha)+\Delta_\psi^2(F_\alpha)\right)
=
2\left(1+\sum_{i=1}^{3}p_i^2\right).
\end{align}
The state dependence cancels upon combining the Cartan and root contributions:
\begin{align}\label{eq:sl3_anti_hermitian_variance_sum}
\Delta_{\mathrm{def}}^2(\mathfrak{su}(3))
 \coloneqq
 \frac12\left(
\Delta_\psi^2(H_1)+\Delta_\psi^2(H_2)
+\sum_{\alpha\in\Phi^+}
\left(\Delta_\psi^2(E_\alpha)+\Delta_\psi^2(F_\alpha)\right)
\right)
=2.
\end{align}
This identity holds for every pure state in the defining representation.
By Eq.~\eqref{eq:physical_variance_sum}, the same value is obtained for the corresponding Hermitian observables. Thus, the quantum uncertainty is constant over pure states, and its optimal lower bound is
\begin{align}
\min_{\psi}\Delta_{\mathrm{def}}^2(\mathfrak{su}(3))=2.
\end{align}

Thus, despite realizing the same algebra in the same representation, the Chevalley and orthogonal anti-Hermitian bases yield different optimal uncertainty bounds:
\begin{align}
    \min_{\psi}\Delta_{\mathrm{def}}^2(\mathfrak{sl}(3,\mathbb C))
    <
    \min_{\psi}\Delta_{\mathrm{def}}^2(\mathfrak{su}(3)).
\end{align}

The preceding example shows that optimal uncertainty bounds can depend on the operator basis even when the underlying algebra and its representation remain unchanged. 
This raises a natural question: 
can uncertainty be defined intrinsically, independently of the basis used to express it? 
We answer this question by combining the Gram matrix, which encodes the geometry of the operator basis, with the covariance matrix, which describes uncertainty in a given quantum state. 
We begin by defining the covariance matrix for a set of operators.

\begin{mydef}{Covariance Matrix}{Covariance_Matrix}
Let $\mathbf X\coloneqq\{X_i\}_{i}$ be a family of operators on space $V$, not necessarily Hermitian. 
For a normalized state $\ket{\psi}$, their covariance is
\begin{align}\label{eq:general_nonhermitian_covariance_definition}
    \mathrm{Cov}_\psi(X_i,X_j)
    \coloneqq
    \langle X_i^\dagger X_j\rangle_\psi
    -\langle X_i^\dagger\rangle_\psi\langle X_j\rangle_\psi,
\end{align}
with the corresponding covariance matrix
\begin{align}\label{eq:general_nonhermitian_covariance_matrix}
    \mathrm{Cov}_\psi(\mathbf X)
    \coloneqq
    \left(\mathrm{Cov}_\psi(X_i,X_j)\right)_{i,j}.
\end{align}
\end{mydef}

The diagonal entries of $\mathrm{Cov}_\psi(\mathbf X)$ recover the operator variance defined in Eq.~\eqref{eq:non_hermitian_variance}:
\begin{align}\label{eq:general_nonhermitian_variance_diagonal}
    \mathrm{Cov}_\psi(X_i,X_i)=\Delta_\psi^2(X_i),
\end{align}
while its off-diagonal entries capture the correlations between their uncertainties.
More generally, the covariance matrix determines the uncertainty of every linear combination of the operators. 
For any state $\ket{\phi}$,
\begin{align}\label{eq:b5_linear_combination_variance}
    \bra{\phi}\mathrm{Cov}_\psi(\mathbf X)\ket{\phi}
    =\Delta_\psi^2\!\left(\sum_{i}\phi_iX_i\right)\geqslant0.
\end{align}
Thus, $\mathrm{Cov}_\psi(\mathbf X)$ is Hermitian and positive semidefinite.

To define the Gram matrix $\mathrm{G}(\mathbf X)$ of an operator basis $\mathbf X$, we first equip the Lie algebra $\g$ with a positive-definite Hermitian inner product. 
The normalized Killing form $\kappa$ in Eq.~\eqref{eq:normal_Killing} is complex bilinear and therefore does not itself provide such an inner product. 
The chosen compact real form $\g_{\mathbb R}$ supplies the additional structure needed for this construction.
On $\g_{\mathbb R}$, the form $\kappa$ is negative definite, so $(X,Y)_{\g_{\mathbb R}}\coloneqq-\kappa(X,Y), X,Y\in\g_{\mathbb R}$ defines a positive-definite real inner product (see Eq.\eqref{eq:Inner_product_algebra}). 
To extend it to $\g$, recall that every $x\in\g$ can be written uniquely as $x=X+\mathrm{i}Y$, with $X,Y\in\g_{\mathbb R}$. 
The compact real form therefore determines an antilinear involution
\begin{align}
    \tau(X+\mathrm{i}Y)=X-\mathrm{i}Y,
\end{align}
which, in the Chevalley basis, satisfies $\tau(h_i)=-h_i$ and $\tau(e_\alpha)=-f_\alpha$, as in Eq.~\eqref{eq:tau_conjugation}. We then define
\begin{align}\label{eq:inner_product_g}
    (x,y)_{\g}
    \coloneqq-\kappa(\tau(x),y),\quad \text{for}\,\,x,y\in\g.
\end{align}
The resulting form is Hermitian, conjugate linear in its first argument and linear in its second. It is positive definite because
\begin{align}\label{eq:g_Inner_Prodduct}
    (x,x)_{\g}
    =-\kappa(X,X)-\kappa(Y,Y)>0\quad
    \text{for}\,\,x=X+\mathrm{i}Y\ne0,
\end{align}
and it reduces to $-\kappa$ on $\g_{\mathbb R}$.
Here, $(\cdot,\cdot)_{\g}$ denotes a unified inner product on $\g$, whose restrictions to $\g_{\mathbb R}$ and $\mathrm{i}\g_{\mathbb R}$ recover their respective inner products.

With the inner product $(\cdot,\cdot)_{\g}$ on $\g$ established, we can now define the Gram matrix of a collection of elements of $\g$. 
Although our focus is on Lie algebras, this definition applies to any inner product space.

\begin{mydef}{Gram Matrix}{Gram_Matrix_Generator_Basis}
For a family of operators $\mathbf X\coloneqq\{X_i\}_{i}$ in $\g$, we define the associated Gram matrix $\mathrm{G}(\mathbf X)$ by
\begin{align}\label{eq:Gram_matirx}
    \mathrm{G}(\mathbf X)
    \coloneqq
    \left((X_i,X_j)_{\g}\right)_{i, j}.
\end{align}
\end{mydef}

When $\mathbf X$ forms a basis of the Lie algebra $\g$, its Gram matrix $\mathrm{G}(\mathbf X)$ is positive semidefinite, allowing us to formulate the generalized uncertainty relation as follows.

\begin{mydef}{Generalized Lie Algebraic Quantum Uncertainty}{Generalized_Lie_Algebraic_Uncertainty_Relations}
Let $\g$ be a Lie algebra with basis $\mathbf X$. We define the corresponding quantum uncertainty by
\begin{align}\label{eq:GUR_Lie}
    \Delta_\psi^2(\g)
    \coloneqq\Tr[\mathrm{G}(\mathbf X)^{-1}\mathrm{Cov}_\psi(\mathbf X)].
\end{align}
\end{mydef}

If $\mathbf X$ forms an orthogonal basis with respect to $(\cdot,\cdot)_{\g}$ (see Eq.~\eqref{eq:inner_product_g}); 
that is
\begin{align}
    (X_i,X_j)_{\g}=c_i\delta_{ij},
\end{align}
then its Gram matrix is diagonal
\begin{align}
\mathrm{G}(\mathbf X)=
\begin{pmatrix}
c_1 & 0 & \cdots & 0\\
0 & c_2 & \cdots & 0\\
\vdots & \vdots & \ddots & \vdots\\
0 & 0 & \cdots & c_n
\end{pmatrix},
\end{align}
whose inverse takes the form
\begin{align}
\mathrm{G}(\mathbf X)^{-1}
&=
\begin{pmatrix}
c_1^{-1} & 0 & \cdots & 0\\
0 & c_2^{-1} & \cdots & 0\\
\vdots & \vdots & \ddots & \vdots\\
0 & 0 & \cdots & c_n^{-1}
\end{pmatrix}.
\end{align}
In this case, the uncertainty reduces to a weighted sum of variances,
\begin{align}
    \Delta_\psi^2(\g)
    =\sum_{i}\frac{\Delta_\psi^2(X_i)}{c_i}.
\end{align}
If all basis elements have equal norm, so that $c_1=\cdots=c_n=c$, this recovers the standard sum of variances up to an overall factor:
\begin{align}
    \Delta_\psi^2(\g)
    =\frac{1}{c}\sum_{i}\Delta_\psi^2(X_i).
\end{align}

Although we define quantum uncertainty using a basis $\mathbf X$, the notation on the left-hand side of Eq.~\eqref{eq:GUR_Lie} omits any reference to this choice. 
This reflects the basis independence of the definition, established by the following theorem.

\begin{mythm}{Basis Independence of Generalized Lie Algebraic Quantum Uncertainty}{Basis_Independence_Orthogonal_Representation}
The quantum uncertainty $\Delta_\psi^2(\g)$ defined in Eq.~\eqref{eq:GUR_Lie} is independent of the choice of basis $\mathbf X$ of the Lie algebra $\g$.
\end{mythm}

\begin{proof}
The inner product and covariance are both sesquilinear, with conjugate linearity in the first argument. 
Thus, under an arbitrary change of basis $\mathbf Y=M\mathbf X$, where $M\in\mathrm{GL}_n(\mathbb C)$, we have
\begin{align}\label{eq:b5_covariance_basis_change}
    \mathrm{G}(\mathbf Y)&=\overline M \mathrm{G}(\mathbf X)M^{\T},
\end{align}
and
\begin{align}
    \mathrm{Cov}_\psi(\mathbf Y)&=\overline M\mathrm{Cov}_\psi(\mathbf X)M^{\T},
\end{align}
with $\overline M$ denoting entrywise complex conjugation. 
This implies
\begin{align}
    \mathrm{G}(\mathbf Y)^{-1}\mathrm{Cov}_\psi(\mathbf Y)
    =(M^{\T})^{-1}
    \mathrm{G}(\mathbf X)^{-1}\mathrm{Cov}_\psi(\mathbf X)M^{\T},
\end{align}
which leads to the same trace:
\begin{align}\label{eq:intrinsic_uncertainty_basis_independence}
    \Tr\!\left[\mathrm{G}(\mathbf Y)^{-1}\mathrm{Cov}_\psi(\mathbf Y)\right]
    =\Tr\!\left[\mathrm{G}(\mathbf X)^{-1}\mathrm{Cov}_\psi(\mathbf X)\right].
\end{align}
Hence $\Delta_\psi^2(\g)$ is independent of the chosen basis.
\end{proof}

By Thm.~\ref{thm:Basis_Independence_Orthogonal_Representation}, the three bases considered so far -- the Chevalley basis, the anti-Hermitian basis and the observable basis -- yield the same quantum uncertainty. 
To gain further insight, we derive its explicit form in each basis, highlighting complementary aspects of the underlying algebraic structure.
We begin with the Chevalley basis 
\begin{align}\label{eq:GUR_Chevalley}
    \mathbf X_{\mathrm{Ch}}
    \coloneqq
    \left\{h_i, e_\alpha, f_\alpha\right\}_{i, \alpha}
\end{align}
whose Gram matrix is given by
\begin{align}
    \mathrm{G}(\mathbf X_{\mathrm{Ch}})
    =
    \mathrm{G}_{\mathfrak h}
    \oplus
    \begin{pmatrix}
    \frac{2}{\langle\alpha,\alpha\rangle} & 0 & \cdots & 0\\
    0 & \frac{2}{\langle\alpha,\alpha\rangle} & \cdots & 0\\
    \vdots & \vdots & \ddots & \vdots\\
    0 & 0 & \cdots & \frac{2}{\langle\alpha,\alpha\rangle}
    \end{pmatrix}
    \oplus
    \begin{pmatrix}
    \frac{2}{\langle\alpha,\alpha\rangle} & 0 & \cdots & 0\\
    0 & \frac{2}{\langle\alpha,\alpha\rangle} & \cdots & 0\\
    \vdots & \vdots & \ddots & \vdots\\
    0 & 0 & \cdots & \frac{2}{\langle\alpha,\alpha\rangle}
    \end{pmatrix}.
\end{align}
The corresponding covariance matrix takes the block form
\begin{align}\label{eq:Covariance_Chevallley}
        \mathrm{Cov}_\psi
        (\mathbf X_{\mathrm{Ch}})
        =\begin{pmatrix}
            \mathrm{Cov}_\psi(\mathbf h)&K_\psi\\
            K_\psi^\dagger&\mathrm{Cov}_\psi(\mathbf e\cup \mathbf f)
        \end{pmatrix},
\end{align}
where $\mathbf h\coloneqq\{h_i\}_{i}$, $\mathbf e\coloneqq\{e_{\alpha}\}_{\alpha\in\Phi^+}$ and $\mathbf f\coloneqq\{f_{\alpha}\}_{\alpha\in\Phi^+}$.
Here $K_\psi$ collects the cross-covariance terms between the Cartan and root sectors; 
its explicit form is not needed to evaluate the uncertainty (see Def.~\ref{def:Generalized_Lie_Algebraic_Uncertainty_Relations}). 
In the Chevalley basis $\mathbf X_{\mathrm{Ch}}$ (see Eq.~\eqref{eq:GUR_Chevalley}), the generalized quantum uncertainty therefore takes the form
\begin{align}
    \Delta_\psi^2(\g)
    \coloneqq
    &\Tr\,[
    \mathrm{G}(\mathbf X_{\mathrm{Ch}})^{-1}
    \mathrm{Cov}_\psi(\mathbf X_{\mathrm{Ch}})]\\
    =
    &\Tr\,[
    \mathrm{G}_{\mathfrak h}^{-1}
    \oplus
    \begin{pmatrix}
    \frac{\langle\alpha,\alpha\rangle}{2} & 0 & \cdots & 0\\
    0 & \frac{\langle\alpha,\alpha\rangle}{2} & \cdots & 0\\
    \vdots & \vdots & \ddots & \vdots\\
    0 & 0 & \cdots & \frac{\langle\alpha,\alpha\rangle}{2}
    \end{pmatrix}
    \oplus
    \begin{pmatrix}
    \frac{\langle\alpha,\alpha\rangle}{2} & 0 & \cdots & 0\\
    0 & \frac{\langle\alpha,\alpha\rangle}{2} & \cdots & 0\\
    \vdots & \vdots & \ddots & \vdots\\
    0 & 0 & \cdots & \frac{\langle\alpha,\alpha\rangle}{2}
    \end{pmatrix}
    \cdot
    \mathrm{Cov}_\psi(\mathbf X_{\mathrm{Ch}})]\\
    =
    &
    \Tr\,[\mathrm{G}_{\mathfrak h}^{-1}\mathrm{Cov}_\psi(\mathbf h)]
    +
    \sum_{\alpha\in\Phi^+}\frac{\langle\alpha,\alpha\rangle}{2}
    \left(\Delta_\psi^2(e_\alpha)+\Delta_\psi^2(f_\alpha)\right).
\end{align}

Combining these expressions, we obtain the generalized quantum uncertainty in the Chevalley basis:

\begin{mycor}{Generalized Lie Algebraic Quantum Uncertainty in the Chevalley Basis}{GUR_Chevalley}
In the Chevalley basis $\mathbf X_{\mathrm{Ch}}$ (see Eq.~\eqref{eq:GUR_Chevalley}), the quantum uncertainty $\Delta_\psi^2(\g)$ defined in Eq.~\eqref{eq:GUR_Lie} is given by
\begin{align}\label{eq:chevalley_uncertainty_explicit}
    \Delta_\psi^2(\g)
    =
    \sum_{i,j=1}^{r}(\mathrm{G}_{\mathfrak h}^{-1})_{ij}\mathrm{Cov}_\psi(h_i,h_j)
    +\sum_{\alpha\in\Phi^+}\frac{\langle\alpha,\alpha\rangle}{2}
    \left(\Delta_\psi^2(e_\alpha)+\Delta_\psi^2(f_\alpha)\right).
\end{align}
\end{mycor}

The anti-Hermitian $\mathbf X_{\mathrm{aH}}$ (see Def.~\ref{def:Anti_Hermitian_Basis}) and observable $\mathbf X_{\mathrm{ob}}$ bases,
\begin{align}
    \mathbf X_{\mathrm{aH}}
    &\coloneqq
    \left\{H_i, E_\alpha, F_\alpha\right\}_{i, \alpha},\label{eq:GUR_anti_Hermitian}\\
    \mathbf X_{\mathrm{ob}}
    &\coloneqq
    \left\{\mathrm{i}H_i, \mathrm{i}E_\alpha, \mathrm{i}F_\alpha\right\}_{i, \alpha},\label{eq:GUR_observable}
\end{align}
share the Gram matrix 
\begin{align}
    \mathrm{G}(\mathbf X_{\mathrm{aH}})=
    \mathrm{G}(\mathbf X_{\mathrm{ob}})=
    2\1,
\end{align}
with inverse 
\begin{align}
    \mathrm{G}(\mathbf X_{\mathrm{aH}})^{-1}=
    \mathrm{G}(\mathbf X_{\mathrm{ob}})^{-1}=
    \frac{1}{2}\1.
\end{align}
For the anti-Hermitian basis, the generalized quantum uncertainty therefore reduces to

\begin{mycor}{Generalized Lie Algebraic Quantum Uncertainty in the Anti-Hermitian Basis}{GUR_anti_Hermitian}
In the observable basis $\mathbf X_{\mathrm{aH}}$ specified in Eq.~\eqref{eq:GUR_anti_Hermitian}, the quantum uncertainty of Eq.~\eqref{eq:GUR_Lie} can be expressed as
\begin{align}\label{eq:anti_hermitian_uncertainty}
    \Delta_\psi^2(\g_{\mathbb R})
    =\frac12\sum_{i=1}^{r}\Delta_\psi^2(H_i)
    +\frac12\sum_{\alpha\in\Phi^+}
    \left(\Delta_\psi^2(E_\alpha)+\Delta_\psi^2(F_\alpha)\right).
\end{align}
\end{mycor}

Finally, multiplying each anti-Hermitian basis element by $i$ yields the corresponding observable basis and leaves the quantum uncertainty unchanged.

\begin{mycor}{Generalized Lie Algebraic Quantum Uncertainty in the Observable Basis}{GUR_observable}
In the observable basis $\mathbf X_{\mathrm{ob}}$ defined in Eq.~\eqref{eq:GUR_observable}, the quantum uncertainty of Eq.~\eqref{eq:GUR_Lie} reduces to
\begin{align}\label{eq:observable_uncertainty}
    \Delta_\psi^2(i\g_{\mathbb R})
    =
    \frac12\sum_{i=1}^{r}
    \Delta_\psi^2(\mathrm{i}H_i)
    +\frac12\sum_{\alpha\in\Phi^+}
    \left(
    \Delta_\psi^2(\mathrm{i}E_\alpha)+
    \Delta_\psi^2(\mathrm{i}F_\alpha)\right).
\end{align}
\end{mycor}

Although the equivalence of Eqs.~\eqref{eq:anti_hermitian_uncertainty} and \eqref{eq:observable_uncertainty} is immediate, their agreement with Eq.~\eqref{eq:chevalley_uncertainty_explicit} is less apparent. 
Theorem \ref{thm:Basis_Independence_Orthogonal_Representation} establishes that all three expressions describe the same quantum uncertainty in different operator bases, as summarized in the following corollary.

\begin{mycor}{Equivalence of Uncertainty in Chevalley, Anti-Hermitian and Observable Bases}{Equivalence_Basis_Representations}
For a Lie algebra $\g$ equipped with an inner product $(\cdot,\cdot)_{\g}$, the quantum uncertainty of a state $\psi$, defined in Eq.~\eqref{eq:inner_product_g}, admits equivalent expressions in the Chevalley, anti-Hermitian and observable bases:
\begin{align}\label{eq:three_formulations_equivalence}
        \Delta_\psi^2(\g)=\Delta_\psi^2(\g_{\mathbb R})=\Delta_\psi^2(\mathrm{i}\g_{\mathbb R}).
    \end{align}
\end{mycor}

Since all three bases yield the same uncertainty, the optimal bound (see Thm.~\ref{thm:Lie_Algebraic_Uncertainty_Relations}) established in the anti-Hermitian basis for an irreducible highest weight representation holds for all three expressions. 
Although derived for pure states, the bound remains valid and optimal for mixed states: 
concavity of the uncertainty preserves the lower bound under mixing, while a pure state still attains it.

We conclude this subsection with the defining representations $\pi_{\mathrm{def}}$ (see Eq.~\eqref{eq:def_rep}) of $\mathfrak{sl}(2, \mathbb C)$ and $\mathfrak{sl}(3, \mathbb C)$ (see Tab.~\ref{tab:su3_three_generator_bases}), illustrating how algebraic structure determines quantum uncertainty (see Def.~\ref{def:Generalized_Lie_Algebraic_Uncertainty_Relations}) and how its three equivalent expressions arise. 
For these type-A examples, the chosen inner product $(\cdot,\cdot)_{\g}$ (see Eq.~\eqref{eq:inner_product_g}) coincides with the Hilbert-Schmidt inner product, allowing the Gram matrices $\mathrm{G}(\mathbf X)$ to be computed directly from the matrix generators. 
Since all roots have squared length $2$, Eqs.~\eqref{eq:anti_hermitian_root_generators_E} and~\eqref{eq:anti_hermitian_root_generators_F} reduce to $E_\alpha=-\mathrm{i}(e_\alpha+f_\alpha)$ and
$F_\alpha=-(e_\alpha-f_\alpha)$, respectively. 

\begin{table}[htbp]
\centering
\scriptsize
\begin{tblr}{
      width = \textwidth,
      colspec = {
        Q[c,m,wd=0.08\textwidth]|
        X[c,m]||
        Q[c,m,wd=0.15\textwidth]|
        X[c,m]||
        Q[c,m,wd=0.15\textwidth]|
        X[c,m,wd=0.15\textwidth]
      },
      row{1} = {bg=gray!50,font=\bfseries},
      row{2} = {bg=gray!20,font=\bfseries},
      hlines,
      vlines,
      cells = {m,c},
      cell{3-5}{2} = {bg=mRed},
      cell{3-5}{4} = {bg=mGreen},
      cell{3-5}{6} = {bg=mBlue},
}

\SetCell[c=6]{m,c}
\textbf{Bases for $\mathfrak{sl}(2,\mathbb C)$, $\mathfrak{su}(2)$, and $\mathrm{i}\mathfrak{su}(2)$}
&
&
&
&
&
\\

Chevalley 
&
Matrix
&
Anti-Hermitian 
&
Matrix
&
Observable
&
Matrix
\\

$h$
&
$\begin{pmatrix}
1&0\\
0&-1
\end{pmatrix}$
&
$H=-\mathrm{i}h$
&
$\begin{pmatrix}
-\mathrm{i}&0\\
0&\mathrm{i}
\end{pmatrix}$
&
$\mathrm{i}H=h$
&
$\begin{pmatrix}
1&0\\
0&-1
\end{pmatrix}
=\sigma_z=2J_z$
\\

$e_\alpha$
&
$\begin{pmatrix}
0&1\\
0&0
\end{pmatrix}$
&
$E_\alpha=-\mathrm{i}(e_\alpha+f_\alpha)$
&
$\begin{pmatrix}
0&-\mathrm{i}\\
-\mathrm{i}&0
\end{pmatrix}$
&
$\mathrm{i}E_\alpha=e_\alpha+f_\alpha$
&
$\begin{pmatrix}
0&1\\
1&0
\end{pmatrix}
=\sigma_x=2J_x$
\\

$f_{\alpha}$
&
$\begin{pmatrix}
0&0\\
1&0
\end{pmatrix}$
&
$F_\alpha=-(e_\alpha-f_\alpha)$
&
$\begin{pmatrix}
0&-1\\
1&0
\end{pmatrix}$
&
$\mathrm{i}F_\alpha=-\mathrm{i}(e_\alpha-f_\alpha)$
&
$\begin{pmatrix}
0&-\mathrm{i}\\
\mathrm{i}&0
\end{pmatrix}
=\sigma_y=2J_y$
\\

\end{tblr}

\caption{
\textbf{Operator Bases for $\mathfrak{sl}(2,\mathbb C)$, $\mathfrak{su}(2)$, and $\mathrm{i}\mathfrak{su}(2)$.}
The Chevalley basis is $\mathbf{X}_{\mathrm{Ch}}=\{h,e_\alpha,f_\alpha\}$. 
The anti-Hermitian generators $H=-\mathrm{i}h$, $E_\alpha=-\mathrm{i}(e_\alpha+f_\alpha)$ and $F_\alpha=-(e_\alpha-f_\alpha)$, i.e., $\mathbf X_{\mathrm{aH}}=\{H, E_\alpha, F_\alpha\}$, span the compact real form $\mathfrak{su}(2)$ and are mutually orthogonal under the Hilbert-Schmidt inner product, each with squared norm $2$.
Multiplication by $\mathrm{i}$ gives the Hermitian observable basis
$\mathbf{X}_{\mathrm{ob}}=\{\mathrm{i}H,\mathrm{i}E_\alpha,\mathrm{i}F_\alpha\}=\{\sigma_z,\sigma_x,\sigma_y\}$.
}
\label{tab:su2_three_generator_bases}
\end{table}

For $\mathfrak{sl}(2,\mathbb C)$ and its compact real form $\mathfrak{su}(2)$, let $\alpha$ denote the unique positive root.
The Chevalley, anti-Hermitian and observable bases are listed in Table~\ref{tab:su2_three_generator_bases}.
In this case, the Cartan Gram matrix is
\begin{align}
    \mathrm{G}_{\mathfrak h}=
    (\Tr[h^2])=
    (2),
\end{align}
with inverse
\begin{align}
    \mathrm{G}_{\mathfrak h}^{-1}=
    (\frac{1}{2}).
\end{align}
Since $\langle\alpha,\alpha\rangle=2$, the quantum uncertainty (see Cor.~\ref{cor:GUR_Chevalley}) reduces to
\begin{align}\label{eq:su2_chevalley_ur}
\Delta_\psi^2(\mathfrak{sl}(2,\mathbb C))
= 
\frac{1}{2}\Delta_\psi^2(h) + \Delta_\psi^2(e_\alpha) + \Delta_\psi^2(f_\alpha)
= 
\frac{1}{2}
\Delta_\psi^2
\begin{pmatrix} 1&0\\0&-1 \end{pmatrix}
+ 
\Delta_\psi^2
\begin{pmatrix} 0&1\\0&0 \end{pmatrix}
+ 
\Delta_\psi^2
\begin{pmatrix} 0&0\\1&0 \end{pmatrix}.
\end{align}
For the anti-Hermitian basis $\mathbf{X}_{\mathrm{ob}}$, the generators satisfy $(X_i,X_j)_{\g}=2\delta_{ij}$ for any $X_i$ and $X_j\in\mathbf{X}_{\mathrm{ob}}$, so the corresponding Gram matrix is given by
\begin{align}
    \mathrm{G}(\mathbf X_{\mathrm{aH}})
    =
    2\1_{3}.
\end{align}
The quantum uncertainty (see Cor.~\ref{cor:GUR_anti_Hermitian}) is consequently given by
\begin{align}\label{eq:su2_anti_hermitian_uncertainty}
\Delta_\psi^2(\mathfrak{su}(2))
= 
\frac{1}{2} 
\left( 
\Delta_\psi^2(H) + \Delta_\psi^2(E_\alpha) + \Delta_\psi^2(F_\alpha) 
\right) 
= 
\frac{1}{2} 
\left(
\Delta_\psi^2\begin{pmatrix} -\mathrm{i}&0\\0&\mathrm{i} \end{pmatrix}
+ \Delta_\psi^2\begin{pmatrix} 0&-\mathrm{i}\\-\mathrm{i}&0 \end{pmatrix}
+ \Delta_\psi^2\begin{pmatrix} 0&-1\\1&0 \end{pmatrix} 
\right).
\end{align}
Multiplying the generators by $\mathrm{i}$ gives the observable basis $\mathbf{X}_{\mathrm{ob}}$.
In this basis, the quantum uncertainty (see Cor.~\ref{cor:GUR_observable}) is
\begin{align}\label{eq:su2_hermitian_uncertainty}
\Delta_\psi^2(\mathrm{i}\mathfrak{su}(2))
= 
\frac{1}{2} 
\left( 
\Delta_\psi^2(\mathrm{i}H) + \Delta_\psi^2(\mathrm{i}E_\alpha) + \Delta_\psi^2(\mathrm{i}F_\alpha) 
\right)
= 
\frac{1}{2} 
\left( 
\Delta_\psi^2
\begin{pmatrix} 1&0\\0&-1 \end{pmatrix}
+ \Delta_\psi^2
\begin{pmatrix} 0&1\\1&0 \end{pmatrix}
+ \Delta_\psi^2
\begin{pmatrix} 0&-\mathrm{i}\\\mathrm{i}&0 \end{pmatrix} \right).
\end{align}
The parallelogram identity
\begin{align}
    \Delta_\psi^2 (X+Y)+\Delta_\psi^2(X-Y)
    =
    2\Delta_\psi^2(X)+2\Delta_\psi^2(Y),
\end{align}
together with phase invariance
\begin{align}
    \Delta_\psi^2(e^{\mathrm{i}\theta}X)
    =
    \Delta_\psi^2(X),
\end{align}
implies
\begin{align}
    \Delta_\psi^2(e_\alpha)+\Delta_\psi^2(f_{\alpha})
    =
    \frac{1}{2}
    \left(
    \Delta_\psi^2(E_\alpha)+
    \Delta_\psi^2(F_{\alpha})
    \right).
\end{align}
Since $H=-\mathrm{i}h$, these identities establish the equivalence of Eqs.~\eqref{eq:su2_chevalley_ur}, \eqref{eq:su2_anti_hermitian_uncertainty} and \eqref{eq:su2_hermitian_uncertainty}, explicitly confirming Cor.~\ref{cor:Equivalence_Basis_Representations},
\begin{equation}
\Delta_\psi^2(\mathfrak{sl}(2,\mathbb C))
=
\Delta_\psi^2(\mathfrak{su}(2))
=
\Delta_\psi^2(\mathrm{i}\mathfrak{su}(2)).
\end{equation}

The Lie algebra $\mathfrak{sl}(3,\mathbb C)$ and its compact real form $\mathfrak{su}(3)$ have rank 2 and dimension 8. 
The three positive roots are ordered as $\alpha_1,\alpha_2,\alpha_3=\alpha_1+\alpha_2$, giving $|\Phi^{+}|=3$. 
Table~\ref{tab:su3_three_generator_bases} lists the three bases in the defining representation on $\mathbb C^3$.
For the Chevalley basis $\mathbf{X}_{\mathrm{Ch}}$, evaluating $(G_{\mathfrak h})_{ij}=\Tr[h_i h_j]$ gives the Cartan Gram matrix
\begin{align}
    \mathrm{G}_{\mathfrak h}=
    \begin{pmatrix} 
    2 & -1 \\ -1 & 2 
    \end{pmatrix},
\end{align}
with inverse
\begin{align}
    \mathrm{G}_{\mathfrak{h}}^{-1} 
    = 
    \frac{1}{3} 
    \begin{pmatrix} 
    2 & 1 \\ 1 & 2 
    \end{pmatrix}.
\end{align}
The inverse Gram matrix is then given by
\begin{align}
    \mathrm{G}(\mathbf X_{\mathrm{Ch}})^{-1}
    =
    \frac{1}{3} 
    \begin{pmatrix} 
    2 & 1 \\ 1 & 2 
    \end{pmatrix}
    \oplus\1_6.
\end{align}
Substituting this into Eq.~\eqref{eq:GUR_Lie} yields the quantum uncertainty
\begin{align}\label{eq:su3_chevalley_uncertainty}
\Delta_\psi^2(\mathfrak{sl}(3,\mathbb C))
&= 
\sum_{i,j=1}^{2} (\mathrm{G}_{\mathfrak{h}}^{-1})_{ij} \mathrm{Cov}_\psi(h_i,h_j) 
+ 
\sum_{\alpha\in\Phi^+} 
\left( 
\Delta_\psi^2(e_\alpha) + \Delta_\psi^2(f_\alpha) 
\right) \\
&= 
\frac{2}{3}\left(
\Delta_\psi^2(h_1) 
+ 
\mathrm{Cov}_\psi(h_1,h_2) 
+ 
\Delta_\psi^2(h_2)
\right) 
+ 
\sum_{\alpha\in\Phi^+} 
\left( 
\Delta_\psi^2(e_\alpha) + \Delta_\psi^2(f_\alpha) 
\right).
\end{align}
Since $h_1$ and $h_2$ are commuting Hermitian operators, their covariance is real and satisfies $\mathrm{Cov}\psi(h_1,h_2)=\mathrm{Cov}\psi(h_2,h_1)$.

The corresponding compact real form, $\mathfrak{su}(3)$, has the anti-Hermitian basis 
\begin{align}
    \mathbf X_{\mathrm{aH}}=\{H_1, H_2, E_{\alpha_1}, F_{\alpha_1}, E_{\alpha_2}, F_{\alpha_2}, E_{\alpha_1+\alpha_2}, F_{\alpha_1+\alpha_2}\},
\end{align}
with Gram matrix
\begin{align}
    \mathrm{G}(\mathbf X_{\mathrm{aH}})
    =
    2\1_8.
\end{align}
Corollary~\ref{cor:GUR_anti_Hermitian} therefore gives
\begin{align}
\label{eq:su3_anti_hermitian_uncertainty}
\Delta_\psi^2(\mathfrak{su}(3))
&= 
\frac{1}{2} 
\left( 
\Delta_\psi^2(H_1) + \Delta_\psi^2(H_2) + \sum_{\alpha\in\Phi^+} 
\left( 
\Delta_\psi^2(E_\alpha) + \Delta_\psi^2(F_\alpha) 
\right) 
\right).
\end{align}

The observable basis $\mathbf X_{\mathrm{ob}}=\mathrm{i}\mathbf X_{\mathrm{aH}}$ coincides, up to ordering, with the eight Gell-Mann matrices ${\lambda_a}_{a=1}^{8}$. 
The quantum uncertainty (see Cor.~\ref{cor:GUR_observable}) is therefore
\begin{align}\label{eq:su3_hermitian_uncertainty}
\Delta_\psi^2(\mathrm{i}\mathfrak{su}(3)) = \frac{1}{2} \sum_{a=1}^{8} \Delta_\psi^2(\lambda_a).
\end{align}

The Cartan sector illustrates explicitly how the inverse Gram matrix accounts for the non-orthogonality of the Chevalley basis. 
Using $H_1=-\mathrm{i}h_1$ and $H_2=-\mathrm{i}(h_1+2h_2)/\sqrt3$, phase invariance and expansion of the variance give
\begin{align}
\Delta_\psi^2(e_\alpha)+\Delta_\psi^2(f_\alpha)
=
\frac12
\left(
\Delta_\psi^2(E_\alpha)+\Delta_\psi^2(F_\alpha)
\right)
=
\frac12
\left(
\Delta_\psi^2(\mathrm{i}E_\alpha)+\Delta_\psi^2(\mathrm{i}F_\alpha)
\right).
\end{align}
recovering the Cartan contribution in Eq.~\eqref{eq:su3_chevalley_uncertainty}.
For each positive root $\alpha\in\Phi^{+}$, the definitions $E_\alpha=-\mathrm{i}(e_\alpha+f_\alpha)$ and $F_\alpha=-(e_\alpha-f_\alpha)$, together with the parallelogram identity and phase invariance, yield
\begin{equation}
\Delta_\psi^2(e_\alpha) + \Delta_\psi^2(f_\alpha)
= \frac{1}{2}\left( \Delta_\psi^2(E_\alpha) + \Delta_\psi^2(F_\alpha) \right)
= \frac{1}{2}\left( \Delta_\psi^2(\mathrm{i}E_\alpha) + \Delta_\psi^2(\mathrm{i}F_\alpha) \right).
\end{equation}
Combining the Cartan and root parts therefore confirms that the three realizations give the same quantum uncertainty:
\begin{equation}\label{eq:su3_three_realizations_same_uncertainty}
\Delta_\psi^2(\mathfrak{sl}(3,\mathbb C)) = \Delta_\psi^2(\mathfrak{su}(3)) = \Delta_\psi^2(\mathrm{i}\mathfrak{su}(3)).
\end{equation}

Theorem~\ref{cor:Equivalence_Basis_Representations}, illustrated explicitly for $\mathfrak{sl}(2,\mathbb C)$ and $\mathfrak{sl}(3,\mathbb C)$, establishes that the three bases yield the same quantum uncertainty for a fixed representation, state and inner product $(\cdot,\cdot)_{\g}$ (see Eq.~\eqref{eq:inner_product_g}):
$\Delta_\psi^2(\mathfrak{g}) 
= \Delta_\psi^2(\mathfrak{g}_{\mathbb R}) 
= \Delta_\psi^2(\mathrm{i}\mathfrak{g}_{\mathbb R})$ (see Cor.~\ref{cor:Equivalence_Basis_Representations}).
The inverse Cartan Gram matrix and root length weights in the Chevalley basis account for the geometry of its generators. 
In the orthogonal anti-Hermitian and Hermitian bases, this normalization appears as an overall factor of $1/2$, since each generator has squared norm $2$.
The resulting uncertainty is thus intrinsic to the represented operator space equipped with the inner product $(\cdot,\cdot)_{\g}$ and evaluated in the given state. 
It is independent of the coordinates used to describe that space, but may change with the representation, state or metric. The following subsections derive its optimal bound in terms of Dynkin coefficients and the Cartan matrix.


\subsection{Dynkin Coefficients and Cartan Matrix}
\label{subsec:Dynkin_Cartan}

The general bound $2\langle\Lambda,\delta\rangle$ for Lie algebraic uncertainty relations has now been established (see Thm.~\ref{thm:Lie_Algebraic_Uncertainty_Relations}). 
Its explicit evaluation, however, still requires reconstructing the Weyl vector $\delta$ (see Def.~\ref{def:Weyl_Vector}) and computing the inner product $\langle\Lambda,\delta\rangle$ in an abstract root space formulation. 
It is therefore useful to develop an alternative approach that expresses the bound directly in coordinates adapted to the highest weight representation. 
This is particularly important for non-simply-laced Lie algebras, where simple roots have different lengths and the associated root length factors must be retained explicitly.
For this purpose, we introduce two additional concepts: 
Dynkin coefficients and the Cartan matrix. 
Together, they provide a complete coordinate description of highest weights and allow the uncertainty bound to be evaluated in an explicit form. 
We begin with the fundamental weights, which provide the natural basis for defining Dynkin coefficients and are closely tied to the Cartan matrix. 
Formally, the fundamental weights are defined as follows.

\begin{mydef}{Fundamental Weights}{Fundamental_Weights}
For the simple roots $\Pi=\{\alpha_1,\ldots,\alpha_r\}$ (see Eq.~\eqref{eq:simple_roots_set}), the fundamental weights $\Omega\coloneqq\{\omega_i\}_{i=1}^r$ are the basis of $\mathfrak{h}^*$ dual to the simple coroots. They satisfy
\begin{align}
\label{eq:fundamental_weights}
   \omega_i(h_{\alpha_j})
   =
   \frac{2\langle \omega_i,\alpha_j\rangle}
        {\langle \alpha_j,\alpha_j\rangle}
   =
   \delta_{ij},
   \quad 1\leqslant i,j\leqslant r,
\end{align}
where $h_{\alpha_j}$ is the coroot associated with the simple root $\alpha_j\in\Pi$. 
\end{mydef}

The fundamental weights therefore form the basis dual to the simple coroots and provide the natural coordinate system for highest weights. In this basis, every dominant integral highest weight can be expressed as a linear combination of fundamental weights, with the corresponding coefficients known as the Dynkin coefficients.

\begin{mydef}{Dynkin Coefficients}{Dynkin_Coefficients}
Consider a finite-dimensional irreducible highest weight representation $V(\Lambda)$ of a complex semisimple Lie algebra $\g$, with highest weight $\Lambda$. 
Once a system of simple roots $\Pi=\{\alpha_i\}_{i=1}^r$ (see Eq.~\eqref{eq:simple_roots_set}) has been fixed, the associated fundamental weights $\Omega=\{\omega_i\}_{i=1}^r$ form the basis dual to the simple coroots $\{h_{\alpha_i}\}_{i=1}^r$, as in Eq.~\eqref{eq:fundamental_weights}. 
They therefore provide the coordinates for highest weights. 
In this basis, $\Lambda$ admits the expansion
\begin{align}
\label{eq:highest_weight_expansion}
    \Lambda
    =\sum_{i=1}^r
    \lambda_i\omega_i,
\end{align}
The coefficients $\{\lambda_i\}_{i=1}^{r}$ are non-negative integers and are called the Dynkin coefficients of the representation $V(\Lambda)$, and they are given by
\begin{align}
\label{eq:dynkin_coefficients}
    \lambda_i
    =\Lambda(h_{\alpha_i})
    =\frac{2\langle\Lambda,\alpha_i\rangle}{\langle\alpha_i,\alpha_i\rangle }.
\end{align}
\end{mydef}

We now turn to the second central concept of this subsection: 
the Cartan matrix. 
Formally, it is defined as follows.

\begin{mydef}{Cartan Matrix}{Cartan_Matrix}
A choice of simple roots $\Pi=\{\alpha_i\}_{i=1}^r$ (see Eq.~\eqref{eq:simple_roots_set}) determines the corresponding simple coroots $\{h_{\alpha_j}\}_{j=1}^r$. 
The Cartan matrix is the matrix $C=(C_{ij})_{1\leqslant i,j\leqslant r}$ defined by
\begin{align}
    C\coloneqq(C_{ij})_{1\leqslant i,j\leqslant r},
\end{align}
with the coefficients
\begin{align}
\label{eq:Cartan_matrix}
    C_{ij} 
    \coloneqq
    \alpha_i(h_{\alpha_j})
    =
    \frac{2\langle\alpha_i,\alpha_j\rangle}{\langle\alpha_j,\alpha_j\rangle}.
\end{align}
In this convention, the $j$-th column of $C$ gives the evaluations of all simple roots on the simple coroot $h_{\alpha_j}$.
\end{mydef}

This subsection introduced Dynkin coefficients (see Def.~\ref{def:Dynkin_Coefficients}) for highest weight representations and the Cartan matrix (see Def.~\ref{def:Cartan_Matrix}) of root system. 
Together, these concepts convert the bound $2\langle\Lambda,\delta\rangle$ into intrinsic Lie-theoretic quantities determined entirely by the highest weight and the underlying root system. 
The  subsection~\ref{subsubsec:exceptional_Formula_Bound} uses this coordinate formulation to derive an explicit Cartan matrix expression for the uncertainty lower bound on compact semisimple Lie algebras.


\subsection{Dynkin Diagram}
\label{subsec:Dynkin_Diagram}

A central objective of this work is to connect the algebraic structure of the observables with the optimal lower bound in the corresponding Lie algebraic uncertainty relation. 
Dynkin diagrams provide the natural language for this connection: 
they encode the root system structure of a simple Lie algebra in a compact graphical form, while retaining precisely the information needed to determine the optimal bound. 
Although these diagrams were introduced in Subsec.~\ref{subsec:Classification_Lie}, the connection between Dynkin diagrams and Cartan matrices has not yet been made explicit. 
We now develop this correspondence in both directions. 
This prepares the ground for the next section, where the optimal uncertainty bounds can be read directly from the Dynkin diagrams.

An $r$-node Dynkin diagram determines an $r\times r$ matrix
$C=(C_{ij})_{1\leqslant i,j\leqslant r}$ (see Def.~\ref{def:Cartan_Matrix}), with one row and one column assigned to each simple root (see Def.~\ref{def:Simple_Roots}). 
The diagonal entries are fixed as
\begin{align}
C_{ii}=2.
\end{align}
The off-diagonal entries encode the pairwise relations between distinct simple roots and, for finite-dimensional simple Lie algebras, take values in
\begin{align}
    C_{ij}\in\{0,-1,-2,-3\},
    \quad
    \forall 1\leqslant i,j\leqslant r,\quad i\neq j.
\end{align}
These entries are read directly from the bonds between the corresponding nodes. 
If two nodes $i$ and $j$ are disconnected, then
\begin{align}
C_{ij}=C_{ji}=0,
\end{align}
whereas a single bond gives
\begin{align}
C_{ij}=C_{ji}=-1.
\end{align}

The simplest example is provided by the type-A Lie algebra (see Tab.~\ref{tab:Simple_Lie_Algebra}). 
For $A_3$, the Dynkin diagram is a linear chain of 3 nodes,
\begin{align}
\raisebox{-0.8ex}{\includegraphics[scale=1]{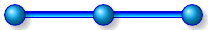}} \, .
\end{align}
so the corresponding Cartan matrix (see Def.~\ref{def:Cartan_Matrix}) takes the form
\begin{align}
C
=
\begin{pmatrix}
 2 & -1 &  0 \\
-1 &  2 & -1 \\
 0 & -1 &  2
\end{pmatrix}.
\end{align}
Thus the graphical data of the Dynkin diagram are faithfully translated into matrix form: 
its nodes determine the dimension of the Cartan matrix, and its bonds encode the off-diagonal entries describing the pairwise relations among simple roots.

For multiple bonds, the bond multiplicity is determined by the product of the two corresponding off-diagonal entries:
\begin{align}
C_{ij}C_{ji}
=
b_{ij}\in\{2,3\},
\end{align}
where $b_{ij}=2$ corresponds to a double bond and $b_{ij}=3$ to a triple bond. 
The arrow points toward the short root. 
With $C_{ij}=2\langle\alpha_i,\alpha_j\rangle/\langle\alpha_j,\alpha_j\rangle$, an arrow from a long-root node $i$ toward a short-root node $j$ means 
\begin{align}
|C_{ij}|>|C_{ji}|,
\end{align}
so that
\begin{align}
(C_{ij},C_{ji})=(-2,-1)
\end{align}
for a double bond, and
\begin{align}
(C_{ij},C_{ji})=(-3,-1)
\end{align}
for a triple bond. 
In this way, the full graphical content of the Dynkin diagram is translated into the Cartan matrix.

A representative non-simplylaced example is provided by the type-B series. 
Consider, in particular, $B_3$. 
We label the three nodes of its Dynkin diagram from left to right by the simple roots $\alpha_1,\alpha_2,\alpha_3$, so that
\begin{align}
 \raisebox{-1.8ex}{\includegraphics[scale=1]{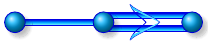}} \, .
\end{align}
The double bond joins the nodes associated with $\alpha_2$ and $\alpha_3$, and its arrow points toward the shorter root, $\alpha_3$. 
Following the convention above, this graphical asymmetry is translated into the Cartan matrix (see Def.~\ref{def:Cartan_Matrix})
\begin{align}
C
=
\begin{pmatrix}
2 & -1 & 0 \\
-1 & 2 & -2 \\
0 & -1 & 2
\end{pmatrix}.
\end{align}
Thus the unequal root lengths encoded by the double bond appear algebraically as the asymmetric pair
$C_{23}=-2$ and $C_{32}=-1$.

This dictionary between Dynkin diagrams (see Tab.~\ref{tab:Simple_Lie_Algebra}) and Cartan matrices (see Def.~\ref{def:Cartan_Matrix}) completes the bridge between graphical root system data and algebraic structure. 
It shows that the information encoded by the simple roots -- their connectivity, relative lengths, and bond multiplicities -- can be translated directly into the Cartan matrix. 
This prepares the ground for Sec.~\ref{sec:Simple_Lie_Algebraic_Quantum_Uncertainty}, where the same data will be used to read off optimal Lie algebraic uncertainty bounds directly from Dynkin diagrams. 
In this way, the geometry of root systems becomes an explicit organizing principle for quantum incompatibility and uncertainty.


\subsection{Explicit Formula for the Bound}
\label{subsubsec:exceptional_Formula_Bound}
The uncertainty bound $2\langle\Lambda,\delta\rangle$ obtained in Thm.~\ref{thm:Lie_Algebraic_Uncertainty_Relations} has a simple form, but its algebraic compactness hides the ingredients needed for direct evaluation. 
To compute the bound explicitly, the highest weight $\Lambda$ and the root system $\Phi$ must be expressed in a common coordinate framework.
This is precisely what the fundamental weights and the Cartan matrix provide. 
In this subsection, we turn the optimal lower bound into an explicit Cartan matrix formula.
Writing $\Lambda=\sum_{i=1}^r\lambda_i\omega_i$, the optimal lower bound is determined entirely by the Dynkin coefficients $\lambda_i$ (see Def.~\ref{def:Dynkin_Coefficients}), the inverse Cartan matrix $C^{-1}$ (see Def.~\ref{def:Cartan_Matrix}), and the simple root lengths (see Def.~\ref{def:Simple_Roots}). 
The abstract quantity $2\langle\Lambda,\delta\rangle$ is therefore converted into a concrete expression read directly from the Lie-theoretic structure of the representation. 
The result is as follows.

\begin{mythm}{Explicit Bounds for Lie Algebraic Uncertainty}{Explicit_Uncertainty_Bound}
Let $V(\Lambda)$ be an irreducible highest weight representation (see Subsec.~\ref{subsec:Highest_Weight_Representations}) of a complex semisimple Lie algebra $\g$ of rank $r$, with compact real form $\g_{\mathbb R}$. 
Write the highest weight in the fundamental weight basis as $\Lambda=\sum_{i=1}^r\lambda_i\omega_i$ (see Def.~\ref{def:Fundamental_Weights}), where $\lambda_i\in\mathbb Z_{\geqslant 0}$ are the Dynkin coefficients (see Def.~\ref{def:Dynkin_Coefficients}) of the representation. 
Let $C$ be the Cartan matrix (see Def.~\ref{def:Cartan_Matrix}) of the  simple roots $\Pi=\{\alpha_i\}_{i=1}^r$ (see Eq.~\eqref{eq:simple_roots_set}), and let $C^{-1}$ denote its inverse. 
Then, for every quantum state $\ket{\psi}\in V(\Lambda)$, the Lie algebraic uncertainty satisfies
\begin{align}
\label{eq:final_bound}
    \Delta^2_{\psi}(\g)
    \geqslant \sum_{i,j=1}^r 
    \lambda_i C^{-1}_{ij} \langle\alpha_j,\alpha_j\rangle.
\end{align}
The bound is optimal and is attained by a highest-weight state (see Def.~\ref{def:Highest_Weight_Vector}) and its compact-group orbit.
Here $C^{-1}_{ij}$ denotes the $(i,j)$-entry of inverse Cartan matrix $C^{-1}$.
\end{mythm}

\begin{proof}
A structural identity now becomes useful.
The Weyl vector (see Def.~\ref{def:Weyl_Vector}) has an especially simple form in the fundamental-weight basis: 
it is the sum of all fundamental weights.
\begin{align}
\label{eq:delta_sum}
    \delta
    =\sum_{i=1}^r \omega_i.
\end{align}
To see this, recall that the Weyl vector satisfies
\begin{align}
    \delta(h_{\alpha_j})
    =1,
\end{align}
for every simple coroot $h_{\alpha_j}$. 
On the other hand, the fundamental weights are defined by (see Eq.~\eqref{eq:fundamental_weights})
\begin{align}
    \omega_i(h_{\alpha_j})
    =\delta_{ij}.
\end{align}
Therefore,
\begin{align}
    \left(\sum_{i=1}^r\omega_i \right) (h_{\alpha_j})
    =
    \sum_{i=1}^r \delta_{ij}
    =1.
\end{align}
for every $j\in\{1,\ldots,r\}$. 
Since the simple coroots form a basis of the Cartan subalgebra $\mathfrak h$, the two functionals have identical values on a basis and hence coincide in $\mathfrak h^*$. 
Thus Eq.~\eqref{eq:delta_sum} follows.

Substituting the fundamental weight expansions of $\delta$ and $\Lambda$ from Eqs.~\eqref{eq:delta_sum} and \eqref{eq:highest_weight_expansion}, respectively, gives
\begin{align}
\label{eq:Lambda_delta_expand}
    2\langle\Lambda,\delta\rangle
    =2\sum_{i,j=1}^r\lambda_i\langle\omega_i,\omega_j\rangle .
\end{align}
It remains to express the inner products of fundamental weights, i.e., $\langle\omega_i,\omega_j\rangle$, in terms of the
Cartan matrix and the simple root lengths $\langle\alpha_j,\alpha_j\rangle$. 
Since the fundamental weights
$\{\omega_i\}_{i=1}^r$ form the basis dual to the simple coroots, each simple
root can be expanded as
\begin{align}
\label{eq:root_weight_relation}
    \alpha_i
    = \sum_{j=1}^r C_{ij}\omega_j ,
\end{align}
where $C_{ij}=\alpha_i(h_{\alpha_j})$ forms the Cartan matrix. 
Inverting this relation gives
\begin{align}
    \omega_i =\sum_{k=1}^r C^{-1}_{ik}\alpha_k.
\end{align}
Taking the inner product with $\omega_j$ and using Eq.~\eqref{eq:fundamental_weights}, we obtain 
\begin{align}
\label{eq:fundamental_inner}
    \langle\omega_i,\omega_j\rangle
    &=
    \sum_{k=1}^r C^{-1}_{ik} \langle\alpha_k,\omega_j\rangle \\
    &=
    \frac12 \sum_{k=1}^r C^{-1}_{ik} \langle\alpha_k,\alpha_k\rangle \delta_{jk} \\
    &=
    \frac12 C^{-1}_{ij} \langle\alpha_j,\alpha_j\rangle.
\end{align}
Equation~\eqref{eq:Lambda_delta_expand} therefore becomes 
\begin{align}
    2\langle\Lambda,\delta\rangle
    =\sum_{i,j=1}^r \lambda_i C^{-1}_{ij} \langle\alpha_j,\alpha_j\rangle,
\end{align}
this gives the claimed Cartan matrix expression for the lower bound and completes the proof.
\end{proof}

To put the optimal lower bound in its most compact form, we introduce two pieces of notation. 
Let
\begin{align}\label{eq:Dynkin_Vector}
    \vec{\lambda}
    \coloneqq
    \begin{pmatrix}
    \lambda_1 \\
    \lambda_2 \\
    \vdots \\
    \lambda_r
    \end{pmatrix}
\end{align}
be the Dynkin coefficients vector of the highest weight (see Def.~\ref{def:Dynkin_Coefficients}), and let
\begin{align}\label{eq:Roots_Vector}
    \vec{\alpha}
    \coloneqq
    \begin{pmatrix}
    \langle\alpha_1,\alpha_1\rangle \\
    \langle\alpha_2,\alpha_2\rangle \\
    \vdots \\
    \langle\alpha_r,\alpha_r\rangle
    \end{pmatrix}
    =
    \begin{pmatrix}
    \|\alpha_1\|^2 \\
    \|\alpha_2\|^2 \\
    \vdots \\
    \|\alpha_r\|^2
    \end{pmatrix}
\end{align}
collect the squared lengths of the simple roots $\Pi=\{\alpha_i\}_{i=1}^r$ (see Eq.~\eqref{eq:simple_roots_set}).
Then the optimal Lie algebraic uncertainty relation takes the compact form

\begin{mycor}{Optimal Lie Algebraic Uncertainty Relations}{Optimal_Lie_Algebraic_Uncertainty_Relations}
For an irreducible highest weight representation $V(\Lambda)$ of a compact semisimple Lie algebra $\g_{\mathbb R}$ and its complexification $\g$ of rank $r$, the optimal state-independent lower bound on the Lie algebraic uncertainty is given by
\begin{align}\label{eq:Lie_Algebraic_UR_Optimal}
    \Delta^2_{\psi}(\g)
    \geqslant
    \vec{\lambda}^{\,\T}\cdot 
    C^{-1}\cdot
    \vec{\alpha},
\end{align}
where $\vec{\lambda}$ (see Eq.~\eqref{eq:Dynkin_Vector}) is the Dynkin coefficient vector, $\vec{\alpha}$ (see Eq.~\eqref{eq:Roots_Vector}) collects the squared simple root lengths, $C^{-1}$ denotes the inverse Cartan matrix (see Def.~\ref{def:Cartan_Matrix}).
\end{mycor}

As a concrete illustration, we evaluate our results for $\mathfrak{su}(3)$ in its representation on $\mathbb C^3$. 
In this representation, the highest weight, Dynkin coefficients, Cartan matrix, and simple root lengths are all explicit. 
The lower bound can therefore be computed directly. 
We denote the computational basis by 
\begin{align}
\ket{0}
=
\begin{pmatrix}
1\\0\\0
\end{pmatrix},
\quad
\ket{1}
=
\begin{pmatrix}
0\\1\\0
\end{pmatrix},
\quad
\ket{2}
=
\begin{pmatrix}
0\\0\\1
\end{pmatrix}.
\end{align}
For the complexified Lie algebra
$\mathfrak{sl}(3,\mathbb C)$,
its Cartan subalgebra $\mathfrak{h}$ is the subspace of traceless diagonal matrices, given by
\begin{align}
    \mathfrak{h} 
    = 
    \left\{ \begin{pmatrix} t_1 & 0 & 0 \\ 0 & t_2 & 0 \\ 0 & 0 & t_3 \end{pmatrix} \;\middle|\; 
    t_1, t_2, t_3 \in \mathbb{C}, \,
    \sum_{i=1}^3 t_i = 0 \right\}
    =
    \mathrm{span}\{h_1,h_2\},
\end{align}
where the Cartan generators $h_1$ and $h_2$ are
\begin{align}
    h_1
    \coloneqq
    \begin{pmatrix}
        1&0&0\\
        0&-1&0\\
        0&0&0
    \end{pmatrix},
    \quad
    h_2
    \coloneqq
    \begin{pmatrix}
        0&0&0\\
        0&1&0\\
        0&0&-1
    \end{pmatrix}.
\end{align}

Let $f_i\in\mathfrak h^*$ denote the functional on the Cartan subalgebra $\mathfrak h$, defined by
\begin{align}
    f_i(h)
    =
    \bra{i-1}h\ket{i-1}=t_i,
    \quad
    \forall h\in\mathfrak h.
\end{align}
where $i=1,2,3$.
Consider the upper-triangular matrices
\begin{align}
    E_{12}
    =
    \begin{pmatrix}
        0&1&0\\
        0&0&0\\
        0&0&0
    \end{pmatrix},
    \quad
    E_{23}
    =
    \begin{pmatrix}
        0&0&0\\
        0&0&1\\
        0&0&0
    \end{pmatrix}, \quad
   E_{13}
    =
    \begin{pmatrix}
        0&0&1\\
        0&0&0\\
        0&0&0
    \end{pmatrix}.
\end{align}
Their commutators with an arbitrary Cartan element $h=\mathrm{diag}\{t_1,t_2,t_3\}\in \mathfrak h$ are
\begin{align}
    [h,E_{12}]
   & =
    (t_1-t_2)E_{12}
    =
    (f_1-f_2)(h)E_{12}\\
     [h,E_{23}]
   & =
    (t_2-t_3)E_{23}
    =
    (f_2-f_3)(h)E_{23}\\
     [h,E_{13}]
    & =
    (t_1-t_3)E_{13}
    =
    (f_1-f_3)(h)E_{13}.
\end{align}
In general, for the matrix $E_{ij}$, whose $(i,j)$-entry is $1$ and whose remaining entries are all zero, one has
\begin{align}
    [h,E_{ij}]=(f_i-f_j)(h)E_{ij}.
\end{align}
Hence the root associated with matrix $E_{ij}$ is $f_i-f_j$.
We now choose the positive simple roots to be
\begin{align}
    \alpha_1
    \coloneqq
    f_1-f_2,
    \quad
    \alpha_2
    \coloneqq
    f_2-f_3.
\end{align}
With this convention, the remaining positive root is
\begin{align}
    \alpha_3
    \coloneqq
    \alpha_1+\alpha_2
    =
    f_1-f_3.
\end{align}
Thus the full root system of $\mathfrak{sl}(3,\mathbb C)$ is
\begin{align}
    \Phi
    =
    \{
    \pm\alpha_1,
    \pm\alpha_2,
    \pm\alpha_3
    \}.
\end{align}
The corresponding positive root system is
\begin{align}
    \Phi^+
    =
    \{
    \alpha_1,\alpha_2,\alpha_3
    \}
    =
    \{
    f_1-f_2,\,
    f_2-f_3,\,
    f_1-f_3
    \}.
\end{align}
In this choice of positive roots, the simple roots (see Def.~\ref{def:Simple_Roots}) are precisely
\begin{align}
    \Pi
    =
    \{\alpha_1,\alpha_2\}.
\end{align}

Returning to the root space decomposition (see Eq.~\eqref{eq:Root_Space_Decomposition}), we can write the Chevalley basis (see Def.~\ref{def:Chevalley_Basis}) of $\mathfrak{sl}(3,\mathbb C)$ explicitly as follows.
\begin{equation}
\left\{
\begin{gathered}
    h_{\alpha_1}\coloneqq h_1,\quad h_{\alpha_2}\coloneqq h_2,\\
    e_{\alpha_1}\coloneqq E_{12},\quad e_{\alpha_2}\coloneqq E_{23}, \quad e_{\alpha_3}\coloneqq E_{13},\\
    f_{\alpha_1}\coloneqq E_{21},\quad f_{\alpha_2}\coloneqq E_{32}, \quad f_{\alpha_3}\coloneqq E_{31}
\end{gathered}
\right\}
\end{equation}
A direct computation then verifies that these generators satisfy the commutation relations.
\begin{align}
    [e_{\alpha_1},f_{\alpha_1}]
    =
    [E_{12},E_{21}]
    =
    E_{11}-E_{22}
    =
    h_1
    =
    h_{\alpha_1}, \\
    [e_{\alpha_2},f_{\alpha_2}]
    =
    [E_{23},E_{32}]
    =
    E_{22}-E_{33}
    =
    h_2
    =
    h_{\alpha_2}.
\end{align}

We next identify the weight space decomposition of the representation (see Thm.~\ref{thm:Weight_Space_Decomposition}). 
For $i\in\{1,2,3\}$, let $V_{\mu_i}$ denote the one-dimensional subspace spanned by the computational basis vector $\ket{i-1}$; 
that is
\begin{align}
    V_{\mu_i}
    \coloneqq
    \mathrm{span}\{\ket{i-1}\}.
\end{align}
Acting with the Cartan generator $h_1$ and $h_2$ on $V_{\mu_1}$ gives
\begin{align}
    h_1\ket{0}=\ket{0},
    \quad
    h_2\ket{0}=0.
\end{align}
Hence the corresponding weight $\mu_1\in\mathfrak h^*$ is characterized by
\begin{align}
    \mu_1(h_1)=1,
    \quad
    \mu_1(h_2)=0.
\end{align}
Since the functional $f_1$ satisfies the same relations, we have
\begin{align}
    \mu_1=f_1.
\end{align}
Repeating the same calculation for the remaining computational basis vectors gives
\begin{align}
    \mu_i=f_i,
    \quad
    \forall i\in\{1,2,3\}.
\end{align}

Therefore the weight space decomposition of the representation is
\begin{align}
    \mathbb C^3=
    \underbrace{V_{f_1}}_{\text{spanned by $\ket{0}$}}
    \oplus 
    \underbrace{V_{f_2}}_{\text{spanned by $\ket{1}$}}
    \oplus 
    \underbrace{V_{f_3}}_{\text{spanned by $\ket{2}$}}.
\end{align}
The positive root generators annihilate $\ket{0}$,
\begin{align}
    e_{\alpha_1}\ket{0}
    =
    e_{\alpha_2}\ket{0}
    =
    e_{\alpha_3}\ket{0}
    =
    0.
\end{align}
Thus $\ket{0}$ is the highest weight vector of the representation, and its weight $f_1$ is the highest weight
\begin{align}
    \Lambda=f_1.
\end{align}
We now determine the fundamental weights in this explicit realization. 

By definition, the fundamental weights satisfy $\omega_i(h_{\alpha_j})=\delta_{ij}$ (see Eq.~\eqref{eq:fundamental_weights}). 
Thus, for the simple roots $\alpha_1, \alpha_2$, the first fundamental weight $\omega_1$ is characterized by
\begin{align}
    \omega_1(h_1)=1,
    \quad
    \omega_1(h_2)=0.
\end{align}
These are precisely the values of the functional $f_1$ on the simple coroots.
Hence, we obtain 
\begin{align}
    \omega_1=f_1.
\end{align}
The same calculation determines the remaining fundamental weight:
\begin{align}
    \omega_2=-f_3.
\end{align}
Collecting the results, the fundamental weights of $\mathfrak{sl}(3,\mathbb C)$ in this convention are
\begin{align}
    \{\omega_1
    =
    f_1,\,
    \omega_2
    =
    -f_3\}.
\end{align}
From the identification of the fundamental weights, the highest weight $\Lambda$ of the representation is precisely the first fundamental weight $\omega_1$, i.e., $\Lambda=\omega_1$, hence the Dynkin coefficient vector (see Eq.~\eqref{eq:Dynkin_Vector}) of the highest weight is
\begin{align}\label{eq:sl_3_Dynkin_Vector}
    \vec{\lambda}
    =
    \begin{pmatrix}
    1 \\
    0 
    \end{pmatrix}.
\end{align}

We proceed to the Cartan matrix $C$, whose entries are $C_{ij}=\alpha_i(h_{\alpha_j})$  (see Eq.~\eqref{eq:Cartan_matrix}).
For $\mathfrak{sl}(3,\mathbb C)$, the simple roots satisfy
\begin{align}
    &C_{11}=
    \alpha_1(h_{\alpha_1})=
    (f_1-f_2)(h_1)=1+1=2,\\
    &C_{12}=
    \alpha_1(h_{\alpha_2})=
    (f_1-f_2)(h_2)=0-1=-1,\\
    &C_{21}=
    \alpha_2(h_{\alpha_1})=
    (f_2-f_3)(h_1)=-1-0=-1,\\
    &C_{22}=
    \alpha_2(h_{\alpha_2})=
    (f_2-f_3)(h_2)=1+1=2.
\end{align}
Therefore, the Cartan matrix is
\begin{align}
    C
    =
    \begin{pmatrix}
        2&-1\\
        -1&2
    \end{pmatrix},
\end{align}
with inverse
\begin{align}\label{eq:sl_3_Inverse_Cartan_Matrix}
    C^{-1}
    =
    \frac13
    \begin{pmatrix}
        2&1\\
        1&2
    \end{pmatrix}.
\end{align}

The final ingredient is the root vector $\vec{\alpha}$ of squared simple-root lengths (see Eq.~\eqref{eq:Roots_Vector}).
For $\mathfrak{sl}(3,\mathbb C)$ with the fixed normalized invariant form, each simple root has squared length $2$:
\begin{align}
    \langle\alpha_i,\alpha_i\rangle=2.
\end{align}
Thus, the vector $\vec{\alpha}$ of squared simple-root lengths is 
\begin{align}\label{eq:sl_3_Root_Vector}
    \vec{\alpha}
    =
    \begin{pmatrix}
    2 \\
    2 
    \end{pmatrix}.
\end{align}

With the Dynkin coefficient vector $\vec{\lambda}$ (see Eq.~\eqref{eq:sl_3_Dynkin_Vector}), inverse Cartan matrix $C^{-1}$ (see Eq.~\eqref{eq:sl_3_Inverse_Cartan_Matrix}), and squared-length vector $\vec{\alpha}$ (see Eq.~\eqref{eq:sl_3_Root_Vector}) now fixed explicitly, the optimal Lie Algebraic uncertainty relation (see Eq.~\eqref{eq:Lie_Algebraic_UR_Optimal}) follows by direct substitution
\begin{align}\label{eq:sl_3_UR_Optimal}
    \Delta^2_{\psi}(\mathfrak{sl}(3,\mathbb C))=
    \Delta^2_{\psi}(\mathfrak{su}(3))\geqslant
    \vec{\lambda}^{\,\T}\cdot 
    C^{-1}\cdot
    \vec{\alpha}=
    \begin{pmatrix} 1 & 0 \end{pmatrix}
    \cdot
    \frac13
    \begin{pmatrix}
        2&1\\
        1&2
    \end{pmatrix}
    \cdot
    \begin{pmatrix}
    2 \\
    2 
    \end{pmatrix}=2.
\end{align}


\section{Simple Lie Algebraic Quantum Uncertainty}
\label{sec:Simple_Lie_Algebraic_Quantum_Uncertainty}

In this section, we further simplify the optimal Lie algebraic uncertainty bound by separating the contribution of quantum uncertainty into two distinct components. 
The first is representation-theoretic and is determined by the highest weight of the representation; 
the second is geometric and depends only on the root system of the underlying Lie algebra, independently of the representation. 
So the Dynkin labels $\vec{\lambda}$ encode the representation, while $C^{-1}\vec{\alpha}$ depends on the normalized roots data.
This separation is developed in Subsec.~\ref{subsec:Further_Reduction}. 
We then apply the general framework systematically to all simple Lie algebras, covering the classical series and the exceptional cases. 
In particular, the optimal uncertainty relations for $A_n$, $B_n$, $C_n$, $D_n$ are treated in Subsecs.~\ref{subsec:A_type}--\ref{subsec:D_type}, respectively, and $G_2$, $F_4$, $E_6$, $E_7$, and $E_8$ are derived in Subsec.~\ref{subsec:exceptional}.
We also illustrate the framework through representative examples of direct relevance to quantum theory.


\subsection{Further Reduction of the Uncertainty Bound}
\label{subsec:Further_Reduction}

Corollary~\ref{cor:Optimal_Lie_Algebraic_Uncertainty_Relations} of Subsec.~\ref{subsubsec:exceptional_Formula_Bound} gives the optimal uncertainty relation for operators that generate the underlying Lie algebra $\g$. 
We now specialize this general result to each simple Lie algebra and derive the corresponding exact bound.
Recall that the optimal lower bound $\vec{\lambda}^{\,\T}\cdot C^{-1}\cdot\vec{\alpha}$ (see Eq.~\eqref{eq:Lie_Algebraic_UR_Optimal}) can be expressed in terms of the Dynkin coefficient vector $\vec{\lambda}$ (Eq.~\eqref{eq:Dynkin_Vector}) of the highest weight $\Lambda$ (see Subsec.~\ref{subsec:Highest_Weight_Representations}), the inverse Cartan matrix $C^{-1}$ (see Def.~\ref{def:Cartan_Matrix}), and the squares simple-root length vector $\vec{\alpha}$ (see Eq.~\eqref{eq:Roots_Vector}). 
To simplify the notation, we absorb the latter two quantities into a single vector $\vec{\xi}$, i.e.,
\begin{align}\label{eq:xi}
    \vec{\xi}=
    \begin{pmatrix}
    \xi_1 \\
    \xi_2 \\
    \vdots \\
    \xi_r
    \end{pmatrix}
    \coloneqq C^{-1}\cdot\vec{\alpha},
\end{align}
where $\xi_i=\sum_{j=1}^{r}C^{-1}_{ij}\langle\alpha_j,\alpha_j\rangle$.
With this convention, the optimal uncertainty relation is determined by the Euclidean inner product $\langle\cdot,\cdot\rangle$ between $\vec{\lambda}$ and $\vec{\xi}$, namely

\begin{mycor}{Simplified Lie Algebraic Uncertainty Relations}{Simplified_Lie_Algebraic_Uncertainty_Relations}
For an irreducible highest weight representation $V(\Lambda)$ of a complex simple Lie algebra $\g$ of rank $r$, the optimal state-independent uncertainty bound admits the expression
\begin{align}\label{eq:Lie_Algebraic_UR_Simplified}
    \Delta^2_{\psi}(\g)
    \geqslant
    \vec{\lambda}^{\,\T}\cdot 
    \vec{\xi}=
    \langle\underbrace{\vec{\lambda}}_{\text{Highest Weight}},\underbrace{\vec{\xi}}_{\text{Root System}}\rangle=
    \sum_{i=1}^{r}\lambda_i\xi_i.
\end{align}
The vector $\vec{\lambda}$ collects the Dynkin coefficients of the highest weight $\Lambda$ (see Eq.~\eqref{eq:Dynkin_Vector}), while $\vec{\xi}=C^{-1}\vec{\alpha}$ contains the normalized root data (see Eq.~\eqref{eq:xi}).
Thus the bound depends on the highest weight, the root system, and the fixed metric. It is independent of the state within $V(\Lambda)$ and is attained on the highest weight group orbit.

\end{mycor}

The preceding discussion has assumed a fixed ordering of the simple roots $\Pi$ (see Def.~\ref{def:Simple_Roots}), equivalently a fixed labeling of the Dynkin diagram nodes. 
Any dependence of the bound on this choice would make the framework unnecessarily redundant, since each relabeling would require a separate analysis. 
The key question is therefore whether the optimal uncertainty bound is invariant under changes of the simple root ordering. 
Fortunately, it is.

\begin{mythm}{Simple Root Relabeling Invariance}{Simple_Root_Relabeling_Invariance}
The optimal state-independent bound in the Lie algebraic uncertainty relation is independent of the chosen ordering of the simple roots and, equivalently, of the labeling convention used for the Dynkin diagram.
\end{mythm}

\begin{proof}
Let $P$ denote the permutation matrix relabeling the simple roots. Under this relabeling, the Cartan matrix $C$ (see Def.~\ref{def:Cartan_Matrix}) and the vector $\vec{\alpha}$ (see Eq.~\eqref{eq:Roots_Vector}) of squared simple root lengths transform as
\begin{align}
    C'=P C P^{\dagger},
\end{align}
and 
\begin{align}
    \vec{\alpha'}
    =
    P\vec{\alpha}.
\end{align}
Consequently, the vector $\vec{\xi}$ (see Eq.~\eqref{eq:xi}) in the relabeled convention is given by
\begin{align}
    \vec{\xi'}
    =
    (C')^{-1}\vec{\alpha'}
    =
    (P C P^{\dagger})^{-1}P\vec{\alpha}
    =
    P C^{-1} P^{\dagger}P\vec{\alpha}
    =
    P\vec{\xi}.
\end{align}

The Dynkin coefficient vector $\vec{\lambda}$ (see Eq.~\eqref{eq:Dynkin_Vector}) is relabeled by the same permutation
\begin{align}
    \vec{\lambda'}
    =
    P\vec{\lambda}.
\end{align}
Hence, in the relabeled convention, the inner product between $\vec{\lambda'}$ and $\vec{\xi'}$ (see Eq.~\eqref{eq:Lie_Algebraic_UR_Simplified}) becomes
\begin{align}
    \langle\vec{\lambda'},\vec{\xi'}\rangle
    =
    (\vec{\lambda'})^{\T}\vec{\xi'}
    =
    (P\vec{\lambda})^{\T}
    (P\vec{\xi})
    =
    \vec{\lambda}^{\, \T}
    P^{\T}P
    \vec{\xi}
    =
    \vec{\lambda}^{\,\T}\vec{\xi}
    =
    \langle\vec{\lambda},\vec{\xi}\rangle.
\end{align}
Thus, the optimal value of the Lie algebraic uncertainty lower bound is invariant under simple roots relabeling.
\end{proof}

Tab.~\ref{tab:bounds} collects the explicit components $\xi_i$ of the vector $\vec{\xi}$ defined in Eq.~\eqref{eq:xi} for each simple Lie algebra. 
The corresponding derivations are presented in the following subsections.

\begin{table}[h]
    \centering
    \begin{tblr}{
      width = \textwidth,
      colspec = {l || c | X[l]},
      row{1} = {bg=gray!50, font=\bfseries, c}, 
      row{2} = {bg=gray!20, font=\bfseries, c},
      hlines,
      vlines,
      cells = {m},
      cell{3-6}{1-3} = {bg=mGreen},
      cell{7-11}{1-3} = {bg=mRed},
    }
      \SetCell[c=3]{c} Explicit Forms of $\vec{\xi}$ for Simple Lie Algebras & & \\
      Category & Lie Algebra & Components $\xi_i$ \\
      
      \SetCell[r=4]{c} Classical 
      & $A_n$ & $\xi_i=i(n+1-i)$,\quad $1 \leqslant i \leqslant n$. \\
      & $B_n\,(n\geqslant 2)$ & $\xi_i=i(2n-i)$,\quad$1 \leqslant i\leqslant n-1$, \quad with \quad $\xi_n=\dfrac{n^2}{2}$. \\
      & $C_n \,(n\geqslant 2) $ & $\xi_i=\dfrac{i(2n+1-i)}{2}$, \quad $1 \leqslant i\leqslant n-1$, \quad with \quad $\xi_n=\dfrac{n(n+1)}{2}$. \\
      & $D_n \,(n\geqslant 4) $& $\xi_i=i(2n-i-1)$, \quad $1 \leqslant i \leqslant n-2$, \quad with \quad $\xi_{n-1}=\xi_n=\dfrac{n(n-1)}{2}$. \\
     
      \SetCell[r=5]{c} Exceptional 
      & $G_2$ & $\xi_1=6$,\, $\xi_2=\dfrac{10}{3}$. \\
      & $F_4$ & $\xi_1=16$,\, $\xi_2=30$,\, $\xi_3=21$,\, $\xi_4=11$. \\
      & $E_6$ & $\xi_1=16$,\, $\xi_2=22$,\, $\xi_3=30$,\, $\xi_4=42$,\, $\xi_5=30$,\, $\xi_6=16$. \\
      & $E_7$ & $\xi_1=34$,\, $\xi_2=49$,\, $\xi_3=66$,\,  
      $\xi_4=96$,\, $\xi_5=75$,\, $\xi_6=52$,\, $\xi_7=27$. \\
      & $E_8$ & $\xi_1=92$,\, $\xi_2=136$,\, $\xi_3=182$,\, $\xi_4=270$,\, $\xi_5=220$,\, $\xi_6=168$,\, $\xi_7=114$,\, $\xi_8=58$. \\
    \end{tblr}
    \caption{\textbf{Explicit Components of $\vec{\xi}$ for Simple Lie Algebras.}
   Here, $\vec{\xi}$ (see Eq.~\eqref{eq:xi}) absorbs the Cartan matrix and root-length factors, with the simple roots ordered as specified in the following corresponding subsection and the normalization chosen such that long roots have squared length $2$.
  }
    \label{tab:bounds}
\end{table}


\subsection{Explicit Bounds for Type-A Lie Algebras}\label{subsec:A_type}

We begin with the first classical series, the type-A Lie algebras; 
see Tab.~\ref{tab:Simple_Lie_Algebra}. 
In this case, $A_n=\mathfrak{sl}(n+1,\mathbb C)$, with compact real form $\mathfrak{su}(n+1)$. 
The construction of $\vec{\xi}$ is governed by the Cartan matrix and the simple root lengths. We therefore start from the Cartan matrix of $A_n$. 
Its Dynkin diagram is a linear chain with $n$ nodes,
\begin{equation}
     \raisebox{-3.2ex}{\includegraphics[scale=1]{DK_A_n_No.pdf}}\, .
\end{equation}
This simply laced structure fixes the Cartan matrix (see Def.~\ref{def:Cartan_Matrix}) to be symmetric and tridiagonal. 
Explicitly, $C_{ii}=2$ for each simple root, $C_{i,i+1}=C_{i+1,i}=-1$ for adjacent roots, and $C_{ij}=0$ whenever $|i-j|>1$. 
Equivalently, the Cartan matrix is
\begin{align}
    C
    =
    \begin{pmatrix}
    2 & -1 & 0 & \cdots & 0\\
    -1 & 2 & -1 & \cdots & 0\\
    0 & -1 & 2 & \ddots & \vdots\\
    \vdots & \vdots & \ddots & \ddots & -1\\
    0 & 0 & \cdots & -1 & 2
    \end{pmatrix}.
\end{align}
Its inverse is given by
\begin{align}
    C^{-1}
    =
    \frac1{n+1}
    \begin{pmatrix}
     n & n-1 & n-2 & \cdots & 1\\
    n-1 & 2(n-1) & 2(n-2) & \cdots & 2\\
    n-2 & 2(n-2) & 3(n-2) & \cdots & 3\\
    \vdots & \vdots & \vdots & \ddots & \vdots\\
     1 & 2 & 3 & \cdots & n
\end{pmatrix}.
\end{align}
Equivalently, the $(i,j)$-entry of the inverse Cartan matrix $C^{-1}$ now reads

\begin{mylem}{Entries of $C^{-1}$ for Type-A Lie Algebras}{Inverse_Cartan_Matrix_A}
For the type-A Lie algebra $A_n$, the $(i,j)$-entry of the inverse Cartan matrix is
\begin{align}
\label{eq:An_inverse}
    C^{-1}_{ij}
    =
    \frac{\min\{i,j\}\,(n+1-\max\{i,j\})}{n+1}.
\end{align}
\end{mylem}

In this case, all simple roots have squared length $2$, namely $\langle\alpha_i,\alpha_i\rangle =2$, so the simple root length vector takes the following simple form.

\begin{mylem}{Root Length Vector $\vec{\alpha}$ for Type-A Lie Algebras}{Root_Length_Vector_A}
For the type-A Lie algebra $A_n$, the squared simple root lengths form vector 
\begin{align}
\label{eq:An_Root_Length_Vector}
    \vec{\alpha}
    =
    \begin{pmatrix}
    2 \\
    2 \\
    \vdots \\
    2
    \end{pmatrix}.
\end{align}
\end{mylem}

With the entrywise form of $C^{-1}$ (see Lem.~\ref{lem:Inverse_Cartan_Matrix_A}) and the simple root length vector $\vec{\alpha}$ (see Lem.~\ref{lem:Root_Length_Vector_A}) fixed, the vector $\vec{\xi}$ (see Eq.~\eqref{eq:xi}) can now be evaluated explicitly. 
For the type-A series, its components are given by the following theorem.

\begin{mythm}{Components of $\vec{\xi}$ for Type-A Lie Algebras}{Root_Length_Vector_A}
For the type-A Lie algebra $A_n$, the components of the vector $\vec{\xi}$ (see Eq.~\eqref{eq:xi}) are
\begin{align}
    \xi_i=i(n+1-i),\quad \forall i\in\{1, \ldots, n\}.
\end{align}
\end{mythm}

\begin{proof}
Substituting the entrywise formula for $C^{-1}$ in Eq.~\eqref{eq:An_inverse} and the simple root length vector $\vec{\alpha}$ in Eq.~\eqref{eq:An_Root_Length_Vector}, the $i$-th component $\xi_i$ of $\vec{\xi}$ is obtained from the corresponding row sum. 
Splitting the sum into the two regions $j\leqslant i$ and $j>i$, we find
\begin{align}
    \xi_i
    &=
    \sum_{j=1}^{n}
    C^{-1}_{ij}\langle \alpha_j,\alpha_j\rangle\\
    &=
    \frac2{n+1}
    \left(
    \underbrace{\left(\sum_{j=1}^{i}
    j(n+1-i)
    \right)}_{j\leqslant i}
    +
    \underbrace{\left(
    \sum_{j=i+1}^{n}
    i(n+1-j)
    \right)}_{j>i}
    \right)\\
    &=
    \frac2{n+1}
    \left(
    (n+1-i)\frac{i(i+1)}2
    +
    i\frac{(n-i)(n-i+1)}2
    \right)\\
    &=
    i(n+1-i).
\end{align}
This establishes the desired expression and completes the proof.
\end{proof}

The preceding analysis gives all the ingredients needed to express the optimal Lie algebraic uncertainty bound for type-A Lie algebras.

\begin{mycor}{Type-A Lie Algebraic Uncertainty Relations}{Type_A_Lie_UR}
For the type-A Lie algebra $A_n$ (see Tab.~\ref{tab:Simple_Lie_Algebra}), the optimal Lie algebraic uncertainty relation (see Def.~\ref{def:Lie_Algebraic_Uncertainty_Relations}) takes the explicit form
\begin{align}
    \Delta^2_{\psi}(\raisebox{-0.6ex}{\includegraphics[scale=0.8]{DK_A_n.pdf}})\coloneqq
    \Delta^2_{\psi}(\mathfrak{sl}(n+1,\mathbb C))=
    \Delta^2_{\psi}(\mathfrak{su}(n+1))\geqslant
    \sum_{i=1}^{n}i(n+1-i)\lambda_i.
\end{align}
The bound is state-independent. 
It depends only on the representation, through the highest weight $\Lambda$ (see Subsec.~\ref{subsec:Highest_Weight_Representations}), or equivalently its Dynkin coefficient vector $\vec{\lambda}$ (see Eq.~\eqref{eq:Dynkin_Vector}), and on the type-A root system geometry encoded in $\vec{\xi}$ (see Eq.~\eqref{eq:xi}). 
Its optimality follows directly from the general result established in Cor.~\ref{cor:Simplified_Lie_Algebraic_Uncertainty_Relations}.
\end{mycor}

The rank dependence of the bound $\sum_{i=1}^{n}i(n+1-i)\lambda_i$ is made explicit in Tab.~\ref{tab:explicit_bounds}, where representative cases of the type-A series are listed.

\begin{table}[h]
\centering
\begin{tblr}{
      width = \textwidth,
      colspec = {c||c|X[c]},
      row{1} = {bg=gray!50, font=\bfseries}, 
      row{2} = {bg=gray!20, font=\bfseries},
      hlines,
      vlines,
      cells = {m},
      cell{3-10}{1} = {bg=gray!20},          
      cell{3-10}{2} = {bg=mRed},            
      cell{3-10}{3} = {bg=mGreen},           
}

\SetCell[c=3]{c}
\textbf{Type-A Lie Algebraic Uncertainty Relations\,\, \raisebox{-0.6ex}{\includegraphics[scale=0.8]{DK_A_n.pdf}}}
&
&
\\

Compact Lie algebra
&
Rank
&
Optimal Bound
\\

$\mathfrak{su}(2)$
&
$1$
&
$\lambda_1$
\\

$\mathfrak{su}(3)$
&
$2$
&
$2\lambda_1+2\lambda_2$
\\

$\mathfrak{su}(4)$
&
$3$
&
$3\lambda_1+4\lambda_2+3\lambda_3$
\\

$\mathfrak{su}(5)$
&
$4$
&
$4\lambda_1+6\lambda_2+6\lambda_3+4\lambda_4$
\\

$\mathfrak{su}(6)$
&
$5$
&
$5\lambda_1+8\lambda_2+9\lambda_3+8\lambda_4+5\lambda_5$
\\

$\mathfrak{su}(7)$
&
$6$
&
$\displaystyle
6\lambda_1+10\lambda_2+12\lambda_3
+12\lambda_4+10\lambda_5+6\lambda_6$
\\

$\vdots$
&
$\vdots$
&
$\vdots$
\\

$\mathfrak{su}(n+1)$
&
$n$
&
$\displaystyle
n\lambda_1
+2(n-1)\lambda_2
+3(n-2)\lambda_3
+\cdots
+3(n-2)\lambda_{n-2}
+2(n-1)\lambda_{n-1}
+n\lambda_n$
\\

\end{tblr}

\caption{
\textbf{Type-A Lie Algebraic Uncertainty Relations.}
For the compact Lie algebra $\mathfrak{su}(n+1)$ with Dynkin coefficient vector $\vec{\lambda}=(\lambda_1,\ldots,\lambda_n)^{\T}$, the state-independent optimal bound is
$\sum_{i=1}^{n} i(n+1-i)\lambda_i \,$.}
\label{tab:explicit_bounds}
\end{table}

Tab.~\ref{tab:explicit_bounds} expresses the uncertainty bounds in terms of Dynkin coefficient vector $\vec{\lambda}$ (see Eq.~\eqref{eq:Dynkin_Vector}). 
In physical applications, however, the representation is often specified more concretely, for example as the defining, adjoint, symmetric, or antisymmetric representation. To connect the general formula of Eq.~\eqref{eq:Lie_Algebraic_UR_Simplified} with these standard settings, Tab.~\ref{tab:An_representations} translates the general Lie algebraic uncertainty bound into explicit values for representations of $\mathfrak{su}(n+1)$. 
Once the Lie algebra and the representation are fixed, the corresponding optimal bound is obtained immediately.

\begin{table}[h]
\centering
\begin{tblr}{
      width = \textwidth,
      colspec = {X[c]||X[c]|X[c]|X[c]},
      row{1} = {bg=gray!50, font=\bfseries}, 
      row{2} = {bg=gray!20, font=\bfseries}, 
      hlines,
      vlines,
      cells = {m},                       
      cell{3-8}{1} = {bg=gray!20},          
      cell{3-8}{2} = {bg=mRed},            
      cell{3-8}{3} = {bg=mGreen},           
      cell{3-8}{4} = {bg=mBlue},            
}
      
\SetCell[c=4]{c} \textbf{Type-A Lie Algebraic Uncertainty Relations Across Representations} & & & \\
Representation & Highest Weight $\Lambda$ & Dynkin Vector $\vec{\lambda}^{\,\T}$ & Optimal Bound \\

Defining & $\omega_1$ & $(1,0,\dots,0)$ & $n$ \\

Dual Defining & $\omega_n$ & $(0,\dots,0,1)$ & $n$ \\

Adjoint & $\omega_1+\omega_n$ & $(1,0,\dots,0,1)$ & $2n$ \\

Second Symmetric Power & $2\omega_1$ & $(2,0,\dots,0)$ & $2n$ \\

Second Antisymmetric Power & $\omega_2$ & $(0,1,0,\dots,0)$ & $2(n-1)$ \\

Fundamental ($i$-th) & $\omega_i$ & $(0,\dots,1,\dots,0)$ & $i(n+1-i)$ \\

\end{tblr}
\caption{\textbf{Lie Algebraic Uncertainty Relations Across Representations of $\mathfrak{su}(n+1)$.} 
For several standard representations of $\mathfrak{su}(n+1)$, the table gives the highest weight $\Lambda$, the transpose of the Dynkin vector $\vec{\lambda}^{\,\T}$, and the corresponding optimal lower bound.
The Dynkin vector displayed in the ``Adjoint'' row applies
for $n\geqslant2$; for $A_1$, it is $(2)$.
The ``Second Antisymmetric Power'' row and Eq.~\eqref{eq:second} assume $n\geqslant2$; for $A_1$, $\bigwedge^2\mathbb C^2$ is the trivial
representation, with optimal uncertainty bound $0$.}
\label{tab:An_representations}
\end{table}

A fundamental representation is specified by a highest weight $\Lambda$ whose Dynkin coefficient vector $\vec{\lambda}$ (see Eq.~\eqref{eq:Dynkin_Vector}) contains exactly one nonzero entry, equal to 1. 
In Tab.~\ref{tab:An_representations}, for instance, the defining representation, the dual defining representation, and the second antisymmetric power representation are all fundamental representations, distinguished by the location of the nonzero Dynkin coefficient. 
Throughout this work, we refer to these cases as the first, last, and second fundamental representations, respectively.
\begin{align}
    \text{Defining}
    \quad\rightleftharpoons\quad
    &\vec{\lambda}^{\,\T}=(1,0,\dots,0)
    \quad\rightleftharpoons\quad
    \text{Fundamental (first)}\\
    \text{Dual Defining}
    \quad\rightleftharpoons\quad
    &\vec{\lambda}^{\,\T}=(0,0,\dots,1)
    \quad\rightleftharpoons\quad
    \text{Fundamental (last)}\\
    \text{Second Antisymmetric Power}
    \quad\rightleftharpoons\quad
    &\vec{\lambda}^{\,\T}=(0,1,\dots,0)
    \quad\rightleftharpoons\quad
    \text{Fundamental (second)}\label{eq:second}
\end{align}

The analysis above establishes the optimal type-A Lie algebraic uncertainty bound in closed form. 
The bound is fixed entirely by the highest weight of the representation and the simple root geometry of the underlying Lie algebra, and is therefore independent of any particular quantum state. 
This shows that the sharp uncertainty limit is an intrinsic algebraic constraint imposed by the representation itself. 
For $\mathfrak{su}(n+1)$, the formula gives explicit optimal bounds for the defining, adjoint, symmetric, antisymmetric, and fundamental representations, turning the abstract highest weight expression into concrete uncertainty limits for the representation spaces most commonly used in quantum theory.

To make the construction more transparent and to show how the general formula yields explicit uncertainty bounds, we consider the adjoint representation as a representative example. 
For a type-A Lie algebra of rank $n$, the Cartan subalgebra $\mathfrak h$ (see Def.~\ref{def:Roots}) has dimension $n$. 
We therefore have (see Tab.~\ref{tab:An_representations})
\begin{align}\label{eq:adj_A_n}
    \Delta^2_{\mathrm{ad}}(\raisebox{-0.6ex}{\includegraphics[scale=0.8]{DK_A_n.pdf}})
    \geqslant2n.
\end{align}
Here the subscript $\mathrm{ad}$ indicates that the corresponding representation is the adjoint representation. 
To make the discussion still more explicit, we illustrate below the first three cases, 
\begin{align}
    \Delta^2_{\mathrm{ad}}( \raisebox{-0.6ex}{\includegraphics[scale=0.8]{DK_B_1.pdf}} )
    &\geqslant2,\\
    \Delta^2_{\mathrm{ad}}(\raisebox{-0.8ex}{\includegraphics[scale=0.8]{DK_A_2.pdf}})
    &\geqslant4,\\
    \Delta^2_{\mathrm{ad}}( \raisebox{-0.8ex}{\includegraphics[scale=0.8]{DK_A_3.pdf}} )
    &\geqslant6.
\end{align}

Another important type-A example is the defining representation, in which the Lie algebra acts through its natural matrix realization. 
For a rank-$n$ Cartan subalgebra $\mathfrak h$ (see Def.~\ref{def:Roots}), we have (see Tab.~\ref{tab:An_representations})
\begin{align}\label{eq:def_A_n}
    \Delta^2_{\mathrm{def}}(\raisebox{-0.6ex}{\includegraphics[scale=0.8]{DK_A_n.pdf}})
    \geqslant n.
\end{align}
As in the preceding discussions, we illustrate the general construction with several low-rank examples.
\begin{align}
    \Delta^2_{\mathrm{def}}( \raisebox{-0.6ex}{\includegraphics[scale=0.8]{DK_B_1.pdf}} )
    &\geqslant1,\\
    \Delta^2_{\mathrm{def}}(\raisebox{-0.8ex}{\includegraphics[scale=0.8]{DK_A_2.pdf}})
    &\geqslant2,\\
    \Delta^2_{\mathrm{def}}( \raisebox{-0.8ex}{\includegraphics[scale=0.8]{DK_A_3.pdf}} )
    &\geqslant3.
\end{align}
This completes our analysis of the type-A series, showing how the general Lie algebraic uncertainty relation reduces to explicit, optimal bounds across its representations.


\subsection{Explicit Bounds for Type-B Lie Algebras}
\label{subsec:B_type}

We next consider the Type-B series, associated with the simple Lie algebra $B_n=\mathfrak{so}(2n+1,\mathbb C)$ and compact real form $\mathfrak{so}(2n+1)$. 
Unlike the simply laced Type-A series, Type-B contains simple roots of unequal length. 
This feature is reflected in its Dynkin diagram:
a linear chain whose final edge is a double bond. 
With the standard ordering, the terminal node corresponds to the short simple root. 
The Dynkin diagram takes the form

\begin{equation}\label{eq:DynkinBn_Diagram}
    \raisebox{-3.2ex}{\includegraphics[scale=1]{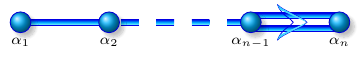}} \, .
\end{equation}

The analysis of type-B representations splits according to the rank. 
We first consider the generic regime $n\geqslant 2$, before treating the rank-one case $n=1$ separately.
With the ordering of Eq.~\eqref{eq:DynkinBn_Diagram}, the Cartan matrix of $B_n$ is tridiagonal but not symmetric.
Its diagonal entries are $C_{ii}=2$ for $1\leqslant i\leqslant n$.
Along the simply laced part of the chain, the nearest neighbor entries are $C_{i,(i+1)}=C_{(i+1),i}=-1$ for $1\leqslant i\leqslant n-2$.
The final double bond produces the asymmetric pair $C_{(n-1),n}=-2$ and $C_{n,(n-1)}=-1$ reflecting the fact that the terminal simple root $\alpha_n$ is short under the convention.
All other off-diagonal entries vanish. Equivalently, the Cartan matrix is
\begin{align}
C
=
\begin{pmatrix}
2 & -1 & 0 & \cdots & 0 & 0\\
-1 & 2 & -1 & \cdots & 0 & 0\\
0 & -1 & 2 & \ddots & \vdots & \vdots\\
\vdots & \vdots & \ddots & \ddots & -1 & 0\\
0 & 0 & \cdots & -1 & 2 & -2\\
0 & 0 & \cdots & 0 & -1 & 2
\end{pmatrix}.
\end{align}
The corresponding inverse is
\begin{align}
C^{-1}
=
\begin{pmatrix}
1 & 1 & 1 & \cdots & 1 & 1\\
1 & 2 & 2 & \cdots & 2 & 2\\
1 & 2 & 3 & \cdots & 3 & 3\\
\vdots & \vdots & \vdots & \ddots & \vdots & \vdots\\
1 & 2 & 3 & \cdots & n-1 & n-1\\
\dfrac12 & 1 & \dfrac32 & \cdots & \dfrac{n-1}{2} & \dfrac n2
\end{pmatrix}.
\end{align}
The entries of the inverse Cartan matrix are given explicitly by the following closed-form formula.

\begin{mylem}{Entries of $C^{-1}$ for Type-B Lie Algebras ($n\geqslant2$)}{Inverse_Cartan_Matrix_B}
For the type-B Lie algebra $B_n$ with $n\geqslant2$, the $(i,j)$-entry of the inverse Cartan matrix is
\begin{align}
\label{eq:Bn_inverse}
C^{-1}_{ij}
=
\begin{cases}
\min\{i,j\},
&
1\leqslant i,j\leqslant n-1,
\\[1.5ex]  
&
1\leqslant i\leqslant n-1,\quad\text{and}\quad j=n,
\\[1.5ex]  
\dfrac{j}{2},
&
i=n,\quad\text{and}\quad 1\leqslant j\leqslant n-1,
\\[2.5ex]  
\dfrac{n}{2},
&
i=j=n.
\end{cases}
\end{align}
\end{mylem}

The Dynkin diagram determines the Cartan matrix and, once the standard normalization of the long roots is fixed, also fixes the squared lengths of the simple roots. 
For type-B, we use the convention that long roots have squared length $2$. 
Hence $\langle\alpha_i,\alpha_i\rangle=2$ for $1\leqslant i\leqslant n-1$, and $\langle\alpha_n,\alpha_n\rangle=1$, where the terminal root $\alpha_n$ is the unique short simple root.
These values define the type-B squared length vector $\vec{\alpha}$.

\begin{mylem}{Root Length Vector $\vec{\alpha}$ for Type-B Lie Algebras ($n\geqslant2$)}{Root_Length_Vector_B} 
For the type-B Lie algebra $B_n$ with $n\geqslant2$, the squared simple root lengths form the vector 
\begin{align}
\label{eq:Bn_Root_Length_Vector}
    \vec{\alpha}
    =
\begin{pmatrix}
2 \\
\vdots \\
2 \\
1
\end{pmatrix}.
\end{align}
\end{mylem}

Combining the entrywise formula for $C^{-1}$ in Lem.~\ref{lem:Inverse_Cartan_Matrix_B} with the simple root squared length vector in Lem.~\ref{lem:Root_Length_Vector_B} gives an explicit expression for $\vec{\xi}$. 
The resulting type-B components are stated in the following theorem.

\begin{mythm}{Components of $\vec{\xi}$ for Type-B Lie Algebras ($n\geqslant2$)}{Root_Length_Vector_B}
For the type-B Lie algebra $B_n$ with $n\geqslant2$, the components of the vector $\vec{\xi}$ (see Eq.~\eqref{eq:xi}) are
\begin{align}
\xi_i=
\begin{cases}
       i(2n-i),\quad  &1\leqslant i\leqslant  n-1,\\[2ex]
       \dfrac{n^2}{2},\quad &i=n.
\end{cases}
\end{align}
\end{mythm}

\begin{proof}
For $1\leqslant i\leqslant n-1$, the entrywise expression in Eq.~\eqref{eq:Bn_inverse} and the simple root length give
\begin{align}
    \xi_i
    &=
    \left(\sum_{j=1}^{n-1}
    C^{-1}_{ij}
    \langle\alpha_j,\alpha_j\rangle\right)
    +
    \left(C^{-1}_{i n}
    \langle\alpha_n,\alpha_n\rangle\right)
    \\
    &=
    2\sum_{j=1}^{n-1}
    \min\{i,j\}
    +
    i\cdot1 .
\end{align}
Decomposing the sum according to the two regimes $j\leqslant i$ and $j>i$, we obtain
\begin{align}
    \xi_i
    &=
    \left(2
    \underbrace{ \sum_{j=1}^{i}j}_{j\leqslant i}
    \right)+
    \left(2\underbrace{\sum_{j=i+1}^{n-1}i}_{j>i}\right)
      +
   \underbrace{i}_{j=n}
    \\
    &=
    2
    \left(
    \frac{i(i+1)}2
    +
    i(n-1-i)
    \right)
    +
    i
    \\
    &=
    i(2n-i).
\end{align}
Thus, for the interior indices $1\leqslant i\leqslant n-1$, 
\begin{align}
    \xi_i=i(2n-i).
\end{align}

It remains to consider the terminal short root node $i=n$, for which we use the final row of $C^{-1}$.

\begin{align}
    \xi_n
    &=
    \left(\sum_{j=1}^{n-1}
    C^{-1}_{n j}
    \langle\alpha_j,\alpha_j\rangle\right)
    +
    \left(C^{-1}_{n n}
    \langle\alpha_n,\alpha_n\rangle\right)
    \\
    &=
    \sum_{j=1}^{n-1}
    \frac{j}{2}\cdot 2
    +
    \frac n2\cdot 1
    \\
    &=
    \frac{n(n-1)}2
    +
    \frac n2\\
    &=
    \frac{n^2}{2}.
\end{align}
This proves the stated expressions and completes the proof.
\end{proof}

With the closed-form coefficients $\xi_i$ for the type-B root system in hand, the abstract uncertainty bound becomes an explicit representation-theoretic inequality.
Substituting the vector $\vec{\xi}$ into Cor.~\ref{cor:Simplified_Lie_Algebraic_Uncertainty_Relations} yields the following type-B Lie-algebraic uncertainty relation.

\begin{mycor}{Type-B Lie Algebraic Uncertainty Relations ($n\geqslant2$)}{Type_B_Lie_UR_2}
For the type-B Lie algebra $B_n$ (see Tab.~\ref{tab:Simple_Lie_Algebra}) with $n\geqslant2$, the optimal Lie algebraic uncertainty relation (see Def.~\ref{def:Lie_Algebraic_Uncertainty_Relations}) takes the explicit form
\begin{align}
   \Delta^2_{\psi}( \raisebox{-1.4ex}{\includegraphics[scale=0.8]{DK_B_n.pdf}} )
   \coloneqq
   \Delta^2_{\psi}(\mathfrak{so}(2n+1,\mathbb C))
   =
   \Delta^2_{\psi}(\mathfrak{so}(2n+1))
   \geqslant
   \left(\sum_{i=1}^{n-1}i(2n-i)\lambda_i\right)+
   \frac{n^2}{2}\lambda_n\,.
\end{align}
The lower bound is state-independent and is fixed entirely by representation-theoretic structure: 
the highest weight $\Lambda$ (see Subsec.~\ref{subsec:Highest_Weight_Representations}), or equivalently its Dynkin coefficient vector $\vec{\lambda}$ (see Eq.\eqref{eq:Dynkin_Vector}), together with the type-B root system geometry encoded in $\vec{\xi}$ (see Eq.~\eqref{eq:xi}). 
Its optimality follows directly from the general result established in Cor.~\ref{cor:Simplified_Lie_Algebraic_Uncertainty_Relations}. 
\end{mycor}

The discussion above covers the generic type-B case $n\geqslant 2$. 
We now consider the exceptional rank-one case $n=1$, for which the Cartan matrix $C$ (see Def.~\ref{def:Cartan_Matrix}) reduces to the $1\times 1$ matrix
\begin{align}
    C
    =
    \begin{pmatrix}
        2
    \end{pmatrix},
\end{align}
with inverse
\begin{align}
    C^{-1}
    =
    \begin{pmatrix}
        \frac{1}{2}
    \end{pmatrix}.
\end{align}
In this rank-one case, the squared root length vector $\vec{\alpha}$ takes the form
\begin{align}
    \vec{\alpha}
    =
    \begin{pmatrix}
    2 
    \end{pmatrix}.
\end{align}
Together with the inverse Cartan matrix above, Eq.~\eqref{eq:xi} gives 
\begin{align}
    \vec{\xi}=
    C^{-1}\cdot\vec{\alpha}=
    \begin{pmatrix}
        1
    \end{pmatrix},
\end{align}
Substituting the vector $\vec{\xi}$ into the general bound $\langle\vec{\lambda},\vec{\xi}\rangle$ (see Eq.~\eqref{eq:Lie_Algebraic_UR_Simplified}) yields the type-B Lie algebraic uncertainty relation for rank one:

\begin{mycor}{Type-B Lie Algebraic Uncertainty Relations ($n=1$)}{Type_B_Lie_UR_1}
For the type-B Lie algebra $B_n$ (see  Tab.~\ref{tab:Simple_Lie_Algebra}) with $n=1$, the optimal Lie algebraic uncertainty relation (see Def.~\ref{def:Lie_Algebraic_Uncertainty_Relations}) takes the explicit form 
\begin{align}
   \Delta^2_{\psi}( \raisebox{-0.6ex}{\includegraphics[scale=0.8]{DK_B_1.pdf}} )
   \coloneqq
   \Delta^2_{\psi}(\mathfrak{so}(3,\mathbb C))
   =
   \Delta^2_{\psi}(\mathfrak{so}(3))
   \geqslant
   \lambda_1\,.
\end{align}
The lower bound is state-independent and, in this rank-one case, is fixed entirely by the leading Dynkin coefficient $\lambda_1$ of the highest weight $\Lambda$ (see Subsec.~\ref{subsec:Highest_Weight_Representations}). 
\end{mycor}

The rank-one endpoint also confirms the consistency of the classification. 
The non-simply-laced structure that distinguishes type-B appears only for $n\geqslant 2$; when $n=1$, the Dynkin diagram contains a single node and coincides with the type-A rank-one diagram. 
With the Dynkin diagram matched, the type-B rank-one bound therefore reduces exactly to the type-A uncertainty relation
\begin{align}
    \min_{\psi}\Delta^2_{\psi}(\mathfrak{sl}(2,\mathbb C))=
    \min_{\psi}\Delta^2_{\psi}(\mathfrak{su}(2))=
    \min_{\psi}\Delta^2_{\psi}( \raisebox{-0.6ex}{\includegraphics[scale=0.8]{DK_B_1.pdf}} )=
    \min_{\psi}\Delta^2_{\psi}(\mathfrak{so}(3,\mathbb C))=
    \min_{\psi}\Delta^2_{\psi}(\mathfrak{so}(3))=
    \lambda_1.
\end{align}
This coincidence of the uncertainty bounds has a clear algebraic origin. 
In rank one, the compact real forms satisfy $\mathfrak{su}(2)\simeq\mathfrak{so}(3)$, so the two descriptions encode the same Lie algebra.
The resulting quantum uncertainty bound is therefore identical. 
This illustrates that isomorphic Lie algebras have the same optimal bound for corresponding representations, with consistently normalized invariant forms.
The spin-$j$ representation provides a direct physical realization of this rank-one result. 
Here the highest weight is $\Lambda=2j\omega_1$ and the Dynkin coefficient is $\lambda_1=2j$; equivalently, $\Lambda(J_z)=j$. 
The rank-one Lie algebraic bound therefore reduces exactly to the familiar angular momentum uncertainty relation, i.e.,
\begin{align}
    \Delta^2_{\psi}( \raisebox{-0.6ex}{\includegraphics[scale=0.8]{DK_B_1.pdf}} )=
    2\left(
    \Delta_{\psi}^2(J_x) + 
    \Delta_{\psi}^2(J_y) + 
    \Delta_{\psi}^2(J_z)
    \right)
    \geqslant2j,
\end{align}
recovering Cor.~\ref{cor:Angular_Momentum_Uncertainty_Relation}.
Tab.~\ref{tab:explicit_bounds_B} gives type-B uncertainty bounds explicitly in terms of the Dynkin coefficient vector $\vec{\lambda}$ (see Eq.~\eqref{eq:Dynkin_Vector}).

\begin{table}[h]
\centering
\begin{tblr}{
width = \textwidth,
colspec = {c||c|X[c]},
row{1} = {bg=gray!50, font=\bfseries}, 
row{2} = {bg=gray!20, font=\bfseries}, 
hlines,
vlines,
cells = {m},
cell{3-11}{1} = {bg=gray!20},         
cell{3-11}{2} = {bg=mRed},             
cell{3-11}{3} = {bg=mGreen},           
}

\SetCell[c=3]{c}
\textbf{Type-B Lie Algebraic Uncertainty Relations\,\,  \raisebox{-1.4ex}{\includegraphics[scale=0.8]{DK_B_n.pdf}} }
&
&
\\

Compact Lie algebra
&
Rank
&
Optimal Bound
\\
$\mathfrak{so}(3)$
&
$1$
&
$\lambda_1$
\\
$\mathfrak{so}(5)$
&
$2$
&
$3\lambda_1+2\lambda_2$
\\

$\mathfrak{so}(7)$
&
$3$
&
$5\lambda_1+8\lambda_2+\dfrac{9}{2}\lambda_3$
\\

$\mathfrak{so}(9)$
&
$4$
&
$7\lambda_1+12\lambda_2+15\lambda_3+8\lambda_4$
\\

$\mathfrak{so}(11)$
&
$5$
&
$9\lambda_1+16\lambda_2+21\lambda_3+24\lambda_4+\dfrac{25}{2}\lambda_5$
\\

$\mathfrak{so}(13)$
&
$6$
&
$11\lambda_1+20\lambda_2+27\lambda_3+32\lambda_4+35\lambda_5+18\lambda_6$
\\

$\mathfrak{so}(15)$
&
$7$
&
$\displaystyle
13\lambda_1+24\lambda_2+33\lambda_3
+40\lambda_4+45\lambda_5+48\lambda_6+\dfrac{49}{2}\lambda_7$
\\

$\vdots$
&
$\vdots$
&
$\vdots$
\\

$\mathfrak{so}(2n+1)$
&
$n$
&
$\displaystyle
(2n-1)\lambda_1
+2(2n-2)\lambda_2
+3(2n-3)\lambda_3
+\cdots
+(n-1)(n+1)\lambda_{n-1}
+\frac{n^2}{2}\lambda_n$
\\

\end{tblr}

\caption{
\textbf{Type-B Lie Algebraic Uncertainty Relations.}
For the compact Lie algebra $\mathfrak{so}(2n+1)$ with Dynkin coefficient vector $\vec{\lambda}=(\lambda_1,\ldots,\lambda_n)^{\T}$, the state-independent optimal bound is 
$(\sum_{i=1}^{n-1} i(2n-i)\lambda_i)+n^2\lambda_n/2$ for $n\geqslant 2$.
For $B_2$, the adjoint representation has highest weight $2\omega_2$ and Dynkin vector $(0,2)$, with optimal bound 4.}

\label{tab:explicit_bounds_B}
\end{table}

This completes the type-B Lie algebraic uncertainty bounds. 
The non-simply-laced structure of $B_n$ enters through the terminal short root, which separates the final Dynkin coefficient from the long root directions. 
The resulting uncertainty bound is state-independent and fixed entirely by the highest weight of the representation together with the geometry of the type-B root system. Tab.~\ref{tab:explicit_bounds_B} makes this structure explicit across ranks, showing how the abstract Lie algebraic result becomes a concrete bound for $\mathfrak{so}(2n+1)$.
Tab.~\ref{tab:Bn_representations} translates the type-B formula into explicit bounds for several representation spaces commonly encountered in physical applications.

\begin{table}[htbp]
\centering
\begin{tblr}{
width = \textwidth,
colspec = {X[c]||X[c]|X[c]|X[c]},
row{1} = {bg=gray!50, font=\bfseries},
row{2} = {bg=gray!20, font=\bfseries},
hlines,
vlines,
cells = {m},
cell{3-7}{1} = {bg=gray!20},
cell{3-7}{2} = {bg=mRed},
cell{3-7}{3} = {bg=mGreen},
cell{3-7}{4} = {bg=mBlue},
}

\SetCell[c=4]{c}
\textbf{Type-B Lie Algebraic Uncertainty Relations Across Representations}
&
&
&
\\

Representation
&
Highest Weight $\Lambda$
&
Dynkin Vector $\vec{\lambda}^{\,\T}$
&
Optimal Bound
\\

Defining
&
$\omega_1$
&
$(1,0,\dots,0)$
&
$2n-1$
\\

Adjoint
&
$\omega_2$ with $n\geqslant 3$
&
$(0,1,0,\dots,0)$
&
$4(n-1)$
\\

Spinor
&
$\omega_n$
&
$(0,\dots,0,1)$
&
$\dfrac{n^2}{2}$
\\

Traceless Symmetric Square
&
$2\omega_1$
&
$(2,0,\dots,0)$
&
$2(2n-1)$
\\

Fundamental ($i$-th)
&
$\omega_i$ with $1\leqslant i<n$
&
$(0,\dots,1,\dots,0)$
&
$i(2n-i)$
\\

\end{tblr}
\caption{\textbf{Lie Algebraic Uncertainty Relations Across Representations of $\mathfrak{so}(2n+1)$.}
For representative standard representations of $\mathfrak{so}(2n+1)$ with $n\geqslant 2$, the table collects the highest weight $\Lambda$, the Dynkin coefficient vector $\vec{\lambda}^{\,\T}$, and the resulting optimal lower bound.
For $n=2$, the adjoint representation has highest
weight $2\omega_2$ and Dynkin coefficient vector
$(0,2)$; its optimal lower bound remains $4(n-1)=4$.
The `` Traceless Symmetric Square'' row refers to
$\operatorname{Sym}^2_0 V$, the irreducible traceless
component of $\operatorname{Sym}^2 V$,
where $V$ is the defining representation.}
\label{tab:Bn_representations}
\end{table}

As in the type-A case discussed in Subsec.~\ref{subsec:A_type}, we use the adjoint representation as a representative example for the type-B series, illustrating how the general framework specializes to explicit optimal uncertainty bounds. 
We begin with the generic $B_n$ case, for which the Dynkin diagram and the optimal lower bound for the adjoint representation are summarized below (see Tab.~\ref{tab:Bn_representations})
\begin{align}\label{eq:adj_B_n}
    \Delta^2_{\mathrm{ad}}( \raisebox{-1.4ex}{\includegraphics[scale=0.8]{DK_B_n.pdf}} )
    \geqslant4(n-1).
\end{align}
Specializing to $n=3,4,5$, the optimal Lie algebraic uncertainty relations are given explicitly by
\begin{align}
    \Delta^2_{\mathrm{ad}}( \raisebox{-1.4ex}{\includegraphics[scale=0.8]{DK_B_3.pdf}} )
    &\geqslant8,\\
    \Delta^2_{\mathrm{ad}}( \raisebox{-1.4ex}{\includegraphics[scale=0.8]{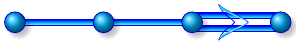}} )
    &\geqslant12,\\
    \Delta^2_{\mathrm{ad}}( \raisebox{-1.4ex}{\includegraphics[scale=0.8]{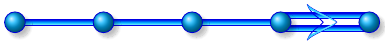}} )
    &\geqslant16.
\end{align}

We next consider the defining representation of type-B Lie algebras. 
For rank $n\geqslant 2$, the Lie algebraic uncertainty relation takes the form (see Tab.~\ref{tab:Bn_representations})
\begin{align}\label{eq:def_B_n}
    \Delta^2_{\mathrm{def}}( \raisebox{-1.4ex}{\includegraphics[scale=0.8]{DK_B_n.pdf}} )
    \geqslant 2n-1.
\end{align}
For the representative low-rank cases $n=3,4,5$, the optimal Lie algebraic uncertainty relations are given explicitly by
\begin{align}
    \Delta^2_{\mathrm{def}}( \raisebox{-1.4ex}{\includegraphics[scale=0.8]{DK_B_3.pdf}} )
    &\geqslant5,\\
    \Delta^2_{\mathrm{def}}( \raisebox{-1.4ex}{\includegraphics[scale=0.8]{DK_B_4.pdf}} )
    &\geqslant7,\\
    \Delta^2_{\mathrm{def}}( \raisebox{-1.4ex}{\includegraphics[scale=0.8]{DK_B_5.pdf}} )
    &\geqslant9.
\end{align}
This completes the type-B analysis, making explicit how the general Lie algebraic uncertainty framework specializes, representation by representation, into state-independent optimal bounds.


\subsection{Explicit Bounds for Type-C Lie Algebras}
\label{subsec:C_type}

We now turn our attention to the type-C series (see Tab.~\ref{tab:Simple_Lie_Algebra}). In this case, the complex Lie algebra is $C_n=\mathfrak{sp}(2n,\mathbb C)$, which possesses the compact real form $\mathfrak{sp}(n)$. 
For $n \geqslant 2$, the corresponding root system is non-simply laced, containing both long and short roots, much like the type-B series (see Eq.~\eqref{eq:DynkinBn_Diagram}).
However, the roles of the long and short roots are interchanged relative to $B_n$: 
here, the final simple root of the Dynkin diagram is long rather than short. 
This inversion of root lengths between $B_n$ and $C_n$ is reflected directly in the structure of the uncertainty relations. 
The Dynkin diagram for $C_n$ is given by
\begin{equation}\label{eq:DynkinCn_Diagram}
     \raisebox{-3.8ex}{\includegraphics[scale=1]{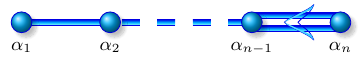}} \, .
\end{equation}
The low-rank examples of the $C_n$ series demonstrate a crucial consistency with the standard accidental isomorphisms of simple Lie algebras. 
It is instructive to examine this correspondence first through the lenses of Dynkin diagrams, Cartan matrices, root lengths, and uncertainty relations.

At rank one, $n=1$, the type-C series reduces to a single-node Dynkin diagram. 
The corresponding lower bound therefore agrees exactly with the $A_1$ and $B_1$ cases, in accordance with the standard low-rank identifications. 
The rank-one type-C uncertainty relation thus reads

\begin{mycor}{Type-C Lie Algebraic Uncertainty Relations ($n=1$)}{Type_C_Lie_UR_1}
For the type-C Lie algebra $C_n$ (see  Tab.~\ref{tab:Simple_Lie_Algebra}) with $n=1$, the optimal Lie algebraic uncertainty relation (see Def.~\ref{def:Lie_Algebraic_Uncertainty_Relations}) takes the explicit form 
\begin{align}
\label{eq:Type_C1}
    \Delta^2_{\psi}( \raisebox{-0.6ex}{\includegraphics[scale=0.8]{DK_B_1.pdf}} )=
    \underbrace{\Delta^2_{\psi}(\mathfrak{su}(2))}_{\text{Type-A}}=
    \underbrace{\Delta^2_{\psi}(\mathfrak{so}(3))}_{\text{Type-B}}=
    \underbrace{\Delta_\psi^2(\mathfrak{sp}(1))}_{\text{Type-C}}
    \geqslant
    \lambda_1
    .
\end{align}
At rank one, the state-independent lower bound is fixed solely by the Dynkin coefficient $\lambda_1$ of the highest weight $\Lambda$; see Subsec.~\ref{subsec:Highest_Weight_Representations}.
\end{mycor}

For $n=2$, the type-C Dynkin diagram consists of two nodes connected by a double bond:
\begin{align}
     \raisebox{-3.2ex}{\includegraphics[scale=1]{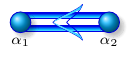}} \, .
\end{align}
In the standard type-C ordering, the first simple root $\alpha_1$ is short and the second root $\alpha_2$ is long; 
accordingly, the arrow points toward $\alpha_1$. 
With the long roots normalized to have squared length 2, the simple root lengths are $\langle\alpha_1,\alpha_1\rangle =1$ and 
$\langle\alpha_2,\alpha_2\rangle=2$. 
The Cartan matrix is therefore
\begin{align}
    C
    =
    \begin{pmatrix}
        2 & -1\\
        -2 & 2
    \end{pmatrix},
\end{align}
and its inverse is given by
\begin{align}
    C^{-1}
    =
    \begin{pmatrix}
        1 & \frac12\\
        1 & 1
    \end{pmatrix}.
\end{align}
The simple root squared lengths are therefore collected in the vector 
\begin{align}
    \vec{\alpha}=
    \begin{pmatrix}
        1\\
        2
    \end{pmatrix}.
\end{align}
Consequently, all root system information entering the uncertainty bound is encoded in the vector $\vec{\xi}$ (see Eq.~\eqref{eq:xi})
\begin{align}
    \vec{\xi}=
    C^{-1}\cdot\vec{\alpha}=
    \begin{pmatrix}
        1 & \frac{1}{2}\\
        1 & 1
    \end{pmatrix}
    \begin{pmatrix}
        1\\
        2
    \end{pmatrix}
    =
    \begin{pmatrix}
        2\\
        3
    \end{pmatrix}.
\end{align}
We therefore obtain the following type-$C_2$ uncertainty relation.

\begin{mycor}{Type-C Lie Algebraic Uncertainty Relations ($n=2$)}{Type_C_Lie_UR_2}
For the type-C Lie algebra $C_n$ (see  Tab.~\ref{tab:Simple_Lie_Algebra}) with $n=2$, the optimal Lie algebraic uncertainty relation (see Def.~\ref{def:Lie_Algebraic_Uncertainty_Relations}) takes the explicit form 
\begin{align}
\label{eq:Type_C2}
    \Delta^2_{\psi}( \raisebox{-1.4ex}{\includegraphics[scale=0.8]{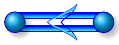}} )\coloneqq
    \Delta_\psi^2(\mathfrak{sp}(2))
    \geqslant
    2\lambda_1+3\lambda_2
    .
\end{align}
For the rank-two case, i.e., $n=2$, the state-independent lower bound is fixed by the two Dynkin coefficients $\lambda_1$ and $\lambda_2$ of the highest weight $\Lambda$; 
see Subsec.~\ref{subsec:Highest_Weight_Representations}.
\end{mycor}

The type-$C_2$ result can be compared directly with the type-$B_2$ bound obtained in Cor.~\ref{cor:Type_B_Lie_UR_2}. 
Setting $n=2$ there gives the type-$B_2$ uncertainty relation
\begin{align}
\label{eq:Type_B2}
    \Delta^2_{\psi}( \raisebox{-1.4ex}{\includegraphics[scale=0.8]{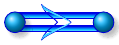}} )\coloneqq
    \Delta_\psi^2(\mathfrak{so}(5))
    \geqslant
    3\lambda_1+2\lambda_2
    .
\end{align}
The coefficient vector in this expression is the reverse of that obtained for type $C_2$ (see Eq.~\eqref{eq:Type_C2}). 
This reversal is not a discrepancy. 
It reflects the standard Dynkin diagram conventions for the two rank-two series: 
in type $B_2$ and type $C_2$, the long and short simple roots are assigned to opposite nodes.

Under the identification $\mathfrak{so}(5)\simeq\mathfrak{sp}(2)$,
corresponding representations have Dynkin labels related by 
$(\lambda_1^{(B)},\lambda_2^{(B)})
=(\lambda_2^{(C)},\lambda_1^{(C)})$.
With the normalization fixed above, Thm.~\ref{thm:Simple_Root_Relabeling_Invariance}
then implies that their optimal uncertainty bounds are identical:
\begin{align}\label{Type_BC_identity}
\min_{\psi}\Delta_\psi^2(\mathfrak{sp}(2))=
\min_{\psi}\Delta^2_{\psi}( \raisebox{-1.4ex}{\includegraphics[scale=0.8]{DK_C_2.pdf}} )=
\min_{\psi}\Delta^2_{\psi}( \raisebox{-1.4ex}{\includegraphics[scale=0.8]{DK_B_2.pdf}} )=
\min_{\psi}\Delta_\psi^2(\mathfrak{so}(5)).
\end{align}
The minima in Eq.~\eqref{Type_BC_identity} are taken over states in the corresponding representation spaces.
To make this correspondence explicit, let us consider the adjoint representation as a representative example. 
In this case, the optimal bound takes the form
\begin{align}
\min_{\psi}\Delta_{\mathrm{ad}}^2(\mathfrak{sp}(2))=
\min_{\psi}\Delta^2_{\mathrm{ad}}( \raisebox{-1.4ex}{\includegraphics[scale=0.8]{DK_C_2.pdf}} )=
\min_{\psi}\Delta^2_{\mathrm{ad}}( \raisebox{-1.4ex}{\includegraphics[scale=0.8]{DK_B_2.pdf}} )=
\min_{\psi}\Delta_{\mathrm{ad}}^2(\mathfrak{so}(5))=4.
\end{align}

For type $C_n$ with $n\geqslant 2$, the non-simply-laced root structure gives a tridiagonal Cartan matrix (see Def.~\ref{def:Cartan_Matrix}) that is no longer symmetric.
With the standard type-C ordering, the diagonal entries are $C_{ii}=2$ for all $i$.
Along the single-bond part of the chain, one has
\begin{align}
    C_{i,i+1}=C_{i+1,i}=-1,\quad 
    1\leqslant i\leqslant n-2.
\end{align}
The terminal double bond, connecting the short root $\alpha_{n-1}$ to the long root $\alpha_n$, gives the asymmetric entries
\begin{align}
    C_{n-1,n}=-1,
    \quad
    C_{n,n-1}=-2.
\end{align}
All remaining entries vanish, namely $C_{ij}=0$ whenever $|i-j|>1$.
Equivalently, the Cartan matrix takes the form
\begin{align}
\label{eq:Cn_Cartan_Matrix}
    C
    =
    \begin{pmatrix}
    2 & -1 & 0 & \cdots & 0 & 0\\
    -1 & 2 & -1 & \cdots & 0 & 0\\
    0 & -1 & 2 & \ddots & \vdots & \vdots\\
    \vdots & \vdots & \ddots & \ddots & -1 & 0\\
    0 & 0 & \cdots & -1 & 2 & -1\\
    0 & 0 & \cdots & 0 & -2 & 2
    \end{pmatrix}.
\end{align}
The asymmetric entries record the fact that, in the standard type-C ordering, the final node corresponds to the long simple root.
The inverse Cartan matrix is
\begin{align}
\label{eq:Cn_inverse_matrix}
     C^{-1}
    =
    \begin{pmatrix}
    1 & 1 & 1 & \cdots & 1 & \frac12\\
    1 & 2 & 2 & \cdots & 2 & 1\\
    1 & 2 & 3 & \cdots & 3 & \frac32\\
    \vdots & \vdots & \vdots & \ddots & \vdots & \vdots\\
    1 & 2 & 3 & \cdots & n-1 & \frac{n-1}{2}\\
    1 & 2 & 3 & \cdots & n-1 & \frac n2
    \end{pmatrix}.
\end{align}
Its entries are given by

\begin{mylem}{Entries of $C^{-1}$ for Type-C Lie Algebras ($n\geqslant2$)}{Entries_Cinverse_Type_C}
For the type-C Lie algebra $C_n$, the $(i,j)$-entry of the inverse Cartan matrix is
\begin{align}
\label{eq:Cn_inverse}
    C^{-1}_{ij}
    =
    \begin{cases}
    \min\{i,j\},
    &
    1\leqslant i,j\leqslant n-1,
    \\[1.5ex]
    \dfrac{i}{2},
    &
    1\leqslant i\leqslant n-1, \quad\text{and}\quad j=n, 
    \\[2.5ex]
    j,
    &
    i=n,\quad\text{and}\quad 1\leqslant j\leqslant n-1,
    \\[1.5ex]
    \dfrac{n}{2},
    &
    i=j=n.
    \end{cases}
\end{align}
\end{mylem}

For $n\geqslant 2$, we adopt the standard normalization in which long roots have squared length $2$.
In the type-C ordering, the first $n-1$ simple roots are short, whereas the final simple root is long.
Thus $\langle\alpha_j,\alpha_j\rangle=1$, for $1\leqslant j\leqslant n-1$, and $\langle\alpha_n,\alpha_n\rangle = 2$.
The corresponding $\vec{\alpha}$ has the following form.

\begin{mylem}{Root Length Vector $\vec{\alpha}$ for Type-C Lie Algebras ($n\geqslant2$)}{Root_Length_Vector_C} 
For the type-C Lie algebra $C_n$ with $n\geqslant2$, the squared simple root lengths form the vector 
\begin{align}
\label{eq:Cn_Root_Length_Vector}
    \vec{\alpha}
    =
    \begin{pmatrix}
1 \\
\vdots \\
1 \\
2
\end{pmatrix} .
\end{align}
\end{mylem}

The inverse Cartan matrix $C^{-1}$ in Lem.~\ref{lem:Entries_Cinverse_Type_C}, together with the squared root lengths in Lem.~\ref{lem:Root_Length_Vector_C}, gives an explicit expression for the components of $\vec{\xi}$.

\begin{mythm}{Components of $\vec{\xi}$ for Type-C Lie Algebras $(n\geqslant2)$}{Components_xi_C}
For the type-C Lie algebra $C_n$ with $n\geqslant2$, the components of the vector $\vec{\xi}$ (see Eq.~\eqref{eq:xi}) are
\begin{align}
     \xi_i
    =
    \begin{cases}
   \dfrac{i(2n+1-i)}{2}\,, &\quad 1\leqslant i\leqslant n-1,
    \\[2ex]
    \dfrac{n(n+1)}{2}, &\quad i=n \,.
    \end{cases}
\end{align}
\end{mythm}

\begin{proof}
For $1\leqslant i\leqslant n-1$, Eq.~\eqref{eq:Cn_inverse}, together with the root length vector $\vec{\alpha}$ in Eq.~\eqref{eq:Cn_Root_Length_Vector}, gives
\begin{align}
    \xi_i
    &=
    \sum_{j=1}^{n-1}
    C^{-1}_{ij}
    \langle\alpha_j,\alpha_j\rangle
    +
    C^{-1}_{i n}
    \langle\alpha_n,\alpha_n\rangle\\
    &=
    \sum_{j=1}^{n-1}
    \min\{i,j\}
    +
    \frac{i}{2}\cdot 2
    \\
    &=
    \sum_{j=1}^{n-1}
    \min\{i,j\}
    +
    i.
\end{align}
Evaluating the remaining sum by separating the contributions with $j\leqslant i$ from those with $j>i$ yields
\begin{align}
    \xi_i
    &=
    \left(
   \underbrace{  \sum_{j=1}^{i} j}_{j\leqslant i}
    +
   \underbrace{  \sum_{j=i+1}^{n-1} i}_{j>i}
    \right)
    +
    i
    \\
    &=
    \frac{i(i+1)}{2}
    +
    i(n-1-i)
    +
    i
    \\
    &=
    \frac{i(2n+1-i)}{2}.
\end{align}
For the final node, corresponding to the long simple root $\alpha_n$, we use the last row of $C^{-1}$:
\begin{align}
    \xi_n
    &=
    \sum_{j=1}^{n-1}
    C^{-1}_{n j}
    \langle\alpha_j,\alpha_j\rangle
    +
    C^{-1}_{n n}
    \langle\alpha_n,\alpha_n\rangle
    \\
    &=
    \sum_{j=1}^{n-1}
    j
    +
    \frac n2\cdot 2
    \\
    &=
    \frac{n(n-1)}{2}
    +
    n
    \\
    &=
    \frac{n(n+1)}{2}.
\end{align}
This completes the proof.
\end{proof}

The closed-form coefficients $\xi_i$ for the type-C root system turn the general uncertainty bound into an explicit inequality.
Substituting this vector into Cor.~\ref{cor:Simplified_Lie_Algebraic_Uncertainty_Relations} yields the type-C Lie algebraic uncertainty relation below.

\begin{mycor}{Type-C Lie Algebraic Uncertainty Relations $(n\geqslant2)$}{Type_C_Lie_UR}
For the type-C Lie algebra $C_n$ (see Tab.~\ref{tab:Simple_Lie_Algebra}) with $n\geqslant2$, the optimal Lie algebraic uncertainty relation (see Def.~\ref{def:Lie_Algebraic_Uncertainty_Relations}) takes the explicit form
\begin{align}
   \Delta^2_{\psi}( \raisebox{-1.4ex}{\includegraphics[scale=0.8]{DK_C_n.pdf}} )
   \coloneqq
   \Delta^2_{\psi}(\mathfrak{sp}(2n,\mathbb C))
   =
    \Delta^2_{\psi}(\mathfrak{sp}(n))
    \geqslant
    \left(\sum_{i=1}^{n-1}\frac{i(2n+1-i)}{2}\lambda_i\right)+
    \frac{n(n+1)}{2}\lambda_n\,.
\end{align}
The bound is state-independent and is fixed entirely by the representation, through the highest weight $\Lambda$; see Subsec.~\ref{subsec:Highest_Weight_Representations}.
Equivalently, it is determined by the Dynkin coefficient vector $\vec{\lambda}$; 
see Eq.~\eqref{eq:Dynkin_Vector}, together with the type-C root system geometry encoded in $\vec{\xi}$; 
see Eq.~\eqref{eq:xi}.
Its optimality follows directly from Cor.~\ref{cor:Simplified_Lie_Algebraic_Uncertainty_Relations}.
\end{mycor}

The compact real form above makes the general type-C structure explicit.
To display how the coefficients grow with the rank, and to include the rank-one case, we summarize the corresponding bounds in Tab.~\ref{tab:explicit_bounds_C}.

\begin{table}[h]
\centering
\begin{tblr}{
width = \textwidth,
colspec = {c||c|X[c]},
row{1} = {bg=gray!50, font=\bfseries},
row{2} = {bg=gray!20, font=\bfseries},
hlines,
vlines,
cells = {m},
cell{3-10}{1} = {bg=gray!20},
cell{3-10}{2} = {bg=mRed},
cell{3-10}{3} = {bg=mGreen},
}

\SetCell[c=3]{c}
\textbf{Type-C Lie Algebraic Uncertainty Relations \quad  \raisebox{-1.4ex}{\includegraphics[scale=0.8]{DK_C_n.pdf}} }
&
&
\\

Compact Lie algebra
&
Rank
&
Optimal Bound
\\

$\mathfrak{sp}(1)$
&
$1$
&
$\lambda_1$
\\

$\mathfrak{sp}(2)$
&
$2$
&
$2\lambda_1+3\lambda_2$
\\

$\mathfrak{sp}(3)$
&
$3$
&
$3\lambda_1+5\lambda_2+6\lambda_3$
\\

$\mathfrak{sp}(4)$
&
$4$
&
$4\lambda_1+7\lambda_2+9\lambda_3+10\lambda_4$
\\

$\mathfrak{sp}(5)$
&
$5$
&
$5\lambda_1+9\lambda_2+12\lambda_3+14\lambda_4+15\lambda_5$
\\

$\mathfrak{sp}(6)$
&
$6$
&
$\displaystyle
6\lambda_1+11\lambda_2+15\lambda_3
+18\lambda_4+20\lambda_5+21\lambda_6$
\\

$\vdots$
&
$\vdots$
&
$\vdots$
\\

$\mathfrak{sp}(n)$
&
$n$
&
$\displaystyle
n\lambda_1
+(2n-1)\lambda_2
+3(n-1)\lambda_3
+\cdots
+\frac{(n-1)(n+2)}{2}\lambda_{n-1}
+\frac{n(n+1)}{2}\lambda_n$
\\

\end{tblr}

\caption{
\textbf{Type-C Lie Algebraic Uncertainty Relations.}
For the compact Lie algebra $\mathfrak{sp}(n)$ with Dynkin coefficient vector $\vec{\lambda}=(\lambda_1,\ldots,\lambda_n)^{\T}$, the state-independent optimal bound is
$\sum_{i=1}^{n-1}(i(2n+1-i)/2)\lambda_i+(n(n+1)/2)\lambda_n$ for $n\geqslant2$.
The rank-one case satisfies $C_1\simeq A_1$ and reproduces the $\mathfrak{su}(2)$ bound $\lambda_1$ (see Cor.~\ref{cor:Type_C_Lie_UR_1}).
}
\label{tab:explicit_bounds_C}
\end{table}

Tab.~\ref{tab:Cn_representations} summarizes representative highest weight representations of $\mathfrak{sp}(n)$ and the corresponding optimal uncertainty bounds.

\begin{table}[h]
\centering
\begin{tblr}{
width = \textwidth,
colspec = {X[c]||X[c]|X[c]|X[c]},
row{1} = {bg=gray!50,font=\bfseries},
row{2} = {bg=gray!20,font=\bfseries},
hlines,
vlines,
cells = {m},
cell{3-7}{1} = {bg=gray!20},
cell{3-7}{2} = {bg=mRed},
cell{3-7}{3} = {bg=mGreen},
cell{3-7}{4} = {bg=mBlue},
}

\SetCell[c=4]{c}
\textbf{Type-C Lie Algebraic Uncertainty Relations Across Representations}
&
&
&
\\

Representation
&
Highest Weight $\Lambda$
&
Dynkin Vector $\vec{\lambda}^{\,\T}$
&
Optimal Bound
\\

Defining
&
$\omega_1$
&
$(1,0,\dots,0)$
&
$n$
\\

Adjoint
&
$2\omega_1$
&
$(2,0,\dots,0)$
&
$2n$
\\

Fundamental (second) 
&
$\omega_2 \quad (n\geqslant 2)$
&
$(0,1,0,\dots,0)$
&
$2n-1$
\\

Fundamental (last)
&
$\omega_n$
&
$(0,\dots,0,1)$
&
$\dfrac{n(n+1)}{2}$
\\

Fundamental ($i$-th)
&
$\omega_i$ with $\ 3\leqslant i\leqslant n-1$
&
$(0,\dots,1,\dots,0)$
&
$\dfrac{i(2n+1-i)}{2}$
\\

\end{tblr}

\caption{
\textbf{Lie Algebraic Uncertainty Relations Across Representations of $\mathfrak{sp}(n)$.}
Representative standard representations of $\mathfrak{sp}(n)$, specified by the highest weight $\Lambda$ and Dynkin coefficient vector $\vec{\lambda}^{\,\T}$, together with the resulting optimal lower bound.}
\label{tab:Cn_representations}
\end{table}

Following the type-A and type-B analyses in Subsecs.~\ref{subsec:A_type} and~\ref{subsec:B_type}, respectively, we take the adjoint representation as a representative example for the type-C series. 
Compared with type $B_n$ (see Eq.~\eqref{eq:DynkinBn_Diagram}), the arrow in the Dynkin diagram is reversed (see Eq.~\eqref{eq:DynkinCn_Diagram}), reflecting the interchange of long and short simple roots. 
This reversal is directly mirrored in the structure of the resulting uncertainty relation.
We begin with the generic $C_n$ case, for which the Dynkin diagram and the optimal lower bound for the adjoint representation are summarized below; see Tab.~\ref{tab:Cn_representations}.
\begin{align}\label{eq:adj_C_n}
    \Delta^2_{\mathrm{ad}}( \raisebox{-1.4ex}{\includegraphics[scale=0.8]{DK_C_n.pdf}} )
    \geqslant2n.
\end{align}
Specializing to $n=3,4,5$, the optimal Lie algebraic uncertainty relations are given explicitly by
\begin{align}
    \Delta^2_{\mathrm{ad}}( \raisebox{-1.4ex}{\includegraphics[scale=0.8]{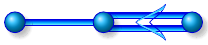}} )
    &\geqslant6,\\
    \Delta^2_{\mathrm{ad}}( \raisebox{-1.4ex}{\includegraphics[scale=0.8]{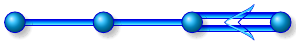}} )
    &\geqslant8,\\
    \Delta^2_{\mathrm{ad}}( \raisebox{-1.4ex}{\includegraphics[scale=0.8]{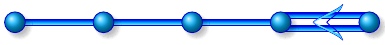}} )
    &\geqslant10.
\end{align}

Therefore, at the level of the adjoint representation, the type-A and type-C series yield identical optimal uncertainty bounds whenever they have the same rank, equivalently the same dimension of the Cartan subalgebra $\mathfrak h$. 
In diagrammatic terms, this correspondence can be expressed as
\begin{align}
    \min_{\psi}\Delta^2_{\mathrm{ad}}( \raisebox{-0.6ex}{\includegraphics[scale=0.8]{DK_A_n.pdf}} )=
    \min_{\psi}\Delta^2_{\mathrm{ad}}( \raisebox{-1.4ex}{\includegraphics[scale=0.8]{DK_C_n.pdf}} ).
\end{align}
For the cases $n=3,4,5$, this equivalently gives
\begin{align}
    \min_{\psi}\Delta^2_{\mathrm{ad}}( \raisebox{-0.8ex}{\includegraphics[scale=0.8]{DK_A_3.pdf}} )
    &=
    \min_{\psi}\Delta^2_{\mathrm{ad}}( \raisebox{-1.4ex}{\includegraphics[scale=0.8]{DK_C_3.pdf}} ),\\
    \min_{\psi}\Delta^2_{\mathrm{ad}}( \raisebox{-0.8ex}{\includegraphics[scale=0.8]{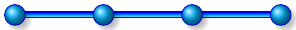}} )
    &=
    \min_{\psi}\Delta^2_{\mathrm{ad}}( \raisebox{-1.4ex}{\includegraphics[scale=0.8]{DK_C_4.pdf}} ),\\
    \min_{\psi}\Delta^2_{\mathrm{ad}}( \raisebox{-0.8ex}{\includegraphics[scale=0.8]{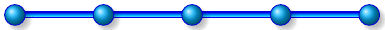}} )
    &=
    \min_{\psi}\Delta^2_{\mathrm{ad}}( \raisebox{-1.4ex}{\includegraphics[scale=0.8]{DK_C_5.pdf}} ).
\end{align}

For type-C algebras, the defining representation leads to a Lie algebraic uncertainty relation. 
At rank $n$, it is given by (see Tab.~\ref{tab:Cn_representations})
\begin{align}\label{eq:def_C_n}
    \Delta^2_{\mathrm{def}}( \raisebox{-1.4ex}{\includegraphics[scale=0.8]{DK_C_n.pdf}} )
    \geqslant n.
\end{align}
The structure of the bound is especially transparent at low rank. 
For $n=3,4,5$, the optimal Lie algebraic uncertainty relations read
\begin{align}
    \Delta^2_{\mathrm{def}}( \raisebox{-1.4ex}{\includegraphics[scale=0.8]{DK_C_3.pdf}} )
    &\geqslant3,\\
    \Delta^2_{\mathrm{def}}( \raisebox{-1.4ex}{\includegraphics[scale=0.8]{DK_C_4.pdf}} )
    &\geqslant4,\\
    \Delta^2_{\mathrm{def}}( \raisebox{-1.4ex}{\includegraphics[scale=0.8]{DK_C_5.pdf}} )
    &\geqslant5.
\end{align}

This mirrors the adjoint representation: 
for the defining representation, the type-A and type-C series yield the same optimal uncertainty bounds at equal rank, equivalently at equal Cartan subalgebra dimension. 
In Dynkin diagram notation, the correspondence reads
\begin{align}
    \min_{\psi}\Delta^2_{\mathrm{def}}(\raisebox{-0.6ex}{\includegraphics[scale=0.8]{DK_A_n.pdf}})=
    \min_{\psi}\Delta^2_{\mathrm{def}}( \raisebox{-1.4ex}{\includegraphics[scale=0.8]{DK_C_n.pdf}} ).
\end{align}
At ranks $n=3,4,5$, this correspondence takes the explicit form
\begin{align}
    \min_{\psi}\Delta^2_{\mathrm{def}}( \raisebox{-0.8ex}{\includegraphics[scale=0.8]{DK_A_3.pdf}} )
    &=
    \min_{\psi}\Delta^2_{\mathrm{def}}( \raisebox{-1.4ex}{\includegraphics[scale=0.8]{DK_C_3.pdf}} ),\\
    \min_{\psi}\Delta^2_{\mathrm{def}}( \raisebox{-0.8ex}{\includegraphics[scale=0.8]{DK_A_4.pdf}} )
    &=
    \min_{\psi}\Delta^2_{\mathrm{def}}( \raisebox{-1.4ex}{\includegraphics[scale=0.8]{DK_C_4.pdf}} ),\\
    \min_{\psi}\Delta^2_{\mathrm{def}}( \raisebox{-0.8ex}{\includegraphics[scale=0.8]{DK_A_5.pdf}} )
    &=
    \min_{\psi}\Delta^2_{\mathrm{def}}( \raisebox{-1.4ex}{\includegraphics[scale=0.8]{DK_C_5.pdf}} ).
\end{align}

The type-C analysis is thereby completed, making explicit how the general Lie algebraic uncertainty framework yields state-independent optimal bounds across its representations.


\subsection{Explicit Bounds for Type-D Lie Algebras}
\label{subsec:D_type}

We next consider the even-orthogonal series; 
see Tab.~\ref{tab:Simple_Lie_Algebra}.
Although the notation $\mathfrak{so}(2n,\mathbb C)$ is defined for all positive integers $n$, the simple type-D series starts only at $n\geqslant 4$.
The first three ranks are exceptional and are governed by the standard low-rank identifications.
First, $\mathfrak{so}(2,\mathbb C)$ is abelian and therefore lies outside the semisimple root system hierarchy.
At the next two ranks, one has
\begin{align}
D_2
=
\mathfrak{so}(4,\mathbb C)
\simeq
\mathfrak{sl}(2,\mathbb C)\oplus\mathfrak{sl}(2,\mathbb C)
=
A_1\oplus A_1,
\end{align}
and
\begin{align}
D_3
=
\mathfrak{so}(6,\mathbb C)
\simeq
\mathfrak{sl}(4,\mathbb C)
=
A_3.
\end{align}
For the compact real forms, these identifications become
\begin{align}
\mathfrak{so}(4)
\simeq
\mathfrak{su}(2)\oplus\mathfrak{su}(2),
\end{align}
and
\begin{align}
\mathfrak{so}(6)
\simeq
\mathfrak{su}(4).  
\end{align}
Thus $\mathfrak{so}(4)$ decomposes into two independent $A_1$ sectors, while $\mathfrak{so}(6)$ is governed by the $A_3$ uncertainty relation, up to the corresponding relabeling of simple roots in Dynkin diagram.
The genuinely new even-orthogonal hierarchy therefore begins with type-D for $n\geqslant 4$.

For the genuine type-D family, with $n\geqslant 4$, one has
$D_n=\mathfrak{so}(2n,\mathbb C)$, with compact real form $\mathfrak{so}(2n)$.
Unlike the non-simply-laced type-B (see Subsec.~\ref{subsec:B_type}) and type-C (see Subsec.~\ref{subsec:C_type}) series, type-D is simply laced; 
all simple roots therefore have the same squared length.
Its Dynkin diagram is a linear chain that splits at the $(n-2)$-th node into two terminal nodes, $\alpha_{n-1}$ and $\alpha_n$:
\begin{equation}
\label{eq:Dynkin_D}
 \raisebox{-8ex}{\includegraphics[scale=1]{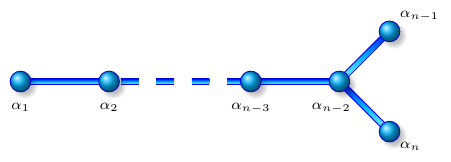}} \, .
\end{equation}

With the standard ordering of the type-D Dynkin diagram, the Cartan matrix (see Def.~\ref{def:Cartan_Matrix}) is
\begin{align}
    C
    =
    \begin{pmatrix}
    2 & -1 & 0 & \cdots & 0 & 0 & 0\\
    -1 & 2 & -1 & \cdots & 0 & 0 & 0\\
    0 & -1 & 2 & \ddots & \vdots & \vdots & \vdots\\
    \vdots & \vdots & \ddots & \ddots & -1 & 0 & 0\\
    0 & 0 & \cdots & -1 & 2 & -1 & -1\\
    0 & 0 & \cdots & 0 & -1 & 2 & 0\\
    0 & 0 & \cdots & 0 & -1 & 0 & 2
    \end{pmatrix},
\end{align}
and its inverse takes the form
\begin{align}
    C^{-1}
    =
    \begin{pmatrix}
    1 & 1 & 1 & \cdots & 1 & \dfrac12 & \dfrac12 \\[10pt]
    1 & 2 & 2 & \cdots & 2 & 1 & 1 \\[6pt]
    1 & 2 & 3 & \cdots & 3 & \dfrac32 & \dfrac32 \\[6pt]
    \vdots & \vdots & \vdots & \ddots & \vdots & \vdots & \vdots \\[6pt]
    1 & 2 & 3 & \cdots & n-2 & \dfrac{n-2}{2} & \dfrac{n-2}{2} \\[8pt]
    \dfrac12 & 1 & \dfrac32 & \cdots & \dfrac{n-2}{2} & \dfrac n4 & \dfrac{n-2}{4} \\[8pt]
    \dfrac12 & 1 & \dfrac32 & \cdots & \dfrac{n-2}{2} & \dfrac{n-2}{4} & \dfrac n4
    \end{pmatrix}.
\end{align}
From this expression, the entries of $C^{-1}$ can be read off directly.

\begin{mylem}{Entries of $C^{-1}$ for Type-D Lie Algebras ($n\geqslant4$)}{Entries_Cinverse_Type_D}
For the type-D Lie algebra $D_n$, the $(i,j)$-entry of the inverse Cartan matrix is
\begin{align}
\label{eq:Dn_inverse}
    C^{-1}_{ij}
    =
    \begin{cases}
    \min\{i,j\},
    &
    1\leqslant i,j\leqslant n-2,
    \\[1.5ex]
    \dfrac{i}{2},
    &
    1\leqslant i\leqslant n-2,
  \quad\text{and}\quad
    j\in\{n-1,n\},
    \\[1.5ex]
    \dfrac{j}{2},
    &
    i\in\{n-1,n\},
   \quad\text{and}\quad
    1\leqslant j\leqslant n-2,
    \\[1.5ex]
    \dfrac{n}{4},
    &
    i=j\in\{n-1,n\},
    \\[1.5ex]
    \dfrac{n-2}{4},
    &
     i,j\in\{n-1,n\} \quad\text{and}\quad i\ne j.
     \end{cases}
\end{align}
\end{mylem}

The inverse Cartan matrix encodes the coupling between Dynkin coefficients, but the uncertainty bound $\langle\vec{\lambda},\vec{\xi}\rangle$ (see Eq.~\eqref{eq:Lie_Algebraic_UR_Simplified}) also depends on the information carried by the simple roots. 
For the type-D series, this second ingredient is especially simple: since the root system is simply laced, all simple roots have equal squared length under the standard normalization. Thus the vector of squared simple root lengths is fixed as follows.

\begin{mylem}{Root Length Vector $\vec{\alpha}$ for Type-D Lie Algebras ($n\geqslant4$)}{Root_Length_Vector_D}
For the type-D Lie algebra $D_n$ with $n\geqslant4$, the root system is simply laced, so all simple roots have the same squared length.
With the standard normalization in which long roots have squared length $2$, i.e., $\langle\alpha_i,\alpha_i\rangle=2$, the squared simple root lengths form the vector 
\begin{align}
\label{eq:Root_Length_Vector_D}
    \vec{\alpha}
    =\begin{pmatrix}
        2\\
        \vdots\\
        2
    \end{pmatrix}\,.
\end{align}
\end{mylem}

The entrywise expression for $C^{-1}$ in Lem.~\ref{lem:Entries_Cinverse_Type_D}, together with the root length vector in Eq.~\eqref{eq:Root_Length_Vector_D}, yields the components of $\vec{\xi}$.

\begin{mythm}{Components of $\vec{\xi}$ for Type-D Lie Algebras ($n\geqslant4$)}{Components_xi_D}
For the type-D Lie algebra $D_n$ with $n\geqslant4$, the components of the vector $\vec{\xi}$ (see Eq.~\eqref{eq:xi}) are
\begin{align}
\label{eq:xi_D}
    \xi_i
    =
   \begin{cases}
   i(2n-i-1),
    &\quad
    1\leqslant i\leqslant n-2, \\[2ex]
    \dfrac{n(n-1)}{2} &\quad i\in\{n-1, n\}.
    \end{cases}
\end{align}
\end{mythm}

\begin{proof}
For the nodes along the linear chain (see Eq.~\eqref{eq:Dynkin_D}), $1\leqslant i\leqslant n-2$, the entrywise form of $C^{-1}$ (see Lem.~\ref{lem:Entries_Cinverse_Type_D}) yields
\begin{align}
    \xi_i
    &=
    2\left(\sum_{j=1}^{n-2}\min\{i,j\}
    +
    \frac{i}{2}
    +
    \frac{i}{2}\right)\\
    &=
    2\left(  \underbrace{\sum_{j=1}^{i}j}_{j\leqslant i}
    +
    \underbrace{\sum_{j=i+1}^{n-2}i}_{i<j\leqslant n-2}
    +
    \underbrace{\frac{i}{2}}_{j=n-1}+\underbrace{\frac{i}{2}}_{j=n}\right)\\
    &=
    2\left(\frac{i(i+1)}{2}
    +
    i(n-2-i)
    +
    i\right)\\
    &=
    i(2n-i-1).
\end{align}

For the first terminal node, i.e., $i=n-1$, the corresponding component $\xi_{n-1}$ is obtained from 
\begin{align}
    \xi_{n-1}
    &=
    2\left(\underbrace{\sum_{j=1}^{n-2}\frac{j}{2}}_{j\leqslant n-2}
    +
    \underbrace{\frac n4}_{j=n-1}
    +
   \underbrace{ \frac{n-2}{4}}_{j=n}\right)\\
    &=
    \left(\sum_{j=1}^{n-2}j\right)
    +
    \frac n2
    +
    \frac{n-2}{2}\\
    &=
    \frac{(n-2)(n-1)}{2}
    +
    n-1\\
    &=
    \frac{n(n-1)}{2}.
\end{align}
By the symmetry between the two terminal nodes (see Eq.~\eqref{eq:Dynkin_D}), namely $\alpha_{n-1}$ and $\alpha_n$, one also obtains
\begin{align}
    \xi_n
    =
    \xi_{n-1}
    =
    \frac{n(n-1)}{2}.
\end{align}
This establishes the claimed expression.
\end{proof}

Having determined the vector $\vec{\xi}$, we are now ready to state the optimal type-D Lie algebraic uncertainty relation.

\begin{mycor}{Type-D Lie Algebraic Uncertainty Relations ($n\geqslant4$)}{Type_D_Lie_UR}
For the type-D Lie algebra $D_n$ (see Tab.~\ref{tab:Simple_Lie_Algebra}) with $n\geqslant4$, the optimal Lie algebraic uncertainty relation (see Def.~\ref{def:Lie_Algebraic_Uncertainty_Relations}) takes the explicit form
\begin{align}
   \Delta^2_{\psi}( \raisebox{-5.2ex}{\includegraphics[scale=0.8]{DK_D_n.pdf}} )
   &\coloneqq
   \Delta^2_{\psi}(\mathfrak{so}(2n,\mathbb C))
   =
    \Delta^2_{\psi}(\mathfrak{so}(2n))\\
   & \geqslant
    \left(\sum_{i=1}^{n-2}i(2n-i-1)\lambda_i\right)+\frac{n(n-1)}{2}(\lambda_{n-1}+\lambda_n)\,.
\end{align}
The bound is state independent and is fixed entirely by the representation, through the highest weight $\Lambda$ (see Subsec.~\ref{subsec:Highest_Weight_Representations}).
Equivalently, it is determined by the Dynkin coefficient vector $\vec{\lambda}$ (see Eq.~\eqref{eq:Dynkin_Vector}).
Its optimality follows directly from Cor.~\ref{cor:Simplified_Lie_Algebraic_Uncertainty_Relations}.
\end{mycor}

Using the results above, we can now summarize the type-D uncertainty bounds in terms of the Dynkin coefficient vector $\vec{\lambda}$ (see Eq.~\eqref{eq:Dynkin_Vector}).
The resulting expressions are collected in Tab.~\ref{tab:explicit_bounds_D}.

\begin{table}[h]
\centering
\begin{tblr}{
width = \textwidth,
colspec = {c||c|X[c]},
row{1} = {bg=gray!50, font=\bfseries},
row{2} = {bg=gray!20, font=\bfseries},
hlines,
vlines,
cells = {m},
cell{3-13}{1} = {bg=gray!20},
cell{3-13}{2} = {bg=mRed},
cell{3-13}{3} = {bg=mGreen},
}

\SetCell[c=3]{c}
\textbf{Type-D Lie Algebraic Uncertainty Relations \quad  \raisebox{-3.2ex}{\includegraphics[scale=0.5]{DK_D_n.pdf}} }
&
&
\\

Compact Lie algebra
&
Rank
&
Optimal Bound
\\

$\mathfrak{so}(2)$
&
$1$
&
$0$
\\

$\mathfrak{so}(4)$
&
$2$
&
$\lambda_1+\lambda_2$
\\

$\mathfrak{so}(6)$
&
$3$
&
$4\lambda_1+3\lambda_2+3\lambda_3$
\\

$\mathfrak{so}(8)$
&
$4$
&
$6\lambda_1+10\lambda_2+6\lambda_3+6\lambda_4$
\\

$\mathfrak{so}(10)$
&
$5$
&
$8\lambda_1+14\lambda_2+18\lambda_3+10\lambda_4+10\lambda_5$
\\

$\mathfrak{so}(12)$
&
$6$
&
$10\lambda_1+18\lambda_2+24\lambda_3+28\lambda_4+15\lambda_5+15\lambda_6$
\\

$\mathfrak{so}(14)$
&
$7$
&
$12\lambda_1+22\lambda_2+30\lambda_3+36\lambda_4+40\lambda_5+21\lambda_6+21\lambda_7$
\\

$\mathfrak{so}(16)$
&
$8$
&
$\displaystyle
14\lambda_1+26\lambda_2+36\lambda_3+44\lambda_4
+50\lambda_5+54\lambda_6+28\lambda_7+28\lambda_8$
\\

$\mathfrak{so}(18)$
&
$9$
&
$\displaystyle
16\lambda_1+30\lambda_2+42\lambda_3+52\lambda_4
+60\lambda_5+66\lambda_6+70\lambda_7
+36\lambda_8+36\lambda_9$
\\

$\vdots$
&
$\vdots$
&
$\vdots$
\\

$\mathfrak{so}(2n)$
&
$n$
&
$\displaystyle
(2n-2)\lambda_1
+2(2n-3)\lambda_2
+3(2n-4)\lambda_3
+\cdots
+(n-2)(n+1)\lambda_{n-2}
+\frac{n(n-1)}{2}
\bigl(\lambda_{n-1}+\lambda_n\bigr)$
\\

\end{tblr}

\caption{
\textbf{Type-D Lie Algebraic Uncertainty Relations.}
The simple type-D series begins at $n=4$; 
the first three ranks are governed by low-rank degeneracies.
For $n=1$, $\mathfrak{so}(2)$ is abelian, the entry 0 is the minimum variance of a Hermitian generator normalized using a separately chosen positive inner product; the Killing-form normalization does not apply.
For $n=2$, $\mathfrak{so}(4)\simeq\mathfrak{su}(2)\oplus\mathfrak{su}(2)$; 
and for $n=3$, $\mathfrak{so}(6)\simeq\mathfrak{su}(4)$.
For $n\geqslant4$, the compact Lie algebra $\mathfrak{so}(2n)$, with Dynkin coefficient vector $\vec{\lambda}=(\lambda_1,\ldots,\lambda_n)^{\T}$, has the state-independent optimal bound
$(\sum_{i=1}^{n-2} i(2n-i-1)\lambda_i)+n(n-1)(\lambda_{n-1}+\lambda_n)/2$.}
\label{tab:explicit_bounds_D}
\end{table}

We have so far derived the optimal type-D uncertainty relation in its general representation-theoretic form.
We now illustrate this result in concrete representation examples.
Tab.~\ref{tab:Dn_representations} collects representative type-D representations, together with the corresponding optimal uncertainty bounds.

\begin{table}[htbp]
\centering
\begin{tblr}{
width = \textwidth,
colspec = {X[c]||X[c]|X[c]|X[c]},
row{1} = {bg=gray!50,font=\bfseries},
row{2} = {bg=gray!20,font=\bfseries},
hlines,
vlines,
cells = {m},
cell{3-8}{1} = {bg=gray!20},
cell{3-8}{2} = {bg=mRed},
cell{3-8}{3} = {bg=mGreen},
cell{3-8}{4} = {bg=mBlue},
}

\SetCell[c=4]{c}
\textbf{Type-D Lie Algebraic Uncertainty Relations Across Representations}
&
&
&
\\

Representation
&
Highest Weight $\Lambda$
&
Dynkin Vector $\vec{\lambda}^{\,\T}$
&
Optimal Bound
\\

Defining
&
$\omega_1$
&
$(1,0,0,\dots,0)$
&
$2n-2$
\\

Adjoint
&
$\omega_2$
&
$(0,1,0,\dots,0)$
&
$4n-6$
\\

First Half-Spinor
&
$\omega_{n-1}$
&
$(0,\dots,0,1,0)$
&
$\dfrac{n(n-1)}{2}$
\\

Second Half-Spinor
&
$\omega_n$
&
$(0,\dots,0,0,1)$
&
$\dfrac{n(n-1)}{2}$
\\

Traceless Symmetric Square
&
$2\omega_1$
&
$(2,0,0,\dots,0)$
&
$4n-4$
\\

General Exterior-Power
&
$\omega_i$ with $\ 1\leqslant i\leqslant n-2$
&
$(0,\dots,1,\dots,0)$
&
$i(2n-i-1)$
\\

\end{tblr}

\caption{
\textbf{Lie Algebraic Uncertainty Relations Across Representations of $\mathfrak{so}(2n)$.}
Representative standard representations of $\mathfrak{so}(2n)$ with $n\geqslant 4$, specified by their highest weight $\Lambda$ and Dynkin coefficient vector $\vec{\lambda}^{\,\T}$, together with the corresponding optimal lower bound.
}
\label{tab:Dn_representations}
\end{table}

Building on the analyses of the type-A, type-B, and type-C series, we now turn to the final classical family, the type-D series, taking the adjoint representation as a representative example. 
We begin with the generic $D_n$ case, for which the Dynkin diagram and the optimal lower bound for the adjoint representation are summarized below; see Tab.~\ref{tab:Dn_representations}.
\begin{align}\label{eq:adj_D_n}
    \Delta^2_{\mathrm{ad}}( \raisebox{-5.2ex}{\includegraphics[scale=0.8]{DK_D_n.pdf}} )
    \geqslant 4n-6.
\end{align}
For the cases $n=4,5,6$, the optimal Lie algebraic uncertainty relations can be written explicitly as follows
\begin{align}
    \Delta^2_{\mathrm{ad}}( \raisebox{-5.2ex}{\includegraphics[scale=0.8]{DK_D_4.pdf}} )
    &\geqslant 10,\\
    \Delta^2_{\mathrm{ad}}(\raisebox{-5.2ex}{\includegraphics[scale=0.8]{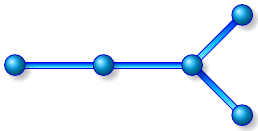}})
    &\geqslant 14,\\
    \Delta^2_{\mathrm{ad}}(\raisebox{-5.2ex}{\includegraphics[scale=0.8]{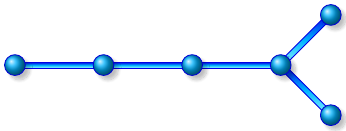}})
    &\geqslant 18.
\end{align}

Turning to the defining representation of type-D Lie algebras, we obtain the following rank-$n$ Lie algebraic uncertainty relation (see Tab.~\ref{tab:Dn_representations}):
\begin{align}\label{eq:def_D_n}
    \Delta^2_{\mathrm{def}}( \raisebox{-5.2ex}{\includegraphics[scale=0.8]{DK_D_n.pdf}} )
    \geqslant 2(n-1).
\end{align}
The low-rank cases $n=4,5,6$ make the structure of the optimal bounds explicit:
\begin{align}
    \Delta^2_{\mathrm{def}}( \raisebox{-5.2ex}{\includegraphics[scale=0.8]{DK_D_4.pdf}} )
    &\geqslant6,\\
    \Delta^2_{\mathrm{def}}( \raisebox{-5.2ex}{\includegraphics[scale=0.8]{DK_D_5.pdf}} )
    &\geqslant8,\\
    \Delta^2_{\mathrm{def}}( \raisebox{-5.2ex}{\includegraphics[scale=0.8]{DK_D_6.pdf}} )
    &\geqslant10.
\end{align}

Having analyzed the adjoint representation for each classical Lie algebra family, from type-A through type-D, we now collect the resulting optimal uncertainty bounds into a unified statement.

\begin{mythm}{Classical Lie Algebraic Uncertainty Relations for Adjoint Representations}{Classical_ad}
For rank $n\geqslant 4$, the optimal uncertainty bounds in the adjoint representation reveal a unified structure across the classical Lie algebra families, from type A through type D, as captured by the following correspondences.
\begin{align}
\underbrace{\min_{\psi}2\Delta^2_{\psi}(\mathfrak{su}(n+1))}_{\text{Type $\mathrm{A}_n$}}=
\underbrace{\min_{\psi}\Delta^2_{\psi}(\mathfrak{so}(2n+1))}_{\text{Type $\mathrm{B}_n$}}+4=
\underbrace{\min_{\psi}2\Delta^2_{\psi}(\mathfrak{sp}(n))}_{\text{Type $\mathrm{C}_n$}}=
\underbrace{\min_{\psi}\Delta^2_{\psi}(\mathfrak{so}(2n))}_{\text{Type $\mathrm{D}_n$}}+6=
4n.
\end{align}
\end{mythm}

In the diagrammatic language of root systems, encoded by Dynkin diagrams, using the explicit bounds derived above Thm.~\ref{thm:Classical_ad}, we obtain the following low-rank identities in Dynkin diagram notation.
\begin{align}
&\underbrace{\min_{\psi}\Delta^2_{\mathrm{ad}}( \raisebox{-0.8ex}{\includegraphics[scale=0.8]{DK_A_3.pdf}} )}_{\text{Type $\mathrm{A}_3$}}=
\underbrace{\min_{\psi}\Delta^2_{\mathrm{ad}}( \raisebox{-1.4ex}{\includegraphics[scale=0.8]{DK_C_3.pdf}} )}_{\text{Type $\mathrm{C}_3$}}=
\underbrace{\min_{\psi}\Delta^2_{\mathrm{def}}( \raisebox{-5.2ex}{\includegraphics[scale=0.8]{DK_D_4.pdf}} )}_{\text{Type $\mathrm{D}_4$}}=6,\\
&\underbrace{\min_{\psi}\Delta^2_{\mathrm{ad}}( \raisebox{-0.8ex}{\includegraphics[scale=0.8]{DK_A_4.pdf}} )}_{\text{Type $\mathrm{A}_4$}}=
\underbrace{\min_{\psi}\Delta^2_{\mathrm{ad}}( \raisebox{-1.4ex}{\includegraphics[scale=0.8]{DK_B_3.pdf}} )}_{\text{Type $\mathrm{B}_3$}}=
\underbrace{\min_{\psi}\Delta^2_{\mathrm{ad}}( \raisebox{-1.4ex}{\includegraphics[scale=0.8]{DK_C_4.pdf}} )}_{\text{Type $\mathrm{C}_4$}}=8,\\
&\underbrace{\min_{\psi}\Delta^2_{\mathrm{ad}}( \raisebox{-0.8ex}{\includegraphics[scale=0.8]{DK_A_5.pdf}} )}_{\text{Type $\mathrm{A}_5$}}=
\underbrace{\min_{\psi}\Delta^2_{\mathrm{ad}}( \raisebox{-1.4ex}{\includegraphics[scale=0.8]{DK_C_5.pdf}} )}_{\text{Type $\mathrm{C}_5$}}=10.
\end{align}

Across the classical families of defining representations, types A-D (see Tab.~\ref{tab:Simple_Lie_Algebra}), the optimal quantum uncertainty bounds obey the following correspondence:

\begin{mythm}{Classical Lie Algebraic Uncertainty Relations for Defining Representations}{Classical_def}
For rank $n\geqslant 4$, the optimal Lie algebraic uncertainty bounds in the defining representation display a unified structure across the classical families of Lie algebras. 
This structure is captured by the following correspondences.
\begin{align}
\underbrace{\min_{\psi}2\Delta^2_{\psi}(\mathfrak{su}(n+1))}_{\text{Type $\mathrm{A}_n$}}=
\underbrace{\min_{\psi}\Delta^2_{\psi}(\mathfrak{so}(2n+1))}_{\text{Type $\mathrm{B}_n$}}+1=
\underbrace{\min_{\psi}2\Delta^2_{\psi}(\mathfrak{sp}(n))}_{\text{Type $\mathrm{C}_n$}}=
\underbrace{\min_{\psi}\Delta^2_{\psi}(\mathfrak{so}(2n))}_{\text{Type $\mathrm{D}_n$}}+2=
2n.
\end{align}
\end{mythm}

Using Dynkin diagrams to encode the underlying root systems, we obtain the following identities in low rank.
\begin{align}
&\underbrace{\min_{\psi}\Delta^2_{\mathrm{def}}( \raisebox{-0.6ex}{\includegraphics[scale=0.8]{DK_B_1.pdf}} )}_{\text{Type $\mathrm{A}_1$}}=1,\\
&\underbrace{\min_{\psi}\Delta^2_{\mathrm{def}}( \raisebox{-0.8ex}{\includegraphics[scale=0.8]{DK_A_2.pdf}} )}_{\text{Type $\mathrm{A}_2$}}=
\underbrace{\min_{\psi}\Delta^2_{\mathrm{def}}( \raisebox{-1.4ex}{\includegraphics[scale=0.8]{DK_C_2.pdf}} )}_{\text{Type $\mathrm{C}_2$}}=2,\\
&\underbrace{\min_{\psi}\Delta^2_{\mathrm{def}}( \raisebox{-0.8ex}{\includegraphics[scale=0.8]{DK_A_3.pdf}} )}_{\text{Type $\mathrm{A}_3$}}=
\underbrace{\min_{\psi}\Delta^2_{\mathrm{def}}( \raisebox{-1.4ex}{\includegraphics[scale=0.8]{DK_B_2.pdf}} )}_{\text{Type $\mathrm{B}_2$}}=
\underbrace{\min_{\psi}\Delta^2_{\mathrm{def}}( \raisebox{-1.4ex}{\includegraphics[scale=0.8]{DK_C_3.pdf}} )}_{\text{Type $\mathrm{C}_3$}}=3,\\
&\underbrace{\min_{\psi}\Delta^2_{\mathrm{def}}( \raisebox{-0.8ex}{\includegraphics[scale=0.8]{DK_A_4.pdf}} )}_{\text{Type $\mathrm{A}_4$}}=
\underbrace{\min_{\psi}\Delta^2_{\mathrm{def}}( \raisebox{-1.4ex}{\includegraphics[scale=0.8]{DK_C_4.pdf}} )}_{\text{Type $\mathrm{C}_4$}}=4,\\
&\underbrace{\min_{\psi}\Delta^2_{\mathrm{def}}( \raisebox{-0.8ex}{\includegraphics[scale=0.8]{DK_A_5.pdf}} )}_{\text{Type $\mathrm{A}_5$}}=
\underbrace{\min_{\psi}\Delta^2_{\mathrm{def}}( \raisebox{-1.4ex}{\includegraphics[scale=0.8]{DK_C_5.pdf}} )}_{\text{Type $\mathrm{C}_5$}}=5.
\end{align}

Taken together, these results reveal further correspondences between different representations and classical Lie algebra families. 
For example, one obtains
\begin{align}
&\underbrace{\min_{\psi}\Delta^2_{\mathrm{ad}}( \raisebox{-0.6ex}{\includegraphics[scale=0.8]{DK_B_1.pdf}} )}_{\text{Type $\mathrm{A}_1$}}=
\underbrace{\min_{\psi}\Delta^2_{\mathrm{def}}( \raisebox{-0.8ex}{\includegraphics[scale=0.8]{DK_A_2.pdf}} )}_{\text{Type $\mathrm{A}_2$}}=
\underbrace{\min_{\psi}\Delta^2_{\mathrm{def}}( \raisebox{-1.4ex}{\includegraphics[scale=0.8]{DK_C_2.pdf}} )}_{\text{Type $\mathrm{C}_2$}}=2\\
&\underbrace{\min_{\psi}\Delta^2_{\mathrm{ad}}( \raisebox{-0.8ex}{\includegraphics[scale=0.8]{DK_A_2.pdf}} )}_{\text{Type $\mathrm{A}_2$}}=
\underbrace{\min_{\psi}\Delta^2_{\mathrm{def}}( \raisebox{-0.8ex}{\includegraphics[scale=0.8]{DK_A_4.pdf}} )}_{\text{Type $\mathrm{A}_4$}}=
\underbrace{\min_{\psi}\Delta^2_{\mathrm{def}}( \raisebox{-1.4ex}{\includegraphics[scale=0.8]{DK_C_4.pdf}} )}_{\text{Type $\mathrm{C}_4$}}=4,\\
&\underbrace{\min_{\psi}\Delta^2_{\mathrm{ad}}( \raisebox{-0.8ex}{\includegraphics[scale=0.8]{DK_A_3.pdf}} )}_{\text{Type $\mathrm{A}_3$}}=
\underbrace{\min_{\psi}\Delta^2_{\mathrm{ad}}( \raisebox{-1.4ex}{\includegraphics[scale=0.8]{DK_C_3.pdf}} )}_{\text{Type $\mathrm{C}_3$}}=
\underbrace{\min_{\psi}\Delta^2_{\mathrm{def}}( \raisebox{-5.2ex}{\includegraphics[scale=0.8]{DK_D_4.pdf}} )}_{\text{Type $\mathrm{D}_4$}}=6\\
&\underbrace{\min_{\psi}\Delta^2_{\mathrm{ad}}( \raisebox{-0.8ex}{\includegraphics[scale=0.8]{DK_A_4.pdf}} )}_{\text{Type $\mathrm{A}_4$}}=
\underbrace{\min_{\psi}\Delta^2_{\mathrm{ad}}( \raisebox{-1.4ex}{\includegraphics[scale=0.8]{DK_B_3.pdf}} )}_{\text{Type $\mathrm{B}_3$}}=
\underbrace{\min_{\psi}\Delta^2_{\mathrm{ad}}( \raisebox{-1.4ex}{\includegraphics[scale=0.8]{DK_C_4.pdf}} )}_{\text{Type $\mathrm{C}_4$}}\notag\\=
&\underbrace{\min_{\psi}\Delta^2_{\mathrm{def}}( \raisebox{-5.2ex}{\includegraphics[scale=0.8]{DK_D_5.pdf}} )}_{\text{Type $\mathrm{D}_5$}}=8.
\end{align}

The examples displayed here represent only a small part of what follows from the general framework. 
They are intended to serve as a guide for exploring further correspondences among classical Lie algebras, exceptional Lie algebras, and other representations.


\subsection{Explicit Bounds for Exceptional Lie Algebras}
\label{subsec:exceptional}

We finally turn to the exceptional Lie algebras.
In contrast to the classical families, these algebras occur as isolated types rather than members of infinite rank families. 
We evaluate their bounds case by case using their Cartan matrices
and simple-root lengths. 
We begin with the smallest exceptional algebra $G_2$, whose Dynkin diagram is
\begin{equation}
\label{eq:Dynkin_G2}
 \raisebox{-2.6ex}{\includegraphics[scale=1]{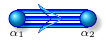}} \, .
\end{equation}
Here, we adopt the convention that $\alpha_1$ denotes the long simple root and $\alpha_2$ the short simple root. 
Accordingly, the multiple bond is always directed from the long root to the short root, with the short root placed at the terminal node of the Dynkin diagram in Eq.~\eqref{eq:Dynkin_G2}. 
Under the normalization, $\langle\alpha_1,\alpha_1\rangle=2$ and $\langle\alpha_2,\alpha_2\rangle=2/3$, the simple root inner product is
\begin{align}
    \langle\alpha_1,\alpha_2\rangle=-1.
\end{align}
With this convention, the Cartan matrix of $G_2$ takes the form
\begin{align}
C=
\begin{pmatrix}
 2 & -3 \\
-1 &  2
\end{pmatrix},
\end{align}
The corresponding inverse is given by
\begin{align}
C^{-1} = 
\begin{pmatrix}
2 & 3 \\
1 & 2
\end{pmatrix}.
\end{align}
For $G_2$, the vector of squared simple-root lengths is 
\begin{align}
    \vec{\alpha}=
    \begin{pmatrix}
        2\\
        \frac{2}{3}
    \end{pmatrix}.
\end{align}
Combining the inverse Cartan matrix and squared root lengths, the vector $\vec{\xi}$ is therefore
\begin{align}
    \vec{\xi}=
    C^{-1}\cdot\vec{\alpha}=
    \begin{pmatrix}
    2 & 3 \\
    1 & 2
    \end{pmatrix}
    \begin{pmatrix}
        2\\
        \frac{2}{3}
    \end{pmatrix}
    =
    \begin{pmatrix}
        6\\
        \frac{10}{3}
    \end{pmatrix}.
\end{align}

The optimal lower bound $\langle\vec{\lambda},\vec{\xi}\rangle$ (see Eq.~\eqref{eq:Lie_Algebraic_UR_Simplified}) is then obtained by pairing the Dynkin coefficient vector $\vec{\lambda}$ with $\vec{\xi}$. 
\begin{align}
    \langle\vec{\lambda},\vec{\xi}\rangle=
    \begin{pmatrix} 
    \lambda_1 & \lambda_2 
    \end{pmatrix}
    \begin{pmatrix}
        6\\
        \frac{10}{3}
    \end{pmatrix}
    =
    6\lambda_1+\frac{10}{3}\lambda_2.
\end{align}
This yields the uncertainty relation for $G_2$.

\begin{mycor}{Type-$\mathrm{G}_2$ Lie Algebraic Uncertainty Relations}{Type_G_Lie_UR}
For the exceptional Lie algebra $G_2$ (see Tab.~\ref{tab:Simple_Lie_Algebra}), the optimal Lie algebraic uncertainty relation (see Def.~\ref{def:Lie_Algebraic_Uncertainty_Relations}) takes the explicit form
\begin{align}
   \Delta^2_{\psi}( \raisebox{-1.4ex}{\includegraphics[scale=1]{DK_G_2.pdf}} )
   \coloneqq
   \Delta^2_{\psi}(\mathfrak{g}_2)
   \geqslant
   6\lambda_1+\frac{10}{3}\lambda_2\,.
\end{align}
The bound is state-independent and is fixed solely by the representation, encoded by the Dynkin coefficient vector $\vec{\lambda}$ (see Eq.~\eqref{eq:Dynkin_Vector}).
\end{mycor}

We next turn to the exceptional Lie algebra $F_4$, whose Dynkin diagram is
\begin{equation}
\label{eq:Dynkin_F4}
\raisebox{-3.2ex}{\includegraphics[scale=1]{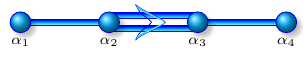}} \, .
\end{equation}
Here, we choose the ordering in which the two left nodes, $\alpha_1$ and $\alpha_2$, correspond to long simple roots, while the two right nodes, $\alpha_3$ and $\alpha_4$, correspond to short simple roots. 
This fixes the orientation of the multiple bond from the long roots toward the short roots, consistent with the convention adopted throughout this work. 
Under the normalization,
\begin{align}
    \langle\alpha_1,\alpha_1\rangle=
    \langle\alpha_2,\alpha_2\rangle=2,
\end{align}
and
\begin{align}
    \langle\alpha_3,\alpha_3\rangle=
    \langle\alpha_4,\alpha_4\rangle=1.
\end{align}
The simple root inner products are
\begin{align}
    \langle\alpha_1,\alpha_2\rangle&=-1,\\
    \langle\alpha_2,\alpha_3\rangle&=-1,\\
    \langle\alpha_3,\alpha_4\rangle&=-\frac{1}{2},
\end{align}
with all other pairs of distinct simple roots being orthogonal.
With this convention, the Cartan matrix of $F_4$ takes the form
\begin{align}
C=
\begin{pmatrix}
2 & -1 & 0 & 0 \\
-1 & 2 & -2 & 0 \\
0 & -1 & 2 & -1 \\
0 & 0 & -1 & 2
\end{pmatrix}.
\end{align}
The inverse Cartan matrix is therefore
\begin{align}
C^{-1}=
\begin{pmatrix}
2 & 3 & 4 & 2 \\
3 & 6 & 8 & 4 \\
2 & 4 & 6 & 3 \\
1 & 2 & 3 & 2
\end{pmatrix}.
\end{align}
In the same convention, the vector of squared simple-root lengths for $F_4$ is
\begin{align}
\vec{\alpha}=
\begin{pmatrix}
2 \\
2 \\
1 \\
1 
\end{pmatrix}.
\end{align}

The components of $\vec{\xi}$ follow directly by combining the entrywise expression for $C^{-1}$ with the root length vector $\vec{\alpha}$.
\begin{align}
    \vec{\xi}=
    C^{-1}\cdot\vec{\alpha}=
    \begin{pmatrix}
    2 & 3 & 4 & 2 \\
    3 & 6 & 8 & 4 \\
    2 & 4 & 6 & 3 \\
    1 & 2 & 3 & 2
    \end{pmatrix}
    \begin{pmatrix}
    2 \\
    2 \\
    1 \\
    1 
    \end{pmatrix}
    =
    \begin{pmatrix} 
    16 \\ 
    30 \\ 
    21 \\ 
    11 
    \end{pmatrix}.
\end{align}

The uncertainty bound is therefore fixed by a simple pairing between highest weight and root system geometry, yielding $\langle\vec{\lambda},\vec{\xi}\rangle$ as in Eq.~\eqref{eq:Lie_Algebraic_UR_Simplified}.
\begin{align}
    \langle\vec{\lambda},\vec{\xi}\rangle=
    \begin{pmatrix} 
    \lambda_1 & \lambda_2 & \lambda_3 & \lambda_4
    \end{pmatrix}
    \begin{pmatrix} 
    16 \\ 
    30 \\ 
    21 \\ 
    11 
    \end{pmatrix}
    =
    16\lambda_1+30\lambda_2+21\lambda_3+11\lambda_4.
\end{align}
This leads to the uncertainty relation for $F_4$.

\begin{mycor}{Type-$\mathrm{F}_4$ Lie Algebraic Uncertainty Relations}{Type_F_Lie_UR}
For the exceptional Lie algebra $F_4$ (see Tab.~\ref{tab:Simple_Lie_Algebra}), the optimal Lie algebraic uncertainty relation (see Def.~\ref{def:Lie_Algebraic_Uncertainty_Relations}) takes the explicit form
\begin{align}
   \Delta^2_{\psi}( \raisebox{-1.2ex}{\includegraphics[scale=0.75]{DK_F_4.pdf}})
   \coloneqq
   \Delta^2_{\psi}(\mathfrak{f}_4)
   \geqslant
   16\lambda_1+30\lambda_2+21\lambda_3+11\lambda_4\,.
\end{align}
The bound is state-independent and is fixed solely by the representation, encoded by the Dynkin coefficient vector $\vec{\lambda}$ (see Eq.~\eqref{eq:Dynkin_Vector}).
\end{mycor}

Beyond $G_2$ and $F_4$, the exceptional family also contains the three simply laced type-E algebras, namely $E_6$, $E_7$, $E_8$.
The first case is $E_6$. 
With the simple roots labeled by their corresponding nodes, its Dynkin diagram is given by
\begin{equation}
\label{eq:Dynkin_E6}
 \raisebox{-3.6ex}{\includegraphics[scale=1]{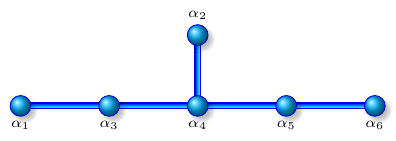}}\, .
\end{equation}

The $E_6$ root system is simply laced. 
Hence all simple roots have the same squared length, and under the standard normalization of long roots one has
\begin{align}
\langle \alpha_1,\alpha_1\rangle
=
\langle \alpha_2,\alpha_2\rangle
=
\langle \alpha_3,\alpha_3\rangle
=
\langle \alpha_4,\alpha_4\rangle
=
\langle \alpha_5,\alpha_5\rangle
=
\langle \alpha_6,\alpha_6\rangle
=
2.
\end{align}
The nonzero inner products between distinct simple roots are determined by the adjacency relations of the Dynkin diagram:
\begin{align}
\langle\alpha_1,\alpha_3\rangle
=
\langle\alpha_2,\alpha_4\rangle
=
\langle\alpha_3,\alpha_4\rangle
=
\langle\alpha_4,\alpha_5\rangle
=
\langle\alpha_5,\alpha_6\rangle
=
-1.
\end{align}
All other pairs of distinct simple roots are orthogonal, and therefore their inner products vanish.

With this convention, the Cartan matrix of $E_6$ is given by
\begin{align}
C
=
\begin{pmatrix}
2 & 0 & -1 & 0 & 0 & 0 \\
0 & 2 & 0 & -1 & 0 & 0 \\
-1 & 0 & 2 & -1 & 0 & 0 \\
0 & -1 & -1 & 2 & -1 & 0 \\
0 & 0 & 0 & -1 & 2 & -1 \\
0 & 0 & 0 & 0 & -1 & 2
\end{pmatrix}.
\end{align}
The corresponding inverse Cartan matrix is
\begin{align}\label{eq:cartan_E6}
C^{-1}
=
\begin{pmatrix}
\dfrac{4}{3} & 1 & \dfrac{5}{3} & 2 & \dfrac{4}{3} & \dfrac{2}{3} \\[8pt]
1 & 2 & 2 & 3 & 2 & 1 \\[8pt]
\dfrac{5}{3} & 2 & \dfrac{10}{3} & 4 & \dfrac{8}{3} & \dfrac{4}{3} \\[8pt]
2 & 3 & 4 & 6 & 4 & 2 \\[8pt]
\dfrac{4}{3} & 2 & \dfrac{8}{3} & 4 & \dfrac{10}{3} & \dfrac{5}{3} \\[8pt]
\dfrac{2}{3} & 1 & \dfrac{4}{3} & 2 & \dfrac{5}{3} & \dfrac{4}{3}
\end{pmatrix}.
\end{align}
Meanwhile, the vector of squared simple-root lengths of $E_6$ is
\begin{align}\label{eq:root_vector_E6}
\vec{\alpha}
=
\begin{pmatrix}
2\\
\vdots\\
2\
\end{pmatrix}.
\end{align}
Combining Eqs.~\eqref{eq:cartan_E6} and \eqref{eq:root_vector_E6} gives the representation-independent vector $\vec{\xi}$
\begin{align}
\vec{\xi}
=
C^{-1}\cdot\vec{\alpha}
&=
\begin{pmatrix}
\dfrac{4}{3} & 1 & \dfrac{5}{3} & 2 & \dfrac{4}{3} & \dfrac{2}{3} \\[8pt]
1 & 2 & 2 & 3 & 2 & 1 \\[8pt]
\dfrac{5}{3} & 2 & \dfrac{10}{3} & 4 & \dfrac{8}{3} & \dfrac{4}{3} \\[8pt]
2 & 3 & 4 & 6 & 4 & 2 \\[8pt]
\dfrac{4}{3} & 2 & \dfrac{8}{3} & 4 & \dfrac{10}{3} & \dfrac{5}{3} \\[8pt]
\dfrac{2}{3} & 1 & \dfrac{4}{3} & 2 & \dfrac{5}{3} & \dfrac{4}{3}
\end{pmatrix}
\begin{pmatrix}
2\\
2\\
\vdots\\
2\\
2
\end{pmatrix}
=
\begin{pmatrix}
16\\
22\\
30\\
42\\
30\\
16
\end{pmatrix}.
\end{align}

Thus, for a highest weight 
$\Lambda
=
\sum_{i=1}^{6}
\lambda_i\omega_i$,
the uncertainty lower bound $\langle\vec{\lambda},\vec{\xi}\rangle$ (see Eq.~\eqref{eq:Lie_Algebraic_UR_Simplified}) can be written explicitly as follows
\begin{align}
\langle\vec{\lambda},\vec{\xi}\rangle
&=
\begin{pmatrix}
\lambda_1 & \lambda_2 & \lambda_3 & \lambda_4 & \lambda_5 & \lambda_6
\end{pmatrix}
\begin{pmatrix}
16\\
22\\
30\\
42\\
30\\
16
\end{pmatrix}
=
16\lambda_1
+
22\lambda_2
+
30\lambda_3
+
42\lambda_4
+
30\lambda_5
+
16\lambda_6.
\end{align}
This gives the following uncertainty relation for $E_6$.

\begin{mycor}{Type-$\mathrm{E}_6$ Lie Algebraic Uncertainty Relations}{Type_E6_Lie_UR}
For the exceptional Lie algebra $E_6$ (see Tab.~\ref{tab:Simple_Lie_Algebra}), the optimal Lie algebraic uncertainty relation (see Def.~\ref{def:Lie_Algebraic_Uncertainty_Relations}) takes the explicit form
\begin{align}
\Delta^2_{\psi}( \raisebox{-0.6ex}{\includegraphics[scale=0.75]{DK_E_6.pdf}})
\coloneqq
\Delta^2_{\psi}(\mathfrak e_6)
\geqslant
16\lambda_1
+
22\lambda_2
+
30\lambda_3
+
42\lambda_4
+
30\lambda_5
+
16\lambda_6 .
\end{align}
The lower bound is independent of the state and is determined entirely by the Dynkin coefficient vector
$\vec{\lambda}$
of the highest weight representation.
\end{mycor}

We next turn to the exceptional Lie algebra $E_7$. 
Adopting the Humphreys numbering convention~\cite{Humphreys1972}, its Dynkin diagram is given by
\begin{equation}
\label{eq:Dynkin_E7}
    \raisebox{-3.4ex}{\includegraphics[scale=1]{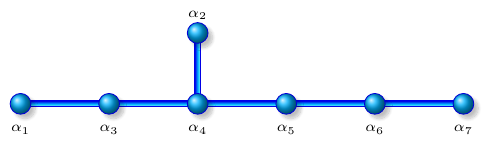}}\, .
\end{equation}

The $E_7$ root system is simply laced.
Consequently, all simple roots have the same length; 
with the standard normalization,
\begin{align}
\langle \alpha_1,\alpha_1\rangle
=
\langle \alpha_2,\alpha_2\rangle
=
\langle \alpha_3,\alpha_3\rangle
=
\langle \alpha_4,\alpha_4\rangle
=
\langle \alpha_5,\alpha_5\rangle
=
\langle \alpha_6,\alpha_6\rangle
=
\langle \alpha_7,\alpha_7\rangle
=
2.
\end{align}
The remaining inner products are fixed entirely by the adjacency structure of the Dynkin diagram. 
The only nonzero off-diagonal entries are
\begin{align}
    \langle\alpha_1,\alpha_3\rangle
    =
    \langle\alpha_2,\alpha_4\rangle
    =
    \langle\alpha_3,\alpha_4\rangle
    =
    \langle\alpha_4,\alpha_5\rangle
    =
    \langle\alpha_5,\alpha_6\rangle
    =
    \langle\alpha_6,\alpha_7\rangle
    =
    -1,
\end{align}
whereas all other pairs of distinct simple roots are orthogonal, and hence have vanishing inner product.

Accordingly, the Cartan matrix takes the form
\begin{align}
C
=
\begin{pmatrix}
2 & 0 & -1 & 0 & 0 & 0 & 0\\
0 & 2 & 0 & -1 & 0 & 0 & 0\\
-1 & 0 & 2 & -1 & 0 & 0 & 0\\
0 & -1 & -1 & 2 & -1 & 0 & 0\\
0 & 0 & 0 & -1 & 2 & -1 & 0\\
0 & 0 & 0 & 0 & -1 & 2 & -1\\
0 & 0 & 0 & 0 & 0 & -1 & 2
\end{pmatrix}.
\end{align}
Its inverse is
\begin{align}
C^{-1}
=
\begin{pmatrix}
2 & 2 & 3 & 4 & 3 & 2 & 1\\[8pt]
2 & \dfrac{7}{2} & 4 & 6 & \dfrac{9}{2} & 3 & \dfrac{3}{2}\\[8pt]
3 & 4 & 6 & 8 & 6 & 4 & 2\\[8pt]
4 & 6 & 8 & 12 & 9 & 6 & 3\\[6pt]
3 & \dfrac{9}{2} & 6 & 9 & \dfrac{15}{2} & 5 & \dfrac{5}{2}\\[8pt]
2 & 3 & 4 & 6 & 5 & 4 & 2\\[8pt]
1 & \dfrac{3}{2} & 2 & 3 & \dfrac{5}{2} & 2 & \dfrac{3}{2}
\end{pmatrix}.
\end{align}
Since $E_7$ is simply laced, the squared simple-root lengths form the vector 
\begin{align}
\vec{\alpha}
=
\begin{pmatrix}
2\\
2\\
\vdots\\
2\\
2
\end{pmatrix}.
\end{align}
Applying $C^{-1}$ to the squared-length vector $\vec{\alpha}$ gives the vector $\vec{\xi}$ entering the uncertainty bound
\begin{align}
    \vec{\xi}
    =
    C^{-1}\cdot\vec{\alpha}
    &=
    \begin{pmatrix}
    2 & 2 & 3 & 4 & 3 & 2 & 1\\[8pt]
    2 & \dfrac{7}{2} & 4 & 6 & \dfrac{9}{2} & 3 & \dfrac{3}{2}\\[8pt]
    3 & 4 & 6 & 8 & 6 & 4 & 2\\[8pt]
    4 & 6 & 8 & 12 & 9 & 6 & 3\\[6pt]
    3 & \dfrac{9}{2} & 6 & 9 & \dfrac{15}{2} & 5 & \dfrac{5}{2}\\[8pt]
    2 & 3 & 4 & 6 & 5 & 4 & 2\\[8pt]
    1 & \dfrac{3}{2} & 2 & 3 & \dfrac{5}{2} & 2 & \dfrac{3}{2}
    \end{pmatrix}
    \begin{pmatrix}
    2\\
    2\\
    \vdots\\
    2\\
    2
    \end{pmatrix}
    =
    \begin{pmatrix}
    34\\
    49\\
    66\\
    96\\
    75\\
    52\\
    27
    \end{pmatrix}.
\end{align}

For a representation with highest weight $\Lambda=\sum_{i=1}^{7}\lambda_i\omega_i$, the uncertainty lower bound $\langle\vec{\lambda},\vec{\xi}\rangle$ (see Eq.~\eqref{eq:Lie_Algebraic_UR_Simplified}) therefore becomes
\begin{align}
    \langle\vec{\lambda},\vec{\xi}\rangle
    &=
    \begin{pmatrix}
    \lambda_1 &
    \lambda_2 &
    \lambda_3 &
    \lambda_4 &
    \lambda_5 &
    \lambda_6 &
    \lambda_7
    \end{pmatrix}
    \begin{pmatrix}
    34\\
    49\\
    66\\
    96\\
    75\\
    52\\
    27
    \end{pmatrix}
    \nonumber\\
    &=
    34\lambda_1
    +
    49\lambda_2
    +
    66\lambda_3
    +
    96\lambda_4
    +
    75\lambda_5
    +
    52\lambda_6
    +
    27\lambda_7.
\end{align}
Thus, the root geometry of $E_7$ fixes the quantum uncertainty through the following explicit relation.

\begin{mycor}{Type-$\mathrm{E}_7$ Lie Algebraic Uncertainty Relations}{Type_E7_Lie_UR}
For the exceptional Lie algebra $E_7$ (see Tab.~\ref{tab:Simple_Lie_Algebra}), the optimal Lie algebraic uncertainty relation (see Def.~\ref{def:Lie_Algebraic_Uncertainty_Relations}) takes the explicit form
\begin{align}
   \Delta^2_{\psi}(\raisebox{-0.8ex}{\includegraphics[scale=0.8]{DK_E_7.pdf}})
   \coloneqq
   \Delta^2_{\psi}(\mathfrak e_7)
   \geqslant
   34\lambda_1
   +
   49\lambda_2
   +
   66\lambda_3
   +
   96\lambda_4
   +
   75\lambda_5
   +
   52\lambda_6
   +
   27\lambda_7 .
\end{align}
The lower bound is state-independent and depends only on the Dynkin coefficient vector
$\vec{\lambda}$
of the highest weight representation.
\end{mycor}

We finally turn to the exceptional Lie algebra $E_8$. 
Throughout, we adopt the Humphreys numbering convention, indicated by the numbered Dynkin diagram below.
\begin{equation}
\label{eq:Dynkin_E8}
    \raisebox{-3.8ex}{\includegraphics[scale=1]{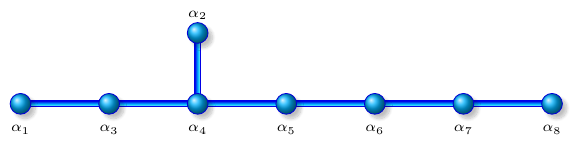}}\, .
\end{equation}

The $E_8$ root system is simply laced. 
Accordingly, with the standard normalization, all simple roots have squared length 
\begin{align}
\langle \alpha_1,\alpha_1\rangle
=
\langle \alpha_2,\alpha_2\rangle
=
\langle \alpha_3,\alpha_3\rangle
=
\langle \alpha_4,\alpha_4\rangle
=
\langle \alpha_5,\alpha_5\rangle
=
\langle \alpha_6,\alpha_6\rangle
=
\langle \alpha_7,\alpha_7\rangle
=
\langle \alpha_8,\alpha_8\rangle
=
2.
\end{align}
The off-diagonal structure is then completely determined by the Dynkin diagram: 
adjacent simple roots have inner product -1, while non-adjacent simple roots are orthogonal.
Thus the only nonzero off-diagonal inner products are
\begin{align}
    \langle\alpha_1,\alpha_3\rangle
    =
    \langle\alpha_2,\alpha_4\rangle
    =
    \langle\alpha_3,\alpha_4\rangle
    =
    \langle\alpha_4,\alpha_5\rangle
    =
    \langle\alpha_5,\alpha_6\rangle
    =
    \langle\alpha_6,\alpha_7\rangle
    =
    \langle\alpha_7,\alpha_8\rangle
    =
    -1,
\end{align}
with all other pairs of distinct simple roots having vanishing inner product.

With this ordering of the simple roots, the Cartan matrix is
\begin{align}
C
=
\begin{pmatrix}
2 & 0 & -1 & 0 & 0 & 0 & 0 & 0\\
0 & 2 & 0 & -1 & 0 & 0 & 0 & 0\\
-1 & 0 & 2 & -1 & 0 & 0 & 0 & 0\\
0 & -1 & -1 & 2 & -1 & 0 & 0 & 0\\
0 & 0 & 0 & -1 & 2 & -1 & 0 & 0\\
0 & 0 & 0 & 0 & -1 & 2 & -1 & 0\\
0 & 0 & 0 & 0 & 0 & -1 & 2 & -1\\
0 & 0 & 0 & 0 & 0 & 0 & -1 & 2
\end{pmatrix}.
\end{align}
whose inverse is
\begin{align}
C^{-1}
=
\begin{pmatrix}
4 & 5 & 7 & 10 & 8 & 6 & 4 & 2\\
5 & 8 & 10 & 15 & 12 & 9 & 6 & 3\\
7 & 10 & 14 & 20 & 16 & 12 & 8 & 4\\
10 & 15 & 20 & 30 & 24 & 18 & 12 & 6\\
8 & 12 & 16 & 24 & 20 & 15 & 10 & 5\\
6 & 9 & 12 & 18 & 15 & 12 & 8 & 4\\
4 & 6 & 8 & 12 & 10 & 8 & 6 & 3\\
2 & 3 & 4 & 6 & 5 & 4 & 3 & 2
\end{pmatrix}.
\end{align}
Because $E_8$ is simply laced, all simple roots have squared length 2. 
Thus the vector $\vec{\alpha}$ is
\begin{align}
    \vec{\alpha}
    =
    \begin{pmatrix}
    2\\
    2\\
    \vdots\\
    2\\
    2
    \end{pmatrix}.
\end{align}
The vector $\vec{\xi}$ that enters the expression for the uncertainty bound is therefore fixed by
\begin{align}
    \vec{\xi}
    =
    C^{-1}\cdot\vec{\alpha}
    &=
    \begin{pmatrix}
    4 & 5 & 7 & 10 & 8 & 6 & 4 & 2\\
    5 & 8 & 10 & 15 & 12 & 9 & 6 & 3\\
    7 & 10 & 14 & 20 & 16 & 12 & 8 & 4\\
    10 & 15 & 20 & 30 & 24 & 18 & 12 & 6\\
    8 & 12 & 16 & 24 & 20 & 15 & 10 & 5\\
    6 & 9 & 12 & 18 & 15 & 12 & 8 & 4\\
    4 & 6 & 8 & 12 & 10 & 8 & 6 & 3\\
    2 & 3 & 4 & 6 & 5 & 4 & 3 & 2
    \end{pmatrix}
    \begin{pmatrix}
    2\\
    2\\
    \vdots\\
    2\\
    2
    \end{pmatrix}
     =
    \begin{pmatrix}
    92\\
    136\\
    182\\
    270\\
    220\\
    168\\
    114\\
    58
    \end{pmatrix}.
\end{align}

A direct calculation then yields the optimal uncertainty bound for $E_8$; that is
\begin{align}
    \langle\vec{\lambda},\vec{\xi}\rangle
    &=
    \begin{pmatrix}
    \lambda_1 &
    \lambda_2 &
    \lambda_3 &
    \lambda_4 &
    \lambda_5 &
    \lambda_6 &
    \lambda_7 &
    \lambda_8
    \end{pmatrix}
    \begin{pmatrix}
    92\\
    136\\
    182\\
    270\\
    220\\
    168\\
    114\\
    58
    \end{pmatrix}
    \nonumber\\
    &=
    92\lambda_1
    +
    136\lambda_2
    +
    182\lambda_3
    +
    270\lambda_4
    +
    220\lambda_5
    +
    168\lambda_6
    +
    114\lambda_7
    +
    58\lambda_8 .
\end{align}
We state this result formally in the following corollary.

\begin{mycor}{Type-$\mathrm{E}_8$ Lie Algebraic Uncertainty Relations}{Type_E8_Lie_UR}
For the exceptional Lie algebra $E_8$ (see Tab.~\ref{tab:Simple_Lie_Algebra}), the optimal Lie algebraic uncertainty relation (see Def.~\ref{def:Lie_Algebraic_Uncertainty_Relations}) takes the explicit form
\begin{align}
  & \Delta^2_{\psi}( \raisebox{-0.8ex}{\includegraphics[scale=0.8]{DK_E_8.pdf}} )\nonumber\\
   \coloneqq
  & \Delta^2_{\psi}(\mathfrak e_8)
   \geqslant
   92\lambda_1
   +
   136\lambda_2
   +
   182\lambda_3
   +
   270\lambda_4
   +
   220\lambda_5
   +
   168\lambda_6
   +
   114\lambda_7
   +
   58\lambda_8 .
\end{align}
The bound is state-independent and depends only on the Dynkin coefficient vector $\vec{\lambda}$ specifying the highest weight representation.
\end{mycor}

The resulting uncertainty relations for the exceptional Lie algebras are collected in Tab.~\ref{tab:exceptional_bounds}. 
For each exceptional type, the table gives the corresponding rank and the explicit optimal lower bound $\langle\vec{\lambda},\vec{\xi}\rangle$ (see Eq.~\eqref{eq:Lie_Algebraic_UR_Simplified}), making clear how the state-independent uncertainty bound is fixed by the Dynkin coefficient vector $\vec{\lambda}$ of the chosen highest weight representation.

\begin{table}[h]
\centering
\begin{tblr}{
      width = \textwidth,
      colspec = {c||c|X[c]},
      row{1} = {bg=gray!50, font=\bfseries},
      row{2} = {bg=gray!20, font=\bfseries},
      hlines,
      vlines,
      cells = {m},
      cell{3-7}{1} = {bg=gray!20},
      cell{3-7}{2} = {bg=mRed},
      cell{3-7}{3} = {bg=mGreen},
}

\SetCell[c=3]{c}
\textbf{Exceptional Lie Algebraic Uncertainty Relations}
&
&
\\

Exceptional Lie Algebra
&
Rank
&
Optimal Bound
\\

$G_2$
&
$2$
&
$\displaystyle
6\lambda_1+\frac{10}{3}\lambda_2$
\\

$F_4$
&
$4$
&
$\displaystyle
16\lambda_1+30\lambda_2+21\lambda_3+11\lambda_4$
\\

$E_6$
&
$6$
&
$\displaystyle
16\lambda_1+22\lambda_2+30\lambda_3
+42\lambda_4+30\lambda_5+16\lambda_6$
\\

$E_7$
&
$7$
&
$\displaystyle
34\lambda_1+49\lambda_2+66\lambda_3
+96\lambda_4+75\lambda_5+52\lambda_6+27\lambda_7$
\\

$E_8$
&
$8$
&
$\displaystyle
92\lambda_1+136\lambda_2+182\lambda_3+270\lambda_4
+220\lambda_5+168\lambda_6+114\lambda_7+58\lambda_8$
\\

\end{tblr}

\caption{
\textbf{Exceptional Lie Algebraic Uncertainty Relations.}
For each exceptional compact Lie algebra, the table reports the explicit state-independent optimal lower bound in terms of the Dynkin coefficient vector $\vec{\lambda}$.}
\label{tab:exceptional_bounds}
\end{table}

We next give representative examples of standard representations and their associated optimal uncertainty lower bounds. 
These examples illustrate how the general formula specializes to concrete highest weight sectors and how the resulting state-independent bounds are fixed by the corresponding Dynkin coefficient vectors.

\begin{table}[htbp]
\centering
\begin{tblr}{
      width = \textwidth,
      colspec = {Q[c,m]||Q[c,m]|Q[c,m]|X[c,m]|X[c,m]},
      row{1} = {bg=gray!50,font=\bfseries},
      row{2} = {bg=gray!20,font=\bfseries},
      hlines,
      vlines,
      cells = {m,c},
      cell{3}{2-5} = {bg=mRed},
      cell{4}{2-5} = {bg=mGreen},
      cell{5}{2-5} = {bg=mRed},
      cell{6}{2-5} = {bg=mGreen},
      cell{7-9}{2-5} = {bg=mRed},
      cell{8}{2-5} = {bg=mBlue},
      cell{9}{2-5} = {bg=mGreen},
      cell{10}{2-5} = {bg=mRed},
      cell{11}{2-5} = {bg=mGreen},
      cell{12}{2-5} = {bg=mOrange},
      cell{13}{2-5} = {bg=mGreen},
}
      
\SetCell[c=5]{m,c}
\textbf{Exceptional Lie Algebraic Uncertainty Relations Across Representations}
&
&
&
&
\\

Exceptional Algebra
&
Dynkin Vector $\vec{\lambda}^{\,\T}$ \& Rep.
&
Dimension
&
Highest Weight $\Lambda$
&
Optimal Bound
\\

\SetCell[r=2]{m,c}
$G_2$
 \raisebox{-1.1ex}{\includegraphics[scale=0.75]{DK_G_2.pdf}} 
&
$\begin{array}{c}
(0,1)\\[-1mm]
\scriptstyle \text{Defining / Minimal}
\end{array}$
&
$7$
&
$\omega_2$
&
$\dfrac{10}{3}$
\\

&
$\begin{array}{c}
(1,0)\\[-1mm]
\scriptstyle \text{Adjoint}
\end{array}$
&
$14$
&
$\omega_1$
&
$6$
\\

\SetCell[r=2]{m,c}
$F_4$
\raisebox{-1.2ex}{\includegraphics[scale=0.75]{DK_F_4.pdf}}
&
$\begin{array}{c}
(0,0,0,1)\\[-1mm]
\scriptstyle \text{Defining / Minimal}
\end{array}$
&
$26$
&
$\omega_4$
&
$11$
\\

&
$\begin{array}{c}
(1,0,0,0)\\[-1mm]
\scriptstyle \text{Adjoint}
\end{array}$
&
$52$
&
$\omega_1$
&
$16$
\\

\SetCell[r=3]{m,c}
$E_6$
\raisebox{-0.8ex}{\includegraphics[scale=0.75]{DK_E_6.pdf}}
&
$\begin{array}{c}
(1,0,0,0,0,0)\\[-1mm]
\scriptstyle \text{Defining / Minimal}
\end{array}$
&
$27$
&
$\omega_1$
&
$16$
\\
&
$\begin{array}{c}
(0,0,0,0,0,1)\\[-1mm]
\scriptstyle \text{Dual Defining / Dual Minimal}
\end{array}$
&
$27$
&
$\omega_6$
&
$16$
\\
&
$\begin{array}{c}
(0,1,0,0,0,0)\\[-1mm]
\scriptstyle \text{Adjoint}
\end{array}$
&
$78$
&
$\omega_2$
&
$22$
\\

\SetCell[r=2]{m,c}
$E_7$
\raisebox{-0.8ex}{\includegraphics[scale=0.7]{DK_E_7.pdf}}
&
$\begin{array}{c}
(0,0,0,0,0,0,1)\\[-1mm]
\scriptstyle \text{Defining / Minimal}
\end{array}$
&
$56$
&
$\omega_7$
&
$27$
\\

&
$\begin{array}{c}
(1,0,0,0,0,0,0)\\[-1mm]
\scriptstyle \text{Adjoint}
\end{array}$
&
$133$
&
$\omega_1$
&
$34$
\\
\SetCell[r=2]{m,c}
$E_8$
\raisebox{-0.8ex}{\includegraphics[scale=0.7]{DK_E_8.pdf}}
&
$\begin{array}{c}
(1,0,0,0,0,0,0,0)\\[-1mm]
\scriptstyle \text{Fundamental (first)}
\end{array}$
&
$3875$
&
$\omega_1$
&
$92$
\\

&
$\begin{array}{c}
(0,0,0,0,0,0,0,1)\\[-1mm]
\scriptstyle \text{Adjoint / Minimal}
\end{array}$
&
$248$
&
$\omega_8$
&
$58$
\\

\end{tblr}

\caption{
\textbf{Lie Algebraic Uncertainty Relations for Exceptional Cases.}
Representative irreducible representations of the exceptional compact simple Lie algebras. 
For each algebra, we list the minimal nontrivial representation, and the adjoint representation, merging entries whenever these representations coincide. 
For $G_2$, we use the convention that $\alpha_1$ is the long simple root and $\alpha_2$ is the short simple root; 
hence the $7$-dimensional defining/minimal representation has highest weight $\omega_2$, whereas the $14$-dimensional adjoint representation has highest weight $\omega_1$.
For $E_6$, the two $27$-dimensional representations with highest weights $\omega_1$ and $\omega_6$ are dual defining/minimal representations, while the $78$-dimensional adjoint representation has highest weight $\omega_2$.
For $E_8$, unlike the classical Lie algebras, there is no separate natural
low-dimensional defining representation. 
The $248$-dimensional adjoint representation
with highest weight $\omega_8$ is already the smallest nontrivial irreducible representation. 
We also list the first fundamental representation with highest
weight $\omega_1$, which has dimension $3875$.
}
\label{tab:exceptional_representations}
\end{table}

Before closing this subsection, we illustrate the general result through explicit examples, summarized in Tab.~\ref{tab:exceptional_representations}. 
As representative cases, we consider the adjoint representations of the exceptional compact Lie algebras. 
Their optimal state-independent uncertainty relations take the form

\begin{mythm}{Exceptional Lie Algebraic Uncertainty Relations for Adjoint Representations}{Exceptional_ad}
For the adjoint representations of exceptional Lie algebras, the optimal state-independent uncertainty bounds are given explicitly below.
\begin{align}
\underbrace{\Delta^2_{\mathrm{ad}}(\mathfrak{g}_2)}_{\text{Type $\mathrm{G}_2$}}
\geqslant6,\quad
\underbrace{\Delta^2_{\mathrm{ad}}(\mathfrak{f}_4)}_{\text{Type $\mathrm{F}_4$}}
\geqslant16,\quad
\underbrace{\Delta^2_{\mathrm{ad}}(\mathfrak{e}_6)}_{\text{Type $\mathrm{E}_6$}}
\geqslant22,\quad
\underbrace{\Delta^2_{\mathrm{ad}}(\mathfrak{e}_7)}_{\text{Type $\mathrm{E}_7$}}
\geqslant34,\quad
\underbrace{\Delta^2_{\mathrm{ad}}(\mathfrak{e}_8)}_{\text{Type $\mathrm{E}_8$}}
\geqslant58.
\end{align}
\end{mythm}

Expressed in the diagrammatic language of root systems, encoded by Dynkin diagrams, Thm.~\ref{thm:Exceptional_ad} leads to the following uncertainty relations.
\begin{align}
    &\textbf{Type $\mathrm{G}_2$}:\qquad
    \Delta^2_{\mathrm{ad}}( \raisebox{-1.4ex}{\includegraphics[scale=1]{DK_G_2.pdf}} )
    \geqslant6,\label{eq:adj_G_2}\\
    &\textbf{Type $\mathrm{F}_4$}:\qquad
    \Delta^2_{\mathrm{ad}}( \raisebox{-1.4ex}{\includegraphics[scale=0.8]{DK_F_4.pdf}})
    \geqslant16,\label{eq:adj_F_4}\\
    &\textbf{Type $\mathrm{E}_6$}:\qquad
    \Delta^2_{\mathrm{ad}}( \raisebox{-0.8ex}{\includegraphics[scale=0.8]{DK_E_6.pdf}})
    \geqslant22,\label{eq:adj_E_6}\\
    &\textbf{Type $\mathrm{E}_7$}:\qquad
    \Delta^2_{\mathrm{ad}}( \raisebox{-0.8ex}{\includegraphics[scale=0.8]{DK_E_7.pdf}})
    \geqslant34,\label{eq:adj_E_7}\\
    &\textbf{Type $\mathrm{E}_8$}:\qquad
    \Delta^2_{\mathrm{ad}}( \raisebox{-0.8ex}{\includegraphics[scale=0.8]{DK_E_8.pdf}})
    \geqslant58.\label{eq:adj_E_8}
\end{align}
Here the subscript $\mathrm{ad}$ indicates that the corresponding representation is the adjoint representation. 
The resulting lower bounds are optimal: 
each is attained by a highest weight state and therefore cannot be improved within the corresponding representation.

For exceptional Lie algebras (see Tab.~\ref{tab:Simple_Lie_Algebra}), the minimal representations give rise to the following optimal state-independent uncertainty bounds.

\begin{mythm}{Exceptional Lie Algebraic Uncertainty Relations for Minimal Representations}{Exceptional_min}
For exceptional Lie algebras, the minimal representations satisfy the following optimal state-independent uncertainty bounds.
\begin{align}
\underbrace{\Delta^2_{\mathrm{min}}(\mathfrak{g}_2)}_{\text{Type $\mathrm{G}_2$}}
\geqslant\frac{10}{3},\quad
\underbrace{\Delta^2_{\mathrm{min}}(\mathfrak{f}_4)}_{\text{Type $\mathrm{F}_4$}}
\geqslant11,\quad
\underbrace{\Delta^2_{\mathrm{min}}(\mathfrak{e}_6)}_{\text{Type $\mathrm{E}_6$}}
\geqslant16,\quad
\underbrace{\Delta^2_{\mathrm{min}}(\mathfrak{e}_7)}_{\text{Type $\mathrm{E}_7$}}
\geqslant27,\quad
\underbrace{\Delta^2_{\mathrm{min}}(\mathfrak{e}_8)}_{\text{Type $\mathrm{E}_8$}}
\geqslant58.
\end{align}
\end{mythm}

In Dynkin diagram form, Thm.~\ref{thm:Exceptional_min} gives the following uncertainty relations.
\begin{align}
    &\textbf{Type $\mathrm{G}_2$}:\qquad
    \Delta^2_{\mathrm{min}}( \raisebox{-1.4ex}{\includegraphics[scale=1]{DK_G_2.pdf}} )
    \geqslant\frac{10}{3},\label{eq:min_G_2}\\
    &\textbf{Type $\mathrm{F}_4$}:\qquad
    \Delta^2_{\mathrm{min}}( \raisebox{-1.4ex}{\includegraphics[scale=0.8]{DK_F_4.pdf}})
    \geqslant11,\label{eq:min_F_4}\\
    &\textbf{Type $\mathrm{E}_6$}:\qquad
    \Delta^2_{\mathrm{min}}( \raisebox{-0.6ex}{\includegraphics[scale=0.75]{DK_E_6.pdf}})
    \geqslant16,\label{eq:min_E_6}\\
    &\textbf{Type $\mathrm{E}_7$}:\qquad
    \Delta^2_{\mathrm{min}}( \raisebox{-0.6ex}{\includegraphics[scale=0.75]{DK_E_7.pdf}})
    \geqslant27,\label{eq:min_E_7}\\
    &\textbf{Type $\mathrm{E}_8$}:\qquad
    \Delta^2_{\mathrm{min}}( \raisebox{-0.6ex}{\includegraphics[scale=0.75]{DK_E_8.pdf}})
    \geqslant58.\label{eq:min_E_8}
\end{align}
The subscript $\mathrm{min}$ indicates that the representation under consideration is the minimal representation. 
All the lower bounds displayed above are optimal, since each is attained by a highest-weight state and hence cannot be improved within the corresponding representation.

So far we have derived optimal Lie algebraic uncertainty relations across a range of representations. 
In particular, we have analyzed in detail the adjoint (see Eqs.~\eqref{eq:adj_A_n} ($\mathrm{A}_n$),~\eqref{eq:adj_B_n} ($\mathrm{B}_n$),~\eqref{eq:adj_C_n} ($\mathrm{C}_n$),~\eqref{eq:adj_D_n} ($\mathrm{D}_n$),~\eqref{eq:adj_G_2} ($\mathrm{G}_2$),~\eqref{eq:adj_F_4} ($\mathrm{F}_4$),~\eqref{eq:adj_E_6} ($\mathrm{E}_6$),~\eqref{eq:adj_E_7} ($\mathrm{E}_7$),~\eqref{eq:adj_E_8} ($\mathrm{E}_8$)), defining (see Eqs.~\eqref{eq:def_A_n} ($\mathrm{A}_n$),~\eqref{eq:def_B_n} ($\mathrm{B}_n$),~\eqref{eq:def_C_n} ($\mathrm{C}_n$),~\eqref{eq:def_D_n} ($\mathrm{D}_n$)), and minimal (see Eqs.~\eqref{eq:min_G_2} ($\mathrm{G}_2$),~\eqref{eq:min_F_4} ($\mathrm{F}_4$),~\eqref{eq:min_E_6} ($\mathrm{E}_6$),~\eqref{eq:min_E_7} ($\mathrm{E}_7$),~\eqref{eq:min_E_8} ($\mathrm{E}_8$)) representations. 
For completeness, and to provide a convenient reference for the reader, we collect in the following Tab.~\ref{tab:LAU_adj} the optimal uncertainty relations for the adjoint representation of every simple Lie algebra.

\begin{table}[h]
\centering
\begin{tblr}{
      width = \textwidth,
      colspec = {c||X[c]},
      row{1} = {bg=gray!50, font=\bfseries},
      row{2} = {bg=gray!20, font=\bfseries},
      hlines,
      vlines,
      cells = {m},
      cell{3-11}{1} = {bg=gray!20},
      cell{3-6}{2} = {bg=mGreen},
      cell{7-11}{2} = {bg=mRed},
}

\SetCell[c=2]{c} \textbf{Optimal Lie Algebraic Uncertainty Relations for Adjoint Representations}
& \\

Lie Algebra
&
Uncertainty Relation
\\

\textbf{Type $\mathrm{A}_n$} 
(see Eq.~\eqref{eq:adj_A_n})
&
$\Delta^2_{\mathrm{ad}}(\raisebox{-0.8ex}{\includegraphics[scale=0.8]{DK_A_n.pdf}})\geqslant2n$
\\

\textbf{Type $\mathrm{B}_n$}
(see Eq.~\eqref{eq:adj_B_n})
&
$\Delta^2_{\mathrm{ad}}( \raisebox{-1.5ex}{\includegraphics[scale=0.8]{DK_B_n.pdf}} )\geqslant4(n-1)$
\\

\textbf{Type $\mathrm{C}_n$}
(see Eq.~\eqref{eq:adj_C_n})
&
$\Delta^2_{\mathrm{ad}}( \raisebox{-1.4ex}{\includegraphics[scale=0.8]{DK_C_n.pdf}} )\geqslant2n$
\\

\textbf{Type $\mathrm{D}_n$}
(see Eq.~\eqref{eq:adj_D_n})
&
$\Delta^2_{\mathrm{ad}}( \raisebox{-5.8ex}{\includegraphics[scale=0.8]{DK_D_n.pdf}} )\geqslant 4n-6$
\\

\textbf{Type $\mathrm{G}_2$}
(see Eq.~\eqref{eq:adj_G_2})
&
$\Delta^2_{\mathrm{ad}}( \raisebox{-1.8ex}{\includegraphics[scale=1]{DK_G_2.pdf}} )\geqslant6$
\\

\textbf{Type $\mathrm{F}_4$}
(see Eq.~\eqref{eq:adj_F_4})
&
$\Delta^2_{\mathrm{ad}}( \raisebox{-1.6ex}{\includegraphics[scale=0.8]{DK_F_4.pdf}})\geqslant16$
\\

\textbf{Type $\mathrm{E}_6$}
(see Eq.~\eqref{eq:adj_E_6})
&
$\Delta^2_{\mathrm{ad}}( \raisebox{-0.8ex}{\includegraphics[scale=0.8]{DK_E_6.pdf}})\geqslant22$
\\

\textbf{Type $\mathrm{E}_7$}
(see Eq.~\eqref{eq:adj_E_7})
&
$\Delta^2_{\mathrm{ad}}( \raisebox{-0.8ex}{\includegraphics[scale=0.8]{DK_E_7.pdf}})\geqslant34$
\\

\textbf{Type $\mathrm{E}_8$}
(see Eq.~\eqref{eq:adj_E_8})
&
$\Delta^2_{\mathrm{ad}}( \raisebox{-0.8ex}{\includegraphics[scale=0.8]{DK_E_8.pdf}})\geqslant58$
\\
\end{tblr}
\caption{
\textbf{Optimal Adjoint Representation Uncertainty Relations for Simple Lie Algebras.}
The Dynkin diagram in each row identifies the corresponding root system, while the right-hand side gives the exact optimal lower bound for the adjoint representation. 
Here $\Delta^2_{\mathrm{ad}}$ denotes the Lie algebraic quantum uncertainty in the adjoint representation. 
All bounds are tight and are attained by highest weight states.
The classical ranges are $A_n (n\geqslant 1), B_n (n\geqslant 2), C_n (n\geqslant 1)$, and $D_n (n\geqslant 4)$; the rank-one $B_1$ case is covered by $A_1$.} 
\label{tab:LAU_adj}
\end{table} 

We have now completed the analysis of Lie algebraic uncertainty relations for all simple Lie algebras, covering both the four classical series ($\mathrm{A}_n$, $\mathrm{B}_n$, $\mathrm{C}_n$, and $\mathrm{D}_n$) and the five exceptional cases ($\mathrm{G}_2$, $\mathrm{F}_4$, $\mathrm{E}_6$, $\mathrm{E}_7$, and $\mathrm{E}_8$). 
Across these families, the general framework reduces the derivation of state-independent optimal bounds to intrinsic representation-theoretic and root system structure, thereby placing the classical and exceptional cases on the same footing. 
The resulting collection of bounds provides a complete structural picture of optimal uncertainty relations for simple Lie algebras and illustrates how representation theory and quantum uncertainty are tied together through the underlying Dynkin diagram.


\section{Applications of Lie Algebraic Uncertainty Relations}
\label{sec:entanglement}

In Sec.~\ref{subsec:semisimple_lie_algebraic_uncertainty},
we make explicit the factorwise structure of the semisimple uncertainty bound established in
Sec.~\ref{sec:Lie_Algebraic_Quantum_Uncertainty}.
Each simple factor contribution is evaluated on the corresponding state, with its lower bound given by the formulas in
Sec.~\ref{sec:Simple_Lie_Algebraic_Quantum_Uncertainty}.
We then show that the same structure leads to a multipartite entanglement criterion based on local uncertainty violations; see
Sec.~\ref{subsec:Entanglement_Local_Uncertainty_Violations}.


\subsection{Semisimple Lie Algebraic Uncertainty}
\label{subsec:semisimple_lie_algebraic_uncertainty}

We now extend the uncertainty bounds derived for simple Lie algebras to the
semisimple algebras.
Let
\begin{align}\label{eq:Semi_Simple}
\mathfrak g
=
\mathfrak g_1\oplus\cdots\oplus\mathfrak g_s
\end{align}
be a finite-dimensional complex semisimple Lie algebra decomposed into $s$ simple ideals.
Its compact real form decomposes accordingly as
\begin{align}\label{eq:semisimple_compact_real_form}
\mathfrak g_{\mathbb R}
=
\mathfrak g_{1,\mathbb R}
\oplus\cdots\oplus
\mathfrak g_{s,\mathbb R}.
\end{align}

In Eq.~\eqref{eq:semisimple_inner_product_direct_sum}, each factor carries the normalized compact inner product fixed in Eq.~\eqref{eq:normal_Killing}, so that long roots have squared length $2$.
We equip $\g_{\mathbb R}$ with the orthogonal direct sum of the positive invariant inner products on its simple factors. Thus, for
$X=X_1+\cdots+X_s, Y=Y_1+\cdots+Y_s$,
where $X_k,Y_k\in\g_{k,\mathbb R}$, we set
\begin{align}\label{eq:semisimple_inner_product_direct_sum}
(X,Y)_{\g_{\mathbb R}}
\coloneqq
\sum_{k=1}^{s}
(X_k,Y_k)_{\g_{k,\mathbb R}}.
\end{align}
Consequently,
\begin{align}
(X_j,Y_k)_{\g_{\mathbb R}}
=
0,
\end{align}
whenever $j\neq k$. Here,
$X_j\in\g_{j,\mathbb R}$ and
$Y_k\in\g_{k,\mathbb R}$.
Hence distinct simple ideals are mutually orthogonal.
Let $V(\Lambda)$ be a finite-dimensional irreducible representation of $\g$ (see Eq.~\eqref{eq:Semi_Simple}). 
The direct-sum decomposition of $\g$ induces the
tensor product decomposition
\begin{align}\label{eq:semisimple_irrep_tensor_product_uncertainty}
V(\Lambda)
\cong
V_1(\Lambda_1)\otimes\cdots\otimes V_s(\Lambda_s),
\end{align}
where $V_k(\Lambda_k)$ is the irreducible highest-weight representation of
$\mathfrak g_k$ with highest weight $\Lambda_k$. Accordingly,
\begin{align}
\Lambda
=
\Lambda_1\oplus\cdots\oplus\Lambda_s.
\end{align}
If $\delta_k$ denotes the Weyl vector of $\mathfrak g_k$, then
\begin{align}
\delta
=
\delta_1\oplus\cdots\oplus\delta_s.
\end{align}
The induced inner product on the weight space of $\mathfrak{g}$ is additive:
\begin{align}\label{eq:semisimple_inner_product_additive_uncertainty}
\langle\Lambda,\delta\rangle
=
\sum_{k=1}^{s}
\langle\Lambda_k,\delta_k\rangle_k,
\end{align}
where $\langle\cdot,\cdot\rangle_k$ denotes the invariant inner product on the weight space of $\mathfrak g_k$.
For each $\mathfrak g_k$, let
$\{Q_{k,a}\}_{a=1}^{\dim\mathfrak g_k}$ be the normalized anti-Hermitian basis of $\mathfrak g_{k,\mathbb R}$ constructed in
Def.~\ref{def:Anti_Hermitian_Basis}, so that
\begin{align}\label{eq:semisimple_factor_generator_normalization}
(Q_{k,a},Q_{k,b})_{\mathfrak g_{k,\mathbb R}}
=
2\delta_{ab}.
\end{align}
In particular, the Cartan generators contained in this basis are the
orthogonalized anti-Hermitian Cartan generators of $\mathfrak g_{k,\mathbb R}$.
Under the representation Eq.~\eqref{eq:semisimple_irrep_tensor_product_uncertainty},
the embedded generator $\widetilde Q_{k,a}$ is represented by
\begin{align}\label{eq:embedded_simple_factor_generator}
\widetilde Q_{k,a}
\coloneqq
\1_1\otimes\cdots\otimes \1_{k-1}
\otimes Q_{k,a}
\otimes \1_{k+1}\otimes\cdots\otimes \1_s,
\end{align}
where $\1_j$ denotes the identity operator on $V_j(\Lambda_j)$.
Thus $\widetilde Q_{k,a}$ acts nontrivially only on $V_k(\Lambda_k)$.

The direct-sum structure of $\mathfrak g_{\mathbb R}$ separates its generators into mutually orthogonal sectors, one for each simple ideal.
Moreover, each sector acts nontrivially only on the corresponding factor of the tensor product representation. 
Accordingly, the total uncertainty of Def.~\ref{def:Lie_Algebraic_Uncertainty_Relations} has the following factorwise form. 
For anti-Hermitian generators, we have
\begin{align}
    \Delta^2_{\rho}(\widetilde{Q}_{k,a})=\Delta^2_{\rho}(\mathrm{i}\widetilde{Q}_{k,a}).
\end{align}

For a simple Lie algebra, we define the quantum uncertainty by the following expression, consistent with the definition in the preceding sections:

\begin{mydef}{Simple Lie Algebraic Uncertainty Relations}{Simple_Lie_UR}
For a quantum state $\rho$ on $V_k(\Lambda_k)$, we define the quantum uncertainty associated with the semisimple Lie algebra $\g_k$ by (see Def.~\ref{def:Lie_Algebraic_Uncertainty_Relations})
\begin{align}\label{eq:simple_factor_uncertainty_contribution}
\Delta_\rho^2(\mathfrak g_{k,\mathbb R})
\coloneqq
\frac12
\sum_{a=1}^{\dim\mathfrak g_k}
\Delta_\rho^2\!\left(\widetilde Q_{k,a}\right),
\end{align}
where $\widetilde Q_{k,a}$ is the embedded generator defined in Eq.~\eqref{eq:embedded_simple_factor_generator}.
\end{mydef}

The quantum uncertainties associated with the simple components then define the uncertainty of the semisimple Lie algebra $\g$ (see Eq.~\eqref{eq:Semi_Simple}).

\begin{mydef}{Semisimple Lie Algebraic Uncertainty Relations}{Semisimple_Lie_UR}
For a semisimple Lie algebra $\g$ with the decomposition in Eq.~\eqref{eq:Semi_Simple}, we define its quantum uncertainty as the sum of the uncertainties associated with its simple components
\begin{align}\label{eq:semisimple_uncertainty_definition}
    \Delta_\rho^2(\mathfrak g_{\mathbb R})
    \coloneqq
    \frac12
    \sum_{k=1}^{s}\sum_{a=1}^{\dim\mathfrak g_k}
    \Delta_\rho^2\!\left(\widetilde Q_{k,a}\right)
    =
    \sum_{k=1}^{s}
    \Delta_\rho^2(\mathfrak g_{k,\mathbb R}),
\end{align}
where each component uncertainty is defined in Def.~\ref{def:Simple_Lie_UR}.
\end{mydef}

Equation~\eqref{eq:semisimple_uncertainty_definition} reduces the semisimple uncertainty relation to the corresponding uncertainty optimization on each simple part. 
Importantly, no product-state assumption is required. 
Even when $\rho$ is entangled across the representation factors, the contribution of $\mathfrak g_{k,\mathbb R}$ depends only on the reduced state on $V_k(\Lambda_k)$. 
The simple-factor bounds may therefore be applied independently.

\begin{mythm}{Semisimple Lie Algebraic Uncertainty Relations}
{semisimple_LAU_bound}
Let
$\mathfrak g
=
\mathfrak g_1\oplus\cdots\oplus\mathfrak g_s$
be a complex semisimple Lie algebra, and let
$V(\Lambda)
\cong
V_1(\Lambda_1)\otimes\cdots\otimes V_s(\Lambda_s)$
be an irreducible representation of $\mathfrak g$, with highest weight
$\Lambda
=
\Lambda_1\oplus\cdots\oplus\Lambda_s$.
Suppose that, for each $k=1,\ldots,s$, every quantum state $\sigma_k$ on
$V_k(\Lambda_k)$ satisfies the simple-factor uncertainty bound
\begin{align}\label{eq:simple_factor_assumed_uncertainty_bound}
\Delta_{\sigma_k}^2(\mathfrak g_{k,\mathbb R})
\geqslant
\mB_k(\Lambda_k).
\end{align}
Then every quantum state $\rho$ on $V(\Lambda)$ satisfies
\begin{align}\label{eq:semisimple_LAU_bound_componentwise}
\Delta_\rho^2(\mathfrak g_{\mathbb R})
\geqslant
\sum_{k=1}^{s}
\mB_k(\Lambda_k).
\end{align}

In particular, applying the simple Lie algebraic uncertainty bound $\mB_k(\Lambda_k)=2\langle\Lambda_k,\delta_k\rangle_k$ (see Thm.~\ref{thm:Lie_Algebraic_Uncertainty_Relations}),
where $\delta_k$ is the Weyl vector of $\mathfrak g_k$, gives $\Delta_\rho^2(\mathfrak g_{\mathbb R})\geqslant2\sum_{k=1}^{s}\langle\Lambda_k,\delta_k\rangle_k=2\langle\Lambda,\delta\rangle$, where $\delta=\delta_1\oplus\cdots\oplus\delta_s$.
\end{mythm}

\begin{proof}
For each $k$, let 
$\rho_k
\coloneqq
\Tr_{\widehat{k}} \left(\rho\right)$
be the reduced state on $V_k(\Lambda_k)$. 
Since $\widetilde Q_{k,a}$ acts nontrivially only on the $k$-th tensor factor,
\begin{align}
\Delta_\rho^2\!\left(\widetilde Q_{k,a}\right)
=
\Delta_{\rho_k}^2(Q_{k,a}).
\end{align}
Consequently,
\begin{align}
\Delta_\rho^2(\mathfrak g_{k,\mathbb R})
=
\Delta_{\rho_k}^2(\mathfrak g_{k,\mathbb R})
\geqslant
\mB_k(\Lambda_k).
\end{align}
Summing over $k$ 
proves 
Eq.~\eqref{eq:semisimple_LAU_bound_componentwise}.
Variance concavity extends 
$\mB_k(\Lambda_k)
=
2\langle\Lambda_k,\delta_k\rangle_k$ to mixed marginals. 
The additive weight metric (see Eq.~\eqref{eq:Inner_root}) gives $2\sum_{k=1}^{s}\langle\Lambda_k,\delta_k\rangle_k$, attained by a product of highest weight states.
\end{proof}

Theorem~\ref{thm:semisimple_LAU_bound} gives the lower bound in an intrinsic highest-weight form. 
For explicit calculations, each simple-factor contribution may be read directly from the Cartan matrix (see Def.~\ref{def:Cartan_Matrix}) and Dynkin coefficients (see Def.~\ref{def:Dynkin_Coefficients}) formulas established for simple Lie algebras in Eq.~\eqref{eq:Lie_Algebraic_UR_Simplified}.

\begin{mycor}{Simplified Semisimple Lie Algebraic Uncertainty Relations}
{semisimple_Dynkin_coordinate_bound}
Let $\g=\bigoplus_{k=1}^{s}\mathfrak g_k$ be a complex semisimple Lie algebra with compact real form $\g_{\mathbb R}$, and let $V(\Lambda)$ be an irreducible module with highest weight $\Lambda=\bigoplus_{k=1}^{s}\Lambda_k$.
For each simple factor $\mathfrak g_k$, let $C^{(k)}$ denote its Cartan matrix, $\vec{\lambda}^{(k)}$ the Dynkin-label vector of $\Lambda_k$, and $\vec{\xi}^{(k)}$ the corresponding coefficient vector determined by $C^{(k)}$, as tabulated in Tab.~\ref{tab:bounds} and~\ref{tab:exceptional_representations}.
Then Cor.~\ref{cor:Simplified_Lie_Algebraic_Uncertainty_Relations} gives $\mB_k(\Lambda_k)=(\vec{\lambda}^{(k)})^{\T}\vec{\xi}^{(k)}$, and every state $\rho$ on $V(\Lambda)$ satisfies
\begin{align}\label{eq:semisimple_Dynkin_coordinate_bound}
\Delta_\rho^2(\mathfrak g_{\mathbb R})
\geqslant
\sum_{k=1}^{s}
\bigl(\vec{\lambda}^{(k)}\bigr)^{\T}
\vec{\xi}^{(k)}.
\end{align}
Defining the vectors, $\vec{\lambda}\coloneqq\vec{\lambda}^{(1)}\oplus\cdots\oplus\vec{\lambda}^{(s)}$ and $\vec{\xi}\coloneqq\vec{\xi}^{(1)}\oplus\cdots\oplus\vec{\xi}^{(s)}$, we can equivalently write
\begin{align}\label{eq:semisimple_bound_lambda_xi}
\Delta_\rho^2(\mathfrak g_{\mathbb R})
\geqslant
\vec{\lambda}^{\,\T}\vec{\xi}  .
\end{align}
\end{mycor}

The uncertainty bound for a semisimple Lie algebra is obtained by adding the contributions of its simple factors, listed in Tabs.~\ref{tab:bounds}
and~\ref{tab:exceptional_representations}.
To illustrate this construction, consider compact semisimple algebras with two simple factors of type-A. 
Let $V_k$ and $W_k$ carry the defining and adjoint representations of $\mathfrak{su}(k+1)$, respectively. 
All states below are assumed to be normalized.

In the defining representation of each factor, every state $\psi\in V_1\otimes V_2$ satisfies
\begin{align}
 \Delta_\psi^2
\left(\mathfrak{su}(2)\oplus\mathfrak{su}(3)\right)
&\geqslant
\underbrace{\min_{\substack{\phi_1\in V_1 }}\Delta^2_{\mathrm{def}}( \raisebox{-0.6ex}{\includegraphics[scale=0.8]{DK_B_1.pdf}} )}_{\text{Type $\mathrm{A}_1$}}
+
\underbrace{\min_{\substack{\phi_1\in V_2}}\Delta^2_{\mathrm{def}}( \raisebox{-0.8ex}{\includegraphics[scale=0.8]{DK_A_2.pdf}} )}_{\text{Type $\mathrm{A}_2$}}
=1+2=3.
\end{align}
The same construction applied to $\mathfrak{su}(3)\oplus\mathfrak{su}(4)$ gives
\begin{align}
 \Delta_{\psi}^2
\left(\mathfrak{su}(3)\oplus\mathfrak{su}(4)\right)
&\geqslant
\underbrace{\min_{\substack{\phi_1\in V_2}}\Delta^2_{\mathrm{def}}( \raisebox{-0.8ex}{\includegraphics[scale=0.8]{DK_A_2.pdf}} )}_{\text{Type $\mathrm{A}_2$}}
+
\underbrace{\min_{\substack{\phi_1\in V_3}}\Delta^2_{\mathrm{def}}( \raisebox{-0.8ex}{\includegraphics[scale=0.8]{DK_A_3.pdf}} )}_{\text{Type $\mathrm{A}_3$}}
=2+3=5.
\end{align}
Passing to the adjoint representations changes the contribution of each simple factor. 
On $W_1\otimes W_2$, the bound becomes
\begin{align}
 \Delta_\psi^2
\left(\mathfrak{su}(2)\oplus\mathfrak{su}(3)\right)
&\geqslant
\underbrace{\min_{\substack{\phi_1\in W_1}}\Delta^2_{\mathrm{ad}}( \raisebox{-0.6ex}{\includegraphics[scale=0.8]{DK_B_1.pdf}} )}_{\text{Type $\mathrm{A}_1$}}
+
\underbrace{\min_{\substack{\phi_1\in W_2}}\Delta^2_{\mathrm{ad}}( \raisebox{-0.8ex}{\includegraphics[scale=0.8]{DK_A_2.pdf}} )}_{\text{Type $\mathrm{A}_2$}}
=2+4=6.
\end{align}
Replacing the second factor by $\mathfrak{su}(4)$ yields
\begin{align}
 \Delta_\psi^2
\left(\mathfrak{su}(2)\oplus\mathfrak{su}(4)\right)
&\geqslant
\underbrace{\min_{\substack{\phi_1\in W_1}}\Delta^2_{\mathrm{ad}}( \raisebox{-0.6ex}{\includegraphics[scale=0.8]{DK_B_1.pdf}} )}_{\text{Type $\mathrm{A}_1$}}
+
\underbrace{\min_{\substack{\phi_1\in W_3}}\Delta^2_{\mathrm{ad}}( \raisebox{-0.8ex}{\includegraphics[scale=0.8]{DK_A_3.pdf}} )}_{\text{Type $\mathrm{A}_3$}}
=2+6=8.
\end{align}
The construction also accommodates different representations of the two simple factors. 
Combining the defining representation of $\mathfrak{su}(2)$ with the adjoint representation of $\mathfrak{su}(3)$ gives, 
\begin{align}
 \Delta_\psi^2
\left(\mathfrak{su}(2)\oplus\mathfrak{su}(3)\right)
&\geqslant
\underbrace{\min_{\substack{\phi_1\in V_1}}\Delta^2_{\mathrm{def}}( \raisebox{-0.6ex}{\includegraphics[scale=0.8]{DK_B_1.pdf}} )}_{\text{Type $\mathrm{A}_1$}}
+
\underbrace{\min_{\substack{\phi_2\in W_2}}\Delta^2_{\mathrm{ad}}( \raisebox{-0.8ex}{\includegraphics[scale=0.8]{DK_A_2.pdf}} )}_{\text{Type $\mathrm{A}_2$}}
=1+4=5.
\end{align}
Interchanging the representation choices yields, 
\begin{align}
 \Delta_\psi^2
\left(\mathfrak{su}(2)\oplus\mathfrak{su}(3)\right)
&\geqslant
\underbrace{\min_{\substack{\phi_1\in W_1}}\Delta^2_{\mathrm{ad}}( \raisebox{-0.6ex}{\includegraphics[scale=0.8]{DK_B_1.pdf}} )}_{\text{Type $\mathrm{A}_1$}}
+
\underbrace{\min_{\substack{\phi_2\in V_2}}\Delta^2_{\mathrm{def}}( \raisebox{-0.8ex}{\includegraphics[scale=0.8]{DK_A_2.pdf}} )}_{\text{Type $\mathrm{A}_2$}}
=2+2=4.
\end{align}
Likewise, combining the defining representation of $\mathfrak{su}(3)$ with the adjoint representation of $\mathfrak{su}(4)$ gives, for every normalized state $\psi\in V_2\otimes W_3$ 
\begin{align}
 \Delta_\psi^2
\left(\mathfrak{su}(3)\oplus\mathfrak{su}(4)\right)
&\geqslant
\underbrace{\min_{\substack{\phi_1\in V_2}}\Delta^2_{\mathrm{def}}( \raisebox{-0.8ex}{\includegraphics[scale=0.8]{DK_A_2.pdf}} )}_{\text{Type $\mathrm{A}_2$}}
+
\underbrace{\min_{\substack{\phi_2\in W_3}}\Delta^2_{\mathrm{ad}}( \raisebox{-0.8ex}{\includegraphics[scale=0.8]{DK_A_3.pdf}} )}_{\text{Type $\mathrm{A}_3$}}
=2+6=8.
\end{align}
With the representation choices reversed, the bound becomes 
\begin{align}
 \Delta_\psi^2
\left(\mathfrak{su}(3)\oplus\mathfrak{su}(4)\right)
&\geqslant
\underbrace{\min_{\substack{\phi_1\in W_2}}\Delta^2_{\mathrm{ad}}( \raisebox{-0.8ex}{\includegraphics[scale=0.8]{DK_A_2.pdf}} )}_{\text{Type $\mathrm{A}_2$}}
+
\underbrace{\min_{\substack{\phi_2\in V_3}}\Delta^2_{\mathrm{def}}( \raisebox{-0.8ex}{\includegraphics[scale=0.8]{DK_A_3.pdf}} )}_{\text{Type $\mathrm{A}_3$}}
=4+3=7.
\end{align}

The same construction extends to simple factors from different Lie algebras. 
Consider the compact semisimple algebra,
\begin{align}
    \mathfrak{su}(3)\oplus\mathfrak{so}(8)\oplus\mathfrak e_6,
\end{align}
whose factors are of types $A_2$, $D_4$ and $E_6$, respectively. 
Let $V_1$ and $V_2$ carry the defining representations of $\mathfrak{su}(3)$ and $\mathfrak{so}(8)$, and let $V_3$ carry a fixed $27$-dimensional fundamental representation of $\mathfrak e_6$. 
The corresponding adjoint representation spaces are denoted by $W_1$, $W_2$ and $W_3$. 
Throughout these examples, the subscripts $\mathrm{def}$ and $\mathrm{ad}$ distinguish the representations carried by $V_k$ and $W_k$, with $\mathrm{def}$ also used for the chosen fundamental representation of $\mathfrak e_6$. 
Each minimum is taken independently over normalized states in the indicated factor space; 
the resulting bounds apply to all normalized states in the full tensor product, including entangled states.

On $V_1\otimes V_2\otimes V_3$, the three contributions add to give
\begin{align}
 \Delta_\psi^2
\left(\mathfrak{su}(3)\oplus\mathfrak{so}(8)\oplus \mathfrak{e}_6\right)
&\geqslant
\underbrace{\min_{\substack{\phi_1\in V_1}}\Delta^2_{\mathrm{def}}( \raisebox{-0.5ex}{\includegraphics[scale=0.6]{DK_A_2.pdf}} )}_{\text{Type $\mathrm{A}_2$}}
+
\underbrace{\min_{\substack{\phi_2\in V_2}}\Delta^2_{\mathrm{def}}( \raisebox{-3.8ex}{\includegraphics[scale=0.6]{DK_D_4.pdf}} )}_{\text{Type $\mathrm{D}_4$}}
+
\underbrace{\min_{\substack{\phi_3\in V_3}}\Delta^2_{\mathrm{def}}(\raisebox{-0.6ex}{\includegraphics[scale=0.75]{DK_E_6.pdf}})}_{\text{Type $\mathrm{E}_6$}}\notag\\
&=2+6+16=24.
\end{align}
Passing to the adjoint representation of each factor yields
\begin{align}
\Delta_\psi^2
\left(
\mathfrak{su}(3)\oplus\mathfrak{so}(8)\oplus\mathfrak{e}_6
\right)
&\geqslant
\underbrace{
\min_{\substack{\phi_1\in W_1}}
\Delta_{\mathrm{ad}}^2
\bigl( \raisebox{-0.5ex}{\includegraphics[scale=0.6]{DK_A_2.pdf}} \bigr)
}_{\text{Type $\mathrm{A}_2$}}
+
\underbrace{
\min_{\substack{\phi_2\in W_2}}
\Delta_{\mathrm{ad}}^2
\bigl( \raisebox{-3.8ex}{\includegraphics[scale=0.6]{DK_D_4.pdf}} \bigr)
}_{\text{Type $\mathrm{D}_4$}}
+
\underbrace{
\min_{\substack{\phi_3\in W_3}}
\Delta_{\mathrm{ad}}^2
\bigl(\raisebox{-0.6ex}{\includegraphics[scale=0.75]{DK_E_6.pdf}}\bigr)
}_{\text{Type $\mathrm{E}_6$}}
\notag\\
&=4+10+22=36.
\end{align}

The representation of each factor can also be chosen independently. 
Retaining $V_1$ while using the adjoint representations of the other two factors gives
\begin{align}
\label{eq:A2_D4_E6_mixed_bound}
\Delta_\psi^2
\left(
\mathfrak{su}(3)\oplus\mathfrak{so}(8)\oplus\mathfrak e_6
\right)
&\geqslant
\underbrace{
\min_{\substack{\phi_1\in V_1}}
\Delta_{\mathrm{def}}^2
\bigl( \raisebox{-0.5ex}{\includegraphics[scale=0.6]{DK_A_2.pdf}} \bigr)
}_{\text{Type $A_2$}}
+
\underbrace{
\min_{\substack{\phi_2\in W_2}}
\Delta_{\mathrm{ad}}^2
\bigl( \raisebox{-3.8ex}{\includegraphics[scale=0.6]{DK_D_4.pdf}} \bigr)
}_{\text{Type $D_4$}}
+
\underbrace{
\min_{\substack{\phi_3\in W_3}}
\Delta_{\mathrm{ad}}^2
\bigl( \raisebox{-0.6ex}{\includegraphics[scale=0.75]{DK_E_6.pdf}} \bigr)
}_{\text{Type $E_6$}}
\notag\\
&=2+10+22=34.
\end{align}

Alternatively, choosing the adjoint representation for only the first factor gives
\begin{align}
\label{eq:A2_D4_E6_ad_def_def_bound}
\Delta_\psi^2
\left(
\mathfrak{su}(3)\oplus\mathfrak{so}(8)\oplus\mathfrak e_6
\right)
&\geqslant
\underbrace{
\min_{\substack{\phi_1\in W_1}}
\Delta_{\mathrm{ad}}^2
\bigl( \raisebox{-0.5ex}{\includegraphics[scale=0.6]{DK_A_2.pdf}} \bigr)
}_{\text{Type $A_2$}}
+
\underbrace{
\min_{\substack{\phi_2\in V_2}}
\Delta_{\mathrm{def}}^2
\bigl( \raisebox{-3.8ex}{\includegraphics[scale=0.6]{DK_D_4.pdf}} \bigr)
}_{\text{Type $D_4$}}
+
\underbrace{
\min_{\substack{\phi_3\in V_3}}
\Delta_{\mathrm{def}}^2
\bigl(\raisebox{-0.6ex}{\includegraphics[scale=0.75]{DK_E_6.pdf}}\bigr)
}_{\text{Type $E_6$}}
\notag\\
&=4+6+16=26.
\end{align}
Placing the adjoint representation on the second factor instead leads to
\begin{align}
\label{eq:A2_D4_E6_def_ad_def_bound}
\Delta_\psi^2
\left(
\mathfrak{su}(3)\oplus\mathfrak{so}(8)\oplus\mathfrak e_6
\right)
&\geqslant
\underbrace{
\min_{\substack{\phi_1\in V_1}}
\Delta_{\mathrm{def}}^2
\bigl( \raisebox{-0.5ex}{\includegraphics[scale=0.6]{DK_A_2.pdf}} \bigr)
}_{\text{Type $A_2$}}
+
\underbrace{
\min_{\substack{\phi_2\in W_2}}
\Delta_{\mathrm{ad}}^2
\bigl( \raisebox{-3.8ex}{\includegraphics[scale=0.6]{DK_D_4.pdf}} \bigr)
}_{\text{Type $D_4$}}
+
\underbrace{
\min_{\substack{\phi_3\in V_3}}
\Delta_{\mathrm{def}}^2
\bigl(\raisebox{-0.6ex}{\includegraphics[scale=0.75]{DK_E_6.pdf}}\bigr)
}_{\text{Type $E_6$}}
\notag\\
&=2+10+16=28.
\end{align}
Finally, choosing it for the exceptional factor yields
\begin{align}
\label{eq:A2_D4_E6_def_def_ad_bound}
\Delta_\psi^2
\left(
\mathfrak{su}(3)\oplus\mathfrak{so}(8)\oplus\mathfrak e_6
\right)
&\geqslant
\underbrace{
\min_{\substack{\phi_1\in V_1}}
\Delta_{\mathrm{def}}^2
\bigl( \raisebox{-0.5ex}{\includegraphics[scale=0.6]{DK_A_2.pdf}} \bigr)
}_{\text{Type $A_2$}}
+
\underbrace{
\min_{\substack{\phi_2\in V_2}}
\Delta_{\mathrm{def}}^2
\bigl( \raisebox{-3.8ex}{\includegraphics[scale=0.6]{DK_D_4.pdf}} \bigr)
}_{\text{Type $D_4$}}
+
\underbrace{
\min_{\substack{\phi_3\in W_3}}
\Delta_{\mathrm{ad}}^2
\bigl( \raisebox{-0.6ex}{\includegraphics[scale=0.75]{DK_E_6.pdf}}\bigr)
}_{\text{Type $E_6$}}
\notag\\
&=2+6+22=30.
\end{align}
Each bound is sharp and is attained by a tensor product of highest weight states of the corresponding factor representations. 
Concavity of the variance extends these bounds to arbitrary density operators on the same representation spaces.
These examples illustrate how the simple factors and their representations determine the minimum quantum uncertainty through additive contributions evaluated on the reduced states. 
The next subsection uses these local uncertainty bounds to derive an entanglement criterion for multipartite systems.


\subsection{Entanglement Detection via Violations of Local Uncertainty}
\label{subsec:Entanglement_Local_Uncertainty_Violations}

A central application of variance-based uncertainty relations is entanglement detection through the violation of local uncertainty bounds~\cite{PhysRevA.68.032103}. 
The underlying mechanism is simple and rests on two structural properties of variance. 
The first is concavity under classical probabilistic mixing: 
statistical uncertainty cannot be reduced below the average uncertainty of the constituent states.
Thus, for any decomposition $\rho=\sum_i p_i\rho_i$ and any observable $\mO$, one has
\begin{align}
\Delta_\rho^2(\mO)
\geqslant
\sum_i p_i\Delta_{\rho_i}^2(\mO).
\end{align}
The second property is the additivity of local variances on product states. 
For a bipartite product state $\rho\coloneqq\sigma\otimes\tau$ and a sum of local observables
\begin{align}
    \mO
    \coloneqq 
    A\otimes \1
    +
    \1\otimes B ,
\end{align}
the covariance term vanishes, and hence
\begin{align}
    \Delta_{\rho}^2(\mO)
    =
    \Delta_{\sigma}^2(A)
    +
    \Delta_{\tau}^2(B).
\end{align}

These two facts immediately lead to an entanglement criterion.
Let $\{A_k\}$ and $\{B_k\}$ be sets of local observables satisfying local uncertainty relations
\begin{align}
    \sum_k \Delta_{\rho}^2(A_k)\geqslant \mB_1,
\end{align}
and
\begin{align}
    \sum_k \Delta_{\rho}^2(B_k)\geqslant \mB_2,
\end{align}
for all local states. 
Then every separable state $\rho\coloneqq\sum_i p_i\, \sigma_i\otimes\tau_i$ obeys
\begin{align}
\sum_k
\Delta_\rho^2(A_k\otimes \1+\1\otimes B_k)
&\geqslant
\sum_{k, i} p_i\Delta_{\sigma_i\otimes\tau_i}^2
(A_k\otimes \1+\1\otimes B_k)\\
&=
\sum_{i} p_i
\left(
\sum_k\left(\Delta_{\sigma_i}^2(A_k)\right)
+
\sum_k\left(\Delta_{\tau_i}^2(B_k)\right)
\right)\\
&\geqslant
\sum_i p_i(\mB_1+\mB_2)\\
&=
\mB_1+\mB_2 .
\end{align}
Consequently, any violation of this local uncertainty bound certifies that the state is entangled.

We now combine the local uncertainty framework with the Lie algebraic uncertainty relations developed in this work to obtain an entanglement witness for many-body systems.
Let $V(\Lambda)$ be the representation space of the compact real Lie algebra $\g_{\mathbb R}$, and let $V(\Lambda)^{\otimes N}$ denote the corresponding $N$-system Hilbert space. 
For any single-system operator $X\in\mathrm{End}(V)$, we write $X^{(k)}$, with $k\in\{1,\ldots,N\}$, for its embedding into the $k$-th system:
\begin{align}
    X^{(k)}
    \coloneqq
    \1^{\otimes(k-1)}\otimes X \otimes \1^{\otimes(N-k)},
\end{align}
where $\1$ denotes the identity operator on $V(\Lambda)$.
To formulate the many-body witness, we lift each single-system Lie algebraic generator to the $N$-particle Hilbert space and form its collective counterpart by summing over all tensor factors.

\begin{mydef}{Collective Lie Algebraic Operators}{Collective_Operators}
Let $V(\Lambda)^{\otimes N}$ be the $N$-particle Hilbert space, and let $\{H_i,E_\alpha,F_\alpha\}$ (see Def.~\ref{def:Anti_Hermitian_Basis}) denote the anti-Hermitian single-particle generators of $\g_{\mathbb R}$ acting on $V(\Lambda)$. 
The corresponding collective operators are obtained by summing their local embeddings over all tensor factors:
\begin{align}
    \mathfrak{H}_i
    \coloneqq
    \sum_{k=1}^N H_i^{(k)}, \quad
    \mathfrak E_\alpha
    \coloneqq
    \sum_{k=1}^N E_\alpha^{(k)}, \quad
    \mathfrak F_\alpha
    \coloneqq
    \sum_{k=1}^N F_\alpha^{(k)} ,
\end{align}
with $i\in\{1,\ldots,r\}$ and $\alpha\in\Phi^+$. 
Here each sum is an operator on $\mathcal H\coloneqq V(\Lambda)^{\otimes N}$.
\end{mydef}

These collective operators generate the diagonal action of $\mathfrak g_{\mathbb R}$ on $V(\Lambda)^{\otimes N}$.
Their total collective uncertainty in a state $\rho$ is

\begin{align}
\label{eq:total_variance_coll}
\Delta_{\text{Many-Body}}^2(\g)
\coloneqq
\frac{1}{2}\left(\sum_{i=1}^{r}\Delta_{\rho}^2(\mathfrak H_i)
+
\sum_{\alpha\in\Phi^+}\left(\Delta_{\rho}^2(\mathfrak E_\alpha)
+
\Delta_{\rho}^2(\mathfrak F_\alpha)\right)\right),
\end{align}
where $r=\mathrm{rank}(\mathfrak g)$.
For any separable state on $\mH$, this quantity satisfies the following many-body uncertainty bound.

\begin{mythm}{Lie Algebraic Uncertainty Bounds for Separable Many-Body States}{LAU_Sep_Many_Body}
Any fully separable state $\rho\coloneqq\sum_i p_i\rho_i^1\otimes\cdots\otimes\rho^N_i$ on $\mH=V(\Lambda)^{\otimes N}$, with each site carrying the fixed representation $V(\Lambda)$ of $\g$, must satisfy the many-body Lie algebraic uncertainty bound
\begin{align}
\label{eq:many_body_sep}
    \Delta_{\text{Many-Body}}^2(\g)
    \geqslant N\, \langle \vec{\lambda}, \vec{\xi}\rangle.
\end{align}
Here $\vec{\lambda}$ denotes the Dynkin coefficient vector of the highest weight $\Lambda$, as defined in Eq.~\eqref{eq:Dynkin_Vector}, while $\vec{\xi}$ is the vector introduced in Eq.~\eqref{eq:xi}. 
Any violation rules out full separability and therefore certifies entanglement.
\end{mythm}

The proof of Thm.~\ref{thm:LAU_Sep_Many_Body} is a direct many-body extension of the bipartite argument; the details are omitted.

We now turn to an application of Lie algebraic uncertainty relations (see Eq.~\eqref{eq:many_body_sep}) as witnesses of bound entanglement. 
In particular, we consider PPT bound entangled states, whose entanglement is not detected by the positive-partial-transpose (PPT) criterion.
As a concrete example, we use the Tiles unextendible product basis (UPB) in $\mathbb{C}^3 \otimes \mathbb{C}^3$ systems~\cite{BennettUPB,DiVincenzoUPB}.
Throughout, $\{\ket{0},\ket{1},\ket{2}\}$ denotes the computational basis of $\mathbb C^3$, and the associated product vectors are represented using the diagrammatic convention below.

\begin{align}
    \ket{\phi_1} 
    \coloneqq \quad
    \vcenter{\hbox{\includegraphics[scale=1]{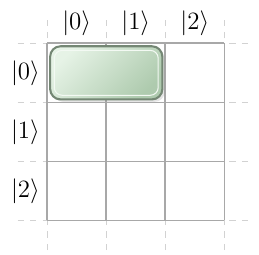}}}
    =
    \ket{0}\otimes \frac{\ket{0}-\ket{1}}{\sqrt2}.
\end{align}
The equation above illustrates our diagrammatic convention for representing product states in a bipartite $\mathbb{C}^3 \otimes \mathbb{C}^3$ quantum system. 
In this $3 \times 3$ grid layout, the rows index the computational basis states of the first subsystem ($\{|0\rangle, |1\rangle, |2\rangle\}$) from top to bottom, while the columns index the computational basis states of the second subsystem ($\{|0\rangle, |1\rangle, |2\rangle\}$) from left to right. 
A localized tile spanning multiple squares visually encodes a superposition state formed by those grid intersections; 
specifically, a green-tinted tile signifies a normalized subtraction state, meaning that the horizontal tile covering the tensor of row $|0\rangle$ with columns $|0\rangle$ and $|1\rangle$ explicitly represents the product vector $|\phi_1\rangle = |0\rangle \otimes (|0\rangle - |1\rangle)/\sqrt{2}$.

The product vectors $\phi_2$, $\phi_3$, and $\phi_4$ are introduced analogously.
In particular, $\phi_2$ is obtained by fixing the $B$ subsystem to $\ket{2}$, while allowing the $A$ subsystem to occupy the two-dimensional subspace spanned by $\ket{0}$ and $\ket{1}$. 
The relative minus sign gives the antisymmetric superposition $(\ket{0}-\ket{1})/\sqrt{2}$.
Diagrammatically, this is represented by the vertical shaded block in the $\ket{2}$ column, spanning the $\ket{0}$ and $\ket{1}$ rows.

\begin{align}
    \ket{\phi_2} 
    \coloneqq 
    \quad
    \vcenter{\hbox{\includegraphics[scale=1]{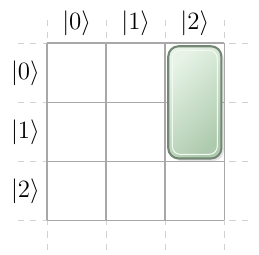}}}
    =
    \frac{\ket{0}-\ket{1}}{\sqrt2}
    \otimes\ket{2}.
\end{align}

The next vector is obtained by the complementary construction: 
instead of fixing the $B$ subsystem and superposing two neighbouring levels of $A$, we now fix the $A$ subsystem and place the relative-minus superposition on $B$. 
This rotates the shaded block from a vertical segment to a horizontal one on the grid, leading to the product vector $\ket{\phi_3}$.

\begin{align}
    \ket{\phi_3} 
    \coloneqq 
    \quad
    \vcenter{\hbox{\includegraphics[scale=1]{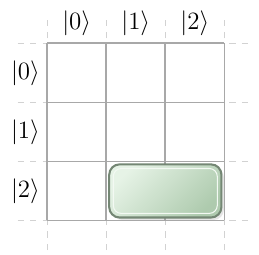}}}
    =
    \ket{2}\otimes
    \frac{\ket{1}-\ket{2}}{\sqrt2}.
\end{align}

The construction then continues by shifting the support to the left edge of the grid. 
While $\ket{\phi_3}$ fixes the $A$ subsystem to $\ket{2}$ and places the relative-minus superposition on the $B$ subsystem, the next vector $\ket{\phi_4}$ again fixes the $B$ subsystem, now to $\ket{0}$, and lets the $A$ subsystem occupy the antisymmetric superposition of $\ket{1}$ and $\ket{2}$. 
Diagrammatically, this turns the horizontal block in the bottom row into a vertical block in the $\ket{0}$ column:

\begin{align}
    \ket{\phi_4} 
    \coloneqq  
    \quad
    \vcenter{\hbox{\includegraphics[scale=1]{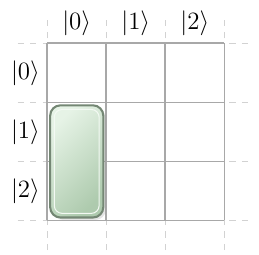}}}
    =
    \frac{\ket{1}-\ket{2}}{\sqrt2}\otimes
    \ket{0}.
\end{align}

Together, the four product vectors $\{\phi_1, \phi_2, \phi_3, \phi_4\}$ form an interlocking pattern on the $3\times 3$ grid. 
Each shaded block represents the two computational-basis components entering the corresponding antisymmetric superposition:
horizontal blocks indicate that the $A$ subsystem is fixed while the $B$ subsystem is superposed, whereas vertical blocks indicate that the $B$ subsystem is fixed while the $A$ subsystem is superposed. 
Arranged in this alternating way, the four states occupy the boundary of the grid in a mutually orthogonal, tile-like configuration, as summarized in the diagram below.
\begin{align}
     \raisebox{-13ex}{\includegraphics[scale=1]{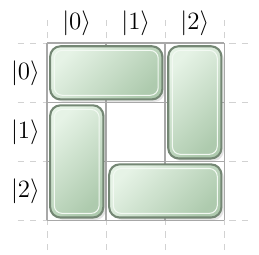}} \, .
\end{align}

The first four product states do not yet form a UPB, as their orthogonal complement still contains product vectors. 
The addition of the fifth state $\ket{\phi_5}$ removes this possibility and completes the unextendible product basis.
\begin{align}
    \ket{\phi_5} 
    \coloneqq 
    \quad
    \vcenter{\hbox{\includegraphics[scale=1]{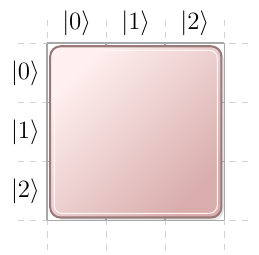}}}
    =
    \frac{\ket{0}+\ket{1}+\ket{2}}{\sqrt3}
    \otimes
    \frac{\ket{0}+\ket{1}+\ket{2}}{\sqrt3} .
\end{align}
Complementing the subtraction rule, the color red in this diagrammatic convention signifies an addition in superposition. 
A red-tinted tile indicates a uniform, positive superposition spanning across its covered basis states. 
As demonstrated by the representation of $\ket{\phi_5}$, when the red tile completely encompasses the entire $3 \times 3$ grid, it geometrically encodes the fully symmetric, normalized addition state vector $(\ket{0}+\ket{1}+\ket{2})\otimes(\ket{0}+\ket{1}+\ket{2})/3$.

Having defined the UPB, we now let $\mS$ denote the subspace generated by its five product vectors:
\begin{align}
\mS
\coloneqq \mathrm{span}\{\ket{\phi_1},\ldots,\ket{\phi_5}\}
=
\mathrm{span}\left\{ \raisebox{-13ex}{\includegraphics[scale=1]{PUB_Phi_1-4.pdf}}, \raisebox{-13ex}{\includegraphics[scale=1]{PUB_Phi_5.pdf}} \right\}.
\end{align}
This subspace sits inside the bipartite Hilbert space $\mH\coloneqq\mathbb C^3\otimes\mathbb C^3$.
The orthogonal projector onto this subspace $\mS$ is
\begin{align}\label{eq:Pi}
\Pi
\coloneqq \sum_{i=1}^{5} \ketbra{\phi_i}{\phi_i}.
\end{align}
Orthonormality immediately implies
\begin{align}
\Pi^2
=
\sum_{i,j=1}^{5}
\ketbra{\phi_i}{\phi_i}
\ketbra{\phi_j}{\phi_j}
=
\sum_{i=1}^{5}
\ketbra{\phi_i}{\phi_i}
=
\Pi,
\end{align}
and
\begin{align}
\Pi^\dagger
=
\Pi.
\end{align}
Thus $\Pi$ is an orthogonal projection, with
\begin{align}
\mathrm{rank}(\Pi)
=
\dim\mS
=
5.
\end{align}

The complementary projector
\begin{align}
\Pi^{^\perp}
\coloneqq
\1-\Pi,
\end{align}
therefore projects onto the four-dimensional subspace
\begin{align}
\mS^\perp
\coloneqq
\left\{
\ket{\psi}\in\mathcal H
\,|\,
\braket{\psi}{\phi_i}=0,\ i=1,\ldots,5
\right\},
\end{align}
and has rank
\begin{align}
\mathrm{rank}(\1-\Pi)
=
9-5
=
4.
\end{align}
This gives the normalized state
\begin{align}\label{eq:UPB_state}
\rho_{\mathrm{UPB}}
\coloneqq
\frac{1}{4}\left(\1-\Pi\right).
\end{align}
We now turn to the state of interest
\begin{align}\label{eq:noisy_UPB_state1}
\rho_p
=
(1-p)\rho_{\mathrm{UPB}}
+
p\frac{\1}{9}.
\end{align}
Since $\Pi^{\T_B}=\Pi$, both $\rho_{\mathrm{UPB}}$ and $\rho_p$ are invariant under partial transposition $\T_{B}$ on $B$. 
Thus $\rho_p^{T_B}=\rho_p\geqslant0$ for every $p\in[0,1]$, and the PPT criterion cannot detect entanglement anywhere in this family.

To detect entanglement in $\rho_p$ (see Eq.~\eqref{eq:noisy_UPB_state1}), we apply Thm.~\ref{thm:LAU_Sep_Many_Body} with $\g_{\mathbb R}=\mathfrak{su}(3)$ and its complexification $\g=\mathfrak{sl}(3,\mathbb C)$. 
We use the Gell-Mann matrices $\mathbf X_{\mathrm{ob}}=\{\lambda_a\}_{a=1}^{8}$ defined in Eqs.~\eqref{eq:Gell_Cartan}-\eqref{eq:su3_hermitian_chevalley_basis}, which are Hermitian, traceless and satisfy
\begin{align}\label{eq:Gell_matrices_norm}
    \Tr[\lambda_a\lambda_b]=2\delta_{ab}.
\end{align}
Their Gram matrix $\mathrm{G}(\mathbf X_{\mathrm{ob}})$ is therefore 
\begin{align}
    \mathrm{G}(\mathbf X_{\mathrm{Ch}})=2\1_8,
\end{align}
with inverse given by
\begin{align}
    \mathrm{G}(\mathbf X_{\mathrm{Ch}})^{-1}=\frac{1}{2}\1_8.
\end{align}
This fixes the overall factor of $1/2$ in the corresponding Lie algebraic uncertainty relation (see Eq.~\eqref{eq:observable_uncertainty}).
Throughout this subsection, $\lambda_a$ denotes a Gell-Mann matrices, while $\vec\lambda$ denotes the Dynkin coefficients vector (see Eq.~\eqref{eq:Dynkin_Vector}).

As listed in Tab.~\ref{tab:An_representations} of Sec.~\ref{subsec:A_type}, the defining representation of the type-A Lie algebra $\mathfrak{su}(3)$ has Dynkin coefficients vector 
\begin{align}
    \vec{\lambda}=
    \begin{pmatrix}
    1\\
    0
    \end{pmatrix}
\end{align}
Every single qutrit state $\phi$ therefore obeys the local uncertainty relation (see Tab.~\ref{tab:explicit_bounds})
\begin{align}\label{eq:single_qutrit_uncertainty_bound_transition}
\frac12 \sum_{a=1}^{8}
\Delta_\phi^2(\lambda_a)
=\Delta^2_{\mathrm{def}}( \raisebox{-0.8ex}{\includegraphics[scale=0.8]{DK_A_2.pdf}} )
\geqslant
2.
\end{align}

To construct an entanglement criterion from this local bound, we first consider the collective observables
\begin{align}
\lambda_a\otimes \1+\1\otimes\lambda_a,
\quad
a\in\{1,\ldots,8\}.
\end{align}
For entanglement detection, however, this fixed pairing need not be optimal: variance depends on the relative orientation of the two local observable frames.
We therefore introduce
\begin{align}
\widetilde{\lambda}_a
\coloneqq
\sum_{b=1}^{8}
O_{ab}\lambda_b.
\end{align}
where
\begin{align}
\mathrm{O}(8)
=\left\{O\in \mathrm{M}_8(\mathbb R)  \middle|
O^{\T}O
=OO^{\T}
=\1_8
\right\}.
\end{align}
Orthogonality preserves the Hilbert--Schmidt normalization 
\begin{align}\label{eq:invariance_inner}
\Tr[\widetilde{\lambda}_a\widetilde{\lambda}_b]
=
2\delta_{ab},
\end{align}
as well as the sum of invariance.

The dependence on $O$ therefore enters entirely through the cross-covariance contribution, allowing us to select the pairing that minimizes the collective uncertainty.

\begin{mydef}{Optimized Collective $\mathfrak{su}(3)$ Uncertainty}{optimized_collective_su3_uncertainty}

For a two-qutrit state $\rho$ on $\mathbb C^3\otimes\mathbb C^3$, we define the optimized collective $\mathfrak{su}(3)$ uncertainty as

\begin{align}\label{eq:optimized_collective_uncertainty}
\mathcal U_{\mathrm{coll}}^{\mathrm{opt}}(\rho)
\coloneqq
\frac12 \min_{O\in \mathrm{O}(8)} \sum_{a=1}^{8} \Delta_\rho^2  \left(\lambda_a\otimes \1
+\1\otimes \sum_{b=1}^{8} O_{ab}\lambda_b \right) .
\end{align}
\end{mydef}

\noindent
Optimization over $\mathrm{O}(8)$ extends the choice of local observable bases beyond those obtainable by $\mathrm{SU}(3)$ conjugation, while preserving the local total variances and their uncertainty bounds.
Although the resulting collective observables need not close
under commutation, the separable state bound follows solely
from the local uncertainty relations and the concavity of variance.
It therefore holds uniformly for all $O\in \mathrm{O}(8)$ and remains
valid after minimization.
A violation of this optimized uncertainty bound thus certifies entanglement.

\begin{mylem}{Entanglement Criterion from Optimized Collective Uncertainty}{qutrit_entangle_threshold}
Every separable two-qutrit state $\rho_{\mathrm{sep}}$ satisfies
\begin{align}\label{eq:qutrit_separable_threshold}
\mathcal U_{\mathrm{coll}}^{\mathrm{opt}}(\rho_{\mathrm{sep}})
\geqslant 4,
\end{align}
where $\mathcal U_{\mathrm{coll}}^{\mathrm{opt}}$ is defined in Eq.~\eqref{eq:optimized_collective_uncertainty}.
Violation of this bound detects entanglement.
\end{mylem}

\begin{proof}
For a fixed $O\in \mathrm{O}(8)$, write
\begin{align}
A_a(O)
\coloneqq\lambda_a\otimes \1 + \1\otimes\widetilde{\lambda}_a .
\end{align}
An orthogonal change of frame leaves the sum of local variances unchanged. Hence, for any single-qutrit state $\phi$ 
\begin{align}
\frac12\sum_{a=1}^{8}
\Delta_\phi^2(\widetilde{\lambda}_a)
=
\frac12\sum_{a=1}^{8}
\Delta_\phi^2(\lambda_a)
=\Delta^2_{\mathrm{def}}( \raisebox{-0.8ex}{\includegraphics[scale=0.8]{DK_A_2.pdf}} )
\geqslant2,
\end{align}
by Eq.~\eqref{eq:single_qutrit_uncertainty_bound_transition}.
Since local variances add on product states and variance is
concave under mixing, every separable state $\rho_{\mathrm{sep}}$
satisfies, for each fixed $O\in \mathrm{O}(8)$,
\begin{align}
\frac12\sum_{a=1}^{8}
\Delta_{\rho_{\mathrm{sep}}}^2
\left(
\lambda_a\otimes\1
+
\1\otimes\widetilde{\lambda}_a
\right)
=\Delta^2_{\mathrm{def}}( \raisebox{-0.8ex}{\includegraphics[scale=0.8]{DK_A_2.pdf}} )+\Delta^2_{\mathrm{def}}( \raisebox{-0.8ex}{\includegraphics[scale=0.8]{DK_A_2.pdf}} )
\geqslant2+2=4.
\end{align}
Minimizing over $\mathrm{O}(8)$ proves the claim.
\end{proof}

To apply this criterion to the noisy UPB states, we first express the sum of the eight Gell-Mann variances for an arbitrary single-qutrit
state $\phi$ in terms of its purity.
The eight generators $\{\lambda_a\}_{a=1}^{8}$ defined in Eqs.~\eqref{eq:Gell_Cartan}-\eqref{eq:su3_hermitian_chevalley_basis}, satisfy
\begin{align}
 \Tr [\lambda_a]=0,
 \quad
 \Tr [\lambda_a\lambda_b]=2\delta_{ab},
\end{align}
together with the quadratic identity
\begin{align}\label{eq:sum_Gell_mann}
 \sum_{a=1}^{8}\lambda_a^2=\frac{16}{3}\1_3.
\end{align}
Any single qutrit state admits the expansion
\begin{align}
 \phi=\frac{\1_3}{3}+\frac12\sum_{a=1}^{8}r_a\lambda_a,
\end{align}
with $r_a=\Tr [\phi\lambda_a]$.
Using tracelessness and trace orthogonality, we obtain
\begin{align}\label{eq:upb_bloch_purity}
 \Tr [\phi^2]
 =\frac19\Tr[\1_3]
   +\frac13\sum_{a=1}^{8}r_a\Tr [\lambda_a]
   +\frac14\sum_{a,b=1}^{8}r_ar_b \Tr [\lambda_a\lambda_b]
=\frac13+\frac12\sum_{a=1}^{8}r_a^2.
\end{align}
Combining Eqs.~\eqref{eq:sum_Gell_mann} and \eqref{eq:upb_bloch_purity} yields
\begin{align}\label{eq:single_qutrit_uncertainty_bound_transition2}
 \frac12\sum_{a=1}^{8}\Delta_\phi^2(\lambda_a)
 =\frac12\Tr [\phi\sum_{a=1}^{8}\lambda_a^2]
   -\frac12\sum_{a=1}^{8}r_a^2
 =\frac83-\left(\Tr [\phi^2]-\frac13\right)
 =3-\Tr [\phi^2],
\end{align}
Thus, the sum of the Gell-Mann variances depends only on the purity
of the state. 
Orthogonal changes of generator basis preserve this variance sum, so
\begin{align} \label{eq:upb_transformed_local_bound}
 \frac12\sum_{a=1}^{8}\Delta_\phi^2(\widetilde\lambda_a)
 =3-\Tr[\phi^2].
\end{align}

For a two-qutrit state $\rho_{AB}$, let $\rho_A\coloneqq\Tr_B[\rho_{AB}]$ and $\rho_B\coloneqq\Tr_A[\rho_{AB}]$ denote its reduced states. 
We collect the local Gell-Mann observables into the two lists
\begin{align}
    \mathbf X^A_{\mathrm{ob}}&\coloneqq\{\lambda_1\otimes\1,\ldots,\lambda_8\otimes\1\},\\
    \mathbf X^B_{\mathrm{ob}}&\coloneqq\{\1\otimes\lambda_1,\ldots,\1\otimes\lambda_8\}.
\end{align}
The real covariance matrix associated with the combined list $(\mathbf X^A_{\mathrm{ob}},\mathbf X^B_{\mathrm{ob}})$ contains an upper-right $8\times8$ block that describes the cross covariances between the two subsystems.
We denote this block by $\mathrm{Cov}_{\rho_{AB}}$, with entries
\begin{align}\label{eq:su3_covariance_matrix}
   (\mathrm{Cov}_{\rho_{AB}})_{ab}
    \coloneqq
    \Tr \left[\rho_{AB}\,\lambda_a\otimes\lambda_b\right]
    -
    \Tr [\rho_A\lambda_a]\,
    \Tr [\rho_B\lambda_b].
\end{align}

Expanding the square of $A_a(O)$ gives
\begin{align}
 A_a(O)^2
 =\lambda_a^2\otimes \1_3
  +\1_3\otimes\widetilde\lambda_a^2
  +2\lambda_a\otimes\widetilde\lambda_a.
\end{align}
Using
$\langle A_a(O)\rangle_{\rho_{AB}}
=\Tr [\rho_A\lambda_a]+\Tr [\rho_B\widetilde\lambda_a]$
and subtracting the squared mean, we obtain
\begin{align}\label{eq:upb_individual_collective_variance}
 \Delta_{\rho_{AB}}^2(A_a(O))
 &=\Delta_{\rho_A}^2(\lambda_a)
   +\Delta_{\rho_B}^2(\widetilde\lambda_a)
   +2\left(\Tr [\rho_{AB}\,\lambda_a\otimes\widetilde\lambda_a]
   -\Tr [\rho_A\lambda_a]\Tr [\rho_B\widetilde\lambda_a]\right)\\
 &=\Delta_{\rho_A}^2(\lambda_a)
   +\Delta_{\rho_B}^2(\widetilde\lambda_a)
   +2\sum_{b=1}^{8}O_{ab}(\mathrm{Cov}_{\rho_{AB}})_{ab}.
\end{align}
Applying Eqs.~\eqref{eq:single_qutrit_uncertainty_bound_transition2}
and~\eqref{eq:upb_transformed_local_bound}, we find that
\begin{align}\label{eq:upb_collective_purity_identity}
 \mathcal U_{\mathrm{coll}}(\rho_{AB};O)
 &=\frac12\sum_{a=1}^{8}\Delta_{\rho_A}^2(\lambda_a)
   +\frac12\sum_{a=1}^{8}\Delta_{\rho_B}^2(\widetilde\lambda_a)
   +\sum_{a,b=1}^{8}O_{ab}(\mathrm{Cov}_{\rho_{AB}})_{ab} \\[4pt]
 &=\bigl(3-\Tr [\rho_A^2]\bigr)
   +\bigl(3-\Tr [\rho_B^2]\bigr)
   +\Tr \left[O^{\T}\mathrm{Cov}_{\rho_{AB}}\right] \\[6pt]
 &=6-\Tr [\rho_A^2]-\Tr[\rho_B^2]
   +\Tr \left[O^{\T}\mathrm{Cov}_{\rho_{AB}}\right].
\end{align}
Here 
\begin{align}
\Tr[O^{\T}\mathrm{Cov}_{\rho_{AB}}]
=\sum_{a,b}O_{ab}(\mathrm{Cov}_{\rho_{AB}})_{ab}.
\end{align}
Thus, $O$ enters only through the cross covariance term.
To minimize this term, consider a real singular-value decomposition
\begin{align}
\mathrm{Cov}_{\rho_{AB}}=L\Sigma R^{\T},
\quad L,R\in \mathrm{O}(8), 
\end{align}
with $\Sigma=\mathrm{diag}(s_1,\ldots,s_8)$.
Here, $s_a\geqslant0$ are the singular values of $\mathrm{Cov}_{\rho_{AB}}$.
The trace norm (Schatten $1$-norm) is therefore 
\begin{align}\label{eq:normal_covriance}
 \|\mathrm{Cov}_{\rho_{AB}}\|_1
 \coloneqq
 \Tr [\sqrt{(\mathrm{Cov}_{\rho_{AB}})^{\T} \mathrm{Cov}_{\rho_{AB}}} ]
 =\Tr [\Sigma]
 =\sum_{a=1}^{8}s_a.
\end{align}
Setting $Q=L^{\T}OR\in \mathrm{O}(8)$, we obtain
\begin{align}\label{eq:trace_OC}
 \Tr[O^{\T}\mathrm{Cov}_{\rho_{AB}}]
 =\Tr[Q^{\T}\Sigma]
 =\sum_{a=1}^{8}s_aQ_{aa}
 \geqslant-\sum_{a=1}^{8}s_a.
\end{align}
Equality is attained by $Q=-\1_8$, or equivalently by $O_*=-LR^{\T}$. 
By Eqs.~\eqref{eq:normal_covriance} and \eqref{eq:trace_OC}, we have 
\begin{align}
 \min_{O\in \mathrm{O}(8)}\Tr[O^{\T}\mathrm{Cov}_{\rho_{AB}}]
 =-\|\mathrm{Cov}_{\rho_{AB}}\|_1.
\end{align}
The optimized uncertainty is therefore
\begin{align}\label{eq:optimized_collective_uncertainty2}
\mathcal U_{\mathrm{coll}}^{\mathrm{opt}}(\rho)
 =6-\Tr [\rho_A^2]-\Tr[\rho_B^2]-\|\mathrm{Cov}_{\rho_{AB}}\|_1.
\end{align}
Combining Eqs.~\eqref{eq:single_qutrit_uncertainty_bound_transition2}
and~\eqref{eq:upb_transformed_local_bound} yields the sufficient entanglement condition
\begin{align}\label{eq:general_covariance_criterion}
 \|\mathrm{Cov}_{\rho_{AB}}\|_1>2-\Tr[\rho_A^2]-\Tr[\rho_B^2].
\end{align}
In the Gell-Mann normalization used here, the corresponding separable-state bound is precisely the trace-norm consequence
(Proposition~3) of the covariance-matrix criterion in~\cite{Guhne2007Covariance}.
Thus, our optimized collective uncertainty relation provides an uncertainty-based derivation of this covariance-based entanglement condition.
Moreover, Eq.~\eqref{eq:general_covariance_criterion} implies the centered-state realignment condition in~\cite{Zhang2008},
\begin{align}
    \|\mathrm{Cov}^{AB}_\rho\|_1
    >
    2\sqrt{\bigl(1-\Tr[\rho_A^2]\bigr)\bigl(1-\Tr [\rho_B^2]\bigr)}.
\end{align}

We next apply this criterion to the noisy two-qutrit state
constructed from the Tiles UPB.
Using  $\rho_{\mathrm{UPB}}=\frac14(\1_9-\Pi)$ (see Eq.~\eqref{eq:noisy_UPB_state1}), we have
\begin{align}
\rho_p
=\frac{9-5p}{36}\1_9-\frac{1-p}{4}\Pi,
\qquad 0\leqslant p\leqslant1.
\end{align}

\begin{figure}[t]
    \centering
    \includegraphics[width=0.85\columnwidth]{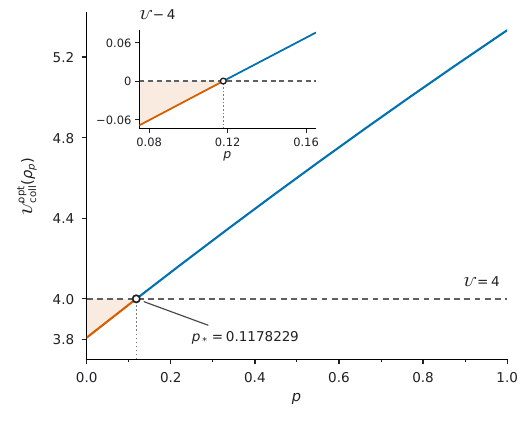}
    \caption{
    \textbf{Entanglement Detection through Lie Algebraic Uncertainty.}
    Optimized collective uncertainty $\mU_{\mathrm{coll}}^{\mathrm{opt}}(\rho_p)$ for the noisy two-qutrit UPB state $\rho_p$ as a function of the noise parameter $p$. 
    The horizontal dashed line marks the separability bound of $4$. 
    The orange part of the line lies below this bound, detecting entanglement for $0\leqslant p<p_*$, where $p_*\simeq0.1178229$ is the detection threshold marked by an open circle.
    The inset shows $\mathcal U_{\mathrm{coll}}^{\mathrm{opt}}(\rho_p)-4$ near the threshold.
    }
    \label{fig:noisy_UPB_collective_uncertainty}
\end{figure}

To evaluate its optimized collective uncertainty, we first write the projector $\Pi$ in Eq.~\eqref{eq:Pi} explicitly as a sum of tensor products of local matrices:
\begin{align}
\Pi
=&
\frac12
\begin{pmatrix}
1&0&0\\
0&0&0\\
0&0&0
\end{pmatrix}
\otimes
\begin{pmatrix}
1&-1&0\\
-1&1&0\\
0&0&0
\end{pmatrix}
+\frac12
\begin{pmatrix}
1&-1&0\\
-1&1&0\\
0&0&0
\end{pmatrix}
\otimes
\begin{pmatrix}
0&0&0\\
0&0&0\\
0&0&1
\end{pmatrix}
+\frac12
\begin{pmatrix}
0&0&0\\
0&0&0\\
0&0&1
\end{pmatrix}
\otimes
\begin{pmatrix}
0&0&0\\
0&1&-1\\
0&-1&1
\end{pmatrix}
\nonumber\\
+&\frac12
\begin{pmatrix}
0&0&0\\
0&1&-1\\
0&-1&1
\end{pmatrix}
\otimes
\begin{pmatrix}
1&0&0\\
0&0&0\\
0&0&0
\end{pmatrix}
+\frac19
\begin{pmatrix}
1&1&1\\
1&1&1\\
1&1&1
\end{pmatrix}
\otimes
\begin{pmatrix}
1&1&1\\
1&1&1\\
1&1&1
\end{pmatrix}.
\label{eq:tiles_upb_projector}
\end{align}
The noisy reduced states of $\rho_p$ are therefore
\begin{align}
 \rho_{p,A}=\rho_{p,B}
 =(1-p)\rho_{\mathrm{UPB},A}+\frac p3\1_3
 =\frac1{24}
 \begin{pmatrix}
  7+p&1-p&-2(1-p)\\
  1-p&10-2p&1-p\\
  -2(1-p)&1-p&7+p
 \end{pmatrix}.
\end{align}
Their common purity becomes
\begin{align}
 \Tr[\rho_{p,A}^2]=\Tr[\rho_{p,B}^2]
 =(1-p)^2\frac{35}{96}+\frac{2p(1-p)}{3}+\frac{p^2}{3} 
 =\frac13+\frac{(1-p)^2}{32}.
 \label{eq:upb_noisy_purity}
\end{align}
Writing $q=1-p$, the cross covariance matrix $\mathrm{Cov}^{AB}_{\rho_p}$ in the standard Gell-Mann ordering $(\lambda_1,\ldots,\lambda_8)$ (see the last column of Tab.~\ref{tab:su3_three_generator_bases}) is
\begingroup
\setlength{\arraycolsep}{3pt}
\begin{align}
    \mathrm{Cov}^{AB}_{\rho_p}
    =
    \frac{q}{576}
    \begin{pmatrix}
        -64-4q & 0 & 6q & -64+8q & 0 & -64-4q & 0
        & -(96+2q)\sqrt3\\[4pt]
        0 & 0 & 0 & 0 & 0 & 0 & 0 & 0\\[4pt]
        144+6q & 0 & 72-9q & -12q & 0 & 6q & 0
        & (-24+3q)\sqrt3\\[4pt]
        -64+8q & 0 & -12q & -64-16q & 0 & -64+8q & 0
        & 4q\sqrt3\\[4pt]
        0 & 0 & 0 & 0 & 0 & 0 & 0 & 0\\[4pt]
        -64-4q & 0 & 144+6q & -64+8q & 0 & -64-4q & 0
        & (48-2q)\sqrt3\\[4pt]
        0 & 0 & 0 & 0 & 0 & 0 & 0 & 0\\[4pt]
        (48-2q)\sqrt3 & 0 & (-24+3q)\sqrt3 & 4q\sqrt3
        & 0 & -(96+2q)\sqrt3 & 0 & 24-3q
    \end{pmatrix}.
\end{align}
\endgroup
The singular values of the cross covariance matrix  $\mathrm{Cov}^{AB}_{\rho_p}$, in nonincreasing order, are
\begin{align}
s_1
    &=\frac{q}{144}
    \left[8+3q-2D\cos\!\left(\theta+\frac{2\pi}{3}\right)\right],\\
s_2&=s_3
    =\frac{q}{2\sqrt2},\\
s_4
    &=\frac{q}{144}\left[-8-3q+2D\cos\theta\right],\\
s_5
    &=\frac{q}{144}\left[8+3q-2D\cos\!\left(\theta-\frac{2\pi}{3}\right)
    \right],\\
s_6&=s_7=s_8=0,
\end{align}
where 
 \begin{align}
    D\coloneqq\sqrt{736+9q^2}, \quad
    \theta\coloneqq
    \frac13\arccos\!\left[
        \frac{-27q^3+5616q-1664}
        {(736+9q^2)^{3/2}}
    \right].
\end{align}
Summing the singular values gives 
\begin{align}
    \left\|\mathrm{Cov}^{AB}_{\rho_p}\right\|_1
    =
    \frac{q}{\sqrt2}
    +
    \frac{q}{144}
    \left[8+3q+4D\cos\theta\right].
\end{align}
Substituting the reduced state purities and the covariance trace norm of $\rho_p$ into Eq.~\eqref{eq:optimized_collective_uncertainty2}
gives
\begin{align}
\mathcal U_{\mathrm{coll}}^{\mathrm{opt}}(\rho_p)
=
\frac{16}{3}
-\frac{1-p}{\sqrt2}
-\frac{1-p}{36}
\left[
5-3p+\sqrt{736+9(1-p)^2}\cos\theta(p)
\right],
\end{align}
where
\begin{align}
\theta(p)
\coloneqq
\frac13\arccos\!\left[
\frac{-27(1-p)^3+5616(1-p)-1664}
{\left[736+9(1-p)^2\right]^{3/2}}
\right].
\end{align}
The optimized collective uncertainty violates the separability bound for $0\leqslant p<p_*$, detecting entanglement in the noisy Tiles UPB state (see Fig.~\ref{fig:noisy_UPB_collective_uncertainty}). 
For $p\geqslant p_*$, no violation occurs, although the state may remain entangled.

Lie algebraic uncertainty bounds thus provide a direct route to entanglement detection through collective observables. 
For the tensor product representations considered here, the total Lie algebraic uncertainty decomposes into local contributions evaluated on the reduced states, with a sharp lower bound given by the sum of the individual subsystem bounds (see Thm.~\ref{thm:LAU_Sep_Many_Body}).
These local bounds determine a collective uncertainty threshold for fully separable states.
Optimizing the pairing of local observables through orthogonal basis transformations preserves this threshold and can strengthen entanglement detection, as demonstrated by the noisy two-qutrit UPB state (see Eq.~\eqref{eq:noisy_UPB_state1}), whose entanglement is detected despite its positive partial transpose.

\end{document}